\documentclass[12pt]{article}
\usepackage{mathtools}
\usepackage{amsmath,amssymb,amsfonts}
\usepackage[top=0.7in, left=0.9in, bottom=0.9in, right=0.9in]{geometry}
\usepackage[dvips]{color}
\usepackage{xcolor}
\usepackage{mathabx}
\usepackage{mathrsfs}
\usepackage{graphicx}
\usepackage{caption}
\usepackage{natbib, har2nat}
\usepackage{rotating}
\usepackage{afterpage}
\usepackage{cancel}
\usepackage{booktabs}
\usepackage{setspace}
\usepackage{enumerate}
\usepackage{float}
\usepackage{footnote}
\usepackage{multirow}
\usepackage[flushleft]{threeparttable}
\usepackage{caption}
\usepackage{subcaption}
\usepackage{tabularx}
\usepackage{tikz}
\usepackage{bigints}
\usepackage{colortbl}
\usepackage{makecell}

\usepackage{url}

\usepackage{amsthm}

\usepackage{setspace}
\usepackage{lipsum}
\usepackage[hang,flushmargin]{footmisc} 

\usepackage[pdftex,bookmarks,colorlinks]{hyperref}

\usepackage{lscape}

\usetikzlibrary{bayesnet}

\usepackage[edges]{forest}
\usetikzlibrary{arrows.meta}

\theoremstyle{definition}
\newtheorem{exmp}{Example}[section]

\theoremstyle{plain}
\newtheorem{algo}{Algorithm}

\usepackage{pdfpages}

\usepackage{siunitx} 

\usepackage{mathpazo} 

\newcommand{\fra}[2]{\frac{\displaystyle #1}{\displaystyle #2}}
\newcommand{\bc}{\begin{center}}
	\newcommand{\ec}{\end{center}}
\newcommand{\be}{\begin{equation}}
\newcommand{\ee}{\end{equation}}
\newcommand{\bea}{\begin{eqnarray}}
\newcommand{\eea}{\end{eqnarray}}
\newcommand{\bean}{\begin{eqnarray*}}
	\newcommand{\eean}{\end{eqnarray*}}
\newcommand{\bt}{\begin{tabular}}
	\newcommand{\et}{\end{tabular}}

\newcolumntype{x}[1]{>{\centering\arraybackslash\hspace{0pt}}p{#1}}

\usetikzlibrary{decorations.pathreplacing}
\tikzset%
{
	myunderbrace/.style={decorate,decoration={brace,raise=4mm,amplitude=2pt,mirror}},
	mytext/.style={below,midway,yshift=-4mm}
}

\newtheorem{theorem}{Theorem}

\newtheorem{assumption}[theorem]{Assumption}

\newtheorem{remark}{Remark}[section]

\newcommand{\argmin}{\operatorname*{argmin}}

\newcounter{saveeqn}

\newcommand\numberthis{\addtocounter{equation}{1}\tag{\theequation}}

\numberwithin{equation}{section}

\hypersetup{
	colorlinks = true,
	citecolor = [rgb]{0.1137, 0.2745, 0.5020},
	linkcolor = [rgb]{0.6824, 0.0431, 0.1647}, 
	urlcolor =  [rgb]{0.1137, 0.2745, 0.5020}}

\makeatletter
\def\@fnsymbol#1{\ensuremath{\ifcase#1\or \dagger\or \ddagger\or
		\mathsection\or \mathparagraph\or \|\or **\or \dagger\dagger
		\or \ddagger\ddagger \else\@ctrerr\fi}}
\makeatother

\begin{document}

\setlength{\abovedisplayskip}{5pt}
\setlength{\belowdisplayskip}{5pt}

\title{ \vspace{-2cm} \bf Incorporating Prior Knowledge of Latent Group Structure in Panel Data Models\footnote{Please check  \href{https://www.dropbox.com/s/pjz8awsd3fiyim0/CBG_latest.pdf?dl=0}{here} for the latest version.}}

\author{
	Boyuan Zhang\thanks{Amazon.com, Seattle, WA Email: \href{mailto:zhang.boyuan@hotmail.com}{zhang.boyuan@hotmail.com}. I am extremely grateful to my advisors Francis X. Diebold and Frank Schorfheide, and my dissertation committee, Xu Cheng and Minchul Shin, for their invaluable guidance and support. I would also like to thank Karun Adusumilli, Siddhartha Chib, Wayne Yuan Gao, Philippe Goulet Coulombe, and Daniel Lewis for helpful comments and suggestions. I further benefited from many helpful discussions with Econometrics Lunch seminar participants at the University of Pennsylvania, as well as participants at the 2022 North American Summer Meeting of the Econometric Society, the 2022 IAAE Annual Conference, the 42nd International Symposium on Forecasting, the 16th International Symposium on Econometric Theory and Applications, the 2022 Asian Meeting of the Econometric Society, the 2022 Australasia Meeting of the Econometric Society, the NBER-NSF Seminar on Bayesian Inference in Econometrics and Statistics, Philly Fed Young Scholars Conference on Machine Learning in Economics and Finance. This paper and its contents are not
		related to Amazon and do not reflect the position of the company and its subsidiaries. All remaining errors are my own.} \\
}

\date{\vspace{-0.5cm}
	\textit{Amazon}\\[2ex]%
	\small
	First Version: June 14, 2022 \\
	This Version: \today \\
	\vspace{0.4cm}
	\large
}

\maketitle
\begin{abstract}
	The assumption of group heterogeneity has become popular in panel data models. We develop a constrained Bayesian grouped estimator that exploits researchers' prior beliefs on groups in a form of pairwise constraints, indicating whether a pair of units is likely to belong to a same group or different groups. We propose a prior to incorporate the pairwise constraints with varying degrees of confidence. The whole framework is built on the nonparametric Bayesian method, which implicitly specifies a distribution over the group partitions, and so the posterior analysis takes the uncertainty of the latent group structure into account. Monte Carlo experiments reveal that adding prior knowledge yields more accurate estimates of coefficient and scores predictive gains over alternative estimators. We apply our method to two empirical applications. In a first application to forecasting U.S. CPI inflation, we illustrate that prior knowledge of groups improves density forecasts when the data is not entirely informative. A second application revisits the relationship between a country’s income and its democratic transition; we identify heterogeneous income effects on democracy with five distinct groups over ninety countries.
\end{abstract}

\noindent JEL CLASSIFICATION: C11, C14, C23, E31


\thispagestyle{empty}
\setcounter{page}{0}
\newpage


\newpage


\section{Introduction} \label{sec:intro}


Numerous studies have examined and demonstrated the important role of panel data models in empirical research throughout the social and business sciences, as the availability of panel data has increased. Using fixed-effects, panel data permits researchers to model unobserved heterogeneity across individuals, firms, regions, and countries as well as possible structural changes over time. As individual heterogeneity is often empirically relevant, fixed-effects are an objective of interest in numerous empirical studies. For example, teachers' fixed-effects are viewed as a measure of teacher quality in the literature on teacher valued-added \citep{rockoff2004, chetty2014}; the heterogeneous coefficients are crucial for panel data forecasting \citep{liu2020, pesaran2022}. In practice, however, researchers may face a short panel where $N$ is large and $T$ is short and fixed. When applying the least squares estimator for the fixed-effects, a significant number of noisy estimates are produced. To alleviate this issue, a popular and parsimonious assumption that has recently been used is to introduce a group pattern into the individual coefficients, so that units within each group have identical coefficients (\citet[BM hereafter]{bonhomme2015}, \citet{su2016}, \citet{bonhomme2022}). 

To recover the group pattern, we essentially face a clustering problem, e.g., dividing $N$ units into several unknown groups. All existing methods for estimating group heterogeneity solve a clustering problem by assuming that units are exchangeable and treating all units equally \textit{a priori}. In a cross-country application of evaluating the impact of climate change on economic growth \citep{hsiang2016, henseler2019,kahn2021}, countries in different climatic zones are assumed to have equal probabilities of being grouped together. The assumption of exchangeability might not be reasonable since correlations are common between observations at proximal locations and researchers could have knowledge of the underlying group structure based on theories or empirical findings. For instance, Sweden and Finland, which share a border, an economic structure, and weather conditions, may have a higher chance of being in the same group than African countries. In such a scenario, it is preferable to use additional information to break the exchangeability between countries to facilitate grouping as opposed to clustering based solely on observations in the sample. The availability of this information drives us to formalize such prior knowledge, which we wish to leverage to improve model performance.


In this paper, we focus on the group heterogeneity in the linear panel data model and develop a nonparametric Bayesian framework to incorporate prior knowledge of groups, which is considered additional information that does not enter the likelihood function. The prior knowledge aids in clustering units into groups and sharpens the inference of group-specific parameters, particularly when units are not well-separated.

The whole framework is built on the nonparametric Bayesian method, where we do not impose a restriction on the number of groups, and model selection is not required. The baseline model is a linear panel data model with an unknown group pattern in fixed-effects, slope coefficients, and cross-sectional error variances. We estimate the model using the stick-breaking representation \cite{sethuraman1994} of the Dirichlet process (DP) \cite{ferguson1973, ferguson1974} prior, a standard prior in nonparametric Bayesian inference. In this framework, the number of groups is considered a random variable and is subject to posterior inference. The number of groups and group membership are estimated together with the heterogeneous coefficients. Moreover, since the DP prior implicitly defines a prior distribution on the group partitionings, the posterior analysis takes the uncertainty of the latent group structure into account. 

The derivation of the proposed prior starts with summarizing prior knowledge in the form of pairwise constraints, which describe a bilateral relationship between any two units. Inspired by the work of \citet{wagstaff2000}, we consider two types of constraints: \textit{positive-link} and \textit{negative-link} constraints, representing the preference of assigning two units to the same group or distinct groups. Instead of imposing these constraints dogmatically, each constraint is given a level of accuracy that shows how confident the researchers are in their choice. There is a hyperparameter that controls the overall strength of the prior knowledge: a small value partially recovers the exchangeability assumption on units, whereas a large value confines the prior distribution of group partitioning around group structure based on prior knowledge. We choose the optimal value for the hyperparameter by maximizing the marginal data density. Summarizing prior knowledge in the form of pairwise constraints is practical and flexible since it eliminates the need to predetermine the number of groups and focuses on the bilateral relationships within any \textit{subset} of units.


The aforementioned pairwise constraints are used to modify the standard DP prior. In particular, the pairwise constraints are combined with the prior distribution of the group partitioning, shrinking the distribution toward my prior knowledge. We refer to the estimator using the proposed prior as the Bayesian group fixed-effects (BGFE) estimator.

We derive a posterior sampling algorithm for the framework with the modified DP prior. Adopting conjugate priors on group-specific coefficients allows for drawing directly from posteriors using a computationally efficient Gibbs sampler. With the newly proposed prior, it can be shown that, compared to the framework that uses a standard DP prior, all that is needed to implement pairwise constraints is a simple modification to the posterior of the group indices.


The pairwise constraint-based framework is closely related and applicable to other models where group structure plays a role. Although we concentrate primarily on the panel data model, the DP prior with pairwise constraints applies to models without the time dimension, such as the standard clustering problem and the estimation of heterogeneous treatment effects. The framework is also applicable to estimating panel VARs \citep{holland1983}, which involves multiple dependent variables. The group structure is used to overcome overparameterization and overfitting issues by clustering the VAR coefficients into groups, and pairwise constraints add additional information to the highly parameterized model. Moreover, the proposed Gibbs sampler with pairwise constraints is connected to the \textit{KMeans}-type algorithm, motivating a frequentist's counterpart of our estimator with a fixed $K$. Essentially, the assignment step in the \textit{Pairwise Constrained-KMeans} algorithm \cite{basu2004}, a constrained version of the \textit{KMeans} algorithm \cite{macqueen1967}, is remarkably similar to the step of drawing a group membership indicator from its posterior. The same exact equivalence can be achieved by applying small-variance asymptotics to the posterior densities under certain conditions. To obtain the frequentist's analog of our pairwise constrained Bayesian estimators, one can utilize the same approach in BM with the \textit{Pairwise Constrained-KMeans} algorithm.

We compare the performance of the BGFE estimator to alternative estimators using simulated data. The Monte Carlo simulation demonstrates that the BGFE estimator generates more accurate estimates of the group-specific parameters and the number of groups than the BGFE estimator without including any constraints. The improved performance is mostly attributable to the precise group structure estimation. The BGFE estimator clearly dominates the estimators that omit the group structure by assuming homogeneity or full heterogeneity. We also evaluate the performance of one-step ahead point, set, and density forecasts. Unsurprisingly, the accurate estimates translate into the predictive power of the underlying model; the BGFE estimator outperforms the rest of the estimators.

We apply the proposed method to two empirical applications. An application to forecasting the inflation of the U.S. CPI sub-indices demonstrates that the suggested predictor yields more accurate density predictions. The better forecasting performance is mostly attributable to three key characteristics: the nonparametric Bayesian prior, prior belief on group structure, and grouped cross-sectional heteroskedasticity. In a second application, we revisit the relationship between a country's income and its democratic transition. This question was originally studied by \citet{acemoglu2008}, who demonstrate that the positive income effect on democracy disappears if country fixed effects are introduced into the model. The proposed framework recovers a group structure with a moderate number of groups. Each group has a clear and distinct path to democracy. In addition, we identify heterogeneous income effects on democracy and, contrary to the initial findings, show that a positive income effect persists in some groups of countries, though quantitatively small.

\noindent \textsc{Literature}. This paper relates to the econometric literature on clustering in panel data models. Early contributions include \citet{sun2005} and \citet{buchinsky2005}. \citet{hahn2010} provide economic and theoretical foundations for fixed effects with a finite support. Most recent works focus on linear\footnote{See \citet{wang2021, bonhomme2022}, among others, for procedures to identify latent group structures in nonlinear panel data models.} panel data models with discrete unobserved group heterogeneity. \citet{lin2012} and \citet{sarafidis2015} apply the \textit{KMeans} algorithm to identify the unobserved group structure of slope coefficients. \citet{bonhomme2015} also use the \textit{KMeans} algorithm to recover the group pattern, but they assume group structure in the additive fixed effects. \citet{bonhomme2022} modify this method and split the procedure into two steps. They first classify individuals into groups using \textit{KMeans} algorithm and then estimate the coefficients. \citet{ando2016} improved on BM's approach by allowing for group structure among the interactive fixed effects. The underlying factor structure in the interactive fixed effects is the key to forming groups. \citet {su2016} develop a new variant of Lasso to shrink individual slope coefficients to unknown group-specific coefficients. This method is then extended by \citet{su2018} and \citet{su2019}. \citet{freeman2022} consider two-way grouped fixed effects that allow for different group patterns in time and cross-sectional dimensions. \citet{okui2021} and \citet{lumsdaine2022} identify structure breaks in parameters along with grouped patterns. From the Bayesian perspective, \citet{kim2019}, \citet{zhang2020}, and \citet{liu2022} adopt the Dirichlet process prior to estimate grouped heterogeneous intercepts in linear panel data models in the semiparametric Bayesian framework. \citet{moon2023} incorporate a version of a spike and slab prior to recover one core group of units. Alternative methods, such as binary segmentation \citep{wang2018} and assumptions, such as multiple latent groups structure \citep{cheng2019, cytrynbaum2020} have also been explored to flourish group heterogeneity literature. 

Our work concerns prior knowledge. \citet{bonhomme2015}'s grouped fixed-effects (GFE) estimator is able to include prior knowledge, but it is plagued by practical issues to some extent. They add a collection of individual group probabilities as a penalty term in the objective function, which is a $N$ by $K$ matrix describing the probability of assigning each unit to all potential groups. This additional penalty term balances the respective weights attached to prior and data information in estimation. The main challenge is providing the set of individual group probabilities for each potential value of $K$ as the underlying \textit{KMeans} algorithm requires model selection. It is rather cumbersome to assess these probabilities for each possible $K$ and to adjust for changes in reallocating probabilities across $K$. 

None but \citet{aguilar2022} explore heterogeneous error variance, and they extend BM's GFE estimator to allow for group-specific error variances. They modify the objective function to avoid the singularity issue in pseudo-likelihood. Despite the fact that their work paves the way for identifying groups in the error variance, their framework is not yet ready to satisfactorily incorporate prior knowledge because they face the same issue as BM. Building on these works, we investigate the value of prior knowledge of group structure.

This paper also relates to the literature of constraint-based semi-supervised clustering in statistics and computer science. Pairwise constraints have been widely implemented in numerous models and have been shown to improve clustering performance. In the past two decades, various pairwise constrained \textit{KMeans} algorithms using prior information have been suggested \cite{wagstaff2001, basu2002, basu2004, bilenko2004, davidson2005, pelleg2007, yoder2017}. Prior information is also introduced in the model-based method. \citet{shental2003} develop a framework to incorporate prior information for the density estimation with Gaussian mixture models. The Dirichlet process mixture model with pairwise constraints has been discussed in \citet{vlachos2008}, \citet{vlachos2009}, \citet{orbanz2008}, \citet{vlachos2010}, \citet{ross2013}. \citet{lu2004}, \citet{lu2007} and \citet{lu2007_2} assume the knowledge on constraints is incomplete and penalize the constraints in accordance with their weights. \citet{law2004} extents \citet{shental2003} to allow for soft constraints in the mixture model by adding another layer of latent variables for the group label. \citet{nelson2007} propose a new framework that samples pairwise constraints given a set of probabilities related to the weights of constraints.


Our paper is closely related to \citet{paganin2021}, who address a similar problem using a novel Bayesian framework. Their proposed method shrinks the prior distribution of group partitioning toward a \textit{full} target group structure, which is an initial clustering of \textit{all} units provided by experts. This is demanding since not every application can have a full target group structure, as their birth defect epidemiology study did. Our framework circumvents this problem by using pairwise constraints, which are flexibly assigned to any two units. In addition, the induced shrinkage of their framework is produced by the distance function defined by Variation of Information \citep{meilua2007}. It can be demonstrated that a partition can readily become caught in local modes, preventing it from ever shrinking toward the prior partition. The use of pairwise relationships in this paper circumvents this issue as well. By fixing the group indices of other pairs, our framework makes sure that the partition with a specific pair that fits our prior belief has a higher prior probability than the partition with a pair that goes against our prior belief.


\noindent \textsc{Outline}. In section \ref{sec:model_prior}, we present the specification of the dynamic panel data model with group pattern in slope coefficients and error variances and provide details on nonparametric Bayesian priors without prior knowledge, which are then extended to accommodate soft pairwise constraints. Section \ref{sec:post} focuses on the posterior analysis, where the posterior sampling algorithm is provided. We also highlight the posterior estimate of group structure and discuss the connection to constrained \textit{KMeans} models. We briefly discuss the extensions of the baseline model in section \ref{subsec:ext}. In section \ref{sec:emp_result}, we present empirical analysis in which we forecast the inflation rate of the U.S. CPI sub-indices and estimate the country's income effect on its democracy. Finally, we conclude in section \ref{sec:conclusion}. Monte Carlo simulations, additional empirical results, and proofs are relegated to the appendix.




\section{Model and Prior Specification} \label{sec:model_prior}


We begin our analysis by setting up a linear panel data model with group heterogeneity in intercepts, slope coefficients, and cross-sectional innovation variance. We then elaborate a nonparametric Bayesian prior for the unknown parameters that takes prior beliefs in the group pattern into account. We briefly highlight several key concepts of a standard nonparametric Bayesian prior as our proposed prior inherits some of its properties.

\subsection{A Basic Linear Panel Data Model}

We consider a panel with observations for cross-sectional units $i=1, \ldots, N$ in periods $t=1, \ldots, T$. Given the panel data set $(y_{it}, x'_{it})$, a basic linear panel data model with grouped heterogeneous slope coefficients and grouped heteroskedasticity takes the following form:
\begin{align} \label{simple_model}
	y_{it} &= \alpha'_{g_{i}} x_{it} + \varepsilon_{i t}, \quad \varepsilon_{i t} \sim N \left(0, \sigma_{g_{i}}^{2}\right),
\end{align}
where $x_{i t}$ are a $p \times 1$ vector of covariates, which may contain intercept, lagged $y_{it}$, other informative covariates. $\alpha_{g_{i}}$ denote the group-specific slope coefficients (including intercepts). $\sigma^2_{g_{i}}$ are the group-specific variance.  $g_{i} \in \{1,..., K \}$ is the latent group index with an unknown number of groups $K$. $\varepsilon_{i t}$ are the idiosyncratic errors that are independent across $i$ and $t$ conditional on $g_i$. They feature zero mean and grouped heteroskedasticity $\sigma_{g_{i}}^{2}$, with cross-sectional homoskedasticity being a special case where $\sigma_{g_{i} }^{2}=\sigma^{2}$. This setting leads to a heterogeneous panel with group pattern modeled through both $\alpha_{g_{i}}$ and $\sigma_{g_{i}}^{2}$.



It is convenient to reformulate the model in (\ref{simple_model}) in matrix form by stacking all observations for unit $i$:
\begin{align}
	\boldsymbol{y_{i}} = \boldsymbol{x_i} \alpha_{g_{i}} + \boldsymbol{\varepsilon_{i}}, \quad \boldsymbol{\varepsilon_{i}} \sim  N\left(\mathbf{0},\sigma_{g_{i}}^{2} I_T \right),
\end{align}
where $\boldsymbol{y_{i}} = \left[y_{i 1}, y_{i 2}, \ldots, y_{i T}\right]^{\prime}$, $ \boldsymbol{x_i} = \left[x_{i 1}, x_{i 2}, \ldots, x_{i T}\right]^{\prime}$, $\boldsymbol{\varepsilon_{i}} = \left[\varepsilon_{i 1}, \varepsilon_{i 2}, \ldots, \varepsilon_{i T}\right]^{\prime}$, and $G = \left[ g_{1}, \ldots, g_{N}\right]$ is a vector of group indices.

Group structure is the key element in our approach. It can be either represented as a vector of group indices $G$ describing to which group each unit belongs or as a collection of disjoint blocks $\mathcal{B} = \left\{B_{1}, B_{2}, \ldots, B_{K}\right\}$ induced by $G$, where $B_{k}$ contains all the units in the $k$-th group and $K$ is the number of groups in the sample of size $N$. $|B_k|$ denotes the cardinality of the set $B_k$ with $\sum_{k=1}^K |B_k| = N$.

\begin{remark}
	Identification issues may arise with certain specifications. If the grouped fixed-effects in $\alpha_{g_{i}}$ are allowed to vary over time, for example, $\sigma_{g_{i}}^{2}$ cannot be identified when the group $g = g_i$ contains only one unit. \citet{aguilar2022} propose a solution, but this problem is beyond the scope of this work. They suggest using the square-root objective function rather than the pseudo-log-likelihood function as the objective function, which replaces the logarithm of $\sigma^2_i$ with the square root of $\sigma^2_i$, to avoid the singularity problem. 
\end{remark}



Following \citet{sun2005}, \citet{lin2012} and BM, we assume that the composition of groups does not change over time. In addition, for any group $k \neq k'$, we assume that they have different slope coefficients, e.g., $\alpha_{k} \neq \alpha_{k'}$, and no single unit can simultaneously belong to these two groups: $B_k \bigcap B_{k'} = \emptyset$. Note that these assumptions are used to simplified the prior construction and are not necessary to incorporate prior knowledge. As we show in Section \ref{subsec:ext}, both assumptions can be relaxed by using slightly different priors.

The primary objective of this paper is to estimate the group-specific slope coefficients $\alpha_{g_i}$, group-specific variance $\sigma_{g_{i}}^{2}$, group membership $G$ as well as the unknown number of groups $K$ using full sample and prior knowledge of the group structure. Given estimates of group-specific coefficients, we are able to offer the point, set, and density forecasts of $y_{i t+h}$ for each unit $i$. Throughout this paper, we will concentrate on the one-step ahead forecast where $h=1$. For multiple-step forecasting, the procedure can be extended by iterating $y_{i T+h}$ in accordance with (\ref{simple_model}) given the estimates of parameters or estimating the model in the style of direct forecasting. The method proposed in this paper is applicable beyond forecasting. In certain applications, the heterogeneous parameters themselves are the objects of interest. For example, the technique developed here can be adapted to infer group-specific heterogeneous treatment effects.

\subsection{Nonparametric Bayesian Prior with Knowledge on $G$} \label{subsec:prior_with_knowledge}

We propose a nonparametric Bayesian prior for the unknown parameters with prior beliefs on the group pattern. Figure \ref{fig:graphic_group_assignment_soft} provides a preview of the procedure for introducing prior knowledge into the model. We propose to use pairwise constraints to summarize researchers' prior knowledge, with each constraint accompanied by a hyperparameter $W$ indicating the researchers' levels of confidence in their choice. The $W$ is then incorporated directly in the prior distribution of the group partition $G$, which is induced from a standard nonparametric Bayesian prior, yielding a new prior. We will elaborate the details throughout this subsection and highlight the clustering properties of the underlying nonparametric Bayesian priors in Section \ref{subsec:prior}.



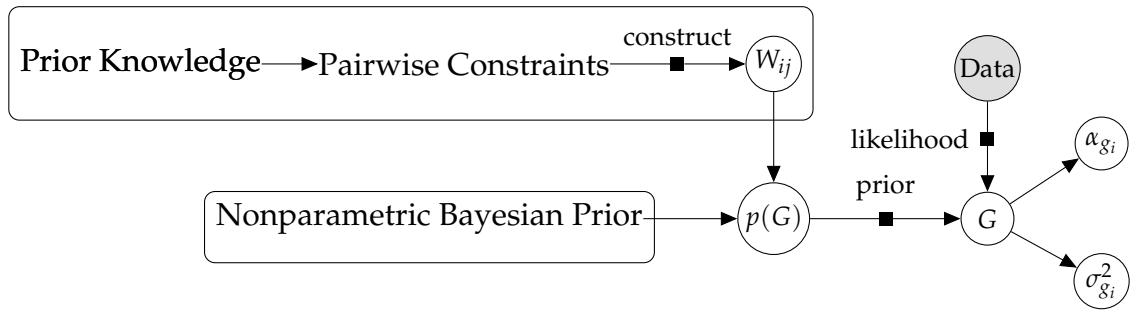
\begin{figure}[htp]
	\centering
	\caption{Graphical Representation of the Proposed Prior} 
	\label{fig:graphic_group_assignment_soft}
		\begin{tikzpicture} 
			\node[latent]    (alpha)  {$\alpha_{g_i}$} ;
			\node[latent, below=of alpha, yshift= -0.1cm]    (sigma)  {$\sigma^2_{g_i}$} ;
			
			\node[latent, left=of alpha, xshift=0.2cm, yshift=-1cm] (G)         {$G$} ;
			\node[latent, left=of G , xshift=-1cm]    (prior-G) {$p(G)$};
			
			\node[latent, above=of prior-G , xshift=0cm, yshift=0.2cm]  (W)  {$W_{ij}$};
			\node[const, left=of W, xshift=-0.8cm]       (crst)         {Pairwise Constraints};
			\node[const, left=of crst, xshift= 0.25cm]    (priorknow) {Prior Knowledge};
			
			\node[const, left=of prior-G , xshift=-0.25cm] (NB)  {Nonparametric Bayesian Prior};

			\node[obs, above=of G, yshift=0.2cm]        (data)  {Data};

			\node[const, left=of crst, xshift= 0.25cm]    (priorknow) {Prior Knowledge};
			
			\factor[right=of crst, xshift=+0.4cm]     {c-f}     {construct} {} {} ;
			

			\factor[left=of G, xshift= -0.5cm]     {prior-name}     {prior} {} {};
			
			\factor[below=of data]     {data-f}     {left:likelihood} {} {} ;
			
			\edge {prior-G, data} {G};
			
			\edge {W} {prior-G};
			
			\edge {G} {alpha};
			
			\edge {G} {sigma};
			
			\edge {NB} {prior-G};
			
			\edge {priorknow} {crst};
			
			\factoredge {crst}  {c-f}  {W}

			\plate {plate1} { %
				(crst)(c-f)(W)(c-f-caption)(priorknow) %
			}{}; %
			
			\plate {plate2} { %
				(NB)
			} {}; %
			
		\end{tikzpicture}
\end{figure}

\subsubsection{A New Prior with Soft Pairwise Constraints} \label{sec:soft_constraint}

The derivation of the proposed prior starts from summarizing prior knowledge in the form of pairwise constraints, which describe a bilateral relationship between any two units. Inspired by the literature on semi-supervised learning \cite{wagstaff2000},\footnote{We essentially follow the same idea of the pairwise constraints in \citet{wagstaff2000}. To better demonstrate the beliefs on constraints, we use different names:  positive-link and negative-link, rather than must-link and cannot-link. } we consider two types of pairwise constraints: (1) positive-link (PL) constraints, $\mathcal{P}$, and (2) negative-link (NL) constraints, $\mathcal{N}$. A positive-link constraint specifies that two units are more likely to be assigned to the same group, whereas a negative-link constraint indicates that the units are prone to be assigned to different groups. 

Instead of imposing these constraints dogmatically, the constraint between units $i$ and $j$ is given a hyperparameter $W_{ij}$ which describes how confident the researchers are in their choice for different types of constraints. $W_{ij}$ is continuously valued on the real line, as depicted in Figure \ref{fig:weights}. On the one hand, the sign of $W_{ij}$ specifies the constraint type, with a positive (negative) value indicating a PL (NL) constraint between $i$ and $j$. On the other hand, the absolute value of $W_{ij}$ reflects the strength of the prior belief. We become increasingly confident in our prior belief on units $i$ and $j$ as $|W_{ij}| \to \infty$. If $|W_{ij}| = \infty$, we essentially impose the constraint, which is known as a \textit{hard} PL/NL constraint. Otherwise, it's a \textit{soft} PL/NL constraint with a nonzero and finite $W_{ij}$. $W_{ij} = 0$ if there is no prior belief in units $i$ and $j$. 

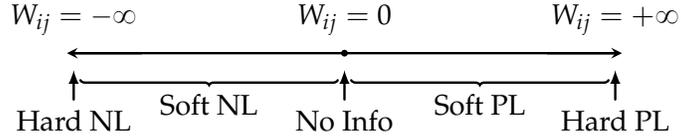
\begin{figure}[ht]
	\centering
	\caption{Relationship Between $W_{ij}$ and Pairwise Constraints}
	\label{fig:weights}
	\scalebox{0.9}{
		\begin{tikzpicture}[line cap=round,line width=0.35mm]
			\draw (0,0) -- (8,0);
			\draw[myunderbrace]   (0.1,0) -- (3.9,0) node[mytext] {\strut Soft NL};
			\draw[myunderbrace]   (4.1,0) -- (7.9,0) node[mytext] {\strut Soft PL};
			
			\foreach\i in {4}
			\fill (\i,0) circle (0.5mm);
			\draw [stealth-](-0.1,0) -- (1,0);
			\draw [-stealth](7,0) -- (8.1,0);
			
	
			\node (psi0) at (0,0.5)    {$W_{i j} = -\infty$};
			\node (hardCL) at (0,-1)    {Hard NL};
			\node (psi05) at (4,0.5)    {$W_{i j} = 0$};
			\node (random) at (4,-1)    {No Info};
			\node (psi1) at (8,0.5)    {$W_{i j} = +\infty$};
			\node (hardML) at (8,-1)    {Hard PL};
			\draw[-latex] (hardCL) --++ (0,0.75);
			\draw[-latex] (random) --++ (0,0.75);
			\draw[-latex] (hardML) --++ (0,0.75);
		\end{tikzpicture}
	}
\end{figure}

\vspace{0.5cm}

We assume the weight $W_{ij}$ is a logit function of two user-defined hyperparameters\footnote{This parametric form is related to the penalized probabilistic clustering proposed by \citet{lu2004, lu2007_2}. See detailed discussion in Appendix \ref{appendix:LuAndLeen}.}, accuracy $\psi_{i j}$ and type $T_{ij}$:
\begin{align} \label{eq:weight_construction}
	W_{i j} = T_{ij} \ln \left( \frac{\psi_{ij}}{1-\psi_{ij}}\right).
\end{align}
Accuracy, $\psi_{i j}  \in [0.5, 1)$, describes the user-specified probability of assigning a constraint for unit $i$ and $j$ being correct given our prior preference. Specifically, $\psi_{i j} = 1$ implies the constraint between $i$ and $j$ must be imposed since we confident that it is accurate, while specifying $\psi_{i j} = 0.5$ is equivalent to a random guess or no information is provided. $\psi_{i j} $ is bounded below by 0.5, following the assumption that leaving the pair unrestricted is more rational than setting a less likely constraint.  The type of constraints is denoted by $T_{i j}$. $T_{i j} = 1$ if unit $i$ and $j$ are specified to be positive-linked, and $T_{i j} = -1$ for a NL constraints. If the pair $(i,j)$ doesn't involve any constraint, we assume $T_{i j} = 0$. 



To incorporate these constraints into the prior, we propose modifying the exchangeable partition probability function (EPPF) or the prior distribution of group indices, $p(G)$, of the baseline Dirichlet process, which we will highlight in the Section \ref{subsec:prior}. The resulting group partition will receive a strictly higher (lower) probability if it is (in)consistent with pairwise constraints. As a result, the induced prior on the group indices $G$ directly depend on the characteristics of user-specific pairwise constraints and is able to increase or decrease the likelihood of a certain $G$.

In the presence of soft constraints, we modify the EPPF by multiplying a function of characteristics of constraints,
\begin{align} \label{eq:prior_G_soft}
	p(G | \psi, T) \; \propto \; p (G) \pi(\psi, T| G) = p (G) \prod_{i, j}  \left( \fra{\psi_{i j}}{1 - \psi_{i j}} \right)^{c T_{i j} \delta_{ij}(G)},
\end{align}
where $\psi_{i j}/ (1 - \psi_{i j} )$ is the prior odds for the constraint between unit $i$ and $j$, $\delta_{i j}(G)$ is a transformed Kronecker delta function such that
\begin{align*} \label{eq:delta}
	\delta_{ij}(G) &=
	\begin{cases}
		1 &  \text{ if } g_{i} = g_j \\
		-1 & \text{ if } g_{i} \ne g_j
	\end{cases}, \numberthis
\end{align*}
and $c$ is a positive number that controls the overall strength of prior belief. For $c \rightarrow 0$, $p(G | \psi, T)$ corresponds to the baseline EPPF $p(G)$, while for $c \rightarrow \infty$, $p\left(G = G^* | \psi, T \right) \rightarrow 1$, where $G^*$ satisfies all pairwise constraints.


\begin{remark}
	Due to the presence of pairwise constraints, the partition probability function presented in (\ref{eq:prior_G_soft}) no longer satisfies the exchangeable assumption as we now distinguish units within each group.
\end{remark}

\begin{remark}
	The current framework enables us to \textit{impose} some constraints. It is an extreme case of soft constraint and thus handy to implement, requiring only setting $\psi_{i j} \to 1$ for the pair $(i,j)$.  Intuitively, any group partition \textit{violating} the pairwise constraint between $i$ and $j$  (i.e., $T_{i j} \delta_{ij}(G) = -1$) will have zero probability, since for such a partition,
	\begin{align*}
		\left( \fra{\psi_{i j}}{1 - \psi_{i j}} \right)^{c T_{i j} \delta_{ij} (G)} \to \left(\frac{1}{\infty}\right) ^ {c} = 0 \text{, this imples } p(G | \psi, T) = 0,
	\end{align*}
	and hence the constraint on $(i,j)$ is imposed and referred to as a \textit{hard constraint} as opposed to soft constraint. By assigning proper $\psi_{i j}$ for the pairs $(i,j)$, we can flexibly combine soft and hard constraints inside a single specification.
\end{remark}

\begin{remark}
	Soft pairwise constraints solve the transitivity issue that might be a problem for hard pairwise constraints. For instance, if we have $(1,2) \in \mathcal{P}$ and $(2,3) \in \mathcal{P}$, we can still have $(1,3) \in \mathcal{N}$ in the framework of soft pairwise constraints since it preserves the possibility of violating any of these constraints. This is not the case in hard pairwise constraints, as $(1,2) \in \mathcal{P}$ and $(2,3) \in \mathcal{P}$ implies $(1,3) \in \mathcal{P}$ by transitivity.
\end{remark}

With the definition of $W_{ij}$, we rewrite the partition probability function defined in (\ref{eq:prior_G_soft}) in terms of $W_{ij}$ to ease notation,
\begin{align} \label{eq:prior_G_soft2}
	p(G | \psi, T) = p(G | W) = \mathcal{C}(W,G,c)^{-1} p (G) \exp \left[ c \sum_{i, j} W_{i j} \delta_{i j} (G) \right],
\end{align}
where
\begin{align}
	\mathcal{C}(W,G,c)  = \sum_{G'} p (G') \exp \left[ c \sum_{i, j} W_{i j} \delta_{i j} (G') \right],
\end{align}
is a normalization constant and we will use the prior $p(G | W)$ hereinafter. In practice, we will first specify $(T_{ij}, \psi_{ij}) = (\text{type, accuracy})$ for the constraint between unit $i$ and $j$ and then construct the corresponding weight $W_{i j}$ via the equation (\ref{eq:weight_construction}).


\begin{remark}
	In the particular case where we don't have any constraint information, $\exp  \left( c \sum_{i, j}  W_{i j} \delta_{i j}\right)$ reduces to 1 as $W_{i j} = 0$ for all $i$ and $j$, and recovers the original DP prior. Hence, our method can cater to all levels of supervision, ranging from hard constraints to a complete lack of constraints.
\end{remark}

\subsubsection{The Effect of Constraints and Scaling Constant on Group Partitioning} \label{subsec:scaling_factor}

 
The function $\pi(\psi, T | G)$ in Equation (\ref{eq:prior_G_soft}) is crucial in shifting the prior probability of $G$. By design, $T_{i j}\delta_{ij} = 1$ when the constraint between $i$ and $j$ is met in a group partitioning defined by $G$. The prior probability for $G$ is therefore increased since $\left[\psi_{i j}/ (1 - \psi_{i j} )\right]^c > 1$. Similarly, if a group partitioning $G$ violates the constraint between $i$ and $j$, then $T_{i j}\delta_{ij} = -1$ and the prior probability for $G$ drops due to $\left[\psi_{i j}/ (1 - \psi_{i j} )\right]^{-c} < 1$. Therefore, with $\pi(\psi, T | G)$, the resulting group partition is shrunk toward our prior knowledge without imposing any constraint.

To fix ideas, consider a simplified scenario with $N = 2$ units where there are at most two groups. For illustrative purposes, we set the concentration parameter to $a = 1$ so that $\Pr (g_1 = g_2) = \Pr (g_1 \ne g_2) = 0.5$\footnote{\citet{antoniak1974} provides analytical formulas for probabilities of more general events with larger $N$. In this example, $\Pr (g_1 = g_2)$ = $\frac{1}{a+1}$.} if no constraint exists. When $N = 2$, listing all partitions $G$ is possible, $G \in \{(1,1), (1,2), (2,1), (2,2)\}$, and we can calculate the probabilities for each $G$ using (\ref{eq:prior_G_soft2}). As a result, we are able to derive the probability of units 1 and 2 belonging to the same or different groups, i.e., analytical formulae for $\Pr(g_1 = g_2)$ and $\Pr(g_1 \ne g_2)$, which neatly demonstrate the effect of $\psi$ and $c$ on group partitioning.

It is straightforward to show $\Pr(g_1 = g_2)$ as a function of  $\psi_{12}, T_{12}$, and $c$:
\begin{align}\label{eq:pr_vs_psi}
	\Pr(g_1 = g_2) = \frac{1}{1 + \exp(-4cW_{12})}  = \frac{1}{1 + \left(\frac{\psi_{12}}{1-\psi_{12}} \right)^{-4cT_{12}}}.
\end{align}

Figure \ref{fig:soft_weight_diff_c} traces out the equation (\ref{eq:pr_vs_psi}) for a range of $c$ values. The left panel (a) displays the curve for a PL constraint. Firstly, observe that when $c = 0$, $\Pr(g_1 = g_2)$ remains unchanged at $0.5$ regardless of the value of $\psi$. This is the situation in which $c$ eliminates the constraint's effect on the prior. Next, given a particular $c$, $\Pr(g_1 = g_2)$ increases in $\psi$, which means that a stronger soft PL constraint between units 1 and 2 leads to higher chance of assigning both units to the same group. When $\psi$ is fixed, increasing $c$ easily results in a higher $\Pr(g_1 = g_2)$, indicating that a larger $c$ value magnifies the effect of the PL constraint. In contrast, panel (b) depicts the curve with a NL constraint. $\psi_{12}$ and $c$ clearly have the opposite effect on $Pr (g_1 = g_2)$: $\Pr(g_1 = g_2)$ drops significantly as $\psi_{12}$ or $c$ increases. Notably, even with a large $\psi_{12}$, the soft constraint framework maintains the possibility of breaching the constraint, which is another important feature that preserves the chance of correctly assigning group indices even if the constraint is erroneous. 


\begin{figure}[htp]
	\caption{$\Pr(g_1 = g_2)$ as a Function of $\psi_{12}$ and $c$}
	\label{fig:soft_weight_diff_c}
	\centering
	\begin{subfigure}[b]{0.45\textwidth}
		\centering
		\includegraphics[width=1\textwidth]{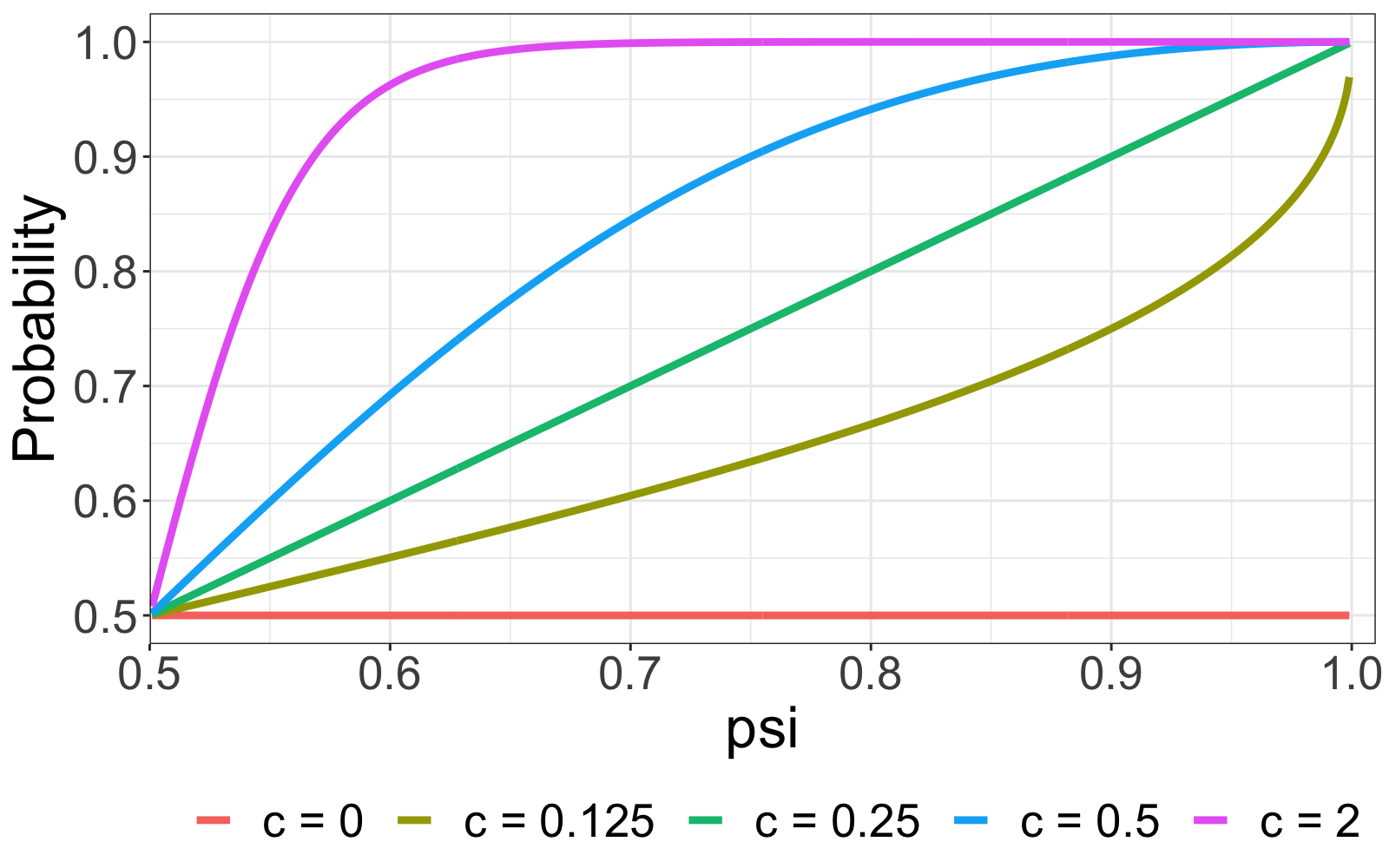}
		\caption{Positive-Link ($T_{12} = 1$)}
	\end{subfigure}
	\begin{subfigure}[b]{0.45\textwidth}
		\centering
		\includegraphics[width=\textwidth]{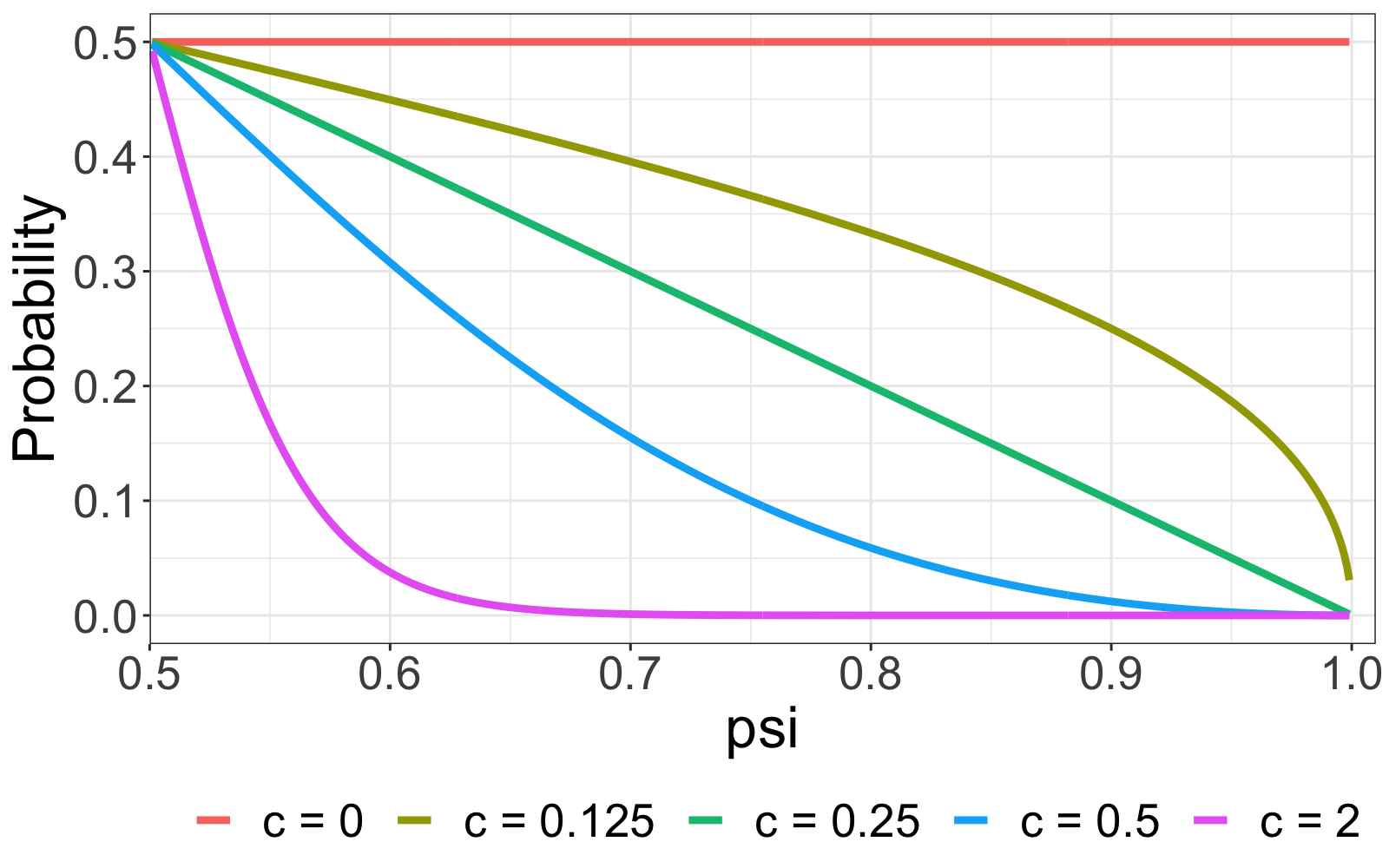}
		\caption{Negative-Link ($T_{12} = -1$)}
	\end{subfigure}
\end{figure}


In the general case where numerous PL and NL constraints are enforced, $c$ concurrently affects all constraints. In other words, the value of $c$ determines the overall ``strength" of the prior belief of $G$. If the prior belief is coherent with the real group partition, it would be preferable to have a large $c$ to intensify the effect on constraints, allowing prior information to take precedence over data information, and vice versa.

\begin{remark}
	We propose to find the optimal $c$ that maximizes marginal data density using grid search, see details in Appendix \ref{app:det_c}. Alternatively, the scaling constant $c$ can be pair-specific and data-driven. \citet{basu2004mrf}, for instance, assume that $c_{ij}$ is a function of observables, i.e., $c_{ij} = \psi (x_i, x_j; T_{ij})$. $\psi (x_i, x_j; T_{ij})$ is monotonically increasing (decreasing) in the distance between units $i$ and $j$ if they are involved in positive-link (negative-link) constraints. This reflects the belief that if two more distant units are assumed to be positive-linked (negative-linked), their constraint should be given more (less) weight.
\end{remark}

\subsubsection{Specification of Soft Pairwise Constraints}  \label{subsec:specify_constraints}

In reality, it is practical to establish soft pairwise constraints based on existing information on group, even if it is not the genuine group partitioning. In the empirical analysis, for instance, we use the official expenditure categories of CPI sub-indices to construct soft pairwise constraints. When information on group partitioning is insufficient, especially when the number of units is large, these official expenditure categories may serve as a trustworthy starting point. Before formalizing the idea, we first introduce the prior similarity matrix $\underline{\pi}^S$ which is a $N \times N$ symmetric matrix describing the prior probability of any two units belonging to the same group, i.e., $\underline{\pi}^S_{ij} = \Pr(g_i = g_j)$ conditional on all hyperparameters in the prior. 


The general idea is to derive soft pairwise constraints using the existing information on a preliminary group partitioning $\underline{G}$, which is allowed to involve only a subset of units. We start with the type of constraints $T_{ij}$ between any two units. Given the preliminary group structure, such as expenditure categories, we specify PL constraints for all pairs of units within the same group and NL constraints for all pairs of units from different groups. This means that we believe the preliminary group structure is correct \textit{a priori}. Despite the fact that more elaborate and subtle constraints might be implemented, this rough specification is usually a great starting point.

The accuracy $\psi_{ij}$ for constraints is then specified. When our prior knowledge is limited or the number of units is large, we cannot specify $\psi_{ij}$ for all pairs with solid knowledge of them. Instead, one desirable yet simple choice is to assume $\psi_{ij}$ again based on preliminary group partitioning $\underline{G}$. More specifically, all units in the same group are positive-linked with identical $\psi^{PL}_{ij} $, i.e., for units $i$ and $j$ from the group $\underline{g}_i = \underline{g}_j = \underline{g}$, we have $\psi^{PL}_{ij} = c_{\underline{g}}$. Units from different groups are assumed to be negative-linked with identical $\psi^{NL}_{ij}$, i.e., for units $i$ and $j$ from distinct groups, we assume $\psi^{NL}_{ij} = c_{\underline{g}_i\underline{g}_j}$ and  $c_{\underline{g}_i\underline{g}_j} = c_{\underline{g}_j\underline{g}_i}$. Following this strategy, $\psi_{ij}$ depends solely on $\underline{G}$ and hence two units from the same group would have identical soft pairwise constraints with other units. Notice that the number of possible distinct $\psi_{ij}$ reduces from $N(N-1)/2$ to $\widebar{K}(\widebar{K}+1)/2$, where $\widebar{K}$ is the number of groups in $\underline{G}$ and $\widebar{K} \ll N$. 


This framework permits no prior belief in certain units. If at least one unit in a pair $(i,j)$ is not included in $\underline{G}$, we assume that this pair of units is free of constraints and we set $T_{ij}$ to 0 or $\psi_{ij}$ to 0.5 in the prior. Note that the absence of a constraint does not ensure that the units $i$ and $j$ are completely unrelated. Instead, if both $i$ and $j$ are involved in constraints with a third unit $k$, or are connected through a series of $l$ constraints, $i \leftrightarrow k_1 \leftrightarrow  \cdots \leftrightarrow k_l \leftrightarrow j$, then the prior probability of $i$ and $j$ belonging to the same group differs from the prior probability without any constraints. If we wish to prevent two units from linking \textit{a priori}, they must not be subject to any constraints with the remaining units.

The aforementioned specification strategy induces a block prior similarity matrix, i.e., for an unit $i$, $\underline{\pi}^S_{ij}  = \underline{\pi}^S_{ik}$ if $\underline{g}_j = \underline{g}_k$. Intuitively, if two units have identical soft pairwise constraints and hence posit an identical relationship with all other units, they are equivalent and exchangeable. As a result, these units should have an equal prior probability of sharing the same group index with any other units. More formally,
\begin{theorem}[Stochastic Equivalence] \label{thm:equal_prior_prob}
	Given two units $j,k$ from the same prior group, if $\psi_{j m} = \psi_{k m}$ for all $m = 1,2,..,N$, then $\Pr (g_i = g_j) = \Pr (g_i = g_k)$ for all unit $i$ in the prior, given weights $W$.
\end{theorem}

Theorem \ref{thm:equal_prior_prob} echos the concept of \textit{stochastic equivalence} \citep{nowicki2001} in stochastic block model\footnote{For a more comprehensive review of the stochastic block model, see \citet{lee2019}.} (SBM) \citep{holland1983}. In less technical terms, for nodes $p$ and $q$ in the same group, $p$ has the same (and independent) probability of connecting with node $r$, as $q$ does. Interestingly, this relationship is not coincidental. The prior draw of group membership with the aforementioned specification of $T_{ij}$ and $\psi_{ij}$ can be viewed as a simulation of a simple SBM. In a simple SBM, there are two essential components: a vector of group memberships and a block matrix, each element of which represents the edge probability of two nodes, given their group memberships. In our case, the preliminary group structure serves as the group membership in SBM. The DP prior and the weight (or $T_{ij}$ and $\psi_{ij}$) of each constraint induce a prior similarity probability comparable to the block matrix.

Let's consider an example of 90 units. The preliminary group structure divides units into 3 groups, with groups 1, 2 and 3 containing 25, 30 and 35 units, respectively. Figure \ref{fig:psm_demo} shows the prior similarity matrix, which is based on the aforementioned specification strategy, so it becomes a block matrix with equal entries in each of nine blocks. Units within the same group are stochastically equivalent, as their prior probabilities of being grouped not only with each other but also with units from other groups are the same. As a result, the similarity probability of each pair depends solely on their preliminary membership (and $\psi$).

\begin{figure}[h]
	\caption{Prior Similarity Matrix under Stochastic Equivalence}
	\label{fig:psm_demo}
	\centering
	\includegraphics[scale= 0.5]{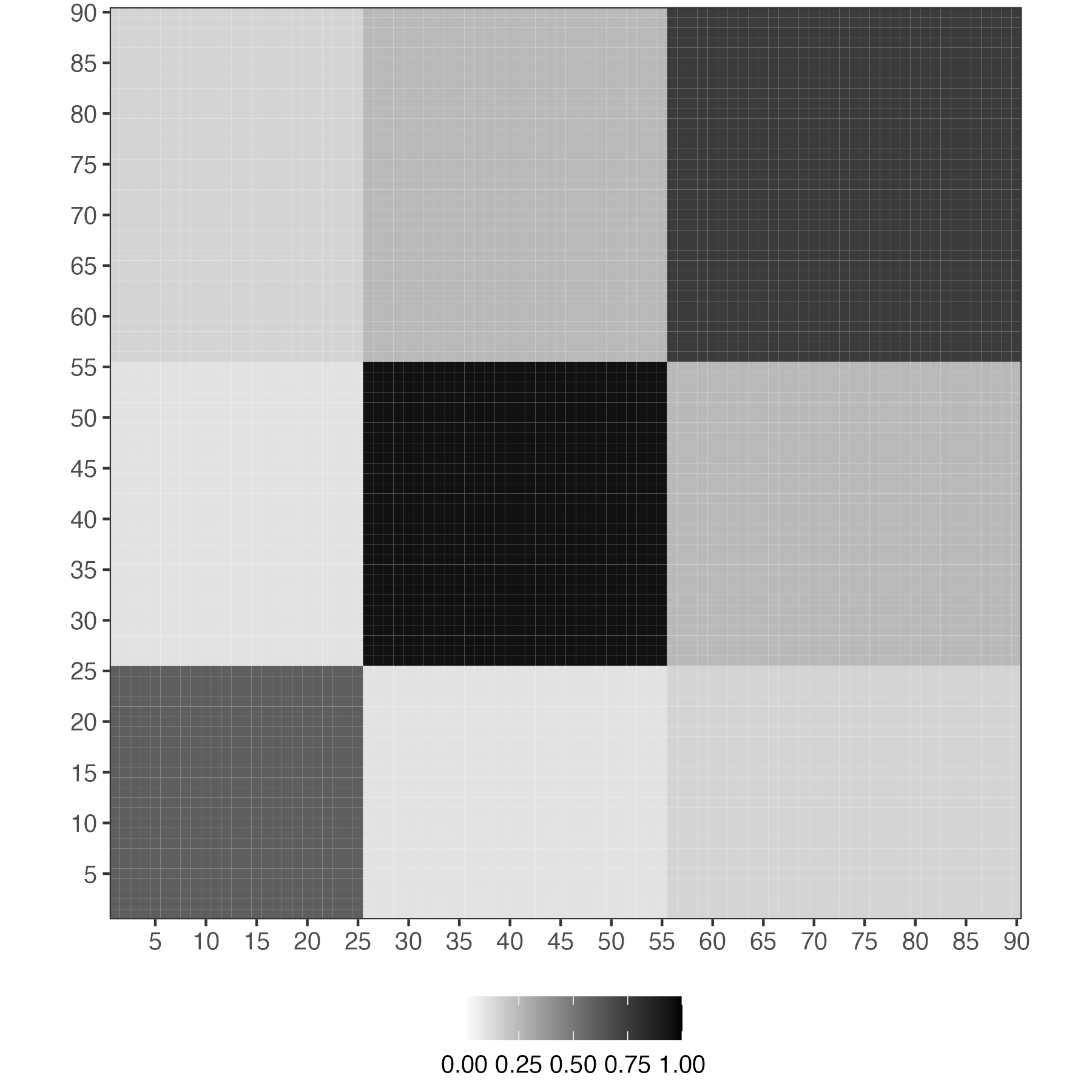}
\end{figure}


\subsubsection{Comparison to Existing Methods} \label{subsec:drawback_F&B}

\citet{bonhomme2015} incorporates prior knowledge of group membership by adding a penalty term to the objective function. They assume that prior information is in the form of probabilities which describe the prior probability of unit $i$ belonging to group $k$ with at most $K$ groups as $\omega^{(K)}_{ik}$. Consequently, the estimated group index is given by:
\begin{align}
	\widehat{g}_{i} (\beta, \alpha)=\underset{k \in\{1, \ldots, K\}}{\operatorname{argmin}} \sum_{t=1}^{T}\left(y_{i t} - \beta^{(K)'} x_{i t}-\alpha^{(K)}_{k t}\right)^{2}-C \ln \omega^{(K)}_{i k},
\end{align}
where $C>0$ is a hyperparameter need to be tuned further and $K$ is the predetermined number of groups. 


The penalty determines the weights assigned to prior and data information in estimation. Due to the fact that $K$ is frequently unknown in advance, this method requires model selection to determine the ideal number of groups. Assume we have $n_K$ alternative options for $K$. We have a $N \times K$ matrix $\omega^{(K)} = \left \{\omega^{(K)}_{ik} \right\}$ for prior information for a given $K$. As a result, in order to pick a model, we must therefore provide $n_K$ sets of prior probability matrix $\omega^{(K)}$, which is cumbersome and inconvenient. For instance, if $K$ has values ranging from $3,5,10$ and $N = 200$, there are 3,600 entries for $\omega$, none of which can be missing or undefined. In addition, the information criteria may be unreliable in the finite-sample results, necessitating further care when selecting an appropriate variant for empirical application.


Summarizing prior knowledge through pairwise constraints is often more practical than the penalty function approach in BM and solves the aforementioned practical issues. Pairwise relationships can be derived intuitively from researchers' input without requiring in-depth knowledge of the underlying groups; researchers do not need to fix the number of groups $K$ or group membership \textit{a priori}.  One only needs to focus on a pair of units each time and specify the preference of assigning them to the same or different groups. Moreover, since the pairwise constraints are incorporated into the DP prior, which implicitly defines a prior distribution on the group partitions, model selection is not required and the posterior analysis also takes the uncertainty of the latent group structure into account.

\citet{paganin2021} offer a statistical framework for including concrete prior knowledge on the partition. Their proposed method aims to shrink the prior distribution towards a \textit{complete} prior clustering structure, which is an initial clustering involving \textit{all} units provided by experts. Specifically, they suggest a prior on group partition that is proportional to a baseline EPPF of the DP prior multiplied by a penalization term,
\begin{align} \label{eq:paganin}
	p\left(G | G_{0}, m \right) \; \propto \; p(G) e^{- m d \left(G, G_{0}\right)}
\end{align}
with $m>0$ is a penalization parameter, $d\left(G, G_{0}\right)$ is a suitable distance measuring how far $G$ is from $G_{0}$, and $p(G)$ indicates the same EPPF as in (\ref{eq:prior_G_soft}). Because of the penalization term, the resulting group indices $G$ shrink toward the initial target $G_0$.

This framework is parsimonious and easy to implement, but it comes with a cost. The method is incapable of coping with an initial clustering in a subset of the units under study or multiple plausible prior partitions; otherwise, the distance measure is not well-defined. In addition, the authors suggest utilizing Variation of Information \citep{meilua2007} as the distance measure. It can be shown that the resulting partition can easily become trapped in local modes, leading the partition to never shrink toward $G_{0}$. They also argue that other available distance methods have flaws. As a result, the penalty term does not function as anticipated.

Our proposed framework with pairwise constraints is more flexible than adopting a target partition. Actually, the target partition in \citet{paganin2021} can be viewed as a special case of the pairwise constants, in which every unit must be involved in at least one PL constraint. Our framework could manage partitions involving arbitrary subsets of the units by tactically specifying the bilateral relationships. Most importantly, when the group indices of other pairs are fixed, our framework ensures that the partition containing a specific pair that is consistent with our prior belief receives a strictly greater prior probability than the partition that is inconsistent with our prior belief. This guarantees that the generated $G$ shrinks in the direction of our prior belief.

\subsection{Nonparametric Bayesian Prior} \label{subsec:prior}



The baseline model contains the parameters listed below: $(\alpha, \sigma^2, \pi, \xi, a, \phi)$. We rely mostly on nonparametric Bayesian models.\footnote{For a more comprehensive review of the nonparametric Bayesian literature, see \citet{ghosal2017} and \citet{mueller2018}.} Bayesian nonparametric models have emerged as rigorous and principled paradigms to bypass the model selection problem in parametric models by introducing a nonparametric prior distribution on the unknown parameters. The prior assumes that a collection of $\alpha$ and $\sigma^2$ is drawn from the Dirichlet process prior.\footnote{There have been some empirical works that use Dirichlet process model with panel data. Dirichlet process mixture prior is specified for either the distribution of the innovations \citep{hirano2002} or intercepts \citep{fisher2022}.} $\pi$ is a vector of mixture probabilities in Dirichlet process that is produced by the stick-breaking approach with stick length $\xi$. $a$ is the concentration parameter in the Dirichlet process, whereas $\phi$ is a collection of hyperparameters in the base measure $B_{0}$. We consider prior distributions in the partially separable form,\footnote{The joint prior includes $\xi$ but not $\pi$. Because the stick-breaking formulation of $\xi$ is a deterministic transformation of $\xi$, knowing $\xi$ is identical to knowing $\pi$.}
\begin{align*}
	p(\alpha, \sigma^2 | a, \phi, \xi) p(\xi | a) p(a).
\end{align*}

We tentatively focus on the random coefficients model where, conditional on $G$, $\alpha_{g_{i}}$ and $\sigma^2_{g_i}$ are independent to the conditional set that includes initial value of each unit $y_{i0}$, the initial values of predetermined variables, and the whole history of exogenous variables. The assumption guarantees that $\alpha_{g_{i}}$ and $\sigma_{g_i}$ can be sampled separately and simplifies the inference of the underlying distribution of $\alpha_{g_{i}}$ and $\sigma_{g_i}$ to an unconditional density estimation problem, therefore lowering computational complexity. The joint distribution of heterogeneous parameters as a function of the conditioning variables can then be modeled to extend the model to the correlated random coefficient model, which is briefed in Section \ref{subsec:corr_rand_model}. A full explanation and derivation for the correlated random coefficient model are provided in the online appendix. 



\subsubsection{Prior on Group-Specific Parameters}


In the nonparametric Bayesian literature, the Dirichlet Process (DP) prior \cite{ferguson1973, ferguson1974, sethuraman1994} is a canonical choice, notable for its capacity to construct group structure and accommodate an infinite number of possible group components.\footnote{See Appendix \ref{app:DP_and_related} for a brief overview of the DP and Appendix \ref{app:CRP} for its clustering properties.} The DP mixture is also known as a ``infinite" mixture model due to the fact that the data indicate a finite number of components, but fresh data can uncover previously undetected components \cite{neal2000}. When the model is estimated, it chooses automatically an appropriate subset of groups to characterize any finite data set. Therefore, there is no need to determine the ``proper" number of groups.

The DP prior can be written as an infinite mixture of point mass with the probability mass function:
\begin{align} \label{prior:a}
	\left( \alpha_i, \sigma_i^2\right) \sim \sum_{k=1}^{\infty} \pi_k \delta_{\left(\alpha_{k}, \sigma_{k}^2\right)} \text{ with } \left(\alpha_{k}, \sigma_{k}^2\right) \sim B_{0}(\phi),
\end{align}
where $\delta_{x}$ denotes the Dirac-delta function concentrated at $x$ and $B_{0}$ is the base distribution. We adopt an Independent Normal Inverse-Gamma (INIG) distribution for the base distribution $B_0$:		
\begin{align}
	B_0(\phi) :=  INIG \left( \mu_{\alpha}, \Sigma_{\alpha}, \frac{\nu_{\sigma}}{2}, \frac{\delta_{\sigma}}{2} \right),
\end{align} 
with a set of hyperparameters $\phi =  \left( \mu_{\alpha}, \Sigma_{\alpha}, \frac{\nu_{\sigma}}{2}, \frac{\delta_{\sigma}}{2} \right)$. 

The group probabilities $\pi_k$ are constructed by an infinite-dimensional stick-breaking process \cite{sethuraman1994} governed by the concentration parameter $a$,
\begin{align} \label{main_text_pi}
	\pi_k &\equiv \xi_k \prod_{j<k} (1-\xi_j) \text{ for } k > 1,  \text{ and } \pi_1 = \xi_1,
\end{align}
where stick lengths $\xi_k$ are independent random variables drawn from the beta distribution,\footnote{Recall that a $Bata(m, n)$ distribution is supported on the unit interval and has mean $m/(m+n)$.} $Beta(1, a)$. The group probability will be random but still satisfy $\sum_{k=1}^{\infty} \pi_{k}=1$ almost surely.



Equation (\ref{main_text_pi}) is essential to understanding how the DP prior controls the number of groups. The building of group probabilities is compared to the breaking of a stick of unit length sequentially, in which the length of each break is assigned to the current value of $\pi_k$. As the number of groups increases, the probability created by the stochastic process decreases because the remaining stick becomes shorter with each break. In practice, the number of groups does not increase as fast as $N$ due to the characteristic of the stick-breaking process that leads the group probability to soon approach zero.

Although in principle we do not restrict the maximum number of groups and allow the number to rise as $N$ increases, a finite number of instances will only occupy a finite number of $K$ components. The concentration parameter $a$ in the prior of $\xi_k$ determines the degree of discretization -- the complexity of the mixture and, consequently, $K$, as also revealed in (\ref{eq:Polya_urn}). As $a \to 0$, the realizations are all concentrated at a single value, however as $a \to \infty$, the realizations become continuous-valued as its based distribution. Specifically, \citet{antoniak1974} derives the relationship between $a$ and the number of unique groups,
\begin{align*}
	E\left( K | a \right) \approx a \log \left(\frac{a+N}{a}\right) \quad \text { and } \quad \operatorname{Var}\left(K | a \right) \approx \alpha \left[ \log \left(\frac{\alpha+N}{\alpha}\right)-1 \right],
\end{align*}
that is, the expected number of unique groups is increasing in both $a$ and the number of units $N$. 


\citet{escobar1995} highlights the importance of specifying $a$ when imposing prior smoothness on an unknown density and demonstrates that the number of estimated groups under a DP prior is sensitive to $a$. This suggests that a data-driven estimate of $a$ is more reasonable. Moreover, \citet{ascolani2022} emphasizes the importance of introducing a prior for $a$ as it is crucial for learning the true number of groups as $N$ increases and hence establishing the posterior consistency. We define a gamma hyperprior for $a$ and update it based on the observed data in order to alter the smoothness level. This step generates a posterior estimate of $a$, which indirectly determines the number of groups $K$ without reestimating the models with different group sizes. Essentially, this represents ``automated" model selection.




Collectively, we specify a DP prior for $\left(\alpha_{i}, \sigma^2_{i}\right)$. The DP prior is a mixture of an infinite number of possible point masses, which can be constructed through the stick-breaking process. The discretization of the underlying distribution is governed by the concentration parameter $a$. With a hyerprior on $a$, we permit the data to determine the number of groups $K$ present in the data, which can expand unboundedly along with the data.

\subsubsection{Prior on Group Partitions}  \label{subsec:partition}

In a formal Bayesian formulation, a prior distribution is specified to partition $\mathcal{B}$ with associated indices $G$. Despite the fact that DP prior does not specify this prior distribution explicitly, we can characterize it using the exchangeable partition probability function (EPPF) \citep{pitman1995}. As we briefly mentioned in the last subsection, the EPPF plays a significant role in connecting the prior belief on group structure to the DP prior, which is included as part of our proposed prior distribution in Equation (\ref{eq:prior_G_soft}). 

The EPPF characterizes the distribution of a partition $\mathcal{B} = \left\{B_{1}, B_{2}, \ldots, B_{K}\right\}$ induced by $G$. As the generic Dirichlet process assumes units are exchangeable, any permutation has no effect on the joint probability distribution of $G$; hence, the EPPF is determined entirely by the number of groups and the size of each group. \citet{pitman1995} demonstrates that the EPPF of the Dirichlet process has the closed form,
\begin{align} \label{eq:EPPF}
	p(G) = \frac{\Gamma(a)}{\Gamma(a+N)} a^{K} \prod_{k=1}^{K} \Gamma\left( |B_k| \right), 
\end{align} 
where $a$ is the concentration parameter and $\Gamma(x) = (x-1)!$ denotes the Gamma function. Note that the partition $\mathcal{B}$ is conceived as a random object and hence the group number $K$ is not predetermined, but rather is a function of $G$, $K = K(G)$. 

\citet{sethuraman1994} and \citet{pitman1996} constructively show that group indices/partitions can be drawn from the EPPF for DP using the stick-breaking process defined in (\ref{main_text_pi}). As a result, the EPPF does not explicitly appear in the posterior analysis in the current setting so long as the priors for the stick lengths are included.

\subsubsection{Correlated Random Coefficients Model} \label{subsec:corr_rand_model}
As suggested by \citet{chamberlain1980}, allowing the individual effects to be correlated with the initial condition can eliminate the omitted variable bias. This subsection presents the first attempt to introduce dependence between grouped effects and covariates, under the presence of group structure in both heterogeneous slope coefficients and cross-sectional variance. The underlying theorems, such as posterior consistency, and the performance of the framework are left for future study. 


We primarily follow the proposed framework in \citet{liu2022} and utilize Mixtures of Gaussian Linear Regressions (MGLRx) for the group-specific parameters. MGLRx prior is discussed in \citet{pati2013} and can be viewed as a Dirichlet Process Mixture (DPM) prior that takes the dependence of covariates into account. Notice that the correlated random coefficients model requires a DPM-type prior for $\alpha_i$ and $\sigma_i^2$. This is because $\alpha_i$ and $\sigma_i^2$ are assumed to be correlated with covariates of each individual, and as such, they are not identical within a group.


Following \citet{liu2022}, we first transform $\sigma_i^2$ and define $l_i=\log \frac{\bar{\sigma}^2\left(\sigma_i^2-\underline{\sigma}^2\right)}{\bar{\sigma}^2-\sigma_i^2}$, where $\underline{\sigma}^2\left(\bar{\sigma}^2\right)$ is some small (large) positive number. This transformation simplifies the prior for $\sigma_i^2$, which is now dependent on covariates, and ensures that a similar prior structures can be applied to both $\alpha_i$ and $l_i$.

In the correlated random coefficients model, the DPM prior for $\alpha_{i}$ or $\sigma_i^2$ is an infinite mixture of Normal densities with the probability density function:
\begin{align} \label{prior:a_CRE}
	\alpha_{i} \sim \sum_{k=1}^{\infty} \pi_k (w_{i0}) N \left(\mu^{\alpha}_{k} [1 \; w'_{i0}]', \; \Omega^{\alpha}_{k} \right), \\
	\sigma_i^2 \sim \sum_{k=1}^{\infty} \pi_k (w_{i0}) N \left(\mu^{\sigma}_{k} [1 \; w'_{i0}]', \; \Omega^{\sigma}_{k} \right), 
\end{align}
where $w_{i0} = [1,  y_{i0}, x_{i,0:T}]'$ is the conditioning set at period 0, which includes initial value of each unit $y_{i0}$, the initial values of predetermined variables, and the whole history of exogenous variables. Notice that $\alpha_{i}$ and $\sigma_i^2$ share the same set of group probabilities $\pi_k (w_{i0})$.

Similar but not identical to the DP prior, it is the component parameters $\left(\mu^{\alpha}_{k}, \Omega^{\alpha}_{k} \right)$ or $\left(\mu^{\sigma}_{k}, \Omega^{\sigma}_{k} \right)$ that are directly drawn from the base distribution $G_0$. $G_0$ is assumed to be a conjugate Matricvariate-Normal-Inverse-Wishart distribution.


The group probabilities are now characterized by a \textit{probit} stick-breaking process \citep{rodriguez2011},
\begin{align} \label{prior:grp_prob_CRE}
	\pi_k (w_{i0}) = \Phi \left(\zeta_{k}\left(w_{i 0}\right)\right) \prod_{j<k}\left(1-\Phi\left(\zeta_{j}\left(w_{i 0}\right)\right)\right),
\end{align}
where the stochastic function $\zeta_{k}$ is drawn from a Gaussian process, $\zeta_{k} \sim G P\left(0, V_{k}\right)$ for $k=1,2, \cdots$. The Gaussian process is assumed to have zero mean and the covariance function $V_{k}$. defined as follows,
\begin{align}
	V_{k}(x, x') = \tau_{v}\exp \left( - A_k \| x - x' \|^2_2 \right),
\end{align}
where $\tau_{v} \sim IG(\frac{\nu_{v}}{2}, \frac{\nu_{v}}{2})$ and $A_k$ has its own hyperprior, see details in \citet{pati2013}.



\section{Posterior Analysis} \label{sec:post}

This section describes the procedure for analyzing posterior distributions for the baseline model described in (\ref{simple_model}) with the priors specified in Section \ref{sec:soft_constraint}. The joint posterior distribution of model parameters is
\begin{align*}\label{eq:posterior}
	& p(\alpha, \sigma^2, \Xi, a, G | Y, X, W, \phi) \\
	\propto \; & p(Y | X, \alpha, \sigma^2, G) p(\alpha, \sigma^2|\phi) p(G | \Xi) p(W | G) p(\Xi | a) p(a), \numberthis
\end{align*}
where $ p(Y | X, \alpha, \sigma^2, G)$ is the likelihood function given by equation (\ref{simple_model}) for an i.i.d. model conditional on group indices $G$, and $p(W | G)$ is the additional term of pairwise constraints with the form $p(W | G) = \prod_{i=1}^{N} \prod_{j=1}^{N}\exp\left(c W_{ij} \delta_{ij}\right)$.

\subsection{Posterior Sampling} \label{subsec:post}

Draws from the joint posterior distribution can be obtained by using blocked Gibbs sampling. The algorithm is derived from \citet{ishwaran2001} and \citet{walker2007}. Due to the use of a finite-dimensional prior and truncation, the method described in \citet{ishwaran2001} cannot truly address our demand for estimating the number of groups without a predetermined value or upper bound. We employ the slice sampler \citep{walker2007}, which is the exact block Gibbs sampler for the posterior computation in infinite-dimensional Dirichlet process models, modifying the block Gibbs sampler of \citet{ishwaran2001} to avoid truncation approximations. \citet{walker2007} augments the posterior distribution with a set of auxiliary variables consisting of i.i.d. standard uniform random variables, i.e., $u_i \stackrel{iid}{\sim} U(0,1)$ for $i = 1,2,..,N$. The augmented posterior is then represented as
\begin{align*}
	& p(\alpha, \sigma^2, \Xi, a, G, u | Y, X, W, \phi) \\
	\propto \; & p(Y | X, \alpha, \sigma^2, G) p(\alpha, \sigma^2|\phi) p(W | G) p(\Xi | a) p(a) \prod_{i} \mathbf{1}(u_i \le \pi_{g_i}). \numberthis
\end{align*}
where $\prod_{i} \mathbf{1}(u_i \le \pi_{g_i})$ is substituted for $p(G | \Xi)$ in the equation (\ref{eq:posterior}). 

To roughly see how slice sampling works, recall that the group probabilities are constructed in a sequential manner, following a stick-breaking procedure. The leftover of the stick after each break gets smaller and smaller. Given the finite number of units, we can always find the smallest $K^*$ such that for all groups $k \ge K^*$, the minimum of $u_i$ among all units is larger than $\pi_k$, which is bounded above by the length of the leftover after $k$ breaks, $1-\sum_{j = 1}^{k} \pi_j$. More formally,
\begin{align} \label{def:K_star}
	K^* = \min_k \left\{ u^* > 1 - \sum_{j=1}^k \pi_j \right\}, \text{ with } u^* = \min_{1 \le i \le N} u_i.
\end{align}
As a result, all units receive strictly zero probability of being assigned to any group $k = K^* + 1, K^* + 2,...,N$ since the indicator function $\mathbf{1}(u_i \le \pi_{k})$ is zero.


There are two advantages to incorporating the auxiliary variable $u$ into the model. First and foremost, $u$ directly determines the largest possible number of groups in each sampling iteration. This reduces the support of $G$ and $\boldsymbol{\Xi}$ to a finite space, enabling us to solve a problem of finite dimensions without truncation. Furthermore, $u$ have no effect on the joint posterior of other parameters because the original posterior can be restored by integrating out $u_i$ for $i = 1,2,...,N$.

The Gibbs sampler are used to simulate the joint posterior distribution of $\left(\alpha, \sigma^2, \Xi, u, a, G\right)$. We break this vector into blocks and sequentially sampling for each block conditional on the current draws for the other parameters and the data. The full conditional distributions for each block are easily derived using the conjugate priors specified in Section \ref{sec:model_prior}. 

For the group-specific parameters, we directly draw samples from their posterior densities as we adapt conjugate priors. The posterior inference with respect to $(\alpha, \sigma^2)$ becomes standard once we condition on the latent group indices $G$. It is essentially a Bayesian panel data regression for each group. The conditional posterior for the stick length $\Xi$ is a beta distribution given $G$, and hence direct sampling is possible. 

We follow \citet{walker2007} to derive the posterior of auxiliary variable $u$. As $u$ are standard uniformly distributed, the posterior is a uniform distribution defined on $\left(0, \pi_{g_i}\right)$, conditional on the group probabilities and group indices. In terms of the concentration parameter $a$, we use a 2-step procedure proposed by \citet{escobar1995}. Following their approach, we first draw a latent variable $\eta$ from $Beta(a+1, N)$. Then, given $\eta$ and number of groups $K^a$ in the current iteration, we directly draw $a$ from a mixture of two Gamma distribution.

It is worth noting that the steps for implementing the DP prior with or without soft pairwise constraints are the same for all parameter besides the group indices $G$. This is due to the fact that soft pairwise constraints only affect other parameters through the group indices. It is handy to sample group indices with soft pairwise constraints conditional on other parameters. The posterior probability of assigning unit $i$ to group $k$ includes additional term $p(W_i | G) = \prod_{j \ne i, g_j = k} \exp \left( 2 c W_{i j} \delta_{i j}\right) $ to rewards (penalizes) the abidance (violation) of constraints, 
\begin{align*} \label{post_G_soft}
	\bar{\pi}_k = p \left(g_{i} = k | \alpha, \sigma^2, G^{(i)}, u, Y, X, W\right) 
	= \frac{p\left(y_{i} | \alpha_{k}, \sigma_{k}^{2}, Y, X\right) \mathbf{1} \left(u_{i} \le \pi_{k} \right) p(W_i | G)}{\sum_{k' = 1}^{K^*} p\left(y_{i} | \alpha_{k'}, \sigma_{k'}^{2}, Y, X\right) \mathbf{1} \left(u_{i} \le \pi_{k'} \right) p(W_i | G)}, \numberthis
\end{align*}
where $K^*$ is the maximal number of groups after generating potential new group-specific slope coefficients and variance. We then draw the group index for unit $i$ from a multinomial distribution:
\begin{align*}
	g_i = k, \text{ with probability } \bar{\pi}_k. \numberthis
\end{align*}


Algorithm \ref{algo:RC_GH} below summarizes the algorithm for the proposed Gibbs sampling. For illustrative purposes, we focus primarily on the posterior densities of major parameters and omit details on step (\ref{algo:RC_GH:new_group}). In short, step (\ref{algo:RC_GH:new_group}) creates potential groups by sampling new $(\alpha_k, \sigma_k^2)$ from the prior if the latest $K^*$ based on newly drawn $u_i$ and $\pi_k$ is larger than previous $K^*$, which indicate the current iteration permits more groups. This Detailed derivations and explanation of each step are provided in Appendix \ref{appendix:post_soft}. 

\begin{algo}\label{algo:RC_GH}(Gibbs Sampler for Random Coefficients Model with Soft Pairwise Constraints)\\
	For each iteration $s = 1,2,..,N_{sim}$, \vspace{-0.25cm}
	\begin{enumerate}[(i)]
		\item Calculate number of active groups: $K^a = \max_{1 \le i \le N} g_i^{(s-1)}.$  \label{algo:RC_GH:new_K}
				
		\item Group-specific slope coefficients: draw $\alpha_k^{(s)}$ from $p \left(\alpha_k | \sigma^{2(s-1)}, G^{(s-1)},Y, X \right)$ for $k = 1,2,...,K^a$. 
		
		\item Group-specific variance: draw $\sigma_{k}^{2^{(s)}}$ from $p \left(\sigma_{k}^2 | \alpha^{(s)}, G^{(s-1)}, Y, X\right)$ for $k = 1,2,...,K^a$.
			
		\item Group ``stick length'': draw $\xi^{(s)}_k$ from $p \left(\xi_k | a^{(s-1)}, G^{(s-1)}\right)$ for $k = 1,2,...,K^a$ and update group probability in accordance to the stick-breaking procedure.
		
				
		\item Auxiliary variable: draw $u^{(s)}_i$ from $p \left( u_i | \Xi^{(s)}, G^{(s-1)} \right)$ for $i = 1,2,...,N$ and calculate $u^* = \min_{1 \le i \le N} u_i$.
	
		\item DP concentration parameter: draw a latent variable $\eta$ from $\textit{Beta} (a+1, N)$ and draw $a^{(s)}$ from $p \left( a | \eta, K^a \right)$.
		
		\item Generate potential groups based on $u^*$ and find the maximal number of groups $K^*$. \label{algo:RC_GH:new_group}
		
		\item[(xi)] Group indices: draw $g_i$ from $p\left(g_{i} = k | \alpha^{(s)}, \sigma^{2(s)}, G^{(i)}, u, Y, X, W\right)$ for $i = 1,2,...,N$ and $k = 1,2,...,K^*$.  \label{algo:RC_GH:group}
	\end{enumerate}
\end{algo}

\subsection{Determining Partition} \label{subsec:det_partition}

In contrast to popular algorithms such as agglomerative hierarchical clustering or the \textit{KMeans} algorithm, which return a single clustering solution, Bayesian nonparametric models provide a posterior over the entire space of partitions, enabling the assessment of statistical properties, such as the uncertainty on the number of groups. 

However, when the group structure is part of the major conclusion of an empirical analysis, the point estimate of group structure becomes crucial. \citet{wade2018} discuss in detail an appropriate point estimate of the group partitioning based on the posterior draws. From the decision theory, the point estimate $G^{*}$ minimizes the posterior expected loss,
\begin{align*}
	G^{*}=\underset{\widehat{G}}{\operatorname{argmin}} \; \mathbb{E}\left[L(G, \widehat{G}) | Y \right] = \underset{\widehat{G}}{\operatorname{argmin}} \sum_{G} L(G, \widehat{G}) p \left(G | Y\right),
\end{align*}
where the loss function $ L(G, \widehat{G})$ is the variation of information by \citet{meilua2007}, which measures the amount of information lost and gained in changing from partition $G$ to $\widehat{G}$.\footnote{Another possible loss function is the 0-1 loss function $L(G, \widehat{G}) = \mathbf{1} (G = \widehat{G})$, which leads to $G^*$ being the posterior mode. This loss function is undesirable since it ignores similarity between two partitions. For instance, a partition that deviates from the truth in the allocation of only one unit is penalized the same as a partition that deviates from the truth in the allocation of numerous units. Furthermore, it is generally recognized that the mode may not accurately reflect the distribution's center.
} The Variation of Information is based on the Shannon entropy $H(\cdot)$, and can be computed as
\begin{align*}
	\begin{aligned}
		\mathrm{VI}\left(G, \widehat{G}\right) &=-H(G)-H\left( \widehat{G} \right)+2 H\left(G \wedge \widehat{G}\right) \\
		&=\sum_{j=1}^K \frac{\lambda_j}{N} \log \left(\frac{\lambda_j}{N}\right)+\sum_{l=1}^{K^{\prime}} \frac{\lambda_l^{\prime}}{N} \log \left(\frac{\lambda_l^{\prime}}{N}\right)-2 \sum_{j=1}^K \sum_{l=1}^{K^{\prime}} \frac{\lambda_{j l}^{\wedge}}{N} \log \left(\frac{\lambda_{j l}^{\wedge}}{N}\right),
	\end{aligned}
\end{align*}
where $\log$ denotes $\log$ base 2, $\lambda_j=\left|B_j\right|$ is the cardinality of the group $j$, and $\lambda_{j l}^{\wedge}$ the size of blocks of the intersection $G \wedge \widehat{G}$ and hence the number of indices in block $j$ under partition $G$ and block $l$ under $\widehat{G}$.

\citet{wade2018} show that the optimal group partitioning can be identified based on the posterior similarity matrix,
\begin{align} \label{eq:VI_emp}
	g^* = \underset{\widehat{g}}{\operatorname{argmin}} \sum_{i=1}^{N} \log \left(\sum_{j=1}^{N} \mathbf{1}\left(\widehat{g}_{j}=\widehat{g}_{i}\right)\right)-2 \sum_{i=1}^{N} \log \left(\sum_{j=1}^{N} P \left(g_{j}=g_{i} | Y,X,W\right)  \mathbf{1}\left(\widehat{g}_{j}=\widehat{g}_{i}\right)\right)
\end{align}
where $P \left(g_{j}=g_{i} | Y,X,W\right)$ is the $(i,j)$ entry of the posterior similarity matrix. We refer to \citet{wade2018} for additional properties and empirical evaluations.


\subsection{Connection to Constrained KMeans Algorithm} \label{sec:connect_kmeans}

The procedure of Gibbs sampling with soft constraints in Algorithm \ref{algo:RC_GH} is closely related to constrained clustering in the computer science literature. In this parallel literature, constrained clustering refers to the process of introducing prior knowledge to guide a clustering algorithm. For a subset of the data, the prior knowledge takes the form of constraints that supplement the information derived from the data via a distance metric. As we shall see below, under several simplifying assumptions, our framework could be reduced to a deterministic method for estimating group heterogeneity using a constrained \textit{KMeans} algorithm. Though this deterministic method may address the practical issues in BM, it only works for certain restricted models and hence is not as general as our proposed framework.


We start with a brief review of the Pairwise Constrained KMeans (\textit{PC-KMeans}) clustering algorithm by \citet{basu2004}, which is a well-known clustering algorithm in the field of semi-supervised machine learning.  It's a pairwise constrained variant of the standard \textit{KMeans} algorithm in which an augmented objective function is used in the assignment step. Given a collection of observations $\left(y_{1}, y_{2}, \ldots, y_{N}\right)$, a set of positive-link constraints $\mathcal{P}$, a set of negative-link constraints $\mathcal{N}$, the cost of violating constraints $w = \{w^p_{i j}, w^n_{i j} \}$ and the number of groups $K$, the \textit{PC-KMeans} algorithm divides $N$ observations into $K$ groups (the \textit{assignment} step) so as to minimize the following objective function,
\begin{align}\label{eq:obj_pckmeans_main}
	& \underbrace{ \frac{1}{2}\sum_{k=1}^{K} \sum_{i \in B_k} \left\|y_i - \mu_{k}\right\|^{2}}_{\text{within-cluster sum of squares}} + \underbrace{\sum_{\left(i,j\right) \in \mathcal{P}} w^p_{ij} \mathbf{1}\left( g_{i} \neq g_j\right) + \sum_{\left(i,j\right) \in \mathcal{N}} w^n_{i j} \mathbf{1}\left( g_{i}=g_j\right) }_{\text{cost of violation}},
\end{align}
where $\mu_{k}$ is the centroid of group $k$, i.e., $\mu_{k}=\frac{1}{\left|B_k\right|} \sum_{i \in B_k} y_i$, $B_k$ is the set of units assigned to group $k$, and $\left|B_k\right|$ is the size of group $k$. The first part is the objective function for the conventional \textit{KMeans} algorithm, while the second part accounts for the incurred cost of violating either PL constraints ($w^p_{i j}$) or NL constraints ($w^n_{i j}$).

Similar to \textit{KMeans},  \textit{PC-KMeans} alternates between reassigning units to groups and recomputing the means. In the assignment step, it determines a disjoint $K$ partitioning that minimizes  (\ref{eq:obj_pckmeans_main}). Then the update step of the algorithm recalculates centroids of observations assigned to each cluster and updates $\mu_{k}$ for all $k$.

By applying asymptotics to the variance of distributions within the model, we demonstrate linkages between the posterior sampler of our constrained BGFE estimator and \textit{KMeans}-type algorithms in Theorem \ref{thm:pc_kmean}. We investigate small-variance asymptotics for posterior densities, motivated by the asymptotic connection between the Gibbs sampling algorithm for the Dirichlet process mixture model and \textit{KMeans} \cite{kulis2011}, and demonstrate that the Gibbs sampling algorithm for the CBG estimator with soft constraints encompasses the constrained clustering algorithm \textit{PC-KMean} in the limit.


\begin{theorem} \label{thm:pc_kmean} {\normalfont (Equivalency between BGFE with Soft Constraints and \textit{PC-KMeans})} \\
	If the following conditions hold,
	\begin{enumerate}[(i)]
		\item group pattern is in fixed-effects but not in slope coefficients, i.e., $x_{it} = 1$. Other covariates might be introduced, but they cannot have grouped effects on $y_{it}$; \label{soft_as0}
		\item The number of group is fixed at $K$; \label{soft_as1}
		\item Homoscedasticity: $\sigma^2_k = \sigma^2$ for all $k = 1,2,...,K$; \label{soft_as2}
		\item Constraint weights is scaled by the variance of errors: $W_{ij} \to W_{ij} / \sigma^2$; \label{soft_as3}
	\end{enumerate}
	then the proposed Gibbs sampling algorithm for the BGFE estimator with soft constraint embodies the \textit{PC-KMeans} clustering algorithm in the limit as $\sigma^2 \to 0$. In particular, the posterior draw of group indices $G$ is the solution to the \textit{PC-KMeans} algorithm.
\end{theorem}

We return to the world of grouped fixed-effects models. In fact, the clustering algorithm is essential for BM and \citet{bonhomme2022}, who use the \textit{KMeans} algorithm to reveal the group pattern in the fixed-effects. With the theorem described above, it motivates a \textit{constrained} version of BM's GFE estimator. We show that it is straightforward to incorporate prior knowledge in the form of soft paired restrictions into the GFE estimator. The \textit{soft pairwise constrained} grouped fixed-effects (SPC-GFE) estimator is defined as the solution to the following minimization problem given the number of groups $K$:
{\small
\begin{align} \label{eq_pcgfe}
	\left( \widehat{\theta}, \widehat{\alpha}, \widehat{G} \right) = \underset{\theta, \alpha, G}{\operatorname{argmin}} \sum_{i=1}^{N} \sum_{t=1}^{T}\left(y_{i t} - x'_{it} \theta  - \alpha_{g_{i} t}\right)^{2}  + c \left[\sum_{(i,j) \in \mathcal{P}} w^p_{i j} \mathbf{1} \left(g_i \ne g_j\right) +  \sum_{(i,j) \in \mathcal{N}}   w^n_{i j} \mathbf{1} \left(g_i = g_j\right) \right],
\end{align}
}where the minimum is taken over all possible partitions $G$ of the $N$ units into $K$ groups, common parameters $\theta$, and group-specific time effects $\alpha$. $w^p_{i j}$ and $w^n_{i j}$ are the user-specified costs on PL and NL constraints.

For given values of $\theta$ and $\alpha$, the optimal group assignment for each individual unit is
{\small
\begin{align} \label{eq_pcgfe_g}
	\widehat{g}_{i}(\theta, \alpha)  = \underset{g \in\{1, \ldots, K\}}{\operatorname{argmin}} \sum_{i=1}^{N} \sum_{t=1}^{T}\left(y_{i t} - x'_{it} \theta  - \alpha_{g_{i} t}\right)^{2}  + c \left[\sum_{(i,j) \in \mathcal{P}} w^p_{i j} \mathbf{1} \left(g_i \ne g_j\right) +  \sum_{(i,j) \in \mathcal{N}}   w^n_{i j} \mathbf{1} \left(g_i = g_j\right) \right],
\end{align}
}where we essentially apply the \textit{PC-KMeans} algorithm to get the group partition. The SPC-GFE estimator of $(\theta, \alpha)$ in (\ref{eq_pcgfe}) can then be written as
\begin{align} \label{eq_pcgfe_para}
	(\widehat{\theta}, \widehat{\alpha})=\underset{\theta, \alpha}{\operatorname{argmin}} \sum_{i=1}^{N} \sum_{t=1}^{T}\left(y_{i t} - x'_{it} \theta  - \alpha_{\widehat{g}_{i} t}\right)^{2},
\end{align}
where $\widehat{g}_{i} = \widehat{g}_{i}(\theta, \alpha)$ is given by (\ref{eq_pcgfe_g}). $\theta$ and $\alpha$ are computed using an OLS regression that controls for interactions of group indices and time dummies. The SPC-GFE estimate of $g_{i}$ is then simply $\widehat{g}_{i}(\widehat{\theta}, \widehat{\alpha})$.

\begin{remark}
	While the SPC-GFE estimator implements soft constraints, it still requires a predetermined number of group $K$ and model selection.
\end{remark}



\section{Empirical Applications} \label{sec:emp_result}

We apply our panel forecasting methods to the following two empirical applications: inflation of the U.S. CPI sub-indices and the income effect on democracy. The first application focuses mostly on predictive performance, whereas the second application focuses primarily on parameter estimation and group structure.



 
\subsection{Model Estimation and Measures of Forecasting Performance}

To accommodate richer assumptions on model, we use variants of the baseline model in Equation (\ref{simple_model}) in this section, either by adding common regressors or allowing for time-variation in the intercept. We use the conjugate prior for all parameters, see details in Appendix \ref{appendix:prior}.

\subsubsection{Specification of Constraints} \label{app:constraints}

In both applications, the prior knowledge of the latent group structure or the pre-grouping structure covers all units. In the first application, CPI sub-indices can be clustered by expenditure category, whereas countries in the second application may be grouped according to their geographic regions. We build positive-link and negative-link constraints given the prior knowledge: all units within the same group are presumed to be positive-linked, while units from different categories are believed to be negative-linked. In terms of the accuracy of constraints, $\psi^{PL}_{ij}$ and $\psi^{NL}_{ij}$ are equal for all constraints with the same type, following the strategy described in Section \ref{subsec:specify_constraints}. In the applications below, we fix $\psi^{PL}_{ij} = 0.65$ and $\psi^{NL}_{ij} = 0.55$ to reflect the belief that PL constraints (attracting forces) play slightly more important role than the NL constraints (repelling forces) and NL constraints cannot be ignored. Finally, we construct weights $W_{ij}$ using (\ref{eq:weight_construction}). Notice that these assumptions on prior and hyperparameters are an example to showcase how the proposed framework works with real data. In practice, we may specify constraints for a subset of units with different levels of weights, either in a data-driven manner (for instance, highly correlated units may fall into the same group with a high level of confidence) or in a model-based manner (i.e., $W$ is a function of covariates). 


\subsubsection{Determining the Scaling Constant $c$} \label{app:det_c}


Given that the dimension of the space of group partitions increases exponentially with the number of units $N$, attention must be given while selecting $c$ across analyses with different $N$. As suggested by \citet{paganin2021}, calibrating the modified prior is computationally intensive. We are facing a trade-off between investing time to get the prior ``exactly right" and letting the constant $c$ be an estimated model parameter. As such, we propose to find the optimal $c$ that maximizes marginal data density using grid search.  

In the Monte Carlo simulation, the value of $c$ is fixed for simplicity, but in the empirical applications, $c$ is determined by marginal data density (MDD). We calculate MDD using the harmonic mean estimator \citep{newton1994}, which defined as
\begin{align}
	\hat{m}^{H M}(y) = \left[ \frac{1}{S} \sum_{j=1}^{S} \frac{1}{p\left(y | \theta^{(j)}\right)} \right]^{-1},
\end{align}
given a sample $\theta^{(j)}$ from the posterior $p(\theta | y)$. The simplicity of the harmonic mean estimator is its main advantage over other more specialized techniques. It uses only within-model posterior samples and likelihood evaluations, which are often available anyway as part of posterior sampling. We finally choose the optimal value for $c$ that maximizes MDD.

\subsubsection{Estimators}

We consider six estimators in the section. The first three estimators are our proposed Bayesian grouped fixed-effects (BGFE) estimator with different assumptions on cross-sectional variance and pairwise constraints. The last three estimator ignore the group structure. 

\begin{enumerate}[(i)]
		
	\item \textit{BGFE-he-cstr}: group-specific slope coefficients and heteroskedasticity \textit{with} constraints. 
	
	\item \textit{BGFE-he}: group-specific slope coefficients and heteroskedasticity \textit{without} constraints. 
	
	\item \textit{BGFE-ho}: homoskedastic version of \textit{BGFE-he}.
	
	\item \textit{Pooled OLS}: fully homogeneous estimator
	
	\item \textit{AR-he}: flat-prior estimator that assumes $p(\alpha_i) \; \propto \; 1$ corresponds to standard AR model with additional regressor $u_t$ in this environment. 
	
	\item \textit{AR-he-PC}: \textit{AR-he} with the lagged value of the first principal component as additional regressor. 
	
\end{enumerate}

In the first application, we focus on inflation forecasting. For the most recent advances in this topic, \citet{faust2013} provide a comprehensive overview of a large set of traditional and recently developed forecasting methods. Among many candidate methods, we choose the AR model as the benchmark and exclusively include it as an alternative estimator in this exercise. This is because the AR is relatively hard to beat and, notably, other popular methods, such as the Atkeson–Ohanian version random walk model \citep{atkeson2001}, UCSV \citep{stock2007}, and TVP-VAR \citep{primiceri2005}, generally do as reasonably well as the AR model, according to \citet{faust2013}.


\subsubsection{Posterior Predictive Densities}


We generate one-step ahead forecasts of $y_{i, T+1}$ for $i = 1,...,N$ conditional on the history of observations
\begin{align*}
	Y &= [y_1, y_2, ..., y_N ], y_i = [y_{i1},y_{i2},...,y_{iT}]', \\
	X &= [x_1, x_2, ..., x_N ], x_i = [x_{i1},x_{i2},...,x_{iT}]',
\end{align*}
and newly available variables $x_{i T+1}$ at $T+1$. 

The posterior predictive distribution for unit $i$ is given by
\begin{align}
	p(y_{i T+1} | Y, X) = \int p(y_{i T+1} | Y, X, \Theta) p(\Theta | Y, X) d \Theta,
\end{align}
where $\Theta$ is a vector of parameters $\Theta = \left( \alpha_{g_i}, \sigma^2_{g_i}, g_i \right)$. This density is the posterior expectation of the following function:
\begin{align}
	p(y_{i T+1} | Y, X, \Theta) = \sum_{k=1}^{K(G)} \mathbf{1}(g_i = k) p\left(y_{i T+1}| Y,  X, \Theta \right),
\end{align}
which is invariant to relabeling the components of the mixture and $K(G)$ is the number of groups in $G$. Given $S$ posterior draws, the posterior predictive distribution estimated from the MCMC draws is
\begin{align}
	\hat{p}(y_{i T+1} | Y,X) = \frac{1}{S} \sum_{j=1}^{S} \left[ \sum_{k=1}^{K^{(j)}(G)} \mathbf{1}(g_i = k) p\left(y_{i T+1}| Y, X, \Theta^{(j)} \right) \right].
\end{align}
where
\begin{align}
	p\left(y_{i T+1}| Y, X, \Theta^{(j)} \right) = \phi \left(y_{i T+1}; \alpha_{g_i}^{(j)'} x_{it+1} + \gamma^{(j)} z_{it+1}, \sigma^{(j)^2}_{g_i} \right).
\end{align}

We can therefore draw samples from $\hat{p}(y_{i T+1} | Y,X)$ by simulating (\ref{simple_model}) forward conditional on the posterior draws of $\Theta $ and observations. Note that MCMC exhibits the true Bayesian predictive distribution, implicitly integrating over the entire underlying parameter space.

\subsubsection{Point Forecasts}
We evaluate the point forecasts via the real-time recursive out-of-sample Root Mean Squared Forecast Error (RMSFE) under the quadratic compound loss function averaged across units. Let $\hat{y}_{i T+1 | T} $ represent the predicted value conditional on the observed data up to period $T$, the loss function is written as
\begin{align}
	L\left(\widehat{y}_{1:N, T+1| T}, y_{1: N, T+1}\right) = \frac{1}{N} \sum_{i=1}^{N} \left(\hat{y}_{i T+1 | T} - y_{i T+1}\right)^{2} =  \frac{1}{N} \sum_{i=1}^{N} \hat{\varepsilon}_{i T+1 | T}^2,
\end{align}
where $y_{i, T+1}$ is the realization at $T+1$ and $\hat{\varepsilon}_{i T+1 | T}$ denote the forecast error.

The optimal posterior forecast under quadratic loss function is obtain by minimizing the posterior risk,
\begin{align*}
	\hat{y}_{1:N, T+1 | T} & = \argmin_{\hat{y}  \in \mathbb{R}^N} \int_{-\infty}^{\infty} L\left(\hat{y} , y_{1: N, T+1}\right) p(y_{1: N, T+1} | Y,X) d y_{1: N, T+1} \\
	& = \argmin_{\hat{y} \in \mathbb{R}^N} \frac{1}{N} \sum_{i=1}^{N} E \left[\left(\hat{y} - y_{i T+1}\right)^{2} | Y,X \right]. \numberthis
\end{align*}
This implies optimal posterior forecast is the posterior mean,
\begin{align}
	\hat{y}_{i, T+1 | T}  = E \left(y_{i T+1} | Y,X \right), \text{ for } i=1, \ldots, N.
\end{align}

Conditional on posterior draws of parameters, the mean forecast can be approximated by the Monte Carlo averaging,
\begin{align}
	\hat{y}_{i, T+1 | T}  \approx \frac{1}{S}\sum_{j=1}^{S} \hat{y}_{i T+1 | T}^{(j)} = \frac{1}{S}\sum_{j=1}^{S} \hat{\alpha}^{(j)'}_{g_i} x_{i T+1}.
\end{align}

Finally, the RMSFE across units is given by
\begin{align}
	RMSFE_{T+1} = \sqrt{ \frac{1}{N} \sum_{i=1}^{N} \left(y_{i, T+1} -\hat{y}_{i, T+1 | T}\right)^2 }.
\end{align}

%

\subsubsection{Density Forecasts}
To compare the performance of density forecasts for various estimators, we report the average log predictive scores (LPS) to assess the performance of the density forecast from the view of the probability distribution function. As suggested in \citet{geweke2010}, the LPS for a panel reads as,
\begin{align}
	LPS_{T+1}  =  &  - \frac{1}{N} \sum_{i=1}^{N} \ln \int p\left(y_{i T+1} | Y, X, \Theta\right) p(\Theta | Y, X) d \Theta,
\end{align}
where the expectation can be approximated using posterior draws:
\begin{align}
	\int p\left(y_{i T+1} | Y, X, \Theta\right) p(\Theta | Y, X) d \Theta & \approx \frac{1}{S} \sum_{j = 1} ^{S} p \left(y_{i T+1} | Y, X, \Theta^{(j)} \right).
\end{align}

The following results are also robust to other metrics such as the continuous ranked probability score \citep{matheson1976, hersbach2000}.

\subsection{Inflation of the U.S. CPI Sub-Indices} \label{sec:emp_result_cpi}

Policymakers and market participants are very interested in the abilities to reliably predict the future disaggregated inflation rate. Central banks predict future inflation trends to justify interest rate decisions, control and maintain inflation around their targets. The Federal Reserve Board forecasts disaggregated price categories for short-term inflation forecasting \citep{bernanke2007}. They rely primarily on the bottom-up approach that focuses on estimating and forecasting price behavior for the various categories of goods and services that make up the aggregate price index. Moreover, investors in fixed-income markets in the private sector wish to forecast future sectoral inflation in order to anticipate future trends in discounted real returns. Some private firms also need to predict specific inflation components in order to forecast price dynamics and reduce risks accordingly.

In this section, we demonstrate the use of constrained BGFE estimators with prior knowledge on the group pattern to forecast inflation rates for the sub-indices of U.S. Consumer Expenditure Index (CPI). We focus primarily on the one-step ahead point and density forecast. Due to space constraints, we only report the group partitioning for the most recent month in the main text. 

\subsubsection{Model Specification and Data}


\noindent{\bf Model:} We start by exploring the out-of-sample forecast performance of a simple, generic Phillips curve model. It is a panel autoregressive distributed lag (ADL) model with a group pattern in the intercept, coefficients, as well as error variance. The model is given by
\begin{align}
	y_{it+1} = & \; \alpha_{g_{i}} + \sum_{j=0}^{p-1} \rho_{g_{i, j}} y_{it-j} +  \beta_{g_i} u_{t} + \varepsilon_{it+1}, \quad \varepsilon_{it+1} \sim N (0, \sigma_{g_i}^2).
\end{align}
where $y_{it}$ is year-over-year inflation rate, i.e., $y_{it} = \log(\text{price}_{it} / \text{price}_{it-12})$, and $u_t$ is the slack measure for the labor market, the unemployment gap. We fix $p$ at 3 because the benchmark AR model would have the best predictive performance.


\vspace{0.5cm}

\noindent{\bf Data:} 
We collect the sub-indices of CPI for all urban consumer (CPI-U) that include food and energy. The raw data is obtained from the U.S. Bureau of Labor Statistics (BLS),  which is recorded on a monthly basis from January 1947 to August 2022. The CPI-U is a hierarchical composite index system that partitions all consumer goods and services into a hierarchy of increasingly detailed categories. It consists of eight major expenditure categories (1) Apparel; (2) Education and Communication; (3) Food and Beverages; (4) Housing; (5) Medical Care; (6) Recreation; (7) Transportation; (8) Other Goods and Services. Each category is composed of finer and finer sub-indexes until the most detailed levels or “leaves” are reached. This hierarchical structure can be represented as a tree structure, as shown in Figure \ref{fig:CPI_tree}. It is important to note that the parent series and its child series may be highly correlated and readily form a group due to the fact that parent series are generated from child series. For instance, the \textit{Energy Services} is expected to correlated with its child series \textit{Utility gas service} and \textit{Electricity}. Due to our focus on group structure, it is vital to eliminate all parent series in order to prevent not just double-counting but also dubious grouping results. More details regarding the data are provided in Appendix \ref{appendix:data_cpi}.

\begin{figure}[h]
	\centering
	\caption{Hierarchical Structure of CPI} \label{fig:CPI_tree}
	\begin{center}
		\resizebox*{.99\linewidth}{!}{%
			\begin{forest}
				forked edges,
				for tree={
					grow = east,
					align = left,
					font=\sffamily,
				}
				[All Items
				[Housing 
				[Shelter
				[Rent of primary residence]
				[Owners' equivalent rent of residences]
				]
				[Fuels and utilities
				[Water sewer and trash]
				[Household energy
				[Fuel oil and other fuels
				[Fuel oil]
				[Propane kerosene and firewood]]
				[Energy services
				[Electricity]
				[Utility gas service]]]
				]
				]
				[Transportation
				[...]
				]
				[Food and beverages
				[...]
				]
				]
			\end{forest}
		}
	\end{center}
\end{figure}
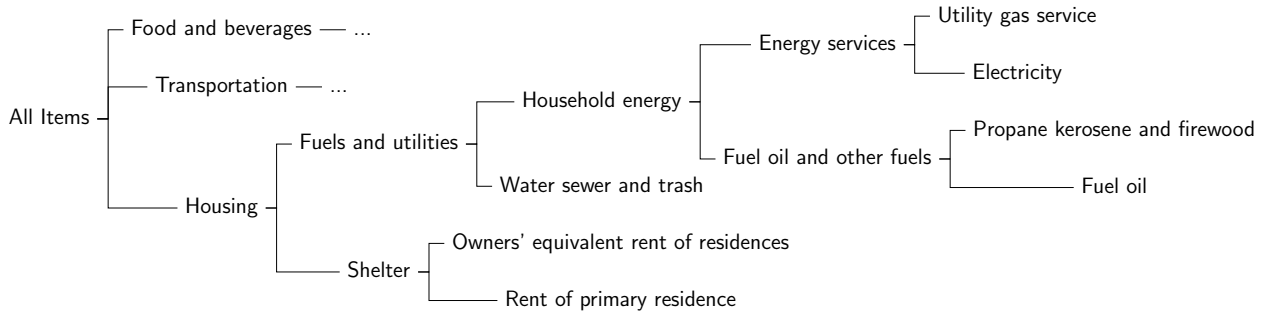

\noindent{\bf Pre-grouping:} The official expenditure categories are used to build PL and NL constraints: all units within the same categories are presumed to be positive-linked, while units from different categories are believed to be negative-linked. 


We focus on the CPI sub-indices after January 1990 for two reasons: (1) the number of sub-indices before 1990 was relatively small, diminishing the importance of the group structure; and (2) the consumption has been changed and more expenditure series were introduced in the 1990s as a result of the popularity of electronic products, food care, etc. After the elimination of all parent nodes, the unbalanced panel consists of 156 sub-indices in eight major expenditure categories. We employ rolling estimation windows of 48 months\footnote{The benchmark \textit{AR-he} model scores the best overall performance with a window size of 48.} and require each estimation sample to be balanced, removing individual series lacking a complete set of observations in a given window. Finally, we generate 329 samples with the first forecast produced for April 1995.



\subsubsection{Results}

We begin the empirical analysis by comparing the performance of point and density forecasts across 329 samples. Throughout the analysis, the \textit{AR-he} estimator serves as the benchmark as it essentially assumes individual effects. 

In Figure \ref{fig:app_cpi_RMSE}, we present the frequency of each estimator with the lowest RMSFE in the panel (a) and the boxplot\footnote{The boundaries of the whiskers is based on the 1.5 IQR value.  All other points outside the boundary of the whiskers are plotted as outliers in red crosses.} of the ratio of RMSFE relative to the \textit{AR-he} estimator in the panel (b). First, the AR-he and AR-he-PC estimators, which rely only on an individual's own past data, are not competitive in point forecasts and perform considerably worse than the others. This implies that it is highly advantageous to explore cross-sectional information to improve point forecasts. Moreover, the BGRE-he-cstr estimator scores the highest frequency of being the best estimator despite the fact that BGFE-he-cstr, BGFE-he, BGFE-ho, and pooled OLS estimators all utilize cross-sectional information. Examining the box plot, we find that the BGFE-ho and pooled OLS estimators, which overlook heteroskedasticity, can achieve greater performance in some samples, but make poorer forecasts more often than the other estimators. BGFE-he-cstr and BGFE-he estimator, on the other hand, typically outperform the others and the benchmark in terms of median RMSFE and the ability to produce forecasts with the lowest RMSFE without also increasing the risk of generating the least accurate forecasts.



\begin{figure}[hp]
	\caption{RMSFE - All Samples}
	\label{fig:app_cpi_RMSE}
	\centering
	\begin{subfigure}[b]{0.45\textwidth}
		\centering
		\caption{Freq. lowest RMSFE}
		\includegraphics[scale= 0.35]{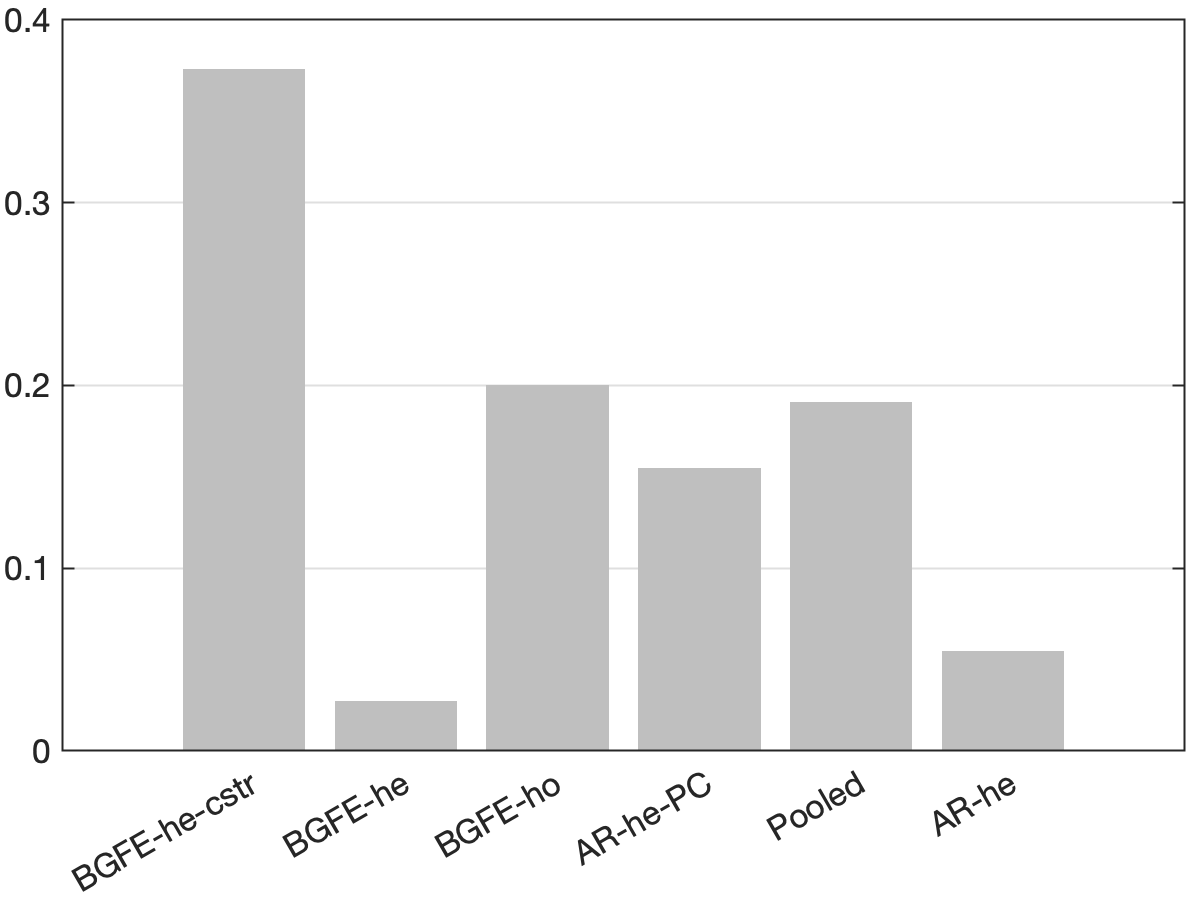}
	\end{subfigure}
	\begin{subfigure}[b]{0.45\textwidth}
		\centering
	    \caption{Boxplot: $\text{RMSE}_{m}$ / $\text{RMSE}_{AR-he}$}
		\includegraphics[scale= 0.35]{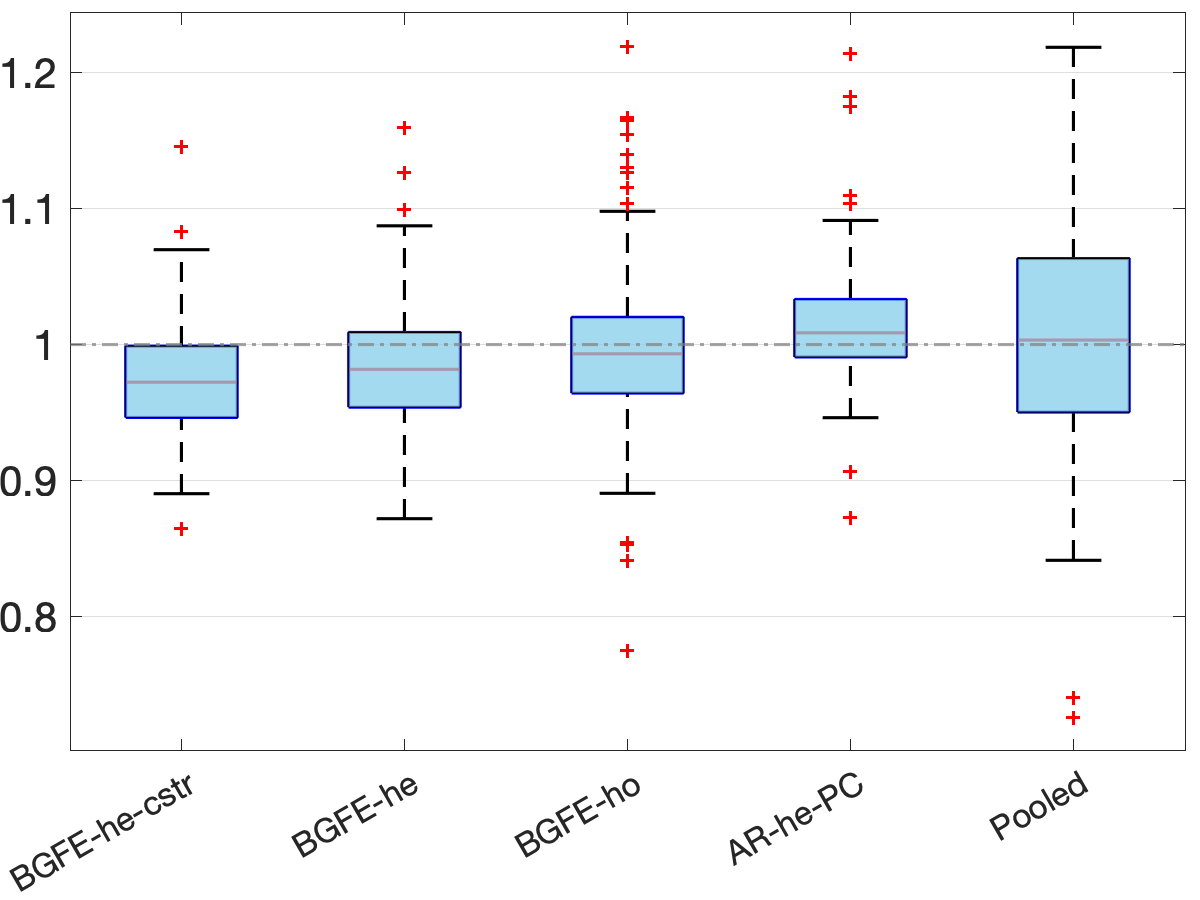}
	\end{subfigure}
\end{figure}


The revealing patterns of the density forecast are significantly distinct from those of the point forecast. Figure \ref{fig:app_cpi_LPS} depicts the log predictive score (LPS) for density forecast. The most notable pattern from the panel (a) is that the BGFE-he-cstr estimator, which incorporates prior knowledge, is dominating and outperforms the rest in over 80\% of the samples. It emerges as the apparent winner in this case. Furthermore, when generating density forecast, the BGFE-ho and pooled OLS are not as accurate as they are in point forecast: they never have the lowest LPS across samples. This also confirms that the heteroscedasticity\footnote{We provide more results in Section \ref{subsec:hetero_vs_homo} to explore the importance of heteroskedasticity in density forecast for the inflation.} is a well-known feature of the inflation time series \citep{clark2015}. In the boxplot, we ignore BGFE-ho and pooled OLS and show the differences in LPS between the respective estimators and the AR-he estimator. As LPS differences represent percentage point differences, BGFE-he-cstr can provide density forecasts that are up to 22\% more accurate compared to the benchmark model. Finally, despite the fact that the BGFE-he-cstr and BGFE-he estimators are mainly based on the same algorithm, the use of prior knowledge on group pattern further enhances the performance, resulting in the BGFE-he-cstr estimator having a lower LPS and scoring the best model with the highest frequency. 

\begin{figure}[htp]
	\caption{Log Predictive Scores - All Samples}
	\label{fig:app_cpi_LPS}
	\centering
	\begin{subfigure}[b]{0.45\textwidth}
		\centering
		\caption{Freq. lowest LPS}
		\includegraphics[scale= 0.35]{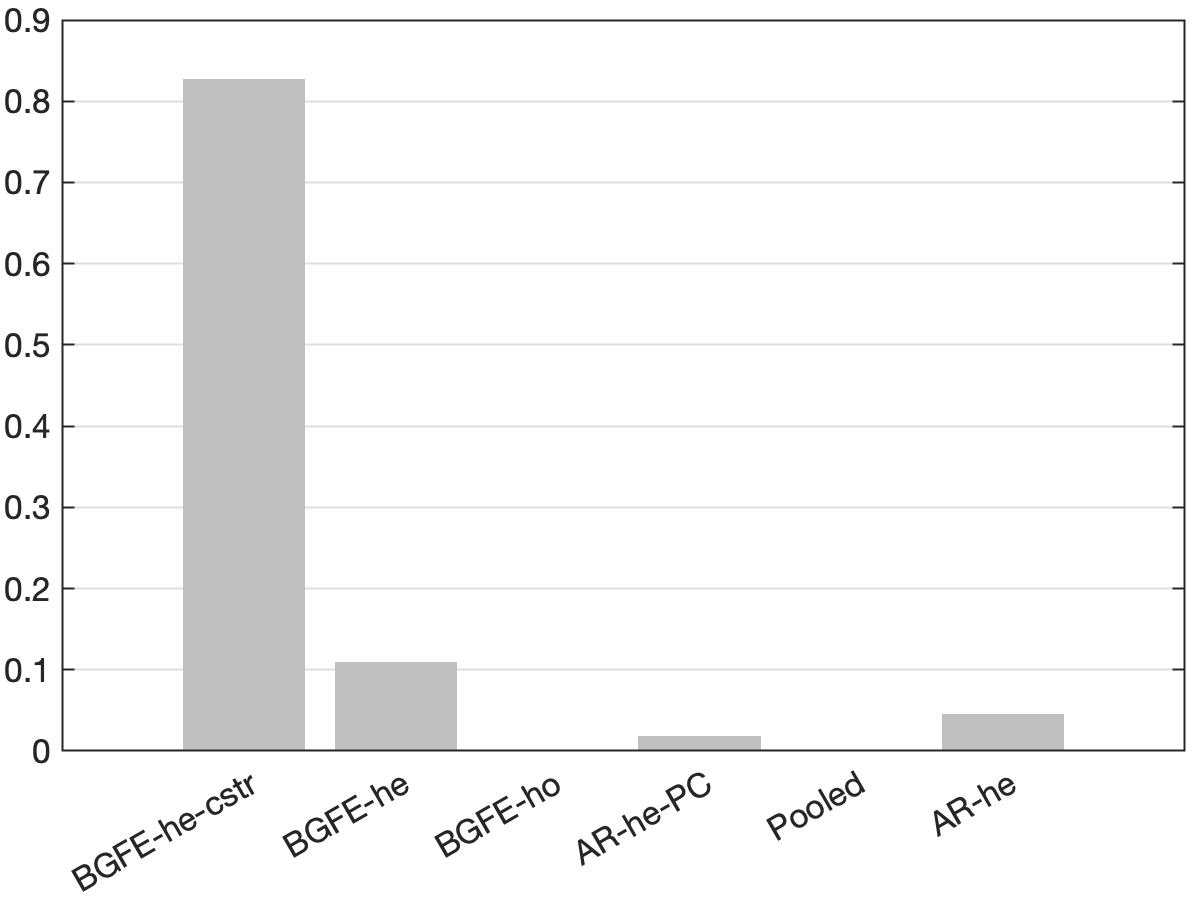}
	\end{subfigure}
	\begin{subfigure}[b]{0.45\textwidth}
		\centering
		\caption{Boxplot: $\text{LPS}_{m}$ - $\text{LPS}_{AR}$}
		\includegraphics[scale= 0.35]{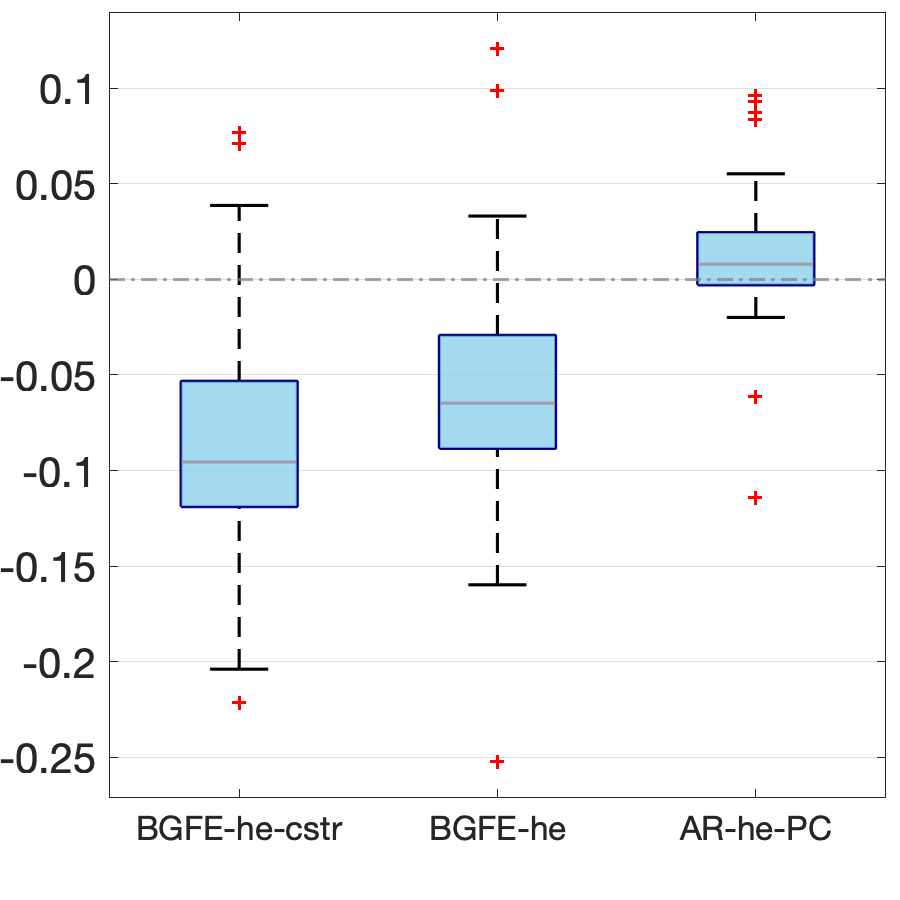}
	\end{subfigure}
\end{figure}

Next, we assess the value of adding prior information about groups by comparing the performance of the BGFE-he-cstr and BGFE-he estimators exclusively. The solid black line in Figure \ref{fig_app:app_cpi_relative_RMSE} represents the ratio of RMSE between BGFE-he-cstr and BGFE-he. The periods during which BGFE-he-cstr, BGFE-he, and all other estimators achieve the lowest RMSE are indicated by pink, blue, and green shaded areas, respectively. Though the BGFE-he-cstr estimator is not always the best across samples, the prior information improves the performance of the Bayesian grouped estimator. The BGFE-he-cstr estimator performs better than the BGFE-he estimator in most samples, with an average improvement of 2\%.

\begin{figure}[h]
	\caption{Relative RMSE, BGFE-he-cstr}
	\label{fig_app:app_cpi_relative_RMSE} \vspace{-0.5cm}
	\begin{center}
		\includegraphics[scale=0.45]{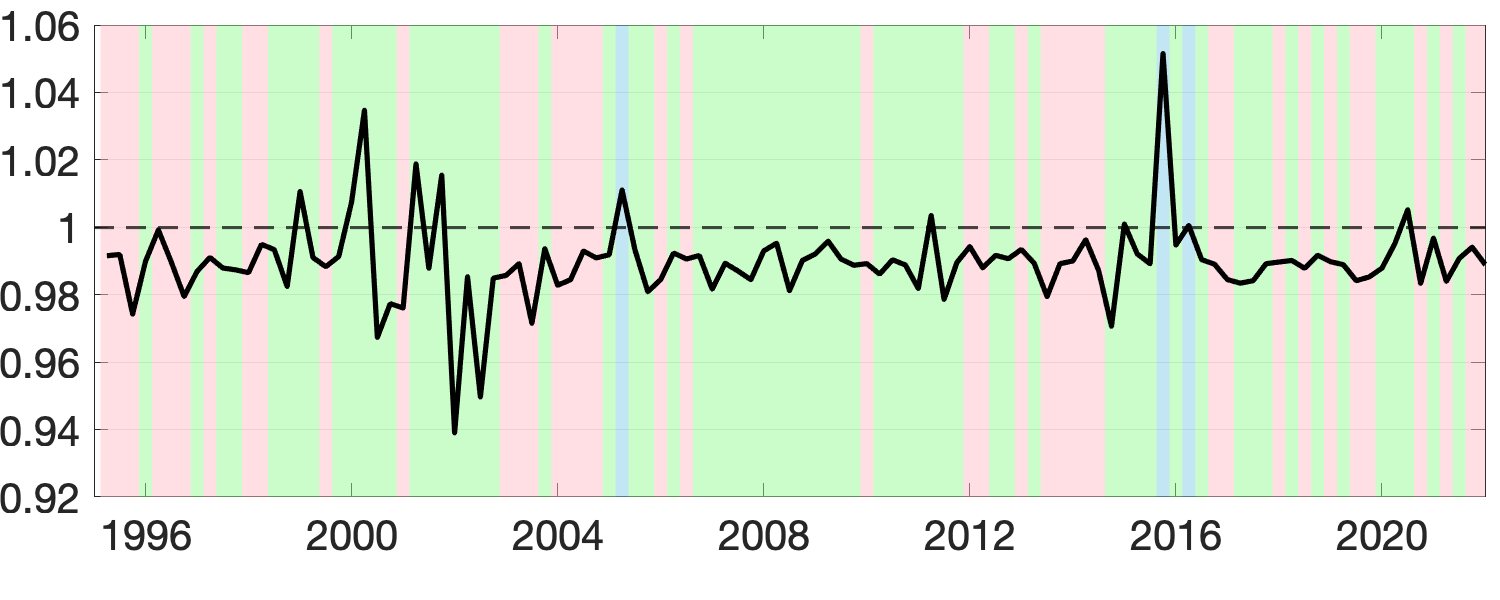}
	\end{center}
\end{figure}

Adding prior information on groups substantially improves the accuracy of density forecasts. Figure \ref{fig_app:app_cpi_relative_LPS} shows the comparison between BGFE-he-cstr and BGFE-he in terms of the difference in LPS. We find the prior information valuable as BGFE-he-cstr outperforms BGFE-he in more than 98\% of the samples. Clearly, the majority of the figure is covered by a pink background, showing that BGFE-he-cstr is typically the best choice. All of these facts demonstrate that adding prior informatio is favorable and essential, especially for density forecasting.

\begin{figure}[h]
	\caption{Relative LPS, BGFE-he-cstr}
	\label{fig_app:app_cpi_relative_LPS} \vspace{-0.5cm}
	\begin{center}
		\includegraphics[scale=0.45]{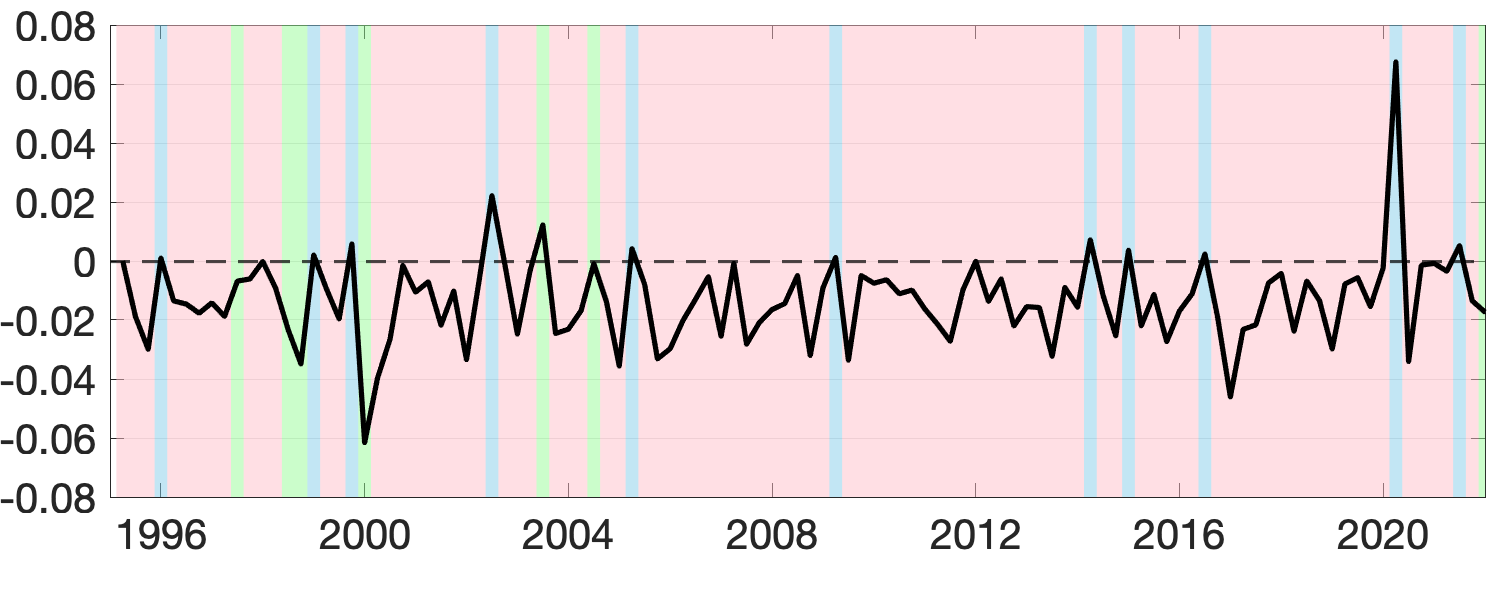}
	\end{center}
\end{figure}


Having specified pairwise constraints across sub-indices, we provide a prior on $G$ that shrinks the group structure toward the eight expenditure categories with equal accuracy for all pairs inside each category. As Theorem \ref{thm:equal_prior_prob} suggested, our prior specification essentially assumes that the prior probability of any two units that come from the same expenditure category being in the same group is equal, and the prior group pattern is actually the expenditure category. We now examine the posterior of group structure to demonstrate how the distribution of $G$ gets updated by data. In order to accomplish this, we construct a posterior similarity matrix (PSM), whose $(i,j)$-element records the posterior probabilities of units $i$ and $j$ being in the same group. For illustrative purposes, we present the results for the last sample, in which we forecast CPI in August 2022. Figure \ref{fig:app_cpi_psm} depicts the PSM generated by BGFE-he-crst for the series in the categories of \textit{Food and Beverages} and \textit{Transportation}. A darker block indicates a higher posterior probability of being in one group. A common pattern emerges: even though some sub-indices are joined together frequently, as shown in the dark diagonal blocks, it is extremely unlikely that all series within the same category belong to the same group. Some series have relatively low or zero probabilities of being grouped together, as suggested by the white and gray off-diagonal blocks. This indicates that the group structure based on official expenditure categories is not optimal, which may result in inaccurate forecasting. Instead, our suggested framework uses information from both prior beliefs and data to reinvent the group pattern, leading to improved forecasting performance.

\begin{figure}[htp]
	\caption{Posterior Similarity Matrices for Selected Categories}
	\label{fig:app_cpi_psm}
	\centering
	\begin{subfigure}[b]{0.49\textwidth}
		\centering
		\caption{Food and beverages (Category 3)}
			\includegraphics[scale= 0.45]{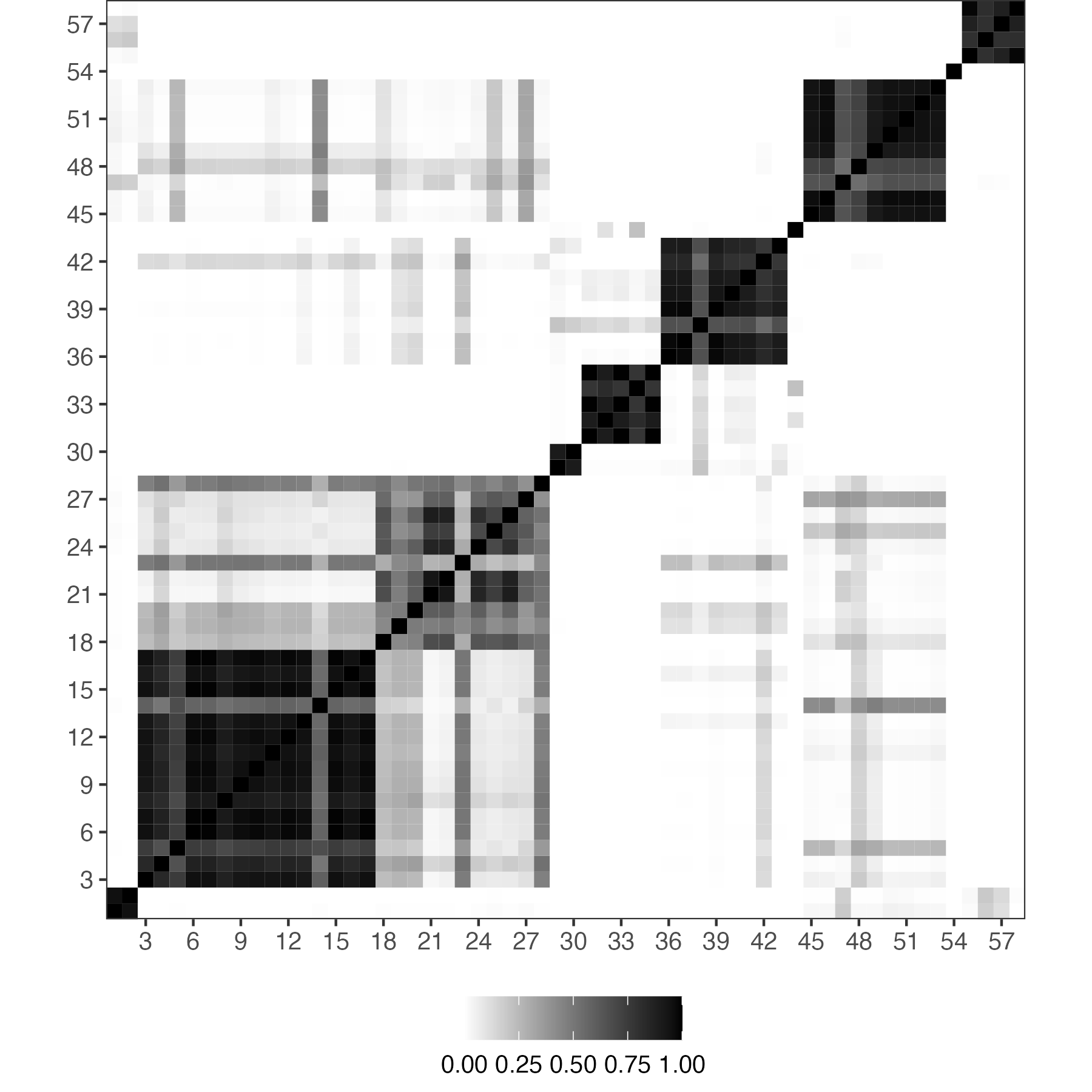}
	\end{subfigure}
	\begin{subfigure}[b]{0.49\textwidth}
		\centering
		\caption{Transportation (Category 7)}
		\includegraphics[scale= 0.45]{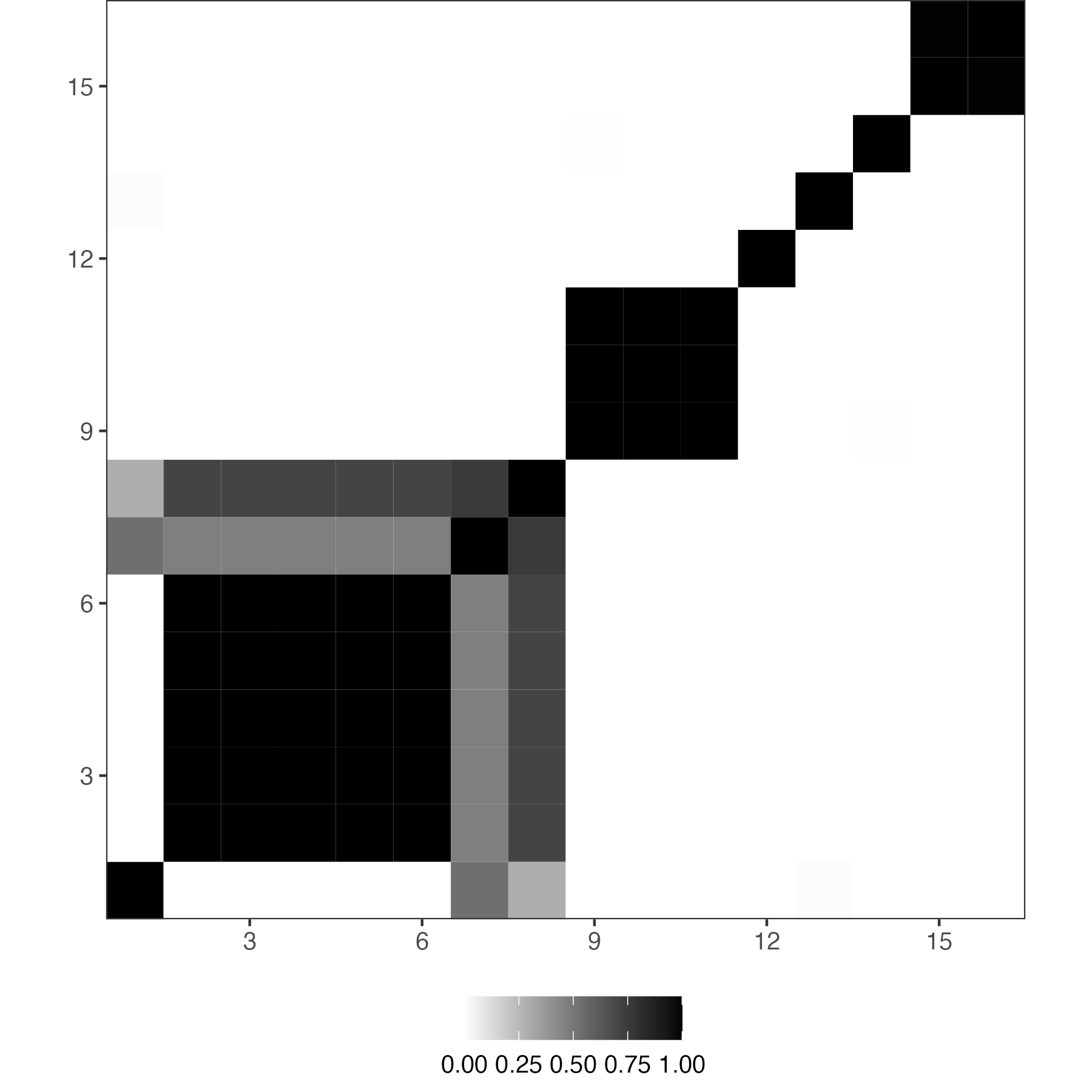}
	\end{subfigure}
	{\footnotesize \raggedright \textit{Notes: This is not the direct output of the algorithm, the ordering of series is changed so that cluster are lying on the diagonals.}\par}
\end{figure}

Finally, we restrict our analysis to the point estimate of group partition, i.e., the single grouping solution, rather than the posterior over the whole universe of partitions.  Figure \ref{fig:app_cpi_grp_202203} depicts the posterior point estimate of $G$ for the last sample ended in August 2022, derived using the approach described in Section \ref{subsec:det_partition}. Eight expenditure categories are divided into twelve groups of varied sizes. Two different forms of groups are generated based on the arrangement of their components. Groups 2, 3, 4, and 5 contain sub-indices from a variety of categories, with no clear dominance. In contrast, the majority of the series in groups 1 and 8, for example,  belong to a certain category. Group 1 may refer to a \textit{Food} group, whereas group 8 is a \textit{Transportation} group. The detailed group 8 components are depicted in Figure \ref{fig:app_cpi_grp_202203_grp12}. There are seven sub-indices from \textit{Transportation}, including car and truck rentals, gasoline (regular, midgrade, and premium), other motor fuels, airline fares, and ship fares, and one series from \textit{Housing} - fuel oil (for residential heating). Clearly, all sub-indices share a common trend and have a close relationship with energy and oil prices, which have increased since the Pandemic. This is an example demonstrating that our proposed algorithm exploits cross-sectional information, not limited to our prior knowledge, and forms meaningful groups for forecasting.

\begin{figure}[h]
	\caption{Posterior Point Estimate of the Group Partition, August 2022}
	\label{fig:app_cpi_grp_202203}
	\centering
	\includegraphics[scale= 0.35]{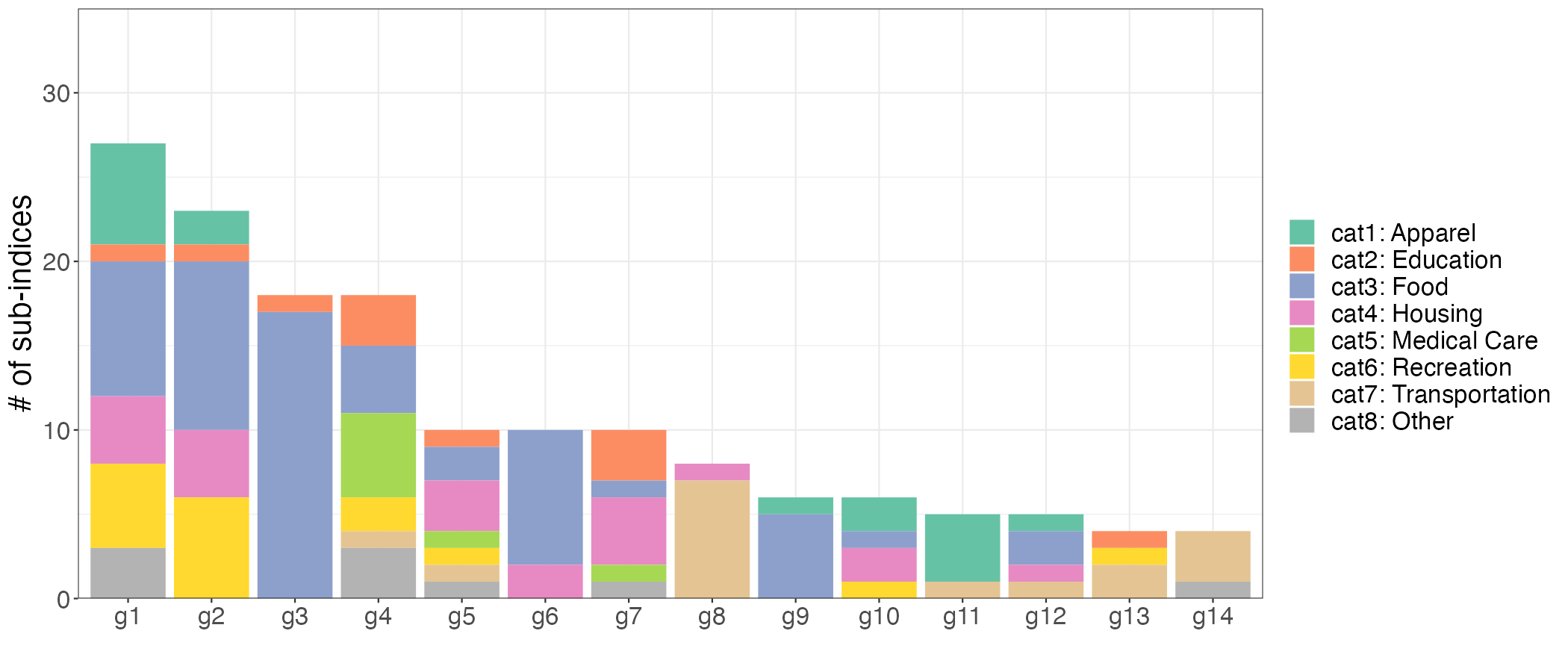}
\end{figure}

\begin{figure}[h]
	\caption{Components in the Group \#9, August 2022}
		\label{fig:app_cpi_grp_202203_grp12}
	\centering
	\includegraphics[scale= 0.35]{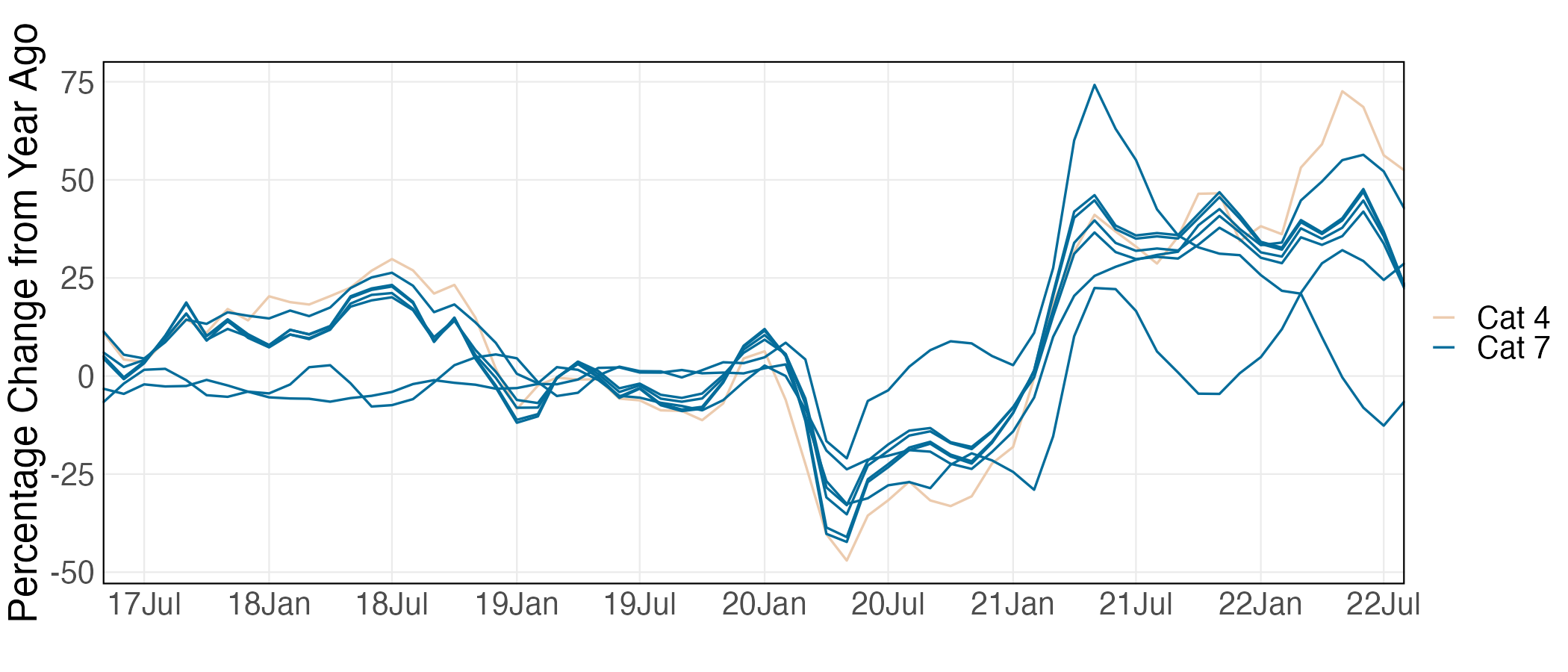}
\end{figure}

\subsubsection{Impact of the Accuracy of Constraints}

We examine how the accuracy of pairwise constraints influences the point estimate of group partitioning. For demonstration purposes, we restrict our analysis to PL constraints solely by setting $\psi^{NL}_{ij} = 0.5$ and changing $\psi^{PL}_{ij}$. We do not select the constant $c$ in the setup since it would balance the impact of the PL restrictions with a different level of accuracy. We set $c$ to 0.5. Again, PL constraints are derived from the official expenditure categories, with the assumption that all units within the same category are positive-linked with equal probability of being in the same group.

\begin{figure}[htp]
	\caption{Impact of the Strength of Constraints}
	\label{fig:app_cpi_strength_ML}
	\centering
	\begin{subfigure}[b]{0.47\textwidth}
		\centering
		\caption{Weaker PL $\left(\psi^{PL}_{ij} = 0.55\right)$}
		\includegraphics[scale= 0.3]{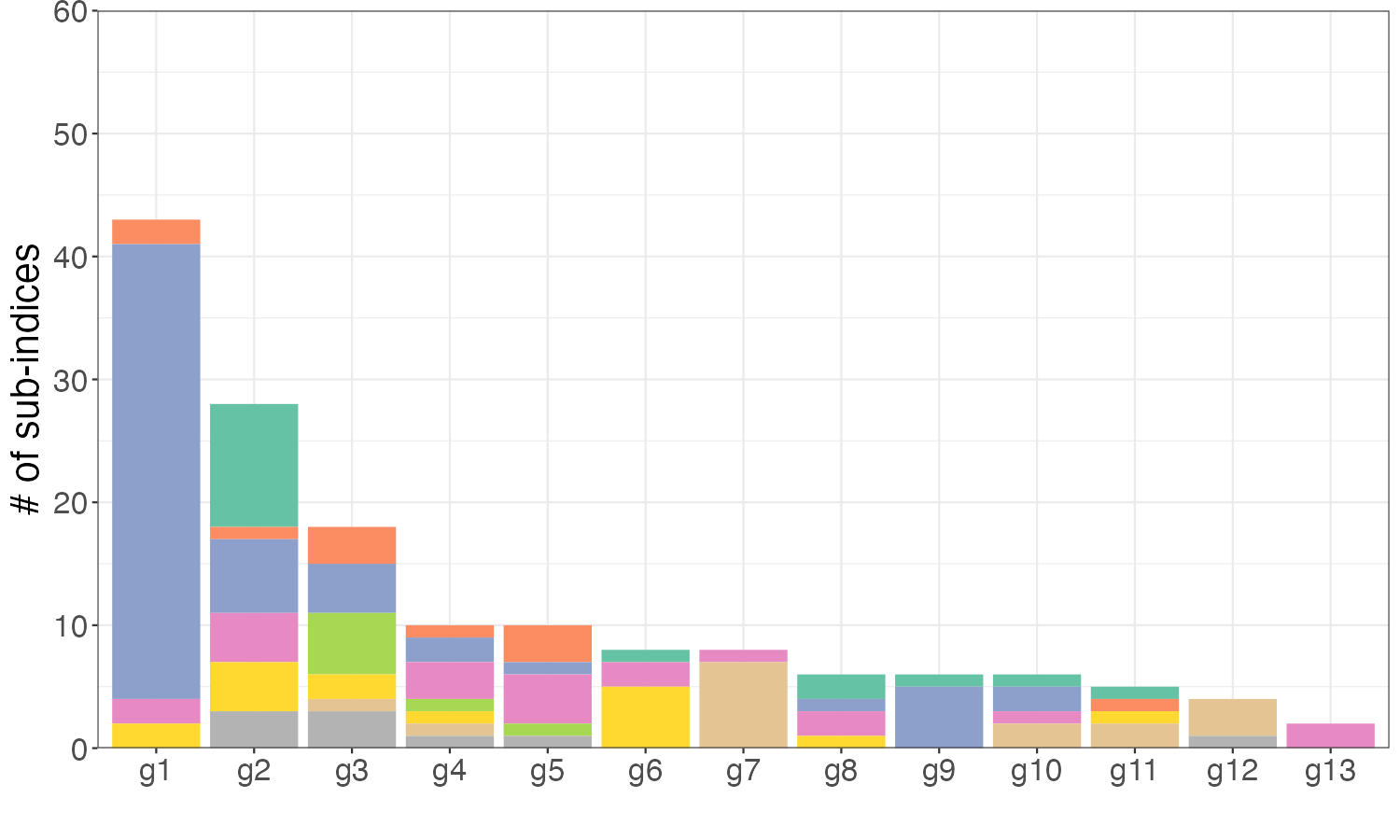}
	\end{subfigure}
	\begin{subfigure}[b]{0.47\textwidth}
		\centering
		\caption{Stronger PL $\left(\psi^{PL}_{ij} = 0.95\right)$}
		\includegraphics[scale= 0.3]{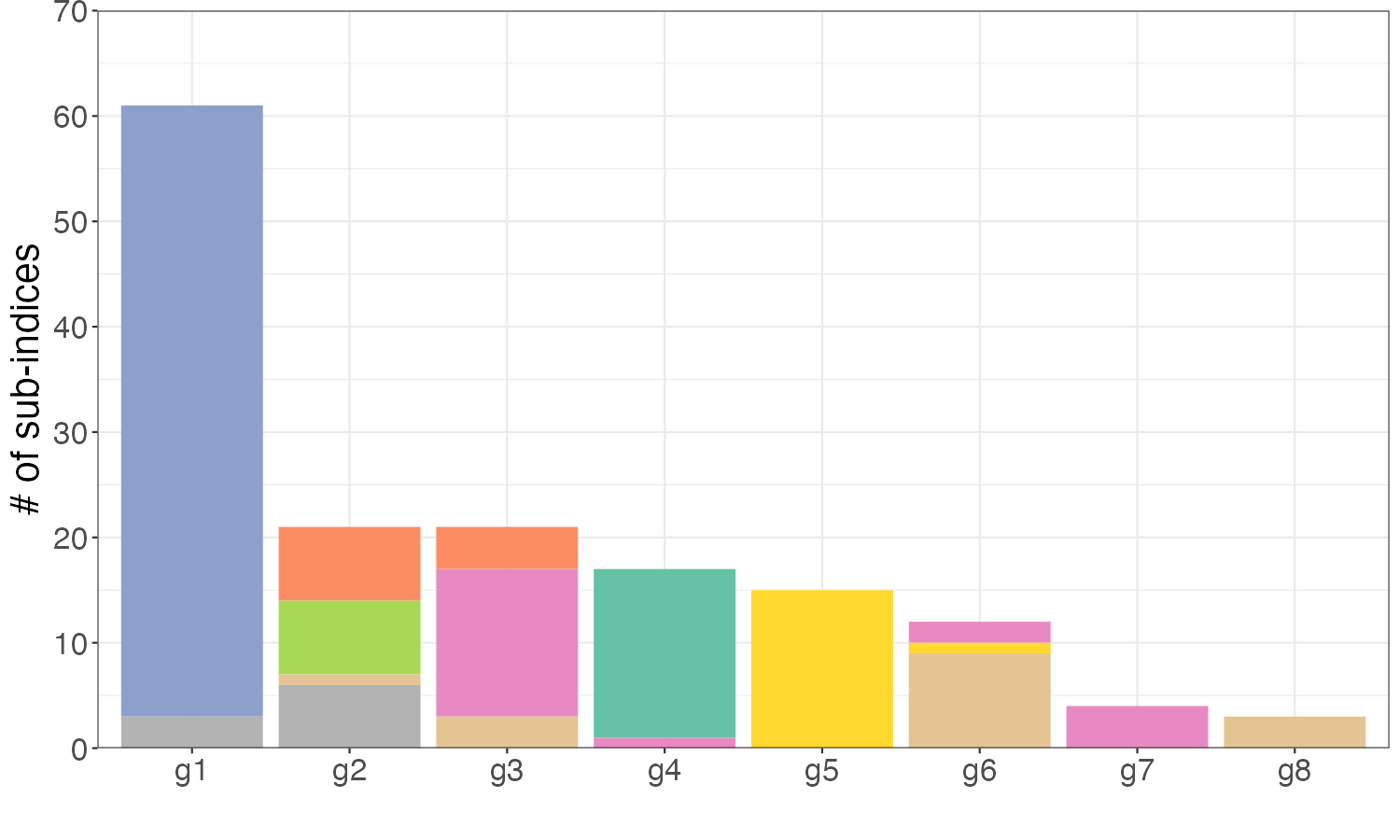}
	\end{subfigure}
\end{figure}

Figure \ref{fig:app_cpi_strength_ML} presents the point estimates of the group structure with two different levels of accuracy. The ``weaker" PL constraints with $\psi^{PL}_{ij} = 0.55$, as shown in panel (a), demonstrate a limited influence of prior knowledge on the group structure. The eleven groups are composites of CPI sub-indices from the various categories, which is diverse from the official spending categories. Panel (b), on the other hand, illustrates the group structure with ``stronger" PL constraints. By setting a high level of accuracy for PL constraints, such as $\psi^{PL}_{ij} = 0.95$, the prior knowledge dominates and pushes the group structure towards the official expenditure categories. As anticipated, panel (b) shows fewer groups, and the majority of CPI sub-indices within each group belong to the same category, bringing the group structure closer to that of the prior.

\subsection{Income and Democracy}

It is well-known in the literature that income per capita is strongly correlated with the level of democracy across countries. This strong empirical regularity is often known as "modernization theory" or Lipset hypothesis \citep{lipset1959}. The theory claims a causal relation: democratic regimes are created and consolidated in affluent societies \citep{lipset1959, przeworski1995, barro1999, epstein2006}. 

In an influential paper, \citet{acemoglu2008} challenge  the casual effect of countries' income on the level of democracy. They argue that it is essential to take into account other factors that affect both economic and political development simultaneously. Their analysis, based on panel data, indicates that the positive relationship between income and democracy disappears when fixed effects are included in the regression. They suggest that this finding is due to historical events, such as the end of feudalism, industrialization, or colonization, which have led countries to follow distinct paths of development. The fixed effects are meant to capture these persistent events. The finding is robust, as it holds for different measures of democracy, various econometric specifications, and additional covariates. Another study by \citet{bonhomme2015} uses a different econometric model but arrives at the same conclusion. Their analysis highlights the presence of diverse group-specific paths of democratization in the data, consistent with the observation that regime types and transitions tend to cluster in time and space \citep{gleditsch2006, ahlquist2012}.

The seminal work of \citet{acemoglu2008} has been subject to critical scrutiny by several recent works, including \citet{moral2012}, \citet{benhabib2013}, and \citet{cervellati2014}. \citet{moral2012} contend that a nonlinear relationship between income and democracy exists, even after accounting for country-specific effects. They show that a positive income-democracy relationship holds only in countries with low levels of income. \citet{benhabib2013} use panel estimation methods that adjust for the censoring of democracy measures at their upper and lower bounds and find that the positive relationship between income and democracy withstands the inclusion of country fixed effects. \citet{cervellati2014} extend the linear estimation framework of \citet{acemoglu2008} and unveil the presence of significant heterogeneity in the income effect on democracy across different subsamples. Specifically, they demonstrate that this effect exhibits an opposite sign for colonies and non-colonies, is substantially different from zero, and is of considerable magnitude. They also argue that the existence of a heterogeneous effect of income suggests that results from a linear framework, such as the finding of a zero effect, may lack robustness since they depend on the composition of the sample.

In this section, we contribute to the literature on the relationship between income and democracy by employing a novel grouped fixed-effects approach. Specifically, we expand on the econometric model proposed by \citet{bonhomme2015} to incorporate group structure not only in time fixed-effects, but also in slope coefficients and the variance of errors. This more complex model allows for a more detailed analysis of the heterogeneous effects of income on democracy across different groups of countries. To identify these groups, we incorporate prior knowledge about the latent group structure by clustering countries based on either geographic location or initial levels of democracy score. By leveraging this information, we are able to identify a moderate number of groups, each of which exhibits a distinct path to democracy.

Our results indicate that the effect of income on democracy is highly varied across countries, a finding which is consistent with previous research by \citet{cervellati2014}. Furthermore, we find that the positive cumulative effect of income on democracy exists in groups of countries with a medium or relative low level of income, in line with the findings of \citet{moral2012}. However, the effect could be relatively small for some groups, suggesting that other factors beyond income may also play important roles in democratic development.


\subsubsection{Model Specification and Data}

\noindent{\bf Model:} To accommodate richer assumptions on models for the real-world applications, we extend the baseline model in Chapter 1 either by adding common regressors and allowing for time-variation in the fixed-effects. Time-variation are essential to this analysis as they capture highly persistent historical shocks. Following \citet{bonhomme2015}, we introduce group-specific time patterns of heterogeneity $\alpha_{g_i t}$ and consider the following two specifications:

\begin{itemize}
	

	\item [SP1:] \emph{Time-varying GFE + grouped slope coefficients}
	\begin{align} \label{eq:emp_app_2_sp2}
		y_{i t} = \alpha_{g_{i}t} + \rho_{g_i} y_{i t-1} + \beta_{g_i} x_{i t-1} +  \varepsilon_{i t}, \; \varepsilon_{i t} \sim N(0, \sigma_{g_i}^2),
	\end{align}
	
	\item [SP2:] \emph{Time-varying GFE}
	\begin{align} \label{eq:emp_app_2_sp1}
		y_{i t} = \alpha_{g_{i}t} + \rho y_{i t-1} + \beta x_{i t-1} +  \varepsilon_{i t}, \; \varepsilon_{i t} \sim N(0, \sigma_{g_i}^2),
	\end{align}

	%
	
\end{itemize}
where $y_{it}$ is the democracy score of country $i$ at time $t$. The lagged value of democracy score $y_{it-1}$ is included to capture persistence in democracy and also potentially mean-reverting dynamics (i.e., the tendency of the democracy score to return to some equilibrium value for the country). The coefficient of main interest is $\beta_{g_i}$ and reflects the effect of the lagged value of $\log$ income per capita $x_{it-1}$ on democracy. In addition, $\alpha_{g_{i} t}$ denote a set of group-specific time fixed-effects; $\varepsilon_{i t}$ is an error term with grouped variance $\sigma_{g_i}^2$, capturing additional transitory shocks to democracy and other omitted factors. We use the conjugate prior for all parameters, see details in Appendix \ref{appendix:prior}.

Specification 2 in (\ref{eq:emp_app_2_sp1}) nests the linear dynamic panel data model in BM as a special case. If we assume homoskedasticity, it is the Equation (22) in BM. This specification enables us to reproduce BM's results and provide fresh insight into their framework. Specification 1 in (\ref{eq:emp_app_2_sp2}), on the other hand, generalizes specification 2 by introducing group-dependent slope coefficients. As we shall demonstrate in the following section, specification 1 yields a more refined group structure and provides a clearer view of the income effects.

\vspace{0.5cm}

\noindent{\bf Data:} 
We use the Freedom House (FH) Political Rights Index as a benchmark for measuring democracy. To standardize the index, we normalize it between 0 and 1, with higher scores indicating higher levels of democracy. FH assesses a country's political rights based on a checklist of questions, such as the presence of free and fair elections, the role of elected officials, the existence of competitive parties or other political groupings, the power of the opposition, and the extent of self-government or participation by minority groups. We measure countries' income using the logarithm of GDP per capita, which is adjusted for purchasing power parity (PPP) in 1996 prices, using data from the Penn World Tables 6.1 \citep{heston2002}. Details on the data can be found in Section 1 of \citet{acemoglu2008}, and all data in this section are from the replication files of \citet{bonhomme2015}.\footnote{ \url{https://www.dropbox.com/s/ssjabvc2hxa5791/Bonhomme_Manresa_codes.zip?dl=0}} Our analysis is based on a five-year panel dataset that includes all independent countries since the postwar period, with observations taken every fifth year from 1970 to 2000. We chose this period for comparability with previous studies. The final dataset consists of a balanced panel of 89 countries.



\vspace{0.5cm}

\noindent{\bf Prior group structure:}  We propose two prior grouping strategies as specified below and assume all units within the same prior group are presumed to be positive-linked, while units from different prior groups are believed to be negative-linked. 


\begin{enumerate}[(i)]
	
	\item Given the countries available in the dataset, we form six groups according to their geographic locations:\footnote{The regions are assigned by the \href{https://www.eiu.com/}{Economist Intelligence Unit}, and may slightly differ from conventional classifications.} (1) North America; (2) Europe; (3) Latin America and the Caribbean; (4) Asia and Australasia; (5) Sub-Saharan Africa; (6)  Middle East and North Africa. We refer this prior to \textit{geo-prior}.
	
	\item Alternately, countries could be categorized according to their initial level of democracy in year 1970. As the Freedom House Index has six possible values, we cluster countries into three primary groups with a reasonable number of countries in each: (1) low democracy, $y_{i,1970} = $ 0 or 0.166; (2) medium democracy, $y_{i,1970} = $ 0.333, 0.5, or 0.667; (3) high democracy, $y_{i,1970} = $ 0.833 or 1. We refer this prior to \textit{dem-prior}.
	
\end{enumerate}

Figure \ref{fig:app_dem_2prior} presents the world maps with countries colored differently according to their respective groups. The panel on the left illustrates the geographic groups, while the panel on the right depicts the democratic groups. All gray nations/regions are excluded from the dataset. We concentrate primarily on the first pre-grouping strategy, as it needs no country-specific knowledge beyond geographic information. We then compare the results using different pre-grouping strategies in Section \ref{sec:emp_result_dem_compare}.

\begin{figure}[h]
	
	\caption{Specifications of Prior Grouping}
	\label{fig:app_dem_2prior}
	\centering
	
	\begin{subfigure}{0.48\textwidth}
		\centering
		\caption{Geographic Location (Geo-Prior)}
		\includegraphics[scale= 0.45]{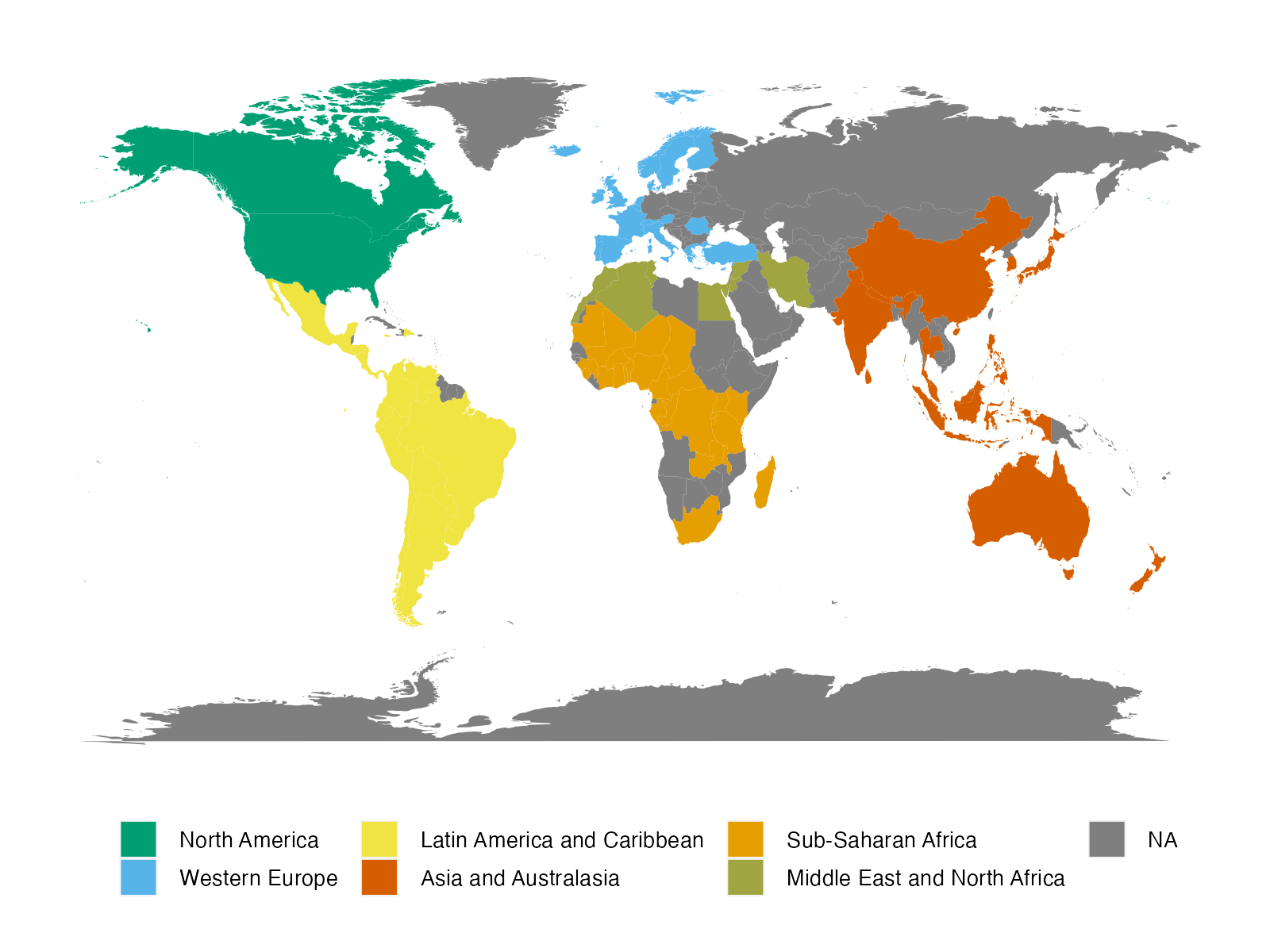}
	\end{subfigure}
	~~
	\begin{subfigure}{0.48\textwidth}
		\caption{Initial Level of Democracy (Dem-Prior)}
		\includegraphics[scale= 0.45]{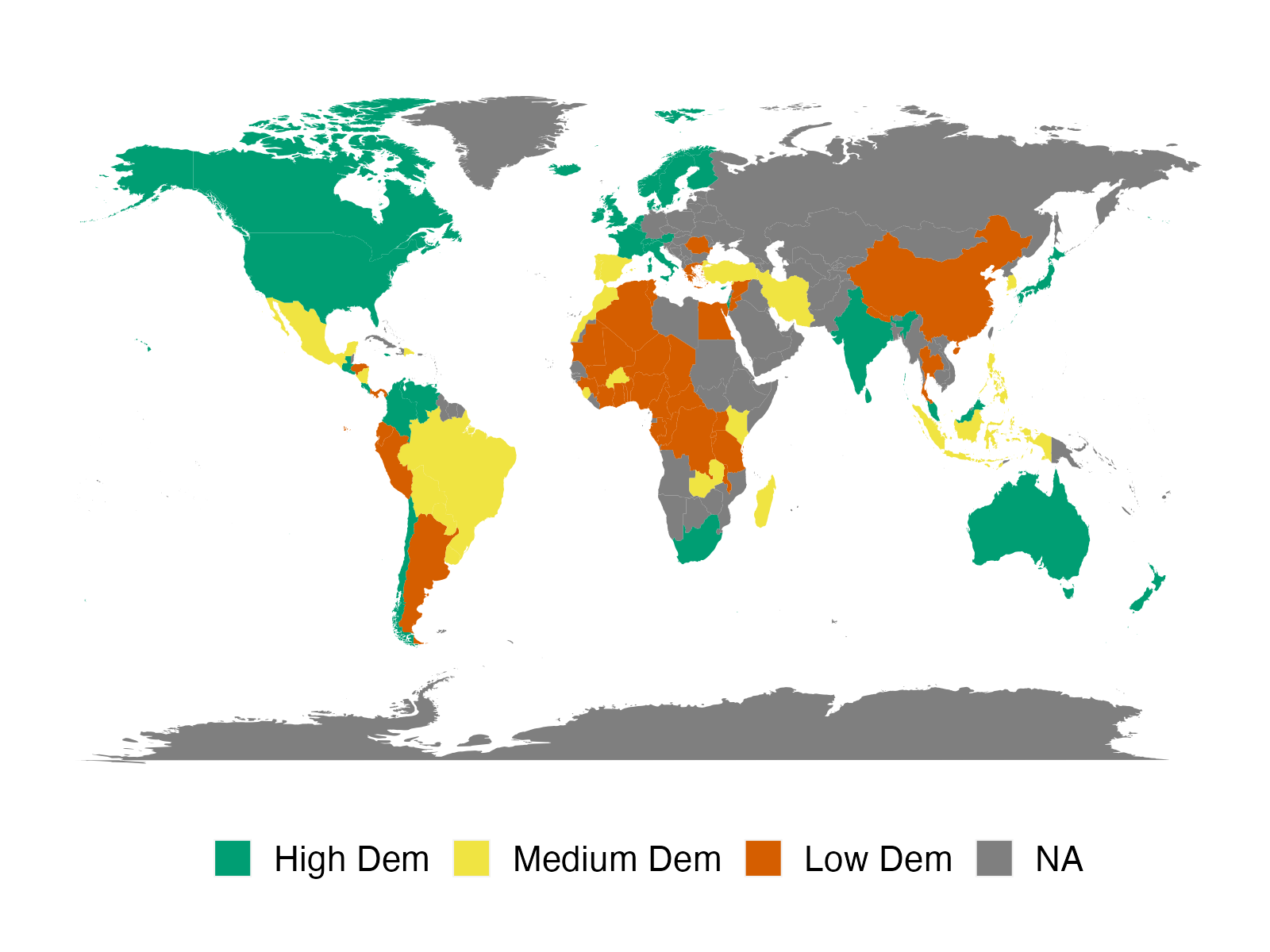}
	\end{subfigure}
\end{figure}

\subsubsection{Results}

\noindent{\bf Specification 1:} We begin with specification 1 where group-specific slope coefficients are allowed and new findings emerge.  Table \ref{tab:emp_app_dem_sp2_ngroup} presents the posterior probability of the number of groups utilizing various estimators. BGFE-ho creates more than 5 groups in all posterior draws. Intriguingly, accounting for heteroskedasticity drastically reduces the number of groups, with BGFE-he identifying four groups. Adding pairwise constraints based on geographic information increases the number of groups to five, whereas six groups are expected in the prior.

\begin{table}[h]
\begin{center}
	\caption{Probability for Number of Groups, Geo-Prior {\color{blue}}}
	\label{tab:emp_app_dem_sp2_ngroup}
	\begin{tabular}{l | c c}
		\toprule
		& BGFE-he-cstr & BGFE-he  \\
		\midrule
		$Pr (K < 4)$ & 0.000 & 0.000 \\
		$Pr (K = 4)$ &  0.000 & {\bf 1.000} \\
		$Pr (K = 5)$ &  {\bf 1.000}& 0.000 \\
		$Pr (K > 5)$ & 0.000 & 0.000  \\
		\bottomrule 
	\end{tabular}
	
\end{center}
\end{table}

The marginal data density (MDD) of each estimator in Table \ref{tab:emp_app_dem_sp2_mdd} provides some insight on different models. Among all the estimators, the BGFE-ho estimator has the lowest MDD; it is even lower than that of specification 1. BGFE-he-cstr and BGFE-he, on the other hand, benefit from the introduction of group-specific slope coefficients, since both achieve substantially greater MDD than in specification 1. BGFE-he-cstr has the highest MDD since the pairwise constraints give direction on grouping and identify the ideal group structure, which BGFE-he cannot uncover without our prior knowledge.

\begin{table}[h]
\begin{center}
	\caption{Marginal Data Density, Geo-Prior}
	\label{tab:emp_app_dem_sp2_mdd}
	\begin{tabular}{ c | c c c}
		\toprule
		& BGFE-he-cstr & BGFE-he & BGFE-ho \\
		\midrule
		SP1 & 544.324 &  501.904 & 327.077 \\
		\midrule
		SP2 & 413.476 & 381.218 & 368.918 \\
		\bottomrule 
	\end{tabular}
\end{center}
\end{table}

We concentrate on the BGFE-he-cstr estimator and use the approach outlined in Section \ref{subsec:det_partition} to identify the unique group partitioning $\widehat{G}$. The left panel of Figure \ref{fig:app_dem_result_sp2} presents the world map colored by $\widehat{G}$, while the right panel present the group-specific averages of democracy index over time. The estimated group structure $\widehat{G}$ features five distinct groups which we refer to as the ``high-democracy", ``low-democracy", ``flawed-democracy", ``late-transition" and ``progressive-transition" group, respectively. With the exception of the ``flawed-democracy" and  ``progressive-transition" group, the group-specific averages of the democracy index are comparable to those in BM for all other groups.  BGFE-he-cstr does not identify the "early transition" group in comparison to BM but instead produces two new groups. Group 3 ("flawed-democracy") comprises primarily of relatively democratic but not the most democratic nations, including India, Sri Lanka, and Venezuela, among others. Group 5 ("progressive transition") contains 30 countries that have had a steady expansion of democracy, including Argentina, Greece, and Panama. Consequently, by incorporating group-specific slope coefficients, we recover a more refined group structure than that of BM.

\begin{figure}[h]
\caption{Posterior Point Estimate of Group Partitioning and Average Democracy}
\label{fig:app_dem_result_sp2}
\centering
\begin{subfigure}[b]{0.45\textwidth}
	\centering
	\includegraphics[scale= 0.4]{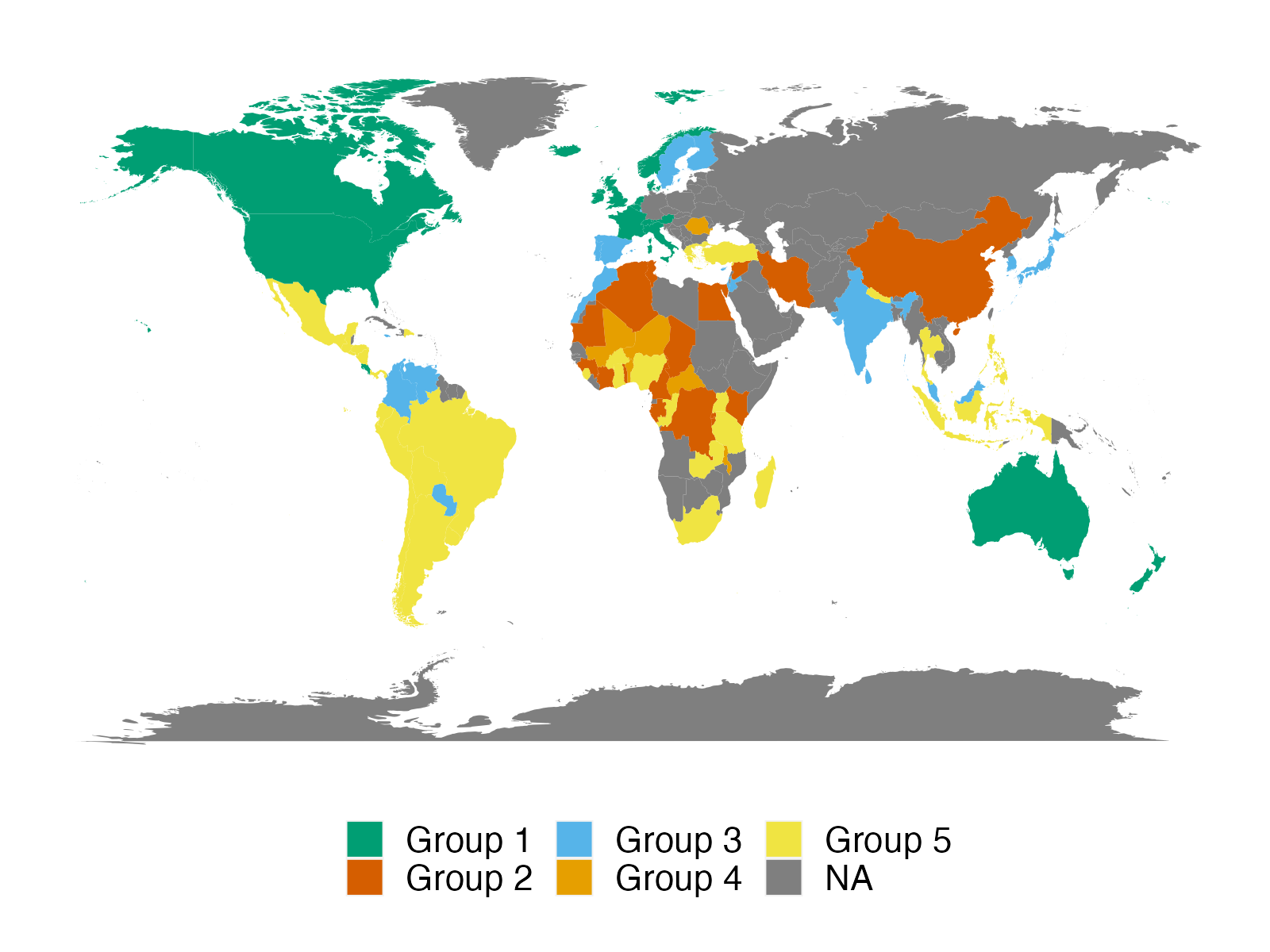}
\end{subfigure}
\begin{subfigure}[b]{0.45\textwidth}
	\centering
	\includegraphics[scale= 0.35]{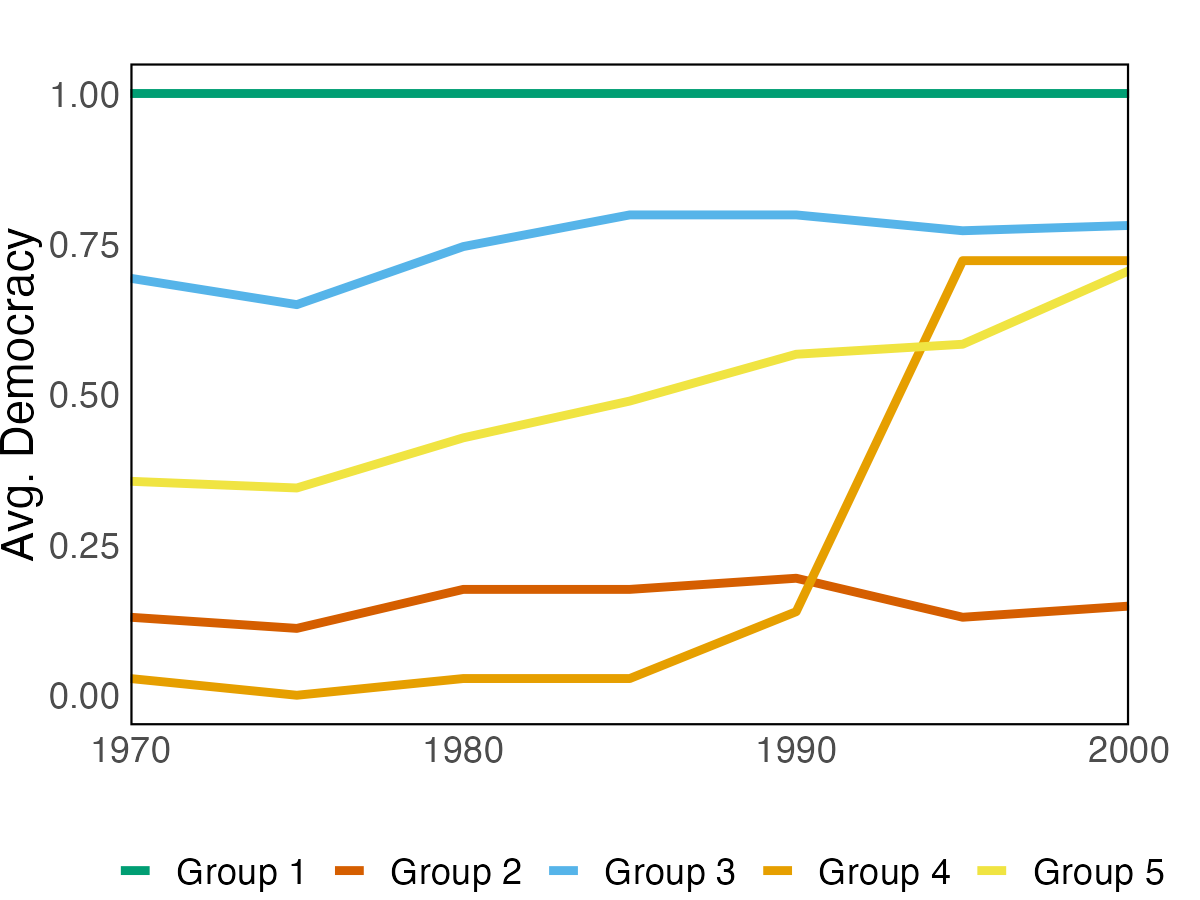}
\end{subfigure}
\end{figure}

Table \ref{tab:app_dem_est_sp2} presents the posterior mean and 90\% credible set for each coefficient across all groups, with $G$ fixed at the point estimate $\widehat{G}$. The key feature of using the specification 1 is that we are able to see distinct (cumulative) income effects across groups as group-specific coefficients are allowed. 

The effect of income on democracy is negligible for group 1 ("full-democracy") and group 4 ("late-transition") as the posterior means of $\beta$ are close to 0 and the associated credible intervals for $\hat{\beta}$ contain 0. Group 1, which we refer to as the "full-democracy" group, mostly contains high-income, high-democracy countries. It includes the United States, Canada, UK, most of European countries, Australia, and New Zealand, but also Costa Rica and Uruguay. These country kept their democracy index at the highest level throughout the sample, demonstrating that income has no effect on democracy. Group 4 is referred to as the "late-transition" group, which consists of Benin, Central African Republic, Mali, Malawi, Niger, and Romania. The transition to democracy for countries in group 4 was primarily driven by historical events in the 90s: Romania began a transition towards democracy after the 1989 Revolution; all other countries involved in the third wave of democratization in sub-Saharan Africa beginning in 1989. The impact of historical events is primarily captured by time fixed-effects, as the credible intervals for $\hat{\beta}$ and $\hat{\rho}$ well cover zero. 

Group 2 ("low-democracy") and group 3 ("flawed-democracy") are two groups that have stable Freedom House scores.  Lagged democracy is highly significant and indicates that there is a considerable degree of persistence in democracy. Log income per capita is also significant and illustrates the well-documented positive relationship between income and democracy. Though statistically significant, the effect of income is quantitatively small. For example, the coefficient of 0.055 for the group 3 implies that a 10 percent increase in GDP per capita is associated with an increase in the Freedom House score of 0.0055, which is very small. Group 2 includes low-democratic countries, China, Singapore, Iran, and a fraction of African countries, whose Freedom House scores remain relatively low throughout the sample. Countries in group 3, however, have Freedom House score stay in the relatively high level. The group covers countries with almost democratic but minor flaws in certain aspects, including India, Japan, South Korea, Finland, Sweden, and Portugal, among others. The cumulative income effects for these two groups are, however, different - it is negligible for group 2 (0.079) and modest for group 3 (0.244).  Group 5 ("progressive-transition"), on the other hand, experiences a continuous increase in Freedom House score from a 0.35 in 1970 to 0.7 in 2020. It has the largest positive income coefficient, although the cumulative income effect is modest (0.156).

\begin{table}[htp]
\begin{center}
	\caption{Coefficient Estimates across Groups, Geo-Prior}
	\label{tab:app_dem_est_sp2}
	\resizebox{\textwidth}{!}{%
		\begin{tabular}{l  p{1.1cm}  p{2.7cm}  p{1.1cm}  p{2.7cm} p{1.1cm}  p{2.7cm}  p{1.1cm}  p{2.4cm} }
			\toprule
			& \multicolumn{2}{c}{Lagged democracy ($\rho$)}  & \multicolumn{2}{c}{Lagged Income ($\beta$)} & \multicolumn{2}{c}{Income Effect ($\beta/(1-\rho)$)} & \multicolumn{2}{c}{Error variance ($\sigma^2$)}  \\
			\cmidrule(lr){2-3} \cmidrule(lr){4-5} \cmidrule(lr){6-7}  \cmidrule(lr){8-9}  \noalign{\smallskip}
			& Coef. & Cred. Set & Coef. & Cred. Set  & Coef. & Cred. Set & Coef. & Cred. Set\\ 
			\midrule
			Group 1 (16) &  0.058 & [-0.263, 0.360] &  0.000 & [-0.012, 0.012] &  0.000 & [-0.013, 0.013] & 0.001 & [0.001, 0.001]  \\ 
			Group 2 (18) &  0.484 & [ 0.354, 0.606] &  0.041 & [ 0.019, 0.062] &  0.079 & [ 0.044, 0.117] & 0.010 & [0.008, 0.011]  \\
			Group 3 (19) & 0.775 & [ 0.703, 0.850] & 0.055 & [ 0.031, 0.078] &  0.249 & [ 0.143, 0.352] & 0.013 & [0.011, 0.016]   \\
			Group 4 (6) &  -0.178 & [-0.468, 0.115] &  -0.025 & [-0.066, 0.017] &  -0.020 & [-0.054, 0.015] & 0.008 & [0.005, 0.011]  \\
			Group 5 (30) &  0.206 & [ 0.091, 0.310] &  0.125 & [ 0.090, 0.163] &  0.157 & [ 0.118, 0.199] & 0.057 & [0.048, 0.066]  \\
			\midrule
			Pooled OLS &  0.667 & [ 0.614, 0.717] &  0.081 & [ 0.064, 0.099] &  0.244 & [ 0.207, 0.281] & 0.039 & [0.035, 0.043]  \\ 
			\bottomrule 
		\end{tabular}
	}
\end{center}
\end{table}

Figure \ref{fig:app_dem_result_income_sp2} depicts the historical average log income for each group of countries between 1970 and 2000. Except for the high income group (group 1) and low income group (group 4), all other three groups with medium or relatively low levels of income reveal positive cumulative income effect, as indicated in Table \ref{tab:app_dem_est_sp2}. This finding is generally in line with the results of \citet{moral2012}, who observe a positive effect in low-income nations. Notice that, their definition of low-income country is quite board, encompassing countries with a GDP per capita below the 80th percentile of the empirical cross-sectional density.. As a result, beside six counties in the group 4 that had the lowest income on average, other countries fit within this definition and confirm that the positive income effect is not prevalent in high-income countries.

\begin{figure}[h]
\caption{Average Log Income by Groups}
\label{fig:app_dem_result_income_sp2}
\begin{center}
	\includegraphics[scale=0.45]{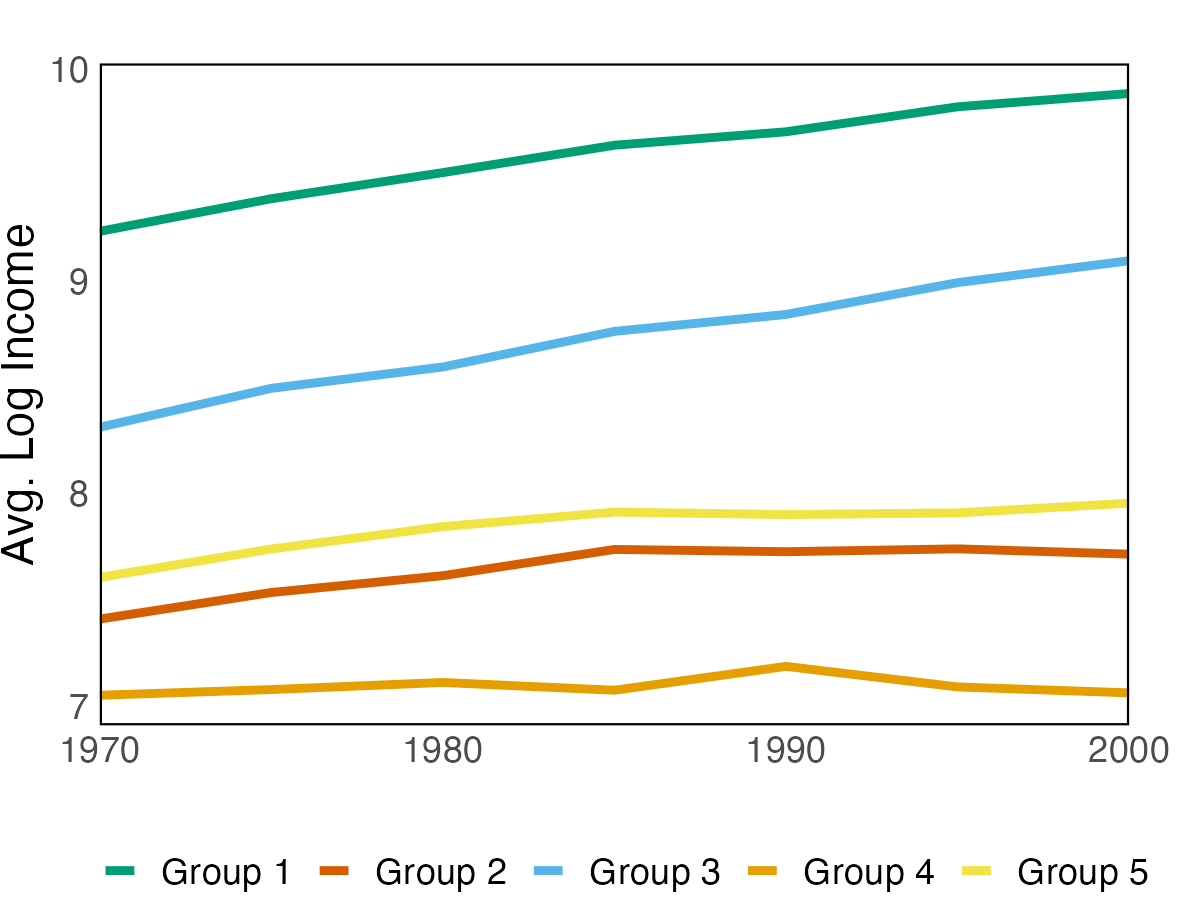}
\end{center}
\end{figure}


\noindent{\bf Specification 2:} The results for specification 2 are reported in Appendix \ref{appendix:res_demo}. In short, the results are comparable to the key findings in BM. BGFE-ho in specification 1 is identical to the main model in BM; it produces eight groups, which is consistent with the upper bound on the number of groups in BM based on BIC. BGFE-he-cstr, on the other hand, is more preferable and has the highest marginal data density as shown in Table \ref{tab:emp_app_dem_sp2_mdd}. The point estimate of group partitioning based on BGFE-he-cstr consists of four groups that all have the similar pattern as BM's group structure. This justifies BM's subjective choice of four groups.  Regarding the estimated coefficients, there is moderate persistence and a positive effect of income on democracy, but the cumulative effect of income is quantitatively small: $\beta / (1-\rho) = 0.08$.

\subsubsection{Impacts of Different Pre-Grouping Strategies}  \label{sec:emp_result_dem_compare}
All results presented thus far are based on pairwise constraints derived from spatial information. We now implement the alternative pre-grouping strategy based on the initial level of democracy. We stick with the BGFE-he-cstr estimator under specification 1.


Different pre-grouping strategies yields different estimates of group patterns. We consider the BGFE-he estimator with three prior group structures: geo-prior, dem-prior, and no prior knowledge, with point estimates of the group partition presented in panel (\subref{fig:app2_geo_prior}), (\subref{fig:app2_dem_prior}), and (\subref{fig:app2_no_prior}) of Figure \ref{fig:app_dem_2prior_post}, respectively. Geo-prior and dem-prior produce comparable group structures, however certain nations are assigned to distinct groups.. They are encircled in the black dashed rectangle, including Portugal, Spain, Romania, Mali, Niger, Central African Republic, Benin, Malawi, and Jordan. Another country is South Korea. As depicted in panel (\subref{fig:app2_no_prior}), without any prior knowledge of groups, the group pattern is quite different, particularly for countries in Asia, Africa, and Latin America. 

The difference in group patterns result in discrepancies in MDD, which are listed in Table \ref{tab:emp_app_dem_diff_prior_mdd}. The BGFE-he estimator with geo-prior has the highest MDD in both specifications, whereas the dem-prior is only informative in specification 1 in comparison to the BGFE-he estimator without prior knowledge.

\begin{table}[h]
	\begin{center}
		\caption{Marginal Data Density, Different Priors}
		\label{tab:emp_app_dem_diff_prior_mdd}
		\begin{tabular}{ c | c c c }
			\toprule
			& geo-prior & dem-prior  & no prior\\
			\midrule
			SP1 & 544.324 &  527.517 & 501.904 \\
			SP2 & 413.476 & 309.369  & 381.218 \\
			\bottomrule 
		\end{tabular}
		
	\end{center}
\end{table}

\begin{figure}[hp]
	
	\caption{Posterior Point Estimate of Group Structure}
	\label{fig:app_dem_2prior_post}
	\centering
	
	\begin{subfigure}{0.48\textwidth}
		\centering
		\caption{Geo-Prior} \vspace{-0.3cm}
		\label{fig:app2_geo_prior}
		\includegraphics[scale= 0.45]{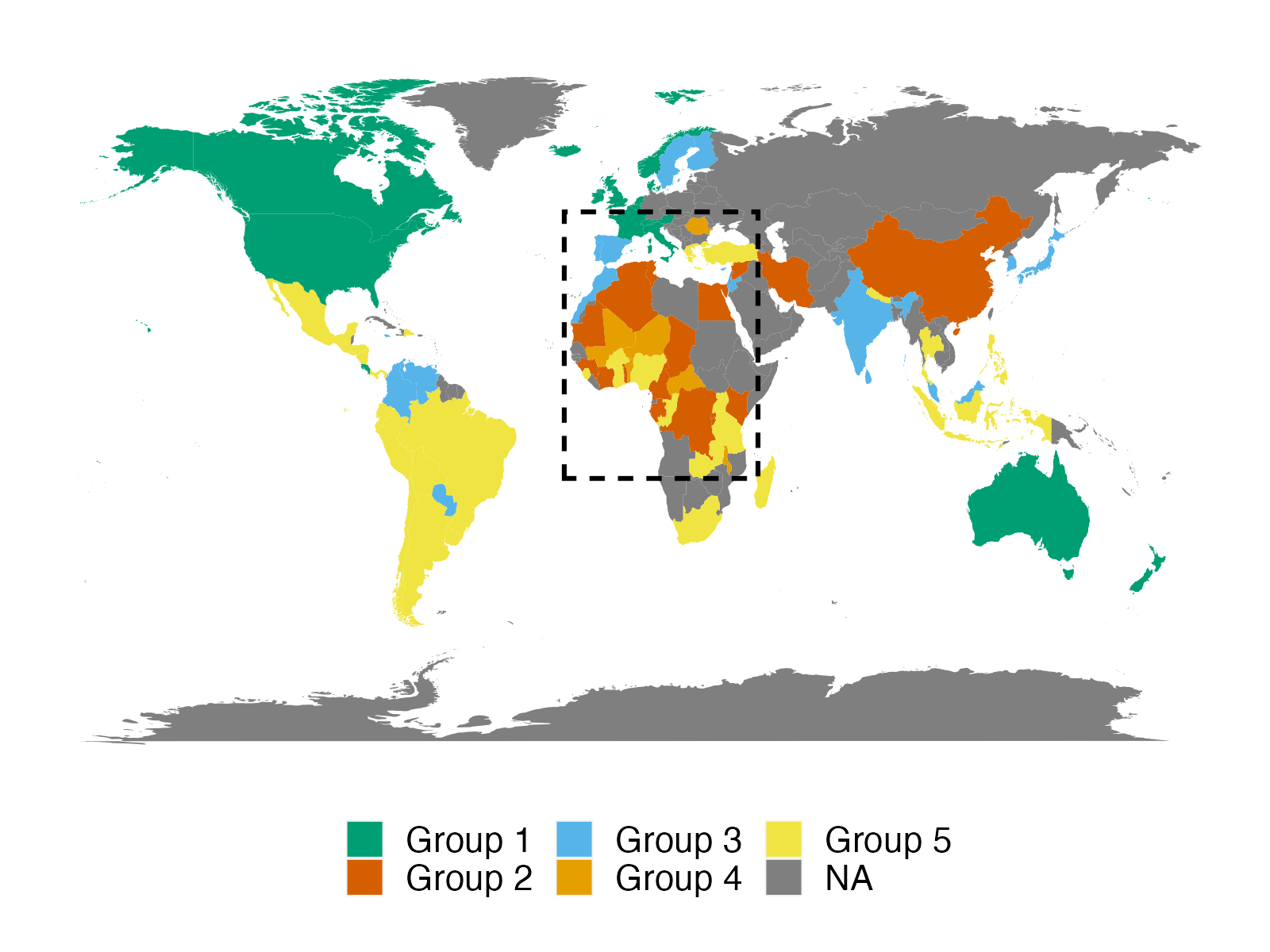}
	\end{subfigure}
	~~  \vspace{0.5cm}
	\begin{subfigure}{0.48\textwidth}
		\caption{Dem-Prior}  \vspace{-0.3cm}
		\label{fig:app2_dem_prior}
		\includegraphics[scale= 0.45]{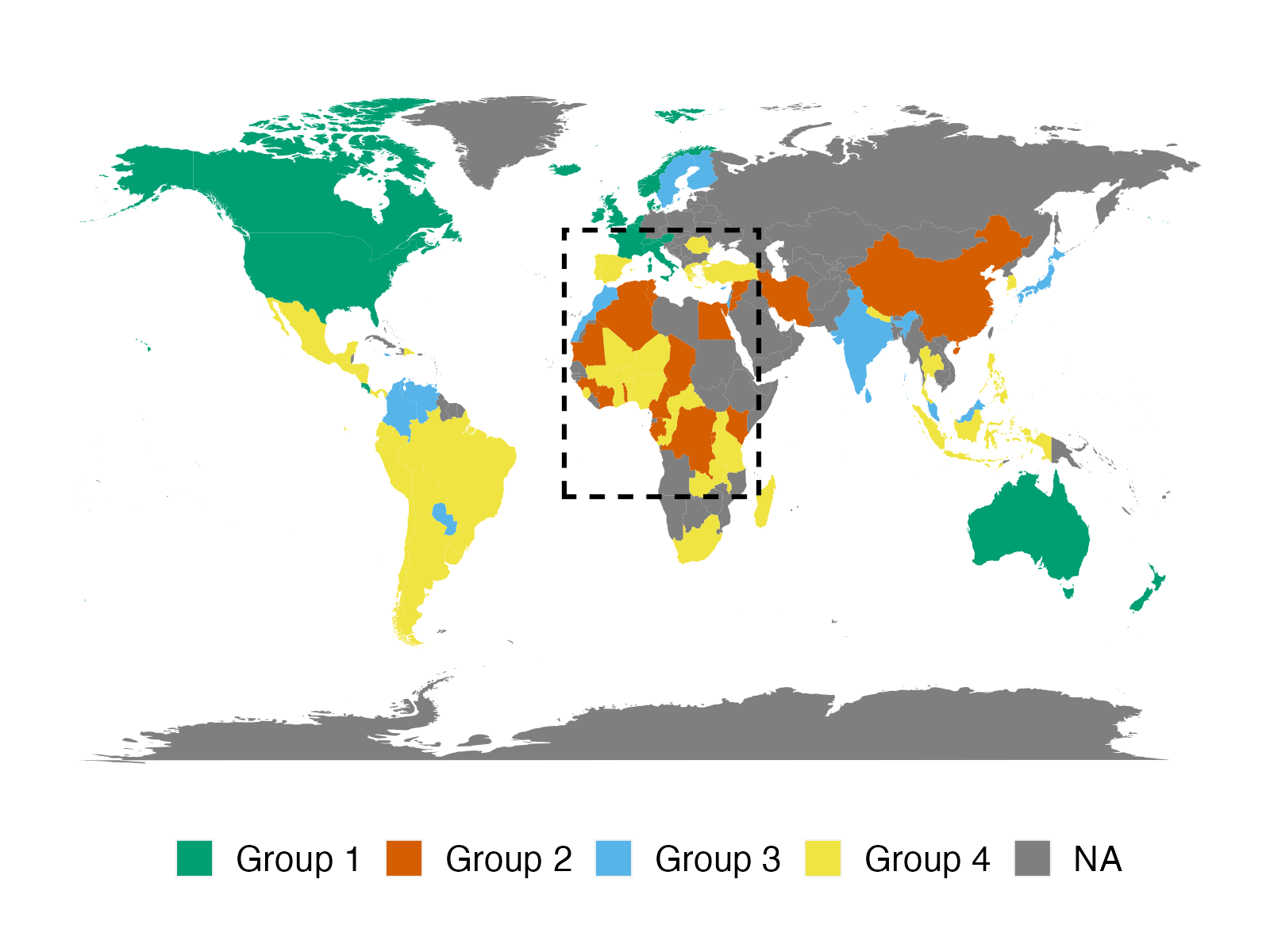}
	\end{subfigure}
	~~
	\begin{subfigure}{0.48\textwidth}
		\caption{No Prior Knowledge}  \vspace{-0.3cm}
		\label{fig:app2_no_prior}
		\includegraphics[scale= 0.45]{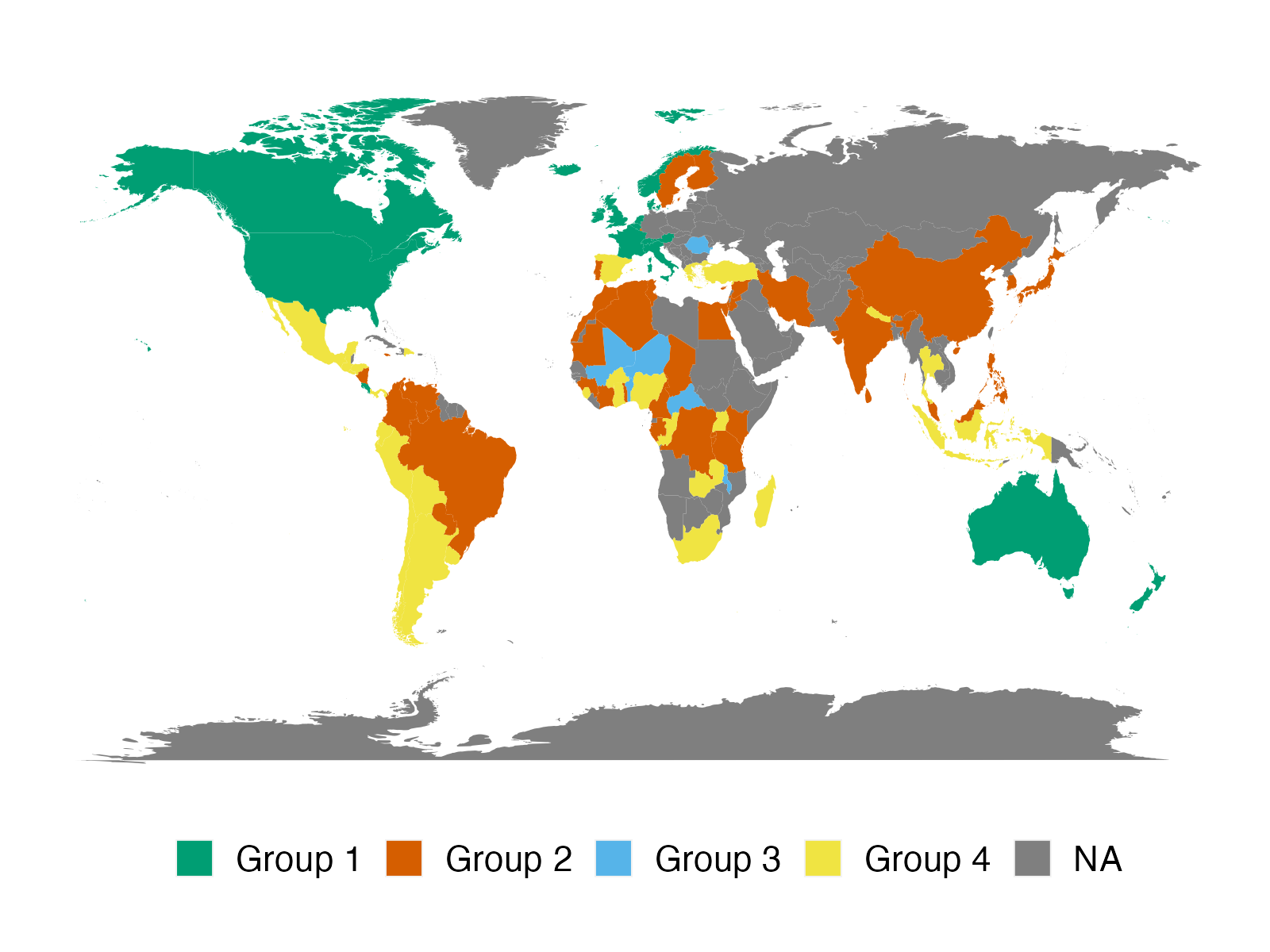}
	\end{subfigure}
\end{figure}

Using the initial level of democracy as prior knowledge results in four groups, as indicated in the panel (\subref{fig:app2_dem_prior}). The dem-prior has two major impacts on the group structure comparing with the geo-prior. It combines the ``late-transition" group (group 4 in geo-prior) with the ``progressive-transition" group (group 5 in geo-prior) to form a bigger and boarder ``progressive-transition" group. Additionally, Portugal and Spain are no longer categorized as ``flawed-democracy" countries, but rather as ``progressive-transition" group in a boarder sense. In terms of the posterior estimates, as shown in Table \ref{tab:app_dem_est_sp2_demo}, we observe similar value for the first three groups because they are merely subject to minor changes in the group structure. However, as ``late-transition" and ``progressive-transition" groups are agglomerated together under the dem-prior, the income effects of countries in the new group becomes larger, forcing some countries to exhibit strong and positive effects even though they are not under the geo-prior.


\begin{table}[htp]
	\begin{center}
		\caption{Coefficient Estimates across Groups, Dem-Prior}
		\label{tab:app_dem_est_sp2_demo}
		\resizebox{\textwidth}{!}{%
			
			\begin{tabular}{l  p{1.1cm}  p{2.7cm}  p{1.1cm}  p{2.7cm} p{1.1cm}  p{2.7cm}  p{1.1cm}  p{2.4cm} }
				\toprule
				& \multicolumn{2}{c}{Lagged democracy ($\rho$)}  & \multicolumn{2}{c}{Lagged Income ($\beta$)} & \multicolumn{2}{c}{Income Effect ($\beta/(1-\rho)$)} & \multicolumn{2}{c}{Error variance ($\sigma^2$)}  \\
				\cmidrule(lr){2-3} \cmidrule(lr){4-5} \cmidrule(lr){6-7}  \cmidrule(lr){8-9}  \noalign{\smallskip}
				& Coef. & Cred. Set & Coef. & Cred. Set  & Coef. & Cred. Set & Coef. & Cred. Set\\ 
				\midrule
				Group 1 (16) &  0.062 & [-0.262, 0.369] &  0.000 & [-0.012, 0.012] &  0.000 & [-0.014, 0.013] & 0.001 & [0.001, 0.001]  \\ 
				Group 2 (19) &  0.513 & [ 0.395, 0.640] &  0.043 & [ 0.022, 0.064] &  0.087 & [ 0.047, 0.126] & 0.010 & [0.008, 0.012]  \\
				Group 3 (15) & 0.802 & [ 0.704, 0.904] & 0.040 & [ 0.016, 0.064] &  0.223 & [ 0.079, 0.337] & 0.012 & [0.009, 0.015]   \\
				Group 4 (39) &  0.302 & [ 0.205, 0.388] &  0.120 & [ 0.092, 0.151] &  0.172 &  [ 0.136, 0.210] & 0.054 & [0.047, 0.062]  \\
				\midrule
				Pooled OLS &  0.667 & [ 0.614, 0.717] &  0.081 & [ 0.064, 0.099] &  0.244 & [ 0.207, 0.281] & 0.039 & [0.035, 0.043]  \\ 
				\bottomrule 
			\end{tabular}
		}
	\end{center}
\end{table}


\section{Extensions} \label{subsec:ext}


Within the domain of panel data models, the proposed constrained-based BGFE framework can be extended in multiple directions to allow for more subtle group structures or more covariates. In addition, the DP prior with soft pairwise constraints also applies to other related topics and models, such as clustering problems, heterogeneous treatment effects, and panel VARs.

\subsection{Subtle Group Structure}

Through the Dirichlet process defines a prior that possesses the clustering property and is flexible enough to incorporate pairwise constraints, the group structure itself is elementary. Aside from our prior belief on the group, the group structure, which is introduced in all $\alpha_{i}$ and $\sigma_i^2$, is entirely governed by the stick-breaking process defined in Equation (\ref{main_text_pi}). The stick length $\xi_k$, on which we have a prior, is independent of any regressors or time. Consequently, each unit is associated with a single group, and the membership remains constant across time.

To create an even more flexible and richer group structure, we provide insight into three possible extensions, each of which requires a set of more distinctive nonparametric priors. (1) overlapping group and (2) time-varying group and (3) dependent group.

Overlapping group structures allow for multi-dimensional grouping. This is a natural extension without having to greatly modify the proposed DP prior. Following \citet{cheng2019}, each of $\alpha_i$'s and $\sigma_i^2$ may have its own group structure and a separate Dirichlet process is specified to each of them. As a result, units simultaneously belong to multiple groups based on the heterogeneous effects among regressors or cross-sectional heteroskedasticity. 

Time-varying group structures allow the membership of the group to change over time. We could replace the DP by variants of the hierarchical Dirichlet process \citep{Teh2006} to achieve this feature. In short, the hierarchical Dirichlet process (HDP), a nonparametric Bayesian approach to clustering grouped data, is now the foundation of the prior. The time dimension naturally divides the panel data into $T$ groups, and a Dirichlet process is assumed for each group, with all Dirichlet processes having the same base distribution, which is distributed according to a global base distribution. The HDP allows each group to have its own cluster, but most importantly, these clusters are shared across groups. This lays the groundwork for time-varying group structures, as it assumes that the number of clusters remains constant over time, while cluster memberships are subject to change. Variants of the HPD are then proposed to capture the time-persistence in group structures, including dynamic HDP \citep{ren2008} and sticky HDP \citep{fox2008, fox2011}. A closely related area in the frequentists' methods is to identify structure breaks in parameters with grouped patterns, see \citet{okui2021, lumsdaine2022}.

Dependent group structures allow the prior group probability to rely directly on a collection of characteristics. The dependence is introduced through a modification of the stick-breaking representation for DPs, where the group probabilities vary with the characteristics. \citet{rodriguez2011} introduced the probit-stick breaking (PSB) process where the Beta random variables are replaced by normally distributed random variables transformed using the standard normal CDF. The PSB is defined by,
\begin{align}
	\pi_k \left(w_{i}\right) = \Phi\left(\zeta_k\left(w_{i}\right)\right) \prod_{j<k} \left[1-\Phi\left(\zeta_j\left(w_{i}\right)\right)\right],
\end{align}
where stochastic function $\zeta_k$ is drawn from Gaussian process $\zeta_k \sim G P\left(0, V_k\right)$ for $k=1,2, \cdots$ and $w_{i}$ is the set of characteristics that are informative to the latent group. Other forms of dependence are also available, see \citet{quintana2022} for a comprehensive review. A caveat of this approach is that analysis of group structure is confined to $w_{i}$ observed by the researcher. The approach requires researchers to know possible key characteristics, be able to observe them and ensure they are informative. In many cases, however, these characteristics might be hard to justify by researchers.

\subsection{Beyond Panel Data Models}

Although we concentrate on panel data model, our framework of the DP prior with soft pairwise constraints applies to other models where the group structure are crucial.


\vspace{0.5cm}

\noindent{\bf Gaussian Mixture Model}

\noindent If we ignore covariates and focus exclusively on group membership, we essentially face a classical clustering problem with an infinite-dimensional mixture model. A typical probabilistic model is the infinite Gaussian mixture model \citep{rasmussen1999}, where the data itself is assumed to be drawn from a mixture of Gaussian components
\begin{align}
	y_i \sim \sum_{k=1}^{\infty} \pi_k N (\mu_k, \Sigma_k),
\end{align}
where $\pi_k$ are the mixture weights. With soft pairwise constraints, observations are clustered in accordance with prior belief.

\vspace{0.5cm}

\noindent{\bf Heterogeneous Treatment Effects}

Following the potential outcomes framework of \citet{rubin1974}, we posit the existence of potential outcomes $y_i(1)$ and $y_i(0)$ corresponding respectively to the response the $i$ th subject would have experienced with and without the treatment, and define the treatment effect at $x$ as
\begin{align} \label{eq:CATE}
	\tau(x) = \mathbb{E}\left[y_i(1) - y_i(0) | x_i = x\right].
\end{align}

Previous works on Bayesian analysis for the treatment effects include \citet{chib2000, chib2002, chib2007a, chib2007b, heckman2014}, etc. Another strand of literature propose to estimate (\ref{eq:CATE}) by using several machine learning algorithms \citep{hill2011, athey2016, wager2018}. These methods are built on the idea that researchers find the subsamples across which the effect of a treatment differs out of all possible subsamples on the basis of the values of $x_i$. Instead of trying to discover valid subsets of the data, \citet{shiraito2016} directly models the outcome as a function of the treatment and pre-treatment covariates $x_i$ and estimate of the distribution of conditional average treatment effects (CATE) across units by employing the Dirichlet process.
\begin{align}
	y_{i} = \tau_{g_i} T_{i} + \gamma'_{g_i} x_{i}+ \varepsilon_{i}, \quad \varepsilon_{i} \stackrel{\text { iid}}{\sim} \mathcal{N}\left(0, \sigma_{g_i}^{2}\right),
\end{align}
where $T$ is the binary treatment variable. Our approaches fit in this methods by adding side information on the treatment groups into the prior.

\vspace{0.5cm}

\noindent{\bf Panel VARs}

\noindent Panel VARs \citep[and references therein.]{holtz1988, canova2013} has been widely used in macroeconomic analysis and policy evaluations to capture the interdependency across sectors, markets, and countries. Nevertheless, the large dimension of panel VARs typically makes the curse of dimensionality a severe problem. \citet{billio2019} propose nonparametric Bayesian priors that cluster the VAR coefficients and induce group-level shrinkage. Our paradigm with the DP prior with soft pairwise constraints is applicable to their method and injects prior information on groups into the underlying Granger causal networks.

Panel VARs have the same structure as VAR models, in the sense that all variables are assumed to be endogenous and interdependent, but a cross-sectional dimension is added to the representation. Thus, let $Y_t$ be the stacked version of $y_{i t}$, the vector of $J$ variables for each unit $i=1, \ldots, N$, i.e., $Y_t=\left(y_{1 t}^{\prime}, y_{2 t}^{\prime}, \ldots y_{N t}^{\prime}\right)^{\prime}$. Then a panel VAR is
\begin{align}
	Y_{t}=A_{0} + A_1  Y_{t-1} + A_2  Y_{t-2} + ...+ A_p  Y_{t-p} + u_{t}, \quad i=1, \ldots, N,
\end{align}
where $u_{t}$ is a $J \times 1$ vector of idiosyncratic errors and $A_{0}$ and $A_j$ are $NJ \times NJ$ matrices of coefficients. 

The main feature of \citet{billio2019} is to specify a prior that blends the DP prior with Lasso prior for each of $A_{0}$ and $A_j$, such that the VAR coefficients are either shrunk toward 0 or clustered at multiple non-zero locations. Our proposed DP prior with soft pairwise constraints, in the meantime, fit into their framework by replacing the original DP prior and permitting richer structure within each coefficient matrix. As the nonzero coefficients form Granger causal networks, equipping with soft pairwise constraints may result in a more plausible network by taking researchers' expertise into account. 




\section{Concluding Remarks} \label{sec:conclusion}

This paper proposes a Bayesian framework for estimating and forecasting in panel data models when prior group knowledge is available and informative for the group pattern. We include prior knowledge in the form of soft pairwise constraints into the Dirichlet process prior. Then, an intuitive and coherent prior is presented. The constrained grouped estimator proposed examines both heteroskedasticity and heterogeneous slope coefficients to endogenously reveal group structure. Our framework immediately estimates the number of groups as opposed to relying on ex-post model selection, and the structure of pairwise restrictions circumvents the computational difficulties and limitations that afflict conventional approaches. In addition, when utilizing small-variance asymptotics, the suggested Gibbs sampler with pairwise constraint contains a clustering procedure comparable to that of the constrained \textit{KMeans} algorithm. In Monte Carlo simulations, we demonstrate that constrained Bayesian grouped estimators outperform conventional estimators even in the presence of incorrect prior knowledge. Our empirical application to forecasting sub-indices of CPI inflation rates demonstrates that incorporating prior knowledge on the latent group structure yields more accurate density predictions. The better forecasting performance is mostly attributable to the key characteristics: nonparametric Bayesian prior and grouped cross-sectional variance. The method proposed in this paper is applicable beyond forecasting. In a second application, we revisit the relationship between a country's income and its democratic transition, where estimation of heterogeneous parameters is the object of interest. We recover a reasonable cluster pattern with a moderate number of groups and identify heterogeneous income effects on democracy.

The current work raises exciting questions for future research. It is desirable to investigate overlapping group structures, in which a unit might belong to many groups. This would allow us to increase the flexibility of a panel data model, potentially enhancing its predictive performance. Second, the assumption that an individual cannot change its group identity for the entire sample time can be amended, resulting in a specification that is even more flexible. Thirdly, our method is applicable to other econometric models, such as panel VARs with latent group structures in macro series.

\clearpage
\setstretch{1}
\bibliography{ref_cbg_new}

\begin{thebibliography}{127}
\newcommand{\enquote}[1]{``#1''}
\expandafter\ifx\csname natexlab\endcsname\relax\def\natexlab#1{#1}\fi

\bibitem[\protect\citeauthoryear{Acemoglu, Johnson, Robinson, and
  Yared}{Acemoglu et~al.}{2008}]{acemoglu2008}
\textsc{Acemoglu, D., S.~Johnson, J.~A. Robinson, and P.~Yared} (2008):
  \enquote{Income and Democracy,} \emph{American Economic Review}, 98, 808--42.

\bibitem[\protect\citeauthoryear{Aguilar and Boot}{Aguilar and
  Boot}{2022}]{aguilar2022}
\textsc{Aguilar, J. and T.~Boot} (2022): \enquote{Grouped Heterogeneity in
  Linear Panel Data Models with Heterogeneous Error Variances,} \emph{Available
  at SSRN 4031841}.

\bibitem[\protect\citeauthoryear{Ahlquist and Wibbels}{Ahlquist and
  Wibbels}{2012}]{ahlquist2012}
\textsc{Ahlquist, J.~S. and E.~Wibbels} (2012): \enquote{Riding the Wave: World
  Trade and Factor-Based Models of Democratization,} \emph{American Journal of
  Political Science}, 56, 447--464.

\bibitem[\protect\citeauthoryear{Aldous}{Aldous}{1985}]{aldous1985}
\textsc{Aldous, D.~J.} (1985): \enquote{Exchangeability and Related Topics,} in
  \emph{{\'E}cole d'{\'E}t{\'e} de Probabilit{\'e}s de Saint-Flour
  XIII—1983}, Springer, 1--198.

\bibitem[\protect\citeauthoryear{Anderson and Hsiao}{Anderson and
  Hsiao}{1982}]{anderson1982}
\textsc{Anderson, T.~W. and C.~Hsiao} (1982): \enquote{Formulation and
  Estimation of Dynamic Models using Panel Data,} \emph{Journal of
  Econometrics}, 18, 47--82.

\bibitem[\protect\citeauthoryear{Ando and Bai}{Ando and Bai}{2016}]{ando2016}
\textsc{Ando, T. and J.~Bai} (2016): \enquote{Panel Data Models with Grouped
  Factor Structure under Unknown Group Membership,} \emph{Journal of Applied
  Econometrics}, 31, 163--191.

\bibitem[\protect\citeauthoryear{Antoniak}{Antoniak}{1974}]{antoniak1974}
\textsc{Antoniak, C.~E.} (1974): \enquote{Mixtures of Dirichlet Processes with
  Applications to Bayesian Nonparametric Problems,} \emph{The Annals of
  Statistics}, 2, 1152--1174.

\bibitem[\protect\citeauthoryear{Ascolani, Lijoi, Rebaudo, and
  Zanella}{Ascolani et~al.}{2022}]{ascolani2022}
\textsc{Ascolani, F., A.~Lijoi, G.~Rebaudo, and G.~Zanella} (2022):
  \enquote{Clustering Consistency with Dirichlet Process Mixtures,} \emph{arXiv
  preprint arXiv:2205.12924}.

\bibitem[\protect\citeauthoryear{Athey and Imbens}{Athey and
  Imbens}{2016}]{athey2016}
\textsc{Athey, S. and G.~Imbens} (2016): \enquote{Recursive Partitioning for
  Heterogeneous Causal Effects,} \emph{Proceedings of the National Academy of
  Sciences}, 113, 7353--7360.

\bibitem[\protect\citeauthoryear{Atkeson, Ohanian et~al.}{Atkeson
  et~al.}{2001}]{atkeson2001}
\textsc{Atkeson, A., L.~E. Ohanian, et~al.} (2001): \enquote{Are Phillips
  Curves Useful for Forecasting Inflation?} \emph{Federal Reserve Bank of
  Minneapolis Quarterly Review}, 25, 2--11.

\bibitem[\protect\citeauthoryear{Barro}{Barro}{1999}]{barro1999}
\textsc{Barro, R.~J.} (1999): \enquote{Determinants of Democracy,}
  \emph{Journal of Political economy}, 107, S158--S183.

\bibitem[\protect\citeauthoryear{Basu, Banerjee, and Mooney}{Basu
  et~al.}{2002}]{basu2002}
\textsc{Basu, S., A.~Banerjee, and R.~Mooney} (2002): \enquote{Semi-Supervised
  Clustering by Seeding,} in \emph{Proceedings of 19th International Conference
  on Machine Learning}, Citeseer.

\bibitem[\protect\citeauthoryear{Basu, Banerjee, and Mooney}{Basu
  et~al.}{2004{\natexlab{a}}}]{basu2004}
\textsc{Basu, S., A.~Banerjee, and R.~J. Mooney} (2004{\natexlab{a}}):
  \enquote{Active Semi-Supervision for Pairwise Constrained Clustering,} in
  \emph{Proceedings of the 2004 SIAM international conference on data mining},
  SIAM, 333--344.

\bibitem[\protect\citeauthoryear{Basu, Bilenko, and Mooney}{Basu
  et~al.}{2004{\natexlab{b}}}]{basu2004mrf}
\textsc{Basu, S., M.~Bilenko, and R.~J. Mooney} (2004{\natexlab{b}}):
  \enquote{A Probabilistic Framework for Semi-Supervised Clustering,} in
  \emph{Proceedings of the Tenth ACM SIGKDD International Conference on
  Knowledge Discovery and Data Mining}, 59--68.

\bibitem[\protect\citeauthoryear{Benhabib, Corvalan, and Spiegel}{Benhabib
  et~al.}{2013}]{benhabib2013}
\textsc{Benhabib, J., A.~Corvalan, and M.~M. Spiegel} (2013): \enquote{Income
  and Democracy: Evidence from Nonlinear Estimations,} \emph{Economics
  Letters}, 118, 489--492.

\bibitem[\protect\citeauthoryear{Bernanke}{Bernanke}{2007}]{bernanke2007}
\textsc{Bernanke, B.~S.} (2007): \enquote{Inflation Expectations and Inflation
  Forecasting,} in \emph{Speech at the Monetary Economics Workshop of the
  National Bureau of Economic Research Summer Institute, Cambridge,
  Massachusetts}, vol.~10.

\bibitem[\protect\citeauthoryear{Bilenko, Basu, and Mooney}{Bilenko
  et~al.}{2004}]{bilenko2004}
\textsc{Bilenko, M., S.~Basu, and R.~J. Mooney} (2004): \enquote{Integrating
  Constraints and Metric Learning in Semi-Supervised Clustering,} in
  \emph{Proceedings of the 21st International Conference on Machine Learning},
  81--88.

\bibitem[\protect\citeauthoryear{Billio, Casarin, and Rossini}{Billio
  et~al.}{2019}]{billio2019}
\textsc{Billio, M., R.~Casarin, and L.~Rossini} (2019): \enquote{Bayesian
  Nonparametric Sparse VAR Models,} \emph{Journal of Econometrics}, 212,
  97--115.

\bibitem[\protect\citeauthoryear{Blackwell and MacQueen}{Blackwell and
  MacQueen}{1973}]{blackwell1973}
\textsc{Blackwell, D. and J.~B. MacQueen} (1973): \enquote{Ferguson
  Distributions via P{\'o}lya Urn Schemes,} \emph{The Annals of Statistics}, 1,
  353--355.

\bibitem[\protect\citeauthoryear{Bonhomme, Lamadon, and Manresa}{Bonhomme
  et~al.}{2022}]{bonhomme2022}
\textsc{Bonhomme, S., T.~Lamadon, and E.~Manresa} (2022): \enquote{Discretizing
  Unobserved Heterogeneity,} \emph{Econometrica}, 90, 625--643.

\bibitem[\protect\citeauthoryear{Bonhomme and Manresa}{Bonhomme and
  Manresa}{2015}]{bonhomme2015}
\textsc{Bonhomme, S. and E.~Manresa} (2015): \enquote{Grouped Patterns of
  Heterogeneity in Panel Data,} \emph{Econometrica}, 83, 1147--1184.

\bibitem[\protect\citeauthoryear{Buchinsky, Hahn, and Hotz}{Buchinsky
  et~al.}{2005}]{buchinsky2005}
\textsc{Buchinsky, M., J.~Hahn, and V.~Hotz} (2005): \enquote{Cluster Analysis:
  A Tool for Preliminary Structural Analysis,} \emph{Unpublished manuscript}.

\bibitem[\protect\citeauthoryear{Canova and Ciccarelli}{Canova and
  Ciccarelli}{2013}]{canova2013}
\textsc{Canova, F. and M.~Ciccarelli} (2013): \enquote{Panel Vector
  Autoregressive Models: A Survey,} in \emph{VAR Models in Macroeconomics--New
  Developments and Applications: Essays in Honor of Christopher A. Sims},
  Emerald Group Publishing Limited.

\bibitem[\protect\citeauthoryear{Cervellati, Jung, Sunde, and
  Vischer}{Cervellati et~al.}{2014}]{cervellati2014}
\textsc{Cervellati, M., F.~Jung, U.~Sunde, and T.~Vischer} (2014):
  \enquote{Income and Democracy: Comment,} \emph{American Economic Review},
  104, 707--719.

\bibitem[\protect\citeauthoryear{Chamberlain}{Chamberlain}{1980}]{chamberlain1980}
\textsc{Chamberlain, G.} (1980): \enquote{Analysis of Covariance with
  Qualitative Data,} \emph{The Review of Economic Studies}, 47, 225--238.

\bibitem[\protect\citeauthoryear{Cheng, Schorfheide, and Shao}{Cheng
  et~al.}{2019}]{cheng2019}
\textsc{Cheng, X., F.~Schorfheide, and P.~Shao} (2019): \enquote{Clustering for
  Multi-Dimensional Heterogeneity,} \emph{Manuscript}.

\bibitem[\protect\citeauthoryear{Chetty, Friedman, and Rockoff}{Chetty
  et~al.}{2014}]{chetty2014}
\textsc{Chetty, R., J.~N. Friedman, and J.~E. Rockoff} (2014):
  \enquote{Measuring the Impacts of Teachers I: Evaluating Bias in Teacher
  Value-Added Estimates,} \emph{American Economic Review}, 104, 2593--2632.

\bibitem[\protect\citeauthoryear{Chib}{Chib}{2007}]{chib2007b}
\textsc{Chib, S.} (2007): \enquote{Analysis of Treatment Response Data without
  the Joint Distribution of Potential Outcomes,} \emph{Journal of
  Econometrics}, 140, 401--412.

\bibitem[\protect\citeauthoryear{Chib and Hamilton}{Chib and
  Hamilton}{2000}]{chib2000}
\textsc{Chib, S. and B.~H. Hamilton} (2000): \enquote{Bayesian Analysis of
  Cross-Section and Clustered Data Treatment Models,} \emph{journal of
  Econometrics}, 97, 25--50.

\bibitem[\protect\citeauthoryear{Chib and Hamilton}{Chib and
  Hamilton}{2002}]{chib2002}
---\hspace{-.1pt}---\hspace{-.1pt}--- (2002): \enquote{Semiparametric Bayes
  Analysis of Longitudinal Data Treatment Models,} \emph{Journal of
  Econometrics}, 110, 67--89.

\bibitem[\protect\citeauthoryear{Chib and Jacobi}{Chib and
  Jacobi}{2007}]{chib2007a}
\textsc{Chib, S. and L.~Jacobi} (2007): \enquote{Modeling and Calculating the
  Effect of Treatment at Baseline from Panel Outcomes,} \emph{Journal of
  Econometrics}, 140, 781--801.

\bibitem[\protect\citeauthoryear{Clark and Ravazzolo}{Clark and
  Ravazzolo}{2015}]{clark2015}
\textsc{Clark, T.~E. and F.~Ravazzolo} (2015): \enquote{Macroeconomic
  Forecasting Performance under Alternative Specifications of Time-Varying
  Volatility,} \emph{Journal of Applied Econometrics}, 30, 551--575.

\bibitem[\protect\citeauthoryear{Cytrynbaum}{Cytrynbaum}{2021}]{cytrynbaum2020}
\textsc{Cytrynbaum, M.} (2021): \enquote{Blocked Clusterwise Regression,}
  \emph{arXiv preprint arXiv:2001.11130}.

\bibitem[\protect\citeauthoryear{Davidson and Ravi}{Davidson and
  Ravi}{2005}]{davidson2005}
\textsc{Davidson, I. and S.~Ravi} (2005): \enquote{Clustering with Constraints:
  Feasibility Issues and the K-Means Algorithm,} in \emph{Proceedings of the
  2005 SIAM International Conference on Data Mining}, SIAM, 138--149.

\bibitem[\protect\citeauthoryear{Epstein, Bates, Goldstone, Kristensen, and
  O'Halloran}{Epstein et~al.}{2006}]{epstein2006}
\textsc{Epstein, D.~L., R.~Bates, J.~Goldstone, I.~Kristensen, and
  S.~O'Halloran} (2006): \enquote{Democratic Transitions,} \emph{American
  journal of political science}, 50, 551--569.

\bibitem[\protect\citeauthoryear{Escobar and West}{Escobar and
  West}{1995}]{escobar1995}
\textsc{Escobar, M.~D. and M.~West} (1995): \enquote{Bayesian Density
  Estimation and Inference using Mixtures,} \emph{Journal of the American
  Statistical Association}, 90, 577--588.

\bibitem[\protect\citeauthoryear{Faust and Wright}{Faust and
  Wright}{2013}]{faust2013}
\textsc{Faust, J. and J.~H. Wright} (2013): \enquote{Forecasting Inflation,} in
  \emph{Handbook of Economic Forecasting}, Elsevier, vol.~2, 2--56.

\bibitem[\protect\citeauthoryear{Ferguson}{Ferguson}{1973}]{ferguson1973}
\textsc{Ferguson, T.~S.} (1973): \enquote{A Bayesian Analysis of Some
  Nonparametric Problems,} \emph{The Annals of Statistics}, 1, 209--230.

\bibitem[\protect\citeauthoryear{Ferguson}{Ferguson}{1974}]{ferguson1974}
---\hspace{-.1pt}---\hspace{-.1pt}--- (1974): \enquote{Prior Distributions on
  Spaces of Probability Measures,} \emph{The Annals of Statistics}, 2,
  615--629.

\bibitem[\protect\citeauthoryear{Fisher and Jensen}{Fisher and
  Jensen}{2022}]{fisher2022}
\textsc{Fisher, M. and M.~J. Jensen} (2022): \enquote{Bayesian Nonparametric
  Learning of How Skill is Distributed across the Mutual Fund Industry,}
  \emph{Journal of Econometrics}, 230, 131--153.

\bibitem[\protect\citeauthoryear{Fox, Sudderth, Jordan, and Willsky}{Fox
  et~al.}{2008}]{fox2008}
\textsc{Fox, E.~B., E.~B. Sudderth, M.~I. Jordan, and A.~S. Willsky} (2008):
  \enquote{An HDP-HMM for Systems with State Persistence,} in \emph{Proceedings
  of the 25th International Conference on Machine Learning}, 312--319.

\bibitem[\protect\citeauthoryear{Fox, Sudderth, Jordan, and Willsky}{Fox
  et~al.}{2011}]{fox2011}
---\hspace{-.1pt}---\hspace{-.1pt}--- (2011): \enquote{A Sticky HDP-HMM with
  Application to Speaker Diarization,} \emph{The Annals of Applied Statistics},
  5, 1020--1056.

\bibitem[\protect\citeauthoryear{Freeman and Weidner}{Freeman and
  Weidner}{2022}]{freeman2022}
\textsc{Freeman, H. and M.~Weidner} (2022): \enquote{Linear Panel Regressions
  with Two-Way Unobserved Heterogeneity,} \emph{arXiv preprint
  arXiv:2109.11911}.

\bibitem[\protect\citeauthoryear{Fruchterman and Reingold}{Fruchterman and
  Reingold}{1991}]{fruchterman1991}
\textsc{Fruchterman, T.~M. and E.~M. Reingold} (1991): \enquote{Graph Drawing
  by Force-Directed Placement,} \emph{Software: Practice and Experience}, 21,
  1129--1164.

\bibitem[\protect\citeauthoryear{Geweke and Amisano}{Geweke and
  Amisano}{2010}]{geweke2010}
\textsc{Geweke, J. and G.~Amisano} (2010): \enquote{Comparing and Evaluating
  Bayesian Predictive Distributions of Asset Returns,} \emph{International
  Journal of Forecasting}, 26, 216--230.

\bibitem[\protect\citeauthoryear{Ghosal and Van~der Vaart}{Ghosal and Van~der
  Vaart}{2017}]{ghosal2017}
\textsc{Ghosal, S. and A.~Van~der Vaart} (2017): \emph{Fundamentals of
  Nonparametric Bayesian Inference}, vol.~44, Cambridge University Press.

\bibitem[\protect\citeauthoryear{Gleditsch and Ward}{Gleditsch and
  Ward}{2006}]{gleditsch2006}
\textsc{Gleditsch, K.~S. and M.~D. Ward} (2006): \enquote{Diffusion and the
  International Context of Democratization,} \emph{International Organization},
  60, 911--933.

\bibitem[\protect\citeauthoryear{Hahn and Moon}{Hahn and Moon}{2010}]{hahn2010}
\textsc{Hahn, J. and H.~R. Moon} (2010): \enquote{Panel Data Models with Finite
  Number of Multiple Equilibria,} \emph{Econometric Theory}, 26, 863--881.

\bibitem[\protect\citeauthoryear{Hamilton}{Hamilton}{2018}]{hamilton2018}
\textsc{Hamilton, J.~D.} (2018): \enquote{Why You Should Never Use the
  Hodrick-Prescott Filter,} \emph{Review of Economics and Statistics}, 100,
  831--843.

\bibitem[\protect\citeauthoryear{Heckman, Lopes, and Piatek}{Heckman
  et~al.}{2014}]{heckman2014}
\textsc{Heckman, J.~J., H.~F. Lopes, and R.~Piatek} (2014): \enquote{Treatment
  Effects: A Bayesian Perspective,} \emph{Econometric Reviews}, 33, 36--67.

\bibitem[\protect\citeauthoryear{Henseler and Schumacher}{Henseler and
  Schumacher}{2019}]{henseler2019}
\textsc{Henseler, M. and I.~Schumacher} (2019): \enquote{The Impact of Weather
  on Economic Growth and its Production Factors,} \emph{Climatic Change}, 154,
  417--433.

\bibitem[\protect\citeauthoryear{Hersbach}{Hersbach}{2000}]{hersbach2000}
\textsc{Hersbach, H.} (2000): \enquote{Decomposition of the Continuous Ranked
  Probability Score for Ensemble Prediction Systems,} \emph{Weather and
  Forecasting}, 15, 559--570.

\bibitem[\protect\citeauthoryear{Heston, Summers, Aten et~al.}{Heston
  et~al.}{2002}]{heston2002}
\textsc{Heston, A., R.~Summers, B.~Aten, et~al.} (2002): \enquote{Penn world
  table version 6.1,} \emph{Center for International Comparisons at the
  University of Pennsylvania (CICUP)}, 18.

\bibitem[\protect\citeauthoryear{Hill}{Hill}{2011}]{hill2011}
\textsc{Hill, J.~L.} (2011): \enquote{Bayesian Nonparametric Modeling for
  Causal Inference,} \emph{Journal of Computational and Graphical Statistics},
  20, 217--240.

\bibitem[\protect\citeauthoryear{Hirano}{Hirano}{2002}]{hirano2002}
\textsc{Hirano, K.} (2002): \enquote{Semiparametric Bayesian Inference in
  Autoregressive Panel Data Models,} \emph{Econometrica}, 70, 781--799.

\bibitem[\protect\citeauthoryear{Holland, Laskey, and Leinhardt}{Holland
  et~al.}{1983}]{holland1983}
\textsc{Holland, P.~W., K.~B. Laskey, and S.~Leinhardt} (1983):
  \enquote{Stochastic Blockmodels: First Steps,} \emph{Social Networks}, 5,
  109--137.

\bibitem[\protect\citeauthoryear{Holtz-Eakin, Newey, and Rosen}{Holtz-Eakin
  et~al.}{1988}]{holtz1988}
\textsc{Holtz-Eakin, D., W.~Newey, and H.~S. Rosen} (1988): \enquote{Estimating
  Vector Autoregressions with Panel Data,} \emph{Econometrica}, 56, 1371--1395.

\bibitem[\protect\citeauthoryear{Hsiang}{Hsiang}{2016}]{hsiang2016}
\textsc{Hsiang, S.} (2016): \enquote{Climate Econometrics,} \emph{Annual Review
  of Resource Economics}, 8, 43--75.

\bibitem[\protect\citeauthoryear{Ishwaran and James}{Ishwaran and
  James}{2001}]{ishwaran2001}
\textsc{Ishwaran, H. and L.~F. James} (2001): \enquote{Gibbs Sampling Methods
  for Stick-Breaking Priors,} \emph{Journal of the American Statistical
  Association}, 96, 161--173.

\bibitem[\protect\citeauthoryear{Kahn, Mohaddes, Ng, Pesaran, Raissi, and
  Yang}{Kahn et~al.}{2021}]{kahn2021}
\textsc{Kahn, M.~E., K.~Mohaddes, R.~N. Ng, M.~H. Pesaran, M.~Raissi, and J.-C.
  Yang} (2021): \enquote{Long-Term Macroeconomic Effects of Climate Change: A
  Cross-Country Analysis,} \emph{Energy Economics}, 104, 105624.

\bibitem[\protect\citeauthoryear{Kim and Wang}{Kim and Wang}{2019}]{kim2019}
\textsc{Kim, J. and L.~Wang} (2019): \enquote{Hidden Group Patterns in
  Democracy Developments: Bayesian Inference for Grouped Heterogeneity,}
  \emph{Journal of Applied Econometrics}, 34, 1016--1028.

\bibitem[\protect\citeauthoryear{Koop}{Koop}{2003}]{koop2003}
\textsc{Koop, G.} (2003): \emph{Bayesian Econometrics}, John Wiley \& Sons.

\bibitem[\protect\citeauthoryear{Kulis and Jordan}{Kulis and
  Jordan}{2011}]{kulis2011}
\textsc{Kulis, B. and M.~I. Jordan} (2011): \enquote{Revisiting K-Means: New
  Algorithms via Bayesian Nonparametrics,} \emph{arXiv preprint
  arXiv:1111.0352}.

\bibitem[\protect\citeauthoryear{Laio and Tamea}{Laio and
  Tamea}{2007}]{laio2007}
\textsc{Laio, F. and S.~Tamea} (2007): \enquote{Verification Tools for
  Probabilistic Forecasts of Continuous Hydrological Variables,}
  \emph{Hydrology and Earth System Sciences}, 11, 1267--1277.

\bibitem[\protect\citeauthoryear{Law, Topchy, and Jain}{Law
  et~al.}{2004}]{law2004}
\textsc{Law, M.~H., A.~Topchy, and A.~K. Jain} (2004): \enquote{Clustering with
  Soft and Group Constraints,} in \emph{Joint IAPR International Workshops on
  Statistical Techniques in Pattern Recognition (SPR) and Structural and
  Syntactic Pattern Recognition (SSPR)}, Springer, 662--670.

\bibitem[\protect\citeauthoryear{Lee and Wilkinson}{Lee and
  Wilkinson}{2019}]{lee2019}
\textsc{Lee, C. and D.~J. Wilkinson} (2019): \enquote{A Review of Stochastic
  Block Models and Extensions for Graph Clustering,} \emph{Applied Network
  Science}, 4, 1--50.

\bibitem[\protect\citeauthoryear{Lin and Ng}{Lin and Ng}{2012}]{lin2012}
\textsc{Lin, C.-C. and S.~Ng} (2012): \enquote{Estimation of Panel Data Models
  with Parameter Heterogeneity When Group Membership is Unknown,} \emph{Journal
  of Econometric Methods}, 1, 42--55.

\bibitem[\protect\citeauthoryear{Lipset}{Lipset}{1959}]{lipset1959}
\textsc{Lipset, S.~M.} (1959): \enquote{Some Social Requisites of Democracy:
  Economic Development and Political Legitimacy,} \emph{American Political
  Science Review}, 53, 69--105.

\bibitem[\protect\citeauthoryear{Liu}{Liu}{2022}]{liu2022}
\textsc{Liu, L.} (2022): \enquote{Density Forecasts in Panel Data Models: A
  Semiparametric Bayesian Perspective,} \emph{Journal of Business \& Economic
  Statistics}, 0, 1--15.

\bibitem[\protect\citeauthoryear{Liu, Moon, and Schorfheide}{Liu
  et~al.}{2020}]{liu2020}
\textsc{Liu, L., H.~R. Moon, and F.~Schorfheide} (2020): \enquote{Forecasting
  with Dynamic Panel Data Models,} \emph{Econometrica}, 88, 171--201.

\bibitem[\protect\citeauthoryear{Lu}{Lu}{2007}]{lu2007}
\textsc{Lu, Z.} (2007): \enquote{Semi-Supervised Clustering with Pairwise
  Constraints: A Discriminative Approach,} in \emph{Artificial Intelligence and
  Statistics}, PMLR, 299--306.

\bibitem[\protect\citeauthoryear{Lu and Leen}{Lu and Leen}{2004}]{lu2004}
\textsc{Lu, Z. and T.~K. Leen} (2004): \enquote{Semi-Supervised Learning with
  Penalized Probabilistic Clustering,} in \emph{NIPS'04 Proceedings of the 17th
  International Conference on Neural Information Processing Systems}, 849--856.

\bibitem[\protect\citeauthoryear{Lu and Leen}{Lu and Leen}{2007}]{lu2007_2}
---\hspace{-.1pt}---\hspace{-.1pt}--- (2007): \enquote{Penalized Probabilistic
  Clustering,} \emph{Neural Computation}, 19, 1528--1567.

\bibitem[\protect\citeauthoryear{Lumsdaine, Okui, and Wang}{Lumsdaine
  et~al.}{2022}]{lumsdaine2022}
\textsc{Lumsdaine, R.~L., R.~Okui, and W.~Wang} (2022): \enquote{Estimation of
  Panel Group Structure Models with Structural Breaks in Group Memberships and
  Coefficients,} \emph{Journal of Econometrics}, Forthcoming.

\bibitem[\protect\citeauthoryear{MacQueen et~al.}{MacQueen
  et~al.}{1967}]{macqueen1967}
\textsc{MacQueen, J. et~al.} (1967): \enquote{Some Methods for Classification
  and Analysis of Multivariate Observations,} in \emph{Proceedings of the Fifth
  Berkeley Symposium on Mathematical Statistics and Probability}, Oakland, CA,
  USA, vol.~1, 281--297.

\bibitem[\protect\citeauthoryear{Matheson and Winkler}{Matheson and
  Winkler}{1976}]{matheson1976}
\textsc{Matheson, J.~E. and R.~L. Winkler} (1976): \enquote{Scoring Rules for
  Continuous Probability Distributions,} \emph{Management Science}, 22,
  1087--1096.

\bibitem[\protect\citeauthoryear{Meil{\u{a}}}{Meil{\u{a}}}{2007}]{meilua2007}
\textsc{Meil{\u{a}}, M.} (2007): \enquote{Comparing Clusterings—An
  Information Based Distance,} \emph{Journal of Multivariate Analysis}, 98,
  873--895.

\bibitem[\protect\citeauthoryear{Miller}{Miller}{2019}]{miller2019}
\textsc{Miller, J.~W.} (2019): \enquote{An Elementary Derivation of the Chinese
  Restaurant Process from Sethuraman’s Stick-Breaking Process,}
  \emph{Statistics \& Probability Letters}, 146, 112--117.

\bibitem[\protect\citeauthoryear{Moral-Benito and Bartolucci}{Moral-Benito and
  Bartolucci}{2012}]{moral2012}
\textsc{Moral-Benito, E. and C.~Bartolucci} (2012): \enquote{Income and
  Democracy: Revisiting the Evidence,} \emph{Economics Letters}, 117, 844--847.

\bibitem[\protect\citeauthoryear{M{\"u}eller, Quintana, and Page}{M{\"u}eller
  et~al.}{2018}]{mueller2018}
\textsc{M{\"u}eller, P., F.~A. Quintana, and G.~Page} (2018):
  \enquote{Nonparametric Bayesian Inference in Applications,} \emph{Statistical
  Methods \& Applications}, 27, 175--206.

\bibitem[\protect\citeauthoryear{Neal}{Neal}{1992}]{neal1992}
\textsc{Neal, R.~M.} (1992): \enquote{Bayesian Mixture Modeling,} in
  \emph{Maximum Entropy and Bayesian Methods}, Springer, 197--211.

\bibitem[\protect\citeauthoryear{Neal}{Neal}{2000}]{neal2000}
---\hspace{-.1pt}---\hspace{-.1pt}--- (2000): \enquote{Markov Chain Sampling
  Methods for Dirichlet Process Mixture Models,} \emph{Journal of Computational
  and Graphical Statistics}, 9, 249--265.

\bibitem[\protect\citeauthoryear{Nelson and Cohen}{Nelson and
  Cohen}{2007}]{nelson2007}
\textsc{Nelson, B. and I.~Cohen} (2007): \enquote{Revisiting Probabilistic
  Models for Clustering with Pairwise Constraints,} in \emph{Proceedings of the
  24th International Conference on Machine Learning}, 673--680.

\bibitem[\protect\citeauthoryear{Newton and Raftery}{Newton and
  Raftery}{1994}]{newton1994}
\textsc{Newton, M.~A. and A.~E. Raftery} (1994): \enquote{Approximate Bayesian
  Inference with the Weighted Likelihood Bootstrap,} \emph{Journal of the Royal
  Statistical Society: Series B (Methodological)}, 56, 3--26.

\bibitem[\protect\citeauthoryear{Nowicki and Snijders}{Nowicki and
  Snijders}{2001}]{nowicki2001}
\textsc{Nowicki, K. and T.~A.~B. Snijders} (2001): \enquote{Estimation and
  Prediction for Stochastic Blockstructures,} \emph{Journal of the American
  Statistical Association}, 96, 1077--1087.

\bibitem[\protect\citeauthoryear{Okui and Wang}{Okui and Wang}{2021}]{okui2021}
\textsc{Okui, R. and W.~Wang} (2021): \enquote{Heterogeneous Structural Breaks
  in Panel Data Models,} \emph{Journal of Econometrics}, 220, 447--473.

\bibitem[\protect\citeauthoryear{Orbanz and Buhmann}{Orbanz and
  Buhmann}{2008}]{orbanz2008}
\textsc{Orbanz, P. and J.~M. Buhmann} (2008): \enquote{Nonparametric Bayesian
  Image Segmentation,} \emph{International Journal of Computer Vision}, 77,
  25--45.

\bibitem[\protect\citeauthoryear{Paganin, Herring, Olshan, and Dunson}{Paganin
  et~al.}{2021}]{paganin2021}
\textsc{Paganin, S., A.~H. Herring, A.~F. Olshan, and D.~B. Dunson} (2021):
  \enquote{Centered Partition Processes: Informative Priors for Clustering,}
  \emph{Bayesian Analysis}, 16, 301--370.

\bibitem[\protect\citeauthoryear{Pati, Dunson, and Tokdar}{Pati
  et~al.}{2013}]{pati2013}
\textsc{Pati, D., D.~B. Dunson, and S.~T. Tokdar} (2013): \enquote{Posterior
  Consistency in Conditional Distribution Estimation,} \emph{Journal of
  multivariate analysis}, 116, 456--472.

\bibitem[\protect\citeauthoryear{Pelleg and Baras}{Pelleg and
  Baras}{2007}]{pelleg2007}
\textsc{Pelleg, D. and D.~Baras} (2007): \enquote{K-Means with Large and Noisy
  Constraint Sets,} in \emph{European Conference on Machine Learning},
  Springer, 674--682.

\bibitem[\protect\citeauthoryear{Pesaran, Pick, and Timmermann}{Pesaran
  et~al.}{2022}]{pesaran2022}
\textsc{Pesaran, M.~H., A.~Pick, and A.~Timmermann} (2022):
  \enquote{Forecasting with Panel Data: Estimation Uncertainty versus Parameter
  Heterogeneity,} \emph{CEPR Discussion Paper No. DP17123}.

\bibitem[\protect\citeauthoryear{Pitman}{Pitman}{1995}]{pitman1995}
\textsc{Pitman, J.} (1995): \enquote{Exchangeable and Partially Exchangeable
  Random Partitions,} \emph{Probability Theory and Related Fields}, 102,
  145--158.

\bibitem[\protect\citeauthoryear{Pitman}{Pitman}{1996}]{pitman1996}
---\hspace{-.1pt}---\hspace{-.1pt}--- (1996): \enquote{Some Developments of the
  Blackwell-MacQueen Urn Scheme,} \emph{Lecture Notes-Monograph Series},
  245--267.

\bibitem[\protect\citeauthoryear{Pitman and Yor}{Pitman and
  Yor}{1997}]{pitman1997}
\textsc{Pitman, J. and M.~Yor} (1997): \enquote{The Two-Parameter
  Poisson-Dirichlet Distribution Derived from a Stable Subordinator,} \emph{The
  Annals of Probability}, 25, 855--900.

\bibitem[\protect\citeauthoryear{Primiceri}{Primiceri}{2005}]{primiceri2005}
\textsc{Primiceri, G.~E.} (2005): \enquote{Time Varying Structural Vector
  Autoregressions and Monetary Policy,} \emph{The Review of Economic Studies},
  72, 821--852.

\bibitem[\protect\citeauthoryear{Przeworski, Limongi, and Giner}{Przeworski
  et~al.}{1995}]{przeworski1995}
\textsc{Przeworski, A., F.~Limongi, and S.~Giner} (1995): \enquote{Political
  Regimes and Economic Growth,} in \emph{Democracy and Development}, Springer,
  3--27.

\bibitem[\protect\citeauthoryear{Quintana, M{\"u}ller, Jara, and
  MacEachern}{Quintana et~al.}{2022}]{quintana2022}
\textsc{Quintana, F.~A., P.~M{\"u}ller, A.~Jara, and S.~N. MacEachern} (2022):
  \enquote{The Dependent Dirichlet Process and Related Models,}
  \emph{Statistical Science}, 37, 24--41.

\bibitem[\protect\citeauthoryear{Rasmussen}{Rasmussen}{1999}]{rasmussen1999}
\textsc{Rasmussen, C.} (1999): \enquote{The Infinite Gaussian Mixture Model,}
  in \emph{Advances in Neural Information Processing Systems}, ed. by S.~Solla,
  T.~Leen, and K.~M\"{u}ller, MIT Press, vol.~12, 554--560.

\bibitem[\protect\citeauthoryear{Ren, Dunson, and Carin}{Ren
  et~al.}{2008}]{ren2008}
\textsc{Ren, L., D.~B. Dunson, and L.~Carin} (2008): \enquote{The Dynamic
  Hierarchical Dirichlet Process,} in \emph{Proceedings of the 25th
  International Conference on Machine Learning}, 824--831.

\bibitem[\protect\citeauthoryear{Rockoff}{Rockoff}{2004}]{rockoff2004}
\textsc{Rockoff, J.~E.} (2004): \enquote{The Impact of Individual Teachers on
  Student Achievement: Evidence from Panel Data,} \emph{American Economic
  Review}, 94, 247--252.

\bibitem[\protect\citeauthoryear{Rodriguez and Dunson}{Rodriguez and
  Dunson}{2011}]{rodriguez2011}
\textsc{Rodriguez, A. and D.~B. Dunson} (2011): \enquote{Nonparametric Bayesian
  Models through Probit Stick-Breaking Processes,} \emph{Bayesian Analysis}, 6.

\bibitem[\protect\citeauthoryear{Ross and Dy}{Ross and Dy}{2013}]{ross2013}
\textsc{Ross, J. and J.~Dy} (2013): \enquote{Nonparametric Mixture of Gaussian
  Processes with Constraints,} in \emph{International Conference on Machine
  Learning}, PMLR, 1346--1354.

\bibitem[\protect\citeauthoryear{Rubin}{Rubin}{1974}]{rubin1974}
\textsc{Rubin, D.~B.} (1974): \enquote{Estimating Causal Effects of Treatments
  in Randomized and Nonrandomized Studies,} \emph{Journal of Educational
  Psychology}, 66, 688--701.

\bibitem[\protect\citeauthoryear{Sarafidis and Weber}{Sarafidis and
  Weber}{2015}]{sarafidis2015}
\textsc{Sarafidis, V. and N.~Weber} (2015): \enquote{A Partially Heterogeneous
  Framework for Analyzing Panel Data,} \emph{Oxford Bulletin of Economics and
  Statistics}, 77, 274--296.

\bibitem[\protect\citeauthoryear{Sethuraman}{Sethuraman}{1994}]{sethuraman1994}
\textsc{Sethuraman, J.} (1994): \enquote{A Constructive Definition of Dirichlet
  Priors,} \emph{Statistica Sinica}, 4, 639--650.

\bibitem[\protect\citeauthoryear{Shental, Bar-Hillel, Hertz, and
  Weinshall}{Shental et~al.}{2003}]{shental2003}
\textsc{Shental, N., A.~Bar-Hillel, T.~Hertz, and D.~Weinshall} (2003):
  \enquote{Computing Gaussian Mixture Models with EM using Equivalence
  Constraints,} \emph{Advances in Neural Information Processing Systems}, 16,
  465--472.

\bibitem[\protect\citeauthoryear{Shiraito}{Shiraito}{2016}]{shiraito2016}
\textsc{Shiraito, Y.} (2016): \enquote{Uncovering Heterogeneous Treatment
  Effects,} \emph{Visited on. https://shiraito. github. io/research/files/jmp.
  pdf}.

\bibitem[\protect\citeauthoryear{Smith}{Smith}{1973}]{smith1973}
\textsc{Smith, A.~F.} (1973): \enquote{A General Bayesian Linear Model,}
  \emph{Journal of the Royal Statistical Society: Series B (Methodological)},
  35, 67--75.

\bibitem[\protect\citeauthoryear{Stock and Watson}{Stock and
  Watson}{2007}]{stock2007}
\textsc{Stock, J.~H. and M.~W. Watson} (2007): \enquote{Why has US Inflation
  Become Harder to Forecast?} \emph{Journal of Money, Credit and Banking}, 39,
  3--33.

\bibitem[\protect\citeauthoryear{Su and Ju}{Su and Ju}{2018}]{su2018}
\textsc{Su, L. and G.~Ju} (2018): \enquote{Identifying Latent Grouped Patterns
  in Panel Data Models with Interactive Fixed Effects,} \emph{Journal of
  Econometrics}, 206, 554--573.

\bibitem[\protect\citeauthoryear{Su, Shi, and Phillips}{Su
  et~al.}{2016}]{su2016}
\textsc{Su, L., Z.~Shi, and P.~C. Phillips} (2016): \enquote{Identifying Latent
  Structures in Panel Data,} \emph{Econometrica}, 84, 2215--2264.

\bibitem[\protect\citeauthoryear{Su, Wang, and Jin}{Su et~al.}{2019}]{su2019}
\textsc{Su, L., X.~Wang, and S.~Jin} (2019): \enquote{Sieve Estimation of
  Time-Varying Panel Data Models with Latent Structures,} \emph{Journal of
  Business \& Economic Statistics}, 37, 334--349.

\bibitem[\protect\citeauthoryear{Sun}{Sun}{2005}]{sun2005}
\textsc{Sun, Y.} (2005): \enquote{Estimation and Inference in Panel Structure
  Models,} \emph{Available at SSRN 794884}.

\bibitem[\protect\citeauthoryear{Teh, Jordan, Beal, and Blei}{Teh
  et~al.}{2006}]{Teh2006}
\textsc{Teh, Y.~W., M.~I. Jordan, M.~J. Beal, and D.~M. Blei} (2006):
  \enquote{Hierarchical Dirichlet Processes,} \emph{Journal of the American
  Statistical Association}, 101, 1566--1581.

\bibitem[\protect\citeauthoryear{Vlachos, Ghahramani, and Briscoe}{Vlachos
  et~al.}{2010}]{vlachos2010}
\textsc{Vlachos, A., Z.~Ghahramani, and T.~Briscoe} (2010): \enquote{Active
  Learning for Constrained Dirichlet Process Mixture Models,} in
  \emph{Proceedings of the 2010 Workshop on Geometrical Models of Natural
  Language Semantics}, 57--61.

\bibitem[\protect\citeauthoryear{Vlachos, Ghahramani, and Korhonen}{Vlachos
  et~al.}{2008}]{vlachos2008}
\textsc{Vlachos, A., Z.~Ghahramani, and A.~Korhonen} (2008): \enquote{Dirichlet
  Process Mixture Models for Verb Clustering,} in \emph{Proceedings of the ICML
  Workshop on Prior Knowledge for Text and Language}.

\bibitem[\protect\citeauthoryear{Vlachos, Korhonen, and Ghahramani}{Vlachos
  et~al.}{2009}]{vlachos2009}
\textsc{Vlachos, A., A.~Korhonen, and Z.~Ghahramani} (2009):
  \enquote{Unsupervised and Constrained Dirichlet Process Mixture Models for
  Verb Clustering,} in \emph{Proceedings of the Workshop on Geometrical Models
  of Natural Language Semantics}, 74--82.

\bibitem[\protect\citeauthoryear{Wade and Ghahramani}{Wade and
  Ghahramani}{2018}]{wade2018}
\textsc{Wade, S. and Z.~Ghahramani} (2018): \enquote{Bayesian Cluster Analysis:
  Point Estimation and Credible Balls (with Discussion),} \emph{Bayesian
  Analysis}, 13, 559--626.

\bibitem[\protect\citeauthoryear{Wager and Athey}{Wager and
  Athey}{2018}]{wager2018}
\textsc{Wager, S. and S.~Athey} (2018): \enquote{Estimation and Inference of
  Heterogeneous Treatment Effects using Random Forests,} \emph{Journal of the
  American Statistical Association}, 113, 1228--1242.

\bibitem[\protect\citeauthoryear{Wagstaff and Cardie}{Wagstaff and
  Cardie}{2000}]{wagstaff2000}
\textsc{Wagstaff, K. and C.~Cardie} (2000): \enquote{Clustering with
  Instance-Level Constraints,} \emph{AAAI/IAAI Proceedings}, 1097, 577--584.

\bibitem[\protect\citeauthoryear{Wagstaff, Cardie, Rogers, Schroedl
  et~al.}{Wagstaff et~al.}{2001}]{wagstaff2001}
\textsc{Wagstaff, K., C.~Cardie, S.~Rogers, S.~Schroedl, et~al.} (2001):
  \enquote{Constrained K-Means Clustering with Background Knowledge,} in
  \emph{Proceedings of the 18th International Conference on Machine Learning},
  vol.~1, 577--584.

\bibitem[\protect\citeauthoryear{Walker}{Walker}{2007}]{walker2007}
\textsc{Walker, S.~G.} (2007): \enquote{Sampling the Dirichlet Mixture Model
  with Slices,} \emph{Communications in Statistics—Simulation and
  Computation{\textregistered}}, 36, 45--54.

\bibitem[\protect\citeauthoryear{Wang and Dunson}{Wang and
  Dunson}{2011}]{wang2011}
\textsc{Wang, L. and D.~B. Dunson} (2011): \enquote{Fast Bayesian Inference in
  Dirichlet Process Mixture Models,} \emph{Journal of Computational and
  Graphical Statistics}, 20, 196--216.

\bibitem[\protect\citeauthoryear{Wang, Phillips, and Su}{Wang
  et~al.}{2018}]{wang2018}
\textsc{Wang, W., P.~C. Phillips, and L.~Su} (2018): \enquote{Homogeneity
  Pursuit in Panel Data Models: Theory and Application,} \emph{Journal of
  Applied Econometrics}, 33, 797--815.

\bibitem[\protect\citeauthoryear{Wang and Su}{Wang and Su}{2021}]{wang2021}
\textsc{Wang, W. and L.~Su} (2021): \enquote{Identifying Latent Group
  Structures in Nonlinear Panels,} \emph{Journal of Econometrics}, 220,
  272--295.

\bibitem[\protect\citeauthoryear{Yoder and Priebe}{Yoder and
  Priebe}{2017}]{yoder2017}
\textsc{Yoder, J. and C.~E. Priebe} (2017): \enquote{Semi-Supervised
  K-Means++,} \emph{Journal of Statistical Computation and Simulation}, 87,
  2597--2608.

\bibitem[\protect\citeauthoryear{Zhang}{Zhang}{2020}]{zhang2020}
\textsc{Zhang, B.} (2020): \enquote{Forecasting with Bayesian Grouped Random
  Effects in Panel Data,} \emph{arXiv preprint arXiv:2007.02435}.

\end{thebibliography}

\clearpage
\setstretch{1.15}
\appendix
\setcounter{saveeqn}{\value{section}}\renewcommand{\theequation}{\mbox{%
                \Alph{saveeqn}.\arabic{equation}}} \setcounter{saveeqn}{1} %
\setcounter{equation}{0}
\renewcommand*\thepage{A-\arabic{page}}
\setcounter{page}{1}
\renewcommand*\thetable{A-\arabic{table}}
\setcounter{table}{0}
\renewcommand*\thefigure{A-\arabic{figure}}
\setcounter{figure}{0}

\bc

{\Large {\bf Supplemental Appendix to \\
		``Unobserved Grouped Patterns in Panel Data \\and Prior Wisdom'' }}

{\bf Boyuan Zhang}

\ec




\numberwithin{equation}{section}
\numberwithin{table}{section}
\numberwithin{figure}{section}

\section{Definitions and Terminology}

\subsection{Dirichlet Process and Related Stochastic Processes} \label{app:DP_and_related}


All unknown quantities in a model must be assigned prior distributions in Bayesian inference. A nonparametric prior can be used to reflect uncertainty about the parametric form of the prior distribution. Because of its richness, computational ease, and interpretability, the Dirichlet process (DP) is one of the most often used random probability measures. It can be used to model the uncertainty about the functional form of the prior distribution for parameters in a model.

The DP, which was first established using Kolmogorov consistency conditions \cite{ferguson1973}, can be defined from a number of views. \citet{ferguson1973} shows that the DP can be obtained by normalizing a gamma process. By using exchangeability, the Pólya urn method leads to the GP \cite{blackwell1973}. The Chinese restaurant process (CRP) \cite{aldous1985, pitman1996}, a distribution over partitions, is a similarly related sequential process that produces the DP when each partition is assigned an independent parameter with a common distribution. \citet{sethuraman1994} provided a constructive definition of the DP, which led to the characterization as a stick-breaking prior \cite{ishwaran2001}.

Construction of the DP using a stick-breaking process or a gamma process represents the DP as a countably infinite sum of atomic measures. These approaches suggest that a DPM model can be seen as a mixture model with infinitely many components. The distribution of parameters imposed by a DP can also be obtained as a limiting case of a parametric mixture model \cite{neal1992, rasmussen1999, neal2000}. This approach shows that a DPM can easily be defined as an extension of a parametric mixture model without the need to do model selection for determining the number of components to be used.


\subsubsection{Dirichlet Process}

\citet{ferguson1973} defines a DP with two parameters, a positive scalar $a$ and a probability measure $B_{0}$, referred to as the concentration parameter and the base measure, respectively. The base distribution $B_{0}$ is the parameter on which the nonparametric distribution is centered, which can be thought of as the prior guess \cite{antoniak1974}. The concentration parameter $a$ expresses the strength of belief in $B_{0}$. For small values of $a$, samples from a DP is likely to be composed of a small number of atomic measures with large weights. For large values, most samples are likely to be distinct, thus concentrated on $B_{0}$.

Technically, a nonparametric prior is a probability distribution on $\mathcal{P}$, the space of all probability measures (say on the real line). Measurable sets in $\mathcal{P}$ are of the form $\{A \colon P(A) < 1\}$. We could specify a prior distribution over $\left(P\left(A_{1}\right), P\left(A_{2}\right), \ldots, P\left(A_{K}\right)\right)$ where $A_1$, $A_2$,..., $A_K$ are measurable finite partition of the measurable set $A$. Denote
\begin{align*}
	P \sim D P\left(a, B_{0}\right)
\end{align*}
for all partition $\left(A_{1}, \cdots, A_{K}\right)$, then,
\begin{align*}
	\left(P\left(A_{1}\right), \cdots, P\left(A_{K}\right)\right) \sim \operatorname{Dir}\left(a B_{0}\left(A_{1}\right), \cdots, a B_{0}\left(A_{K}\right)\right)
\end{align*}
Dir $(\cdot)$ stands for the Dirichlet distribution with probability distribution function being
\begin{align*}
	p \left(x_{1}, \cdots, x_{K} ; \eta_{1}, \cdots, \eta_{K}\right) = \frac{\Gamma\left(\sum_{k=1}^{K} \eta_{k}\right)}{\prod_{k=1}^{K} \Gamma\left(\eta_{k}\right)} \prod_{k=1}^{K} x_{k}^{\eta_{k}-1}
\end{align*}
where $x_{i}\in(0,1)$ and $\sum_{i=1}^{K} x_{i}=1$. This is a multivariate generalization of the Beta distribution and the infinite-dimensional generalization of the Dirichlet distribution is the Dirichlet process.

We next define a typical DP prior for $(\alpha, \sigma^2)$:
\begin{align*} \label{eq:DP_prior}
	(\alpha_i, \sigma_i^2) & \sim B,  \\
	B & \sim \operatorname{DP}\left(a, B_{0}\right), \numberthis
\end{align*} 
where $B$ is a random distribution. There are two parameters. The base distribution $B_{0}$ is a distribution over the same space as $B$. For example, if $B_0$ is a distribution on reals then $B$ must be a distribution on reals too. The concentration parameter $a$ is a positive scalar. One property of the DP is that random distributions $B$ are discrete, and each places its mass on a countably infinite collection of atoms drawn from $B_{0}$.

We adopt an Independent Normal Inverse-Gamma (INIG) distribution for the base distribution $B_0$:		
\begin{align}
	B_0(\phi) :=  INIG \left( \mu_{\alpha}, \Sigma_{\alpha}, \frac{\nu_{\sigma}}{2}, \frac{\delta_{\sigma}}{2} \right),
\end{align} 
with a set of hyperparameters $\phi =  \left( \mu_{\alpha}, \Sigma_{\alpha}, \frac{\nu_{\sigma}}{2}, \frac{\delta_{\sigma}}{2} \right)$.

The form of the base distribution and the value of the concentration parameter are critical aspects of model selection that influence modeling performance. Given a murky prior distribution, the concentration parameter's value can be derived from the data. It is more difficult to choose the base distribution because the model's performance is largely dependent on its parametric form, even if it is constructed hierarchically for robustness. The choice of the base distribution is determined largely by mathematical and practical convenience. For computational ease, conjugate distributions are recommended.

A draw from DP is, by definition, a discrete distribution. In this sense, given the baseline model, imposing a DP prior on the distribution of $\alpha_{g_i}$ entails limiting the intercepts to some discrete values and assuming that intercepts within a group are identical, which may not be appealing for some empirical applications. A natural extension is to suppose that $\alpha_i$ has a continuous parametric distribution $f(\alpha_i ; \theta)$, with $\theta$ as parameters, and to use a DP prior for the distribution of $\theta$. The parameters $\theta$ are then discrete and has group structure, whereas group heterogeneity $\alpha_{g_i}$ has a continuous distribution, i.e., $\alpha_i$ within a group can be different, but they are all derived from the same distribution. This additional layer of mixing is the general idea of the Dirichlet process mixture (DPM) model.


\subsubsection{Stick Breaking Process}
A nonparametric prior can also be defined as the distribution of a random variable $P$ taking values in $\mathcal{P}$. A construction of $\mathrm{DP}$ follows the stick-breaking process \cite{sethuraman1994},
\begin{align*}
	P(A) &=\sum_{k=1}^{\infty} \pi_{k} \mathbf{1}_{\alpha_k}\left(A\right), \\
	\alpha_k & \sim B_{0}, \quad k=1,2, \cdots, \\
	\pi_{k} &= \begin{cases}\zeta_{1}, & k=1 \\
		\prod_{j<k}\left(1-\zeta_{j}\right) \zeta_{k}, & k=2,3, \cdots\end{cases} \\
	& \text { where } \zeta_{k} \sim \operatorname{Beta}(1, a), \quad k=1,2, \cdots
\end{align*}

The stick breaking process distinguishes the roles of $B_{0}$ and $a$ in that the former governs component value $\alpha_k$ while the latter guides the choice of component probability $\pi_{k}$. Roughly speaking, the DP concentration parameter $a$ is linked to the number of unique components in the mixture density and thus determines and reflects the flexibility of the mixture density. 



\subsubsection{Chinese Restaurant / Pólya Urn Process} \label{app:CRP}

Another widely used representation of the DP prior is the Chinese restaurant process (CRP). To set the stage, imagine that we have a Chinese restaurant that has infinitely many tables that can each seat infinitely many customers. When a new customer, say the $n$-th, enters the restaurant, the probability of them sitting at the table $k$ with $n_k$ other customers proportional to $n_k$, and the probability of this customer sitting alone at a new table is related to $a$ (the concentration parameter in DP),
\begin{align*}
	p\left((\alpha_i, \sigma_i^2) = (\alpha_k, \sigma_k^2) | \alpha_{1:i-1}, \sigma_{1:i-1}^2, a, B_{0} \right)
	& \; \propto
	\left\{
	\begin{array}{ll}
		n_{k} & \text { if } k \text { is an existing table} \\
		a & \text { if } k \text{ is a new table}.
	\end{array}
	\right.
\end{align*}

Upon marginalizing out the random mixing measure $B$, one obtains the conditional distribution of $(\alpha_i, \sigma_i^2)$ given $(\alpha_{1:i-1}, \sigma_{1:i-1}^2)$, which follows a Polya urn distribution,
\begin{align} \label{eq:Polya_urn}
	\alpha_i, \sigma_i^2 | \alpha_{1:i-1}, \sigma_{1:i-1}^2, a, B_{0}\sim \frac{1}{a+i-1} \sum_{j=1}^{i-1} \delta_{(\alpha_j, \sigma_j^2)}+ \frac{a}{a+i-1} B_{0} (\cdot).
\end{align}


Equation (\ref{eq:Polya_urn}) reveals the  \textit{clustering property} of the DP prior: there is a positive probability that each $(\alpha_i, \sigma_i^2)$ will take on the value of another $(\alpha_j, \sigma_j^2)$, leading some of the variables to share values. The probability of sharing values is proportional to $n_{g_j}$. This self-reinforcing property is sometimes expressed as \textit{the rich get richer}. This equation also reveals the roles of scaling parameter $a$ and base distribution $B_{0}$. The unique values contained in $(\alpha_{1:N}, \sigma_{1:N}^2)$ are drawn independently from $B_{0}$, and the parameter $a$ determines how likely $\alpha_{i}, \sigma_{i}^2$ is to be a newly drawn value from $B_{0}$ rather than take on one of the values from $\alpha_{1:i-1}, \sigma_{1:i-1}^2$.

Chinese restaurant process shares the same characteristics as the Pólya urn process which can be extend to the two-parameter Pitman–Yor process \cite{pitman1997}. Here is the basic idea of Pólya urn process. Imagine that we have an urn with possibly infinitely many colors. Let $a$ (again, the concentration parameter in DP) be the initial number of balls with each color. The urn evolves in discrete time steps - at each step, one ball is sampled uniformly at random and put it back to the urn; The color of the withdrawn ball is observed, and one additional ball with the same color is returned to the urn. This process is then repeated.

Equation (\ref{eq:Polya_urn}) is also called the Blackwell-MacQueen prediction rule - the conditional distribution of $\theta_{n}$ given previous sampled $\theta_{1:n-1}$ from the Dirichlet process prior. It characterizes the Chinese restaurant process/Pólya urn process and serves as the key component in the Pólya urn Gibbs sampler \cite{ishwaran2001}.

Prior literature shows the equivalence between Chinese restaurant process/Pólya urn process and aforementioned processes. \citet{blackwell1973} present the equivalence between Pólya urn process and Dirichlet process. \citet{miller2019}  formally prove that the Chinese restaurant process is equivalent to the stick breaking process.


\subsection{Exchangeable Partition Probability Function}

A predominant feature of the the exchangeable partition probability function (EPPF) in (\ref{eq:EPPF}) is that it defines a prior distribution over $G$. To obtain the prior from EPPF, we must first identify all possible group partitions of $N$ units. This problem can be recast as a prototypical ``balls and urns" problem: what are the ways of putting $N$ distinguishable balls into $N$ indistinguishable urns if empty urns are allowed? 

\begin{exmp} \label{exmp_dp}
	Consider a simple case in which $N = 3$ and $a = 1$. It is easy to show that there are five ways to group three units. Then the prior distribution over $G$ under Dirichlet process is given by,
	\begin{align*}
		\Pr (g_1 = g_2 = g_3 = 1) & = \frac{\Gamma(1)}{\Gamma(4)}  \Gamma(3)= \frac{1}{3}, \\
		\Pr (g_1 = g_2 = 1, g_3 = 2) & = \frac{\Gamma(1)}{\Gamma(4)}  \Gamma(2) \Gamma(1) = \frac{1}{6}, \\
		\Pr (g_1 = g_3 = 1, g_2 = 2) & = \frac{\Gamma(1)}{\Gamma(4)}  \Gamma(2) \Gamma(1) = \frac{1}{6}, \\
		\Pr (g_2 = g_3 = 1, g_1 = 2)& = \frac{\Gamma(1)}{\Gamma(4)}  \Gamma(2) \Gamma(1) = \frac{1}{6}, \\
		\Pr (g_1 = 1, g_2 = 2, g_3 = 3)& = \frac{\Gamma(1)}{\Gamma(4)}  \Gamma(1) \Gamma(1)  \Gamma(1) = \frac{1}{6}.
	\end{align*} 
\end{exmp}

Finding all solutions to the "balls and urns" problem with big $N$ is computationally impossible in general. For a certain number of groups $K$, the number of ways to assign $N$ unit to $K$ groups is described by the \textit{Stirling number of the second kind},
\begin{align}
	\mathcal{S}_{N, K} = \frac{1}{K !} \sum_{j=0}^{K}(-1)^{j} C_{K}^j (K-j)^{N}.
\end{align}

The sum of $\mathcal{S}_{N, K}$ over all possible $K$, also known as the $N$-th Bell number, $\mathcal{B}_{N}=\sum_{K=1}^{N} \mathcal{S}_{N, K}$ describes the number of all possible partitions of $N$ balls. Owing to the rapid growth of the space, listing all feasible partitions becomes computationally impossible. For example from a moderate $N = 12$ to 15, the number of partitions increases from $\mathcal{B}_{12} = 4, 213, 597$ to $\mathcal{B}_{15} = 1,382,958,545$. \citet{sethuraman1994} and \citet{pitman1996} constructively show that group indices/partitions can be drawn from the EPPF for DP using the stick-breaking process defined in (2.8). As a result, the EPPF does not explicitly appear in the posterior analysis in the current setting so long as the priors for the stick lengths are included.

In the example below, we demonstrate how pairwise constraints affect the prior density of the group partition using Equation (\ref{eq:prior_G_soft}):
\begin{align}
	p(G | \psi, T) \; \propto \; p (G) \prod_{i, j}  \left( \fra{\psi_{i j}}{1 - \psi_{i j}} \right)^{c T_{i j} \delta_{ij}(G)}.
\end{align}

\begin{exmp}
	Consider again the three-unit case in Example \ref{exmp_dp}  with $a = 1$. Assume there is a positive-link constraint between units 1 and 2 and that is the only constraint in this example. Then the prior probabilities of the five partitions are adjusted to account for the effect of $W_{12}$:
	\begin{align*}
		\Pr (g_1 = g_2 = g_3 = 1) & = \frac{1}{3} \cdot \frac{2 \exp(4cW_{12})}{\exp(4cW_{12}) + 1} > \frac{1}{3}, \\
		\Pr (g_1 = g_2 = 1, g_3 = 2) & = \frac{1}{3} \cdot \frac{\exp(4cW_{12})}{\exp(4cW_{12}) + 1} > \frac{1}{6}, \\
		\Pr (g_1 = g_3 = 1, g_2 = 2) & = \frac{1}{3} \cdot \frac{1}{\exp(4cW_{12}) + 1} < \frac{1}{6},\\
		\Pr (g_2 = g_3 = 1, g_1 = 2)& = \frac{1}{3} \cdot \frac{1}{\exp(4cW_{12}) + 1} < \frac{1}{6},\\
		\Pr (g_1 = 1, g_2 = 2, g_3 = 3)& = \frac{1}{3} \cdot \frac{1}{\exp(4cW_{12}) + 1} < \frac{1}{6}.
	\end{align*} 
	Note that $c > 0$ and $W_{12} > 0$. Comparing to the results in Example \ref{exmp_dp}, the probabilities of the first two partitions become higher since they all meet the PL constraint between units 1 and 2, while the rest of the partitions violate the constraint and hence the probabilities drop.
\end{exmp}

%

\subsection{Hierarchical Dirichlet Process} \label{appendix:HDP}

The hierarchical Dirichlet process (HDP) was developed by \citet{Teh2006}. The HDP is a nonparametric Bayesian approach to clustering grouped data, with the known group membership. It equips a Dirichlet process for each group of data, with the Dirichlet processes for all groups sharing a base distribution which is itself drawn from a Dirichlet process. This method allows groups to share statistical strength via sharing of clusters across groups. The base distribution being drawn from a Dirichlet process is important, because draws from a Dirichlet process are atomic probability measures, and the atoms will appear in all group-level Dirichlet processes. Since each atom corresponds to a cluster, clusters are shared across all groups.

The HDP is parameterized by a base distribution $H$ that governs the prior distribution over data items, and a number of concentration parameters that govern the prior number of clusters and amount of sharing across groups. Assume that we have $J$ groups of data, each consisting of $N_j$ data points, $y_{j 1}, \ldots y_{j N_j}$.  The process defines a set of random probability measures $(B_j)_{j=1}^J$, one for each group. The random probability measure $B_j$ for the $j$-th group is distributed as a Dirichlet process:
\begin{align}
	B_j | B_0 \sim \operatorname{DP}\left(\gamma, B_0\right),
\end{align}
where $\gamma$ is the concentration parameter and $B_0$ is the base distribution shared across all groups. The distribution of the global random probability measure $B_0$ is given by,
\begin{align}
	B_0 \sim \operatorname{DP}\left(\alpha_0, H\right),
\end{align}
with concentration parameter $\alpha_0$ and base distribution $H$. 

A hierarchical Dirichlet process can be used as the prior distribution over the parameters for grouped data. For each $j$, let $\left(\phi_{j i}\right)_{i=1}^{n_j}$ be i.i.d. random variables distributed as $B_j$. Each $\phi_{j i}$ is a parameter corresponding to a single observation $y_{j i}$. The likelihood is given by,
\begin{align*}
	&\phi_{j i} | B_j \sim B_j, \\
	&y_{j i} | \phi_{j i} \sim p \left(\phi_{j i}\right). \numberthis
\end{align*}
The resulting model above is called a HDP mixture model, with the HDP referring to the hierarchically linked set of Dirichlet processes, and the mixture model referring to the way the Dirichlet processes are related to the data items.

To understand how the HDP implements a clustering model, and how clusters become shared across groups, recall that draws from a Dirichlet process are atomic probability measures with probability one. The base distribution $B_0$ can be expressed using a stick-breaking representation,
\begin{align}
	B_0=\sum_{k=1}^{\infty} \beta_{k} \delta_{\theta_k},
\end{align}
where there are an infinite number of atoms, $\theta_k \sim H$, $k=1,2, \ldots$. Each atom is associated with a mass $\beta_{k}$ and $\beta=\left(\beta_i\right)_{i=1}^{\infty} \sim \operatorname{Stick}(\gamma)$ are mutually independent. Since $B_0$ is the base distribution for the group specific Dirichlet processes, each $B_j$ has the same atoms as $B_0$ and can be written in the form,
\begin{align}
	B_j = \sum_{k=1}^{\infty} \pi_{j k} \delta_{\theta_k}.
\end{align}

Let $\pi_j=\left(\pi_{j k}\right)_{k=1}^{\infty}$. Note that the weights $\pi_j$ are independent given $\beta$ (since the $B_j$ are independent given $B_0$). It can be shown that the connection between the weights $\pi_j$ and the global weights $\beta$ is
\begin{align}
	\pi_j \mid \alpha_0, \beta \sim \operatorname{DP}\left(\alpha_0, \beta\right).
\end{align}

Thus the set of atoms is shared across all groups, with each group having its own group-specific atom masses. Relating this representation back to the observed data, we see that each data item is described by a mixture model,
\begin{align}
	y_{j i} | B_j \sim \sum_{k=1}^{\infty} \pi_{j k} f\left(\theta_k\right),
\end{align}
where the atoms $\theta_k$ play the role of the mixture component parameters, while the masses $\pi_{j k}$ play the role of the mixing proportions. As a result, each group of data is modeled using a mixture model, with mixture components shared across all groups and group-specific mixing weights. 

\vspace{0.5cm}

\noindent {\bf Chinese Restaurant Franchise} \\
\citet{Teh2006} have also described the marginal probabilities obtained from integrating over the random measures $B_0$ and $(B_j)_{j=1}^J$. They show that these marginals can be described in terms of a Chinese restaurant franchise (CRF) that is an analog of the Chinese restaurant process. 

Recall that $\phi_{j i}$ are random variables with distribution $B_j$. In the following discussion, we will let $\theta_1, \ldots, \theta_K$ denote $K$ i.i.d. random variables distributed according to $H$, and, for each $j$, we let $\psi_{j 1}, \ldots, \psi_{j T_j}$ denote $T_j$ i.i.d. variables distributed according to $B_0$.

Each $\phi_{j i}$ is associated with one $\psi_{j t}$, while each $\psi_{j t}$ (table id) is associated with one $\theta_k$. Let $t_{j i}$ be the index of the $\psi_{j t}$ associated with $\phi_{j i}$, and let $k_{j t}$ (dish id) be the index of $\theta_k$ associated with $\psi_{j t}$. Let $n_{j t}$ be the number of $\phi_{j i}$'s associated with $\psi_{j t}$, while $m_{j k}$ is the number of $\psi_{j t}$'s associated with $\theta_k$. Define $m_k=\sum_j m_{j k}$ as the number of $\psi_{j t}$'s associated with $\theta_k$ over all $j$. Notice that while the values taken on by the $\psi_{j t}$'s need not be distinct, they are distributed according to a discrete random probability measure $B_0 \sim \mathrm{DP}(\gamma, H)$, we are denoting them as distinct random variables.

First consider the conditional distribution for $\phi_{j i}$ given $\phi_{j 1}, \ldots, \phi_{j i-1}$ and $B_0$, where $B_j$ is integrated out, we have,
\begin{align}
	\phi_{j i} | \phi_{j 1}, \ldots, \phi_{j i-1}, \alpha_0, G_0 \sim \sum_{t=1}^{T_j} \frac{n_{j t}}{i-1+\alpha_0} \delta_{\psi_{j t}}+\frac{\alpha_0}{i-1+\alpha_0} G_0,
\end{align}
This is a mixture, and a draw from this mixture can be obtained by drawing from the terms on the right-hand side with probabilities given by the corresponding mixing proportions. If a term in the first summation is chosen, then we set $\phi_{j i}=\psi_{j t}$ and let $t_{j i}=t$ for the chosen $t$. If the second term is chosen, then we increment $T_j$ by one, draw $\psi_{j T_j } \sim B_0$ and set $\phi_{j i}=\psi_{j T_j}$ and $t_{j i}=T_j$. 

Now we proceed to integrate out $B_0$. Notice that $B_0$ appears only in its role as the distribution of the variables $\psi_{j t}$. Since $B_0$ is distributed according to a Dirichlet process, we can integrate it out and writing the conditional distribution of $\psi_{j t}$ directly:
\begin{align}
	\psi_{j t} | \psi_{11}, \psi_{12}, \ldots, \psi_{21}, \ldots, \psi_{j t-1}, \gamma, H \sim \sum_{k=1}^K \frac{m_k}{\sum_k m_k+\gamma} \delta_{\theta_k}+\frac{\gamma}{\sum_k m_k+\gamma} H.
\end{align}
If we draw $\psi_{j t}$ via choosing a term in the summation on the right-hand side of this equation, we set $\psi_{j t}=\theta_k$ and let $k_{j t}=k$ for the chosen $k$. If the second term is chosen, we increment $K$ by one, draw $\theta_K \sim H$ and set $\psi_{j t}=\theta_K, k_{j t}=K$.

In short, the CRF is comprised of $J$ restaurants with a shared menu across the restaurants. Each restaurant corresponds to an HDP group, and an infinite buffet line of dishes common to all restaurants. The process of seating customers at tables, however, is restaurant specific. Each customer is preassigned to a given restaurant determined by that customer's group $j$. Upon entering the $j$ th restaurant in the CRF, customer $y_{j i}$ sits at currently occupied tables $t_{j i}$ with probability proportional to the number of currently seated customers, or starts a new table $T_j+1$ with probability proportional to $\alpha$. The first customer to sit at a table goes to the buffet line and picks a dish $k_{j t}$ for their table, choosing the dish with probability proportional to the number of times that dish has been picked previously, or ordering a new dish $\theta_{K+1}$ with probability proportional to $\gamma$. The intuition behind this predictive distribution is that integrating over the global dish probabilities $\beta$ results in customers making decisions based on the observed popularity of the dishes throughout the entire franchise.

\subsection{Random Effects vs. Fixed-Effects}

Regarding the connection between Bayesian and frequentists' panel data model, according to \citet{koop2003}, if we impose a
\begin{enumerate}[(i)]
	\item non-hierarchical prior (such as Normal prior without hyperpriors) on the intercept $\alpha_{i0}$, the resulting panel data model is equivalent to the frequentist fixed-effects model. This is basically a Bayesian linear regression with standard priors on parameters.
	
	\item hierarchical prior on the intercept $\alpha_{i0}$, the resulting panel data model is equivalent to the frequentists' random effects model. A convenient hierarchical prior assumes that, for $i=1, \ldots, N$,
	$$
	\alpha_{i0} \sim N\left(\mu_{\alpha}, \sigma^2_{\alpha}\right).
	$$
	The hierarchical structure of the prior arises if we treat $\mu_{\alpha}$ and $V_{\alpha}$ as unknown parameters which require their own prior. We assume $\mu_{\alpha}$ and $V_{\alpha}$ to be independent of one another with
	$$
	\mu_{\alpha} \sim N\left(\underline{\mu}_{\alpha}, \underline{\sigma}_{\alpha}^{2}\right),
	$$
	and
	$$
	\sigma^2_{\alpha} \sim IG\left(\underline{\tau}_{\alpha}, \underline{\nu}_{\alpha}\right).
	$$
	This is analogous to the random effects model as $\alpha_{i}$ are essentially assumed to draw from the underlying distribution, and data are used to update our prior on the hyperparameters of the underlying distribution.
\end{enumerate}

The discussion in \citet{koop2003} is in line with the hierarchical models discussed in \citet{smith1973}. The panel data model equipped with a non-hierarchical prior is a two‐stage hierarchical model which results in a fixed effects model, while incorporating a hierarchical prior forms a three‐stage hierarchical model that corresponds to a random effects model.

Back to our settings, if the baseline prior for $\alpha_{0i}$ is a DP (DPM) prior, then our proposed nonparametric Bayesian prior is a type of non-hierarchical (hierarchical) prior with latent group structure in intercepts and hence we call our proposed estimator as the constrained grouped fixed (random) effects estimator.


\section{Priors} \label{appendix:prior}



\noindent {\bf Prior on Group-Specific Parameters}
\begin{align}
	\left( \alpha_i, \sigma_i^2\right) \sim \sum_{k=1}^{\infty} \pi_k \delta_{\left(\alpha_{k}, \sigma_{k}^2\right)} \text{ with } \left(\alpha_{k}, \sigma_{k}^2\right) \sim B_{0}(\phi).
\end{align}
$B_0$ is an Independent Normal Inverse-Gamma (INIG) distribution:		
\begin{align}
	B_0 :=  INIG \left( \mu_{\alpha}, \Sigma_{\alpha}, \frac{\nu_{\sigma}}{2}, \frac{\delta_{\sigma}}{2} \right),
\end{align} 
with a set of hyperparameters $\phi =  \left( \mu_{\alpha}, \Sigma_{\alpha}, \frac{\nu_{\sigma}}{2}, \frac{\delta_{\sigma}}{2} \right) = \left(0, 1, 6, 5\right)$.

\vspace{0.5cm}
\noindent {\bf Prior on Stick Lengths}
\begin{align}
	\xi_k \sim Beta \left(1, a\right),
\end{align} 
where $a$ is the concentration parameter.

\vspace{0.5cm}
\noindent {\bf Hyper-prior on Concentration Parameter}
\begin{align}
	a \sim Gamma \left(m, n\right),
\end{align} 
with $(m,n) = (0.4, 10)$.

\vspace{0.5cm}
\noindent {\bf Prior on Common and Individual Slope Coefficients (if any)}
Finally, the prior distribution for the common parameter $\rho$ is chosen to be a normal distribution to stay close to the linear regression framework,
\begin{align}
	\rho \sim N (0, \sigma^2_{\rho}) \text{ with }  \sigma^2_{\rho} = 1.
\end{align}
The prior of heterogeneous parameter $\beta_i$ follows,
\begin{align}
	\beta_i \sim N (0, \Sigma_{\beta}) \text{ with }  \Sigma_{\beta} = 1 \times \mathbf{I}_p.
\end{align}

\section{Posterior Distributions and Algorithms} \label{appendix:algo}

\subsection{Blocked Gibbs Sampler and Algorithm} \label{subsec:algo_soft}
Initialization:
\begin{enumerate}[(i)] \vspace{-0.5cm}
	\item Preset the initial number of active groups as $K_0^a = N$.

	\item Set concentration parameter $a$ to its prior mean.

	\item In ignorance of group heterogeneity ($K = 1$) and heteroskedasticity, use \citet{anderson1982} IV approach to get $\hat{\alpha}_{IV}$ and $\hat{\Sigma}_{\alpha, IV}$. These IV estimators serve as the mean and covariance matrix in the related priors.

	\item Generate $K_0^a$ random sample from the distribution $N(\hat{\alpha}_{IV},\hat{\Sigma}_{\alpha, IV})$.

	\item Initialize group membership $G$ by using assuming no group structure: $G^{(0)} = [1,2,...,N]$.
	%
\end{enumerate}
For each iteration $s = 1,2,..,N_{sim}$
\begin{enumerate}[(i)]
	\item Number of active groups:
	\begin{align*}
		K^a = \max_{1 \le i \le N} g_i^{(s-1)}.
	\end{align*}
	\item Group ``stick length'': for $k = 1,2,...,K^a$, draw $\xi_k$ from a Beta distribution in (\ref{post_Xi_app}):
	\begin{align*}
		\xi_k | a^{(s-1)}, G^{(s-1)} \sim \textit{Beta} \left( |B_k| + 1, a + \sum_{j=1}^N \mathbf{1}(g_j>k) \right),
	\end{align*}
	and calculate group probability in accordance to (\ref{post_p_app}).
	
	\item Group heterogeneity: for $k = 1,2,...,K^a$, draw $\alpha_k^{(s)}$ from a normal distribution in (\ref{post_alpha_app}):
	\begin{align*}
		\alpha_k |\rho^{(s-1)},\beta^{(s-1)}, \Sigma^{(s-1)}, G^{(s-1)},Y, X \sim \textit{N} \left( \bar{\mu}_{\alpha_k}, \bar{\Sigma}_{\alpha_k} \right).
	\end{align*}
	
	\item Group heteroscedasticity: for $k = 1,2,...,K^a$ and $t = 1,2,...,T$, draw $\sigma_{k}^{2^{(s)}}$ from an inverse Gamma distribution in (\ref{post_Sigma_app}):
	\begin{align*}
		\sigma_{k}^2 |\rho^{(s-1)},\beta^{(s-1)}, \alpha^{(s)}, G^{(s-1)}, Y, X \sim \textit{IG} \left( \frac{\bar{v}_{\sigma,k}}{2}, \frac{\bar{\delta}_{\sigma,k}}{2} \right).
	\end{align*}

	\item Auxiliary variables: for $i = 1,2,...,N$, draw $u_i$ from an uniform distribution in (\ref{post_u_app}):
	\begin{align*}
		u_i | \Xi^{(s)}, G^{(s-1)} \sim Unif (0, p^{(s)}_{g_i}).
	\end{align*}
	Then calculate $u^*$ according to (\ref{def:ustar_app}).

	\item DP concentration parameter:
	\begin{enumerate}[(a)]
		\item Draw latent	variable $\eta$ from a Beta distribution in (\ref{post_eta_app}):
		\begin{align*}
			\eta \sim \textit{Beta} (a+1, N)
		\end{align*}
		\item Draw concentration parameter $a$ from a mixture of Gamma distribution in (\ref{post_a_app}):
		\[
		a | \eta, K^a \sim \left\{
		\begin{array}{ll}
			\textit{Gamma} \left (m+K^a, n - \log(\eta) \right) & \mbox{with prob. } \pi_a \\
			\textit{Gamma} \left (m+K^a-1, n - \log(\eta) \right) & \mbox{with prob. } 1-\pi_a
		\end{array}
		\right. ,
		\]
		and $\pi_a$ is defined as
		\begin{align*}
			\fra{\pi_a}{1-\pi_a} = \frac{m+K^a-1}{N(n-\log(\eta))}.
		\end{align*}
	\end{enumerate}

		\item Potential groups: start with $\tilde{K} = K^a$,
	\begin{enumerate}[(a)]
		\item Group probabilities:
		\begin{enumerate} [(1)]
			\item if $ \sum_{j=1}^{\tilde{K}} \pi_{j}^{(s)} > 1-u^*$, set $K^* = \tilde{K}$ and stop.
			\item otherwise, let $\tilde{K} = \tilde{K}+1$, draw $\xi_{\tilde{K}} \sim Beta \left(1, \alpha^{(s)}\right),$ update $\pi_{\tilde{K}} = \xi_{\tilde{K}} \prod_{j<\tilde{K}}\left(1- \xi_{j} \right)$ and go to step $(1)$.
		\end{enumerate}
		\item Group parameters:  for $k = K+1, \cdots, K^*$, draw $\alpha_{k}^{(s)}$ and $\sigma_{k}^{2(s)}$ from their prior distributions.
	\end{enumerate}
	
%

	\item[(xi)] Group membership: for $j = 1,2,...,J$ and $k = 1,2,...,K^a$, draw $g_j$ from a multinomial distribution in (\ref{post_G_app_soft}).
\end{enumerate}

\subsection{Random Coefficients Model with Soft Constraints} \label{appendix:post_soft}

We present the conditional posterior distributions of parameters in the time-invariant random effects model with heteroscedasticity, positive-link constraints and negative-link constraints, which is the most complicated scenarios. For other models, such as its homoscedastic counterparts, adjustment can be easily made by assuming common error variances.

\subsubsection{Derivation}

Model:
\begin{align}
	y_{it} &= \alpha'_{g_{i}} x_{it} + \varepsilon_{i t}, \quad \varepsilon_{i t} \sim N \left(0, \sigma_{g_{i}}^{2}\right).
\end{align}

To facilitate derivation, we stack observations and parameters,
\begin{align*}
	\text{Dependent variable: } Y &= [y_1, y_2, \ldots, y_N ], y_i = [y_{i1},y_{i2}, \ldots, y_{iT}]', \\
	\text{Covariates: } X &= [x_1, x_2, \ldots, x_N ], x_i = [x_{i1},x_{i2}, \ldots, x_{iT}]', \\
	\text{Grouped-specific parameters: }\alpha &= \left[\alpha_{1}, \alpha_{2}, \ldots, \alpha_{N} \right], \\
	\text{Error variance: } \Sigma &= \left[\sigma^2_{1}, \sigma^2_{2}, \ldots, \sigma^2_{N} \right],\\
	\text{Stick length: } \Xi &= \left[\xi_{1}, \xi_{2}, \ldots\right], \\
	\text{Group indices: } G &= \left[g_{1}, \ldots, g_{N}\right], \\
	\text{Auxiliary variable: } u &= \left[  u_1, u_2, ..., u_N \right],\\
	\text{Hyper parameters: } \phi &= \left[ \mu_{\alpha}, \Sigma_{\alpha}, \nu_{\sigma}, \delta_{\sigma}\right ].
\end{align*}


In order to write down the posterior of unknown parameters given a set of pairwise constraints, a probabilistic model of how weights of constraints are obtained must be specified. Inspired by \citet{shental2003}, we have the following assumptions:
\begin{assumption}(Data)\\ \vspace{-1cm}
	\begin{enumerate}[(i)]
		\item Data points are first sampled i.i.d from the full probability distribution conditional on $G$.
		\item From this sample, pairs of points are randomly chosen according to a uniform distribution. In case both points in a pair belong to the same source a positive-link constraint is formed and a negative-link if formed when they belong to different sources.
	\end{enumerate}
\end{assumption}

The posterior of unknown objects in the random coefficients model is,
\begin{align*}
	& p(\alpha, \sigma^2, \Xi, a, G | Y, X, W, \phi) \\
	\propto \; & p(Y | X, \alpha, \sigma^2, G) p(\alpha, \sigma^2|\phi) p(G | \Xi, W) p(\Xi | a) p(a). \numberthis
\end{align*}

All priors have been well-defined except for $p(G | \Xi, W)$ - the prior for group indices $G$ conditional on stick lengths $\Xi$ and the weights of constraints $W$. 

Using the Bayes rule, the modified prior for the group indices is
\begin{align}  \label{eq:G_given_W}
	p(G | \Xi, W) = \frac{p(W | G) p (G | \Xi)}{\sum_{G'} p(W | G') p (G' | \Xi)} \; \propto \; p(W | G) p (G | \Xi),
\end{align}
where the sum in the denominator is taken over all possible group partitioning, $p(W | G)$ is the weighting function of the form:
\begin{align*}
	p(W | G) = \prod_{i, j} \exp  \left( c W_{i j} \delta_{ij} \right),
\end{align*}
and $p (G | \Xi)$ is the density of a categorical distribution with probabilities generated by the stick-breaking process.

From (\ref{eq:G_given_W}), the prior of $g_i$ conditional on the group indices of the other $G^{(-i)}$ is
\begin{align}
	p(g_i | \Xi, W_{i}, G^{(-i)}) \; \propto \; p \left(W_{i} | G \right) p(g_i | \Xi),
\end{align}
where $W_{i} = \{W_{ij} | j = 1,...,N\}$ and
\begin{align}
	p(W_{i} | G) = \prod_{j= 1}^N \exp \left( 2 c W_{i j} \delta_{ij} \right).
\end{align}

Given the expression of $p(g_i | \Xi, W_{i}, G^{(-i)})$ and the DP prior specified in Appendix \ref{appendix:prior}, the posterior of unknown objects in the random coefficients model can be written as,
\begin{align*}
	& p(\alpha, \sigma^2, \Xi, a, G | Y, X, W, \phi) \\
	\propto \; & p(Y | X, \alpha, \sigma^2, G) p(\alpha, \sigma^2|\phi) p(G | \Xi, W) p(\Xi | a) p(a)\\
	\propto \; & \prod_{i=1}^{N} p \left(y_i | x_i, \alpha_{g_{i}}, \sigma^2_{g_{i}} \right) 
	\prod_{j=1}^{\infty} p(\alpha_{j}, \sigma^2_{j} | \phi) 
	\prod_{j=1}^{\infty} p(\xi_j | a) 
	\prod_{i=1}^{N} p \left(g_i | \Xi, W_i, G^{(-i)} \right) 
	p(a) \\
	= \; & \left[ \prod_{i=1}^{N} p \left(y_i | x_i, \alpha_{g_{i}}, \sigma^2_{g_{i}} \right) 
	p \left(g_i | \Xi, W_i, G^{(-i)} \right) \right] \left[ \prod_{j=1}^{\infty}  p(\alpha_{j}, \sigma^2_{j} | \phi) p(\xi_j | a)\right]  p(a) \\
	= \; & \left[ \prod_{i=1}^{N} p \left(y_i | x_i, \alpha_{g_{i}}, \sigma^2_{g_{i}} \right)  p\left(W_{i} | G \right) p(g_i | \Xi) \right] \left[ \prod_{j=1}^{\infty}  p(\alpha_{j}, \sigma^2_{j} | \phi) p(\xi_j | a)\right] p(a). \numberthis
\end{align*}

In the following derivation and algorithm, we adopt the slice sampler \citep{walker2007} that avoids approximation in \citet{ishwaran2001}. \citet{walker2007} augments the posterior distribution with a set of auxiliary variables $u = [u_1, u_2, ..., u_N]$, which are i.i.d. standard uniform random variables, i.e, $u_i \stackrel{iid}{\sim} U(0,1)$. Then the augmented posterior is written as,
\begin{align*}
	&  p(\alpha, \sigma^2, \Xi, a, G, u | Y, X, W, \phi) \\
	\propto \; &  \left[ \prod_{i=1}^{N} p \left(y_i | x_i, \alpha_{g_{i}}, \sigma^2_{g_{i}} \right)  \mathbf{1} (u_i < \pi_{g_i}) p\left(W_{i} | G \right) \right] \left[ \prod_{j=1}^{\infty}  p(\alpha_{j}, \sigma^2_{j} | \phi) p(\xi_j | a)\right]  p(a) \\
	= \; & \left[ \prod_{i=1}^{N} p \left(y_i | x_i, \alpha_{g_{i}}, \sigma^2_{g_{i}} \right)  p(u_i | \pi_{g_i}) \pi_{g_i}  p\left(W_{i} | G \right) \right] \left[ \prod_{j=1}^{\infty}  p(\alpha_{j}, \sigma^2_{j} | \phi) p(\xi_j | a)\right]  p(a), \numberthis
\end{align*}
where $\pi_{g_i} = p(g_i | \Xi)$, $p(u_i | \pi_{g_i})$ is a uniform distribution defined on $[0, \pi_{g_i}]$, and $\mathbf{1}(\cdot)$ is the indicator function, which is equal to zero unless the specific condition is met. The original posterior can be recovered by integrating out $u_i$ for $i = 1,2,...,N$. As we don't limit the upper bound of the number of groups, it is impossible to sample from an infinite-dimensional posterior density. The merit of slice-sampling is that it reduces the dimensions and allows us to solve a manageable problem with finite dimensions, which we will see below.

With a set of auxiliary variables $u = [u_1, u_2, ..., u_N ]$, we define the largest possible number of potential components as
\begin{align} \label{def:kstar_app}
	K^* = \min_k \left\{ u^* > 1 - \sum_{j=1}^k \pi_j \right\},
\end{align}
where
\begin{align} \label{def:ustar_app}
	u^* = \min_{1 \le i \le N} u_i.
\end{align}
Such a specification ensures that for any group $k > K^*$ and any unit $i \in \{1,2,...,N\}$, we have $u_i > \pi_k$.\footnote{See proof in theorem \ref{prop:1}.} This crucial property limits the dimension of $(\alpha_k, \sigma^2_{k})$ to $K^*$ as the densities of $(\alpha_k, \sigma^2_{k})$ and equal 0 for $k > K^*$ due to  $\mathbf{1} (u_i < \pi_k) = 0$, which will be clear in the subsequent posterior derivation. Intuitively, the latent variable $u_i$ has an effect of ``dynamically truncating" the number of groups needed to be sampled.

Next, we define the number of active groups
\begin{align} \label{def:k_a_app}
	K^a = \max_{1 \le i \le N} g_i.
\end{align}
It can be shown that $K^a \le K^*$.\footnote{See proof in theorem \ref{prop:1}.}

As the base distribution $B_0$ is the Independent-Normal-Inverse-Gamma distribution, the prior density of $\alpha_i$ and $\sigma^2_i$ are independent.

\noindent {\bf Conditional posterior of $\alpha$ (grouped coefficients)}. 

\begin{align*}
	p(\alpha | \sigma^2, G, Y, X, \phi)
	\; \propto \; & \left[ \prod_{i = 1}^N  p \left(y_i  | x_i, \alpha_{g_{i}}, \sigma^2_{g_{i}} \right)  \right] \left[ \prod_{j=1}^{\infty}  p(\alpha_{j}, \sigma^2_{j} | \phi) \right].
\end{align*}

\noindent For $k = 1,2,...,K^a$, define a set of units that belong to the group $k$,
\begin{align} \label{def:Sk_app}
	B_k = \left \{i | g_i = k, i \in \{1,2,...,N\} \right \},
\end{align}
then the posterior density for $\alpha_k$ read as
\begin{align*}
	& p(\alpha_k | \sigma_k^2, G, Y, X, \phi) \\
	\propto \; & \left[ \prod_{i \in B_k} p(y_i |  x_i, \alpha_{g_{i}}, \sigma^2_{g_{i}} ) \right] p(\alpha_{k}| \phi) \\
	\propto \; & \exp \left[ - \frac{1}{2\sigma^2_k}\sum_{i \in B_k} \left( y_i - x_i \alpha_k \right)' \left( y_i - x_i \alpha_k \right) \right] \exp \left[ - \frac{1}{2} \left( \alpha_k - \mu_\alpha \right)'\Sigma_\alpha^{-1} \left( \alpha_k - \mu_\alpha \right) \right] \\
	\propto \; & \exp \left[ - \frac{1}{2} \left( \alpha_k - \bar{\mu}_{\alpha_k} \right)' \bar{\Sigma}_{\alpha_k}^{-1} \left( \alpha_k - \bar{\mu}_{\alpha_k} \right) \right].
\end{align*}

Assuming an independent normal conjugate prior for $\alpha_k$, the posterior for $\alpha_k$ is given by
\begin{align} \label{post_alpha_app}
	\alpha_k |  \sigma_k^2, G, Y, X, \phi \sim N \left( \bar{\mu}_{\alpha_k}, \bar{\Sigma}_{\alpha_k} \right).
\end{align}
where
\begin{align*}
	\bar{\Sigma}_{\alpha_k} &= \left( \Sigma_\alpha^{-1} + \sigma_k^{-2} \sum_{i \in B_k} x_i' x_i \right)^{-1}, \\
	\bar{\mu}_{\alpha_k} &= \bar{\Sigma}_{\alpha_k} \left(\Sigma_\alpha^{-1} \mu_\alpha + \sigma_k^{-2} \sum_{i \in B_k} x_i'  y^{\alpha}_i \right), \\
	y^{\alpha}_i &= y_i -x_i \alpha_{g_i}.
\end{align*}

\noindent {\bf Conditional posterior of $\sigma^2$ (grouped variance)}.
Under the assumption of cross-sectional independence, for $k = 1,2,...,K^a$,
\begin{align*}
	p(\sigma^2_k| \alpha_k, G, Y, X, \phi)
	\; \propto \; & \left[ \prod_{i \in B_k}  p \left(y_i  | x_i,\alpha_{g_i}, \sigma^2_{g_i} \right)  \right] p(\sigma_{k}^2 | \phi).
\end{align*}
With a inverse-gamma prior $\sigma_{k}^2 \sim IG \left( \frac{v_\sigma}{2}, \frac{\delta_\sigma}{2} \right)$, the posterior distribution of $\sigma_{k}^2$ is
\begin{align*}
	& p(\sigma^2_k | \alpha_k, G, Y, X, \phi) \\
	\propto \; & \prod_{i \in B_k} \left[ \left(\sigma_{k}^2 \right)^{-\frac{T}{2}} \exp \left( - \frac{1}{2\sigma^2_k} \left( y_i - x_i \alpha_k \right)' \left( y_i - x_i \alpha_k \right) \right)\right]  \left( \frac{1}{\sigma_{k}^2} \right)^{\frac{v_\sigma}{2}+1} \exp \left( - \frac{\delta_\sigma}{2\sigma_{k}^2} \right) \\
	= \; & \left( \frac{1}{\sigma_{k}^2} \right)^{\frac{v_\sigma + T|B_k|}{2}+1} \exp \left[ - \frac{\delta_\sigma + \sum_{i \in B_k} \left( y_i - x_i \alpha_k \right)' \left( y_i - x_i \alpha_k \right) }{2\sigma_{k}^2} \right].
\end{align*}
This implies
\begin{align} \label{post_Sigma_app}
	\sigma^2_k|\alpha_k, G, Y, X, \phi \sim IG \left( \frac{\bar{v}_{\sigma,k}}{2}, \frac{\bar{\delta}_{\sigma,k}}{2} \right),
\end{align}
where
\begin{align*}
	\bar{v}_{\sigma,k} &= v_\sigma + T|B_k|, \\
	\bar{\delta}_{\sigma,kt} &= \delta_\sigma + \sum_{i \in B_k}  \left( y_i - x_i \alpha_k \right)' \left( y_i - x_i \alpha_k \right), \\
	|B_k| &= \text{\# of units in group } k.
\end{align*}

\noindent {\bf Conditional posterior of $\Xi$ (stick length)}.
\begin{align*}
	& p(\Xi |  a, G) \\
	\propto \; & \left[ \prod_{i=1}^{N} p(u_i | \pi_{g_i}) \pi_{g_i}\right] \left[ \prod_{j=1}^{\infty} p(\xi_j | a)\right] \\
	\propto \; & \left[ \prod_{i=1}^{N} p(u_i | \pi_{g_i}) \xi_{g_i} \prod_{l < g_i} (1- \xi_l) \right] \left[ \prod_{j=1}^{\infty} p(\xi_j | a)\right].
\end{align*}

For $k = 1, 2, ..., K^a$,
\begin{align*}
	p(\Xi |  a, G )\;
	\propto \; & \left( \prod_{i \in B_k} \xi_k \right) (1-\xi_k)^{\sum \limits_{j=1}^N \mathbf{1}(g_j>k)} (1-\xi_k)^{a-1}, \\
	\propto \; & \xi_k^{|B_k|} \left( 1-\xi_k \right)^{a + \sum \limits_{j=1}^N \mathbf{1}(g_j>k) -1}.
\end{align*}
where $B_K$ is the set of units that currently belong to group $k$, see equation (\ref{def:Sk_app}).

Therefore, posterior distribution of $\xi_k$ is
\begin{align} \label{post_Xi_app}
	\xi_k | a, G \sim \textit{Beta} \left( |B_k| + 1, a + \sum_{j=1}^N \mathbf{1}(g_j>k) \right).
\end{align}

Give $\Xi = [\xi_1, \xi_2, ..., \xi_{K^a}]$, update group probabilities $\pi_1, \pi_2,...,\pi_{K^a}$:

\begin{align} \label{post_p_app}
	\pi_{k} | G,\Xi = \left\{
	\begin{array}{ll}
		{\xi_{1},} & {k = 1} \\
		{\xi_{k} \prod_{j<k}\left(1-\xi_{j}\right),} & {k = 2, \ldots, K^a}
	\end{array}
	\right..
\end{align}

\noindent {\bf Conditional posterior of $a$ (concentration parameter)}.
Regarding the DP concentration parameter, the standard posterior derivation doesn't work due to the unrestricted number of components in the current sampler. Instead, we implement the 2-step procedure proposed by \citet{escobar1995} (p.8-9). Following their approach, we first draw a latent variable $\eta$,
\begin{align} \label{post_eta_app}
	\eta \sim \textit{Beta } (a+1, J).
\end{align}
Then, conditional on $\eta$ and $K^a$, we draw $a$ from a mixture of two Gamma distribution:
\begin{align} \label{post_a_app}
	p(a | \eta, K^a) = \pi_a \textit{Gamma } (m+K^a, n-\log(\eta)) + (1 - \pi_a) \textit{Gamma } (m+K^a-1, n-\log(\eta)),
\end{align}
with the weights $\pi_a$ defined by
\begin{align*}
	\fra{\pi_a}{1-\pi_a} = \frac{m+K^a-1}{N[n-\log(\eta)]}.
\end{align*}

\noindent {\bf Conditional posterior of $u$ (auxiliary variable)}.
Conditional on the group ``stick lengths'' $\xi_k$ and group indices $G$, it is straightforward to show that the posterior density of $u_i$ is a uniform distribution defined on $(0, \pi_{g_i})$, that is
\begin{align} \label{post_u_app}
	u_i | \Xi, G \sim \textit{Unif }(0, \pi_{g_i}),
\end{align}
where $\pi_{g_i} = \xi_{g_i} \prod_{j<g_j} (1-\xi_j)$. Moreover, it is worth noting that the values for $K^*$ and $u^*$ need to be updated according to equation (\ref{def:kstar_app}) and  (\ref{def:ustar_app}) after this step.\\

\noindent {\bf Conditional posterior of $G$ (group indices)}.
We derive the posterior distribution of $g_i$ consider on $G^{(-i)}$, where $G^{(-i)}$ is a set including all member indices except for $g_i$, i.e., $G^{(-i)} = G \char`\\ g_i$. As a result, for $k = 1, 2, ..., K^*$,
\begin{align*} \label{post_G_app_soft}
	& p \left(g_i = k | y_i, x_i, \alpha_{k}, \sigma^2_{k}, G^{(-i)}, u_i\right) \\
	\propto \; & p \left(y_i | x_i, \alpha_{k}, \sigma^2_{k}\right)  \mathbf{1} (u_i < \pi_{k})  p \left(W_{i} | G \right)  \\
	= \; & p \left(y_i | x_i, \alpha_{k}, \sigma^2_{k} \right)  \mathbf{1} (u_i < \pi_{k}) \prod_{j = 1}^N \exp \left( 2 c W_{i j} \delta_{i j} \right) . \numberthis
\end{align*}
Finally, we normalize the point mass to get a valid distribution.

\section{Technical Proofs}

\subsection{Slice Sampling}
\begin{theorem} \label{prop:1}
	Suppose that we have a model with posterior as given in Appendix \ref{appendix:post_soft}. Given the definition of the number of potential component $K^*$ in (\ref{def:kstar_app}), the minimum of auxiliary variables $u^*$ in (\ref{def:ustar_app}) and the number of active group $K$ in (\ref{def:k_a_app}), we have
	\begin{enumerate}[(i)]
		\item $u_i > \pi_k$ for $\forall i = 1,2,...,n$ and $\forall k > K^*$;
		\item $K < K^*$.
	\end{enumerate}
\end{theorem}

\begin{proof}
	\begin{enumerate}[(i)]
		\item By definition, $u^* = \min \limits_{1\le i \le N} u_i$ for $i = 1,2,...,n$, then,
		\begin{align*}
			u_i \ge u^* > 1 - \sum_{j=1}^{K^*} \pi_j =  \sum_{j=K^*}^{\infty} \pi_j \ge \pi_k, \forall k > K^*.
		\end{align*}
		\item Let $i'$ be an unit $i$ such that $g_{i'} = K$. According to the posterior of $G$, the group $K$ exists if $u_{i'} < \pi_{K}$, otherwise $p(g_i = K | \cdot) = 0$. Then by definition,
		\begin{align*}
			u^* \le u_{i'} < \pi_{K} \Rightarrow 1 - u^* > 1 - \pi_{K} = \sum_{j = 1}^{K-1} \pi_j.
		\end{align*}
		Since $K^*$ is the smallest number $s.t.$ $1 - u^* < \sum \limits_{j = 1}^{K^*} \pi_j$, then $K \le K^*$.
	\end{enumerate}
\end{proof}

\subsection{Connection to \citet{lu2004} and \citet{lu2007_2}} \label{appendix:LuAndLeen}

In this section, we will first show the close connection between the modified prior in the presence of soft constraints defined in (\ref{eq:prior_G_soft}) and the framework of penalized probabilistic clustering proposed by \citet{lu2004} and \citet{lu2007_2}. Then we will discuss the properties of the weights $W_{ij}$.

We start with joint prior odds in (\ref{eq:prior_G_soft}):
\begin{align*} \label{eq:appendix_prior_soft}
	\prod_{i, j}  \left( \fra{\psi_{i j}}{1 - \psi_{i j}} \right)^{c T_{i j} \delta_{ij}} = \prod_{i, j}  \exp \left[ c \delta_{ij} \log \left( \fra{\psi_{i j}}{1 - \psi_{i j}} \right)^ {T_{i j}}\right]. \numberthis
\end{align*}
Define the weight as $W_{i j} = \log \left( \fra{\psi_{i j}}{1 - \psi_{i j}} \right)^ {T_{i j} }$. Then when $T_{ij} = 1$, we have
\begin{align} \label{eq:lu_ml}
	W_{i j} = \log \left( \fra{\psi_{i j}}{1 - \psi_{i j}} \right) \Leftrightarrow \psi_{i j} = \frac{\exp \left( W_{i j} \right)}{1 + \exp \left( W_{i j} \right)}.
\end{align}

When $T_{ij} = -1$, we get
\begin{align} \label{eq:lu_cl}
	W_{i j} = \log \left( \fra{1- \psi_{i j}}{\psi_{i j}} \right) \Leftrightarrow 1- \psi_{i j} = \frac{\exp \left( W_{i j} \right)}{1 + \exp \left( W_{i j} \right)}.
\end{align}

Combining (\ref{eq:lu_ml}) and (\ref{eq:lu_cl}) yields that
\begin{align} \label{eq:lu_combine}
	 \frac{\exp \left( W_{i j} \right)}{1 + \exp \left( W_{i j} \right)} = \psi_{i j}^{\frac{1}{2}(1+T_{i j})} (1- \psi_{i j})^{\frac{1}{2}(1 -T_{i j})}.
\end{align}
This is exactly the equation (7) in \citet{lu2007_2} with $\gamma_{i j} = \psi_{i j}$ and $L_{i j} = \frac{1}{2} \left(T_{i j} + 1\right)$, which uniquely defines the expression for the weights associated with each pairwise constraint given $\gamma_{i j}$ and $L_{i j}$. Since both $L_{i j}$ and $T_{i j}$ are indicators for the type of constraints, they don't affect the formula for $W_{i j}$, thus the following formula weights coincides with the one used in \citet{lu2007_2} and the both frameworks converge,
\begin{align} \label{eq:append_weight}
	W_{ij} & =
	\left\{
	\begin{array}{ll}
		\log \left( \frac{\psi_{ij}}{1-\psi_{ij}}\right)  & \text { if } T_{i j} = 1\\
		\log \left( \frac{1- \psi_{ij}}{\psi_{ij}}\right)  & \text { if } T_{i j} = -1
		\vspace{0.1cm}\\
		0 & \text { if } T_{i j} = 0.
	\end{array}
	\right.
\end{align}

Accordingly, the prior defined in (\ref{eq:appendix_prior_soft}) can be rewritten in term of $W_{i j}$ as 
\begin{align}
	 \prod_{i, j}  \left( \fra{\psi_{i j}}{1 - \psi_{i j}} \right)^{c T_{i j} \delta_{ij}} = \prod_{i, j}  \exp \left( c W_{ij} \delta_{ij}\right).
\end{align}

The weight $W_{i j}$ associated with the constraint between unit $i$ and $j$ as in (\ref{eq:append_weight}) has the following properties:
\begin{enumerate}[(a)]
	\item Unboundedness: $W_{i j} \in (-\infty, \infty)$;
	\item Symmetry: $W_{i j} = W_{j i}$;
	\item Sign reflects constraint's type: If $(i,j) \in \mathcal{P}$ or $L_{i j} = 1$, then $W_{i j} = \log \left( \frac{\psi_{ij}}{1-\psi_{ij}}\right) > 0$; If $(i,j) \in \mathcal{N}$ or $L_{i j} = -1$, then $W_{i j} = \log\left( \frac{1- \psi_{ij}}{\psi_{ij}}\right)  < 0$; If $(i,j)$ doesn't involve in any constraint or $L_{i j} = 0$, then $W_{i j} = 0$.
	\item Absolute value reflects constraint's accuracy: 
	\begin{align*}
		\frac{e^{\left|W_{i j}\right|}}{1+e^{\left|W_{i j}\right|}} = \psi_{i j}.
	\end{align*}
	It is straightforward to show that $\left|W_{i j}\right|$ is monotonically increasing in $ \psi_{i j}$.
\end{enumerate}

\subsection{\textit{Prior Similarity Matrix}}

\begin{proof}[Proof of Theorem \ref{thm:equal_prior_prob}]
	Given equation (\ref{eq:EPPF}) and (\ref{eq:prior_G_soft2}), the prior probability of unit $i$ and $j$ being in the same group is
	\begin{align*}
		& \Pr (g_i = g_j | W ) \\
		= \; & \sum_{G \in \mathcal{G}_{ij}} \frac{1}{M} p (G)  \exp  \left( c \sum_{m, n} W_{m n} \delta_{m n}\right) \\
		= \; &  \sum_{G \in \mathcal{G}_{ij}} \frac{1}{M}  \frac{\Gamma(a)}{\Gamma(a+N)} \left[ \prod_{k=1}^{K} a \Gamma\left( |B_k| \right) \right] \exp  \left( c \sum_{m, n} W_{m n} \delta_{m n} \right) \\
		= \; &  \sum_{G \in \mathcal{G}_{ij}} A(G) \exp  \left( c \sum_{m, n} W_{m n} \delta_{m n} \right) \numberthis
	\end{align*}
	where $\mathcal{G}_{ij}$ is the set of all possible group indices that satisfies $g_i = g_j$ and $M$ is the normalization constant in (\ref{eq:prior_G_soft2}).
	
	$\mathcal{G}_{ij}$ and $\mathcal{G}_{ik}$ are closed related. It is straightforward to see that the numbers of element in $\mathcal{G}_{ij}$ and $\mathcal{G}_{ik}$ are equal since they are all equal to the number of permutation of other $N-2$ units. Moreover, as unit $j$ and $k$ are exchangeable, $\mathcal{G}_{ik}$ can be constructed from $\mathcal{G}_{ij}$ by swapping the group index of unit $j$ and $k$. 
	
	As a result, we can find an one-on-one mapping between $\mathcal{G}_{ij}$ and $\mathcal{G}_{ik}$. That is, for any $G \in \mathcal{G}_{ij}$, if we swap the group index of unit $j$ and $k$, the resulting partition $s_{jk}(G)$ belongs to $\mathcal{G}_{ik}$, and vice versa. As the constant $A(G)$ depends only on the size of partitions, we have $A(G) = A(s_{jk}(G))$.
	
	The properties between $\mathcal{G}_{ij}$ and $\mathcal{G}_{ik}$ enable we to compare each summand in $\Pr (g_i = g_j | W)$ and $\Pr (g_i = g_k | W)$. The difference between these two probabilities is 
	\begin{align*}
		& \Pr (g_i = g_j | W) - \Pr (g_i = g_k | W) \\
		\;=\; & \sum_{G \in \mathcal{G}_{ij}} A(G) \exp  \left( c \sum_{m, n} W_{m n} \delta_{m n} \right) - \sum_{G \in \mathcal{G}_{ik}} A(G) \exp  \left( c \sum_{m, n} W_{m n} \delta_{m n} \right)  \\
		\;=\; & \sum_{G \in \mathcal{G}_{ij}} A(G) \exp  \left( c \sum_{m, n} W_{m n} \delta_{m n} \right) - A(s_{jk}(G)) \exp  \left( c \sum_{m, n} W_{m n} \delta'_{m n} \right) \\
		\;=\; & \sum_{G \in \mathcal{G}_{ij}} A(G) \left[\exp \left( c \sum_{m, n} W_{m n} \delta_{m n} \right) - \exp  \left( c \sum_{m, n} W_{m n} \delta'_{m n} \right) \right]. \numberthis
	\end{align*}
	where $\delta'_{m n}$ is evaluated at $s_{jk}(G)$.
	
	\noindent We can classify a group partitioning $G$ into two cases:
	\begin{enumerate}[(i)]
		\item $G = s_{jk}(G)$. This happens when units $j$ and $k$ are assigned to the same group. Swapping them doesn't affect the group partitioning, which indicates that $\sum_{m, n} W_{m n} \delta_{m n}  = \sum_{m, n} W_{m n} \delta'_{m n}$.
		
		\item $G \ne s_{jk}(G)$. These are the more common cases. We again compare $\sum_{m, n} W_{m n} \delta_{m n}$ with $\sum_{m, n} W_{m n} \delta'_{m n}$. $W_{m n} \delta_{m n}$ and $W_{m n} \delta'_{m n}$ are equal when $m \ne j, k$ and $n \ne j, k$ as these terms remain unchanged regardless of the group indices of units $j$ and $k$. For $m = j,k$, note that $\delta_{jn} = \delta'_{kn}$ and $\delta_{kn} = \delta'_{jn}$ for all $n = 1,2,..,N$. Therefore, under the assumption that $W_{jn} = W_{kn}$ for $\forall n$ , we have,
		\begin{align} \label{eq:proof_eq_prior}
			\sum_{n = 1}^N W_{j n} \delta_{j n} + \sum_{n = 1}^N W_{k n} \delta_{k n} 
			= \sum_{n = 1}^N W_{j n} \delta'_{k n} + \sum_{n = 1}^N W_{k n} \delta'_{jn}  
			= \sum_{n = 1}^N W_{k n} \delta'_{k n} + \sum_{n = 1}^N W_{j n} \delta'_{jn},
		\end{align}
		and hence
		\begin{align*}
			& \sum_{m, n} W_{m n} \delta_{m n}  \\
			\; = \; & \sum_{m, n \not \in (j,k)} W_{m n} \delta_{m n} + 2 \left( \sum_{n = 1}^N W_{j n} \delta_{j n} + \sum_{n = 1}^N W_{k n} \delta_{k n} \right) \\
			\; = \; & \sum_{m, n \not \in (j,k)} W_{m n} \delta'_{m n} + 2 \left( \sum_{n = 1}^N W_{j n} \delta'_{j n} + \sum_{n = 1}^N W_{k n} \delta'_{k n} \right) \\
			\; = \; & \sum_{m, n} W_{m n} \delta'_{m n},
		\end{align*}
		where the first and third equalities use facts that $W_{mn} = W_{nm}$, $\delta_{mn} = \delta_{nm}$, and $W_{nn} = 0$ for $\forall n,m$. The second equality follows the result in (\ref{eq:proof_eq_prior}).
		
		
	\end{enumerate}

	In both cases, we have $\sum_{m, n} W_{m n} \delta_{m n}  = \sum_{m, n} W_{m n} \delta'_{m n}$ for all $G \in \mathcal{G}_{ij}$ and therefore
	\begin{align}
		\Pr (g_i = g_j | W) - \Pr (g_i = g_k | W) = 0.
	\end{align}

\end{proof}

\subsection{\textit{PC-KMeans}}

\begin{proof}[Proof of Theorem \ref{thm:pc_kmean}]
	We start with a brief discussion of \textit{PC-KMeans} algorithm \cite{basu2004}. Given a set of observations $\left(y_{1}, y_{2}, \ldots, y_{N}\right)$, a set of positive-link constraints $\mathcal{P}$, a set of negative-link constraints $\mathcal{N}$, the cost of violating constraints $w = \{w^p_{i j}, w^n_{i j} \}$ and the number of groups $K$, the \textit{PC-KMeans} algorithm aims to partition the $N$ units into $K$ groups so as to minimize the following objective function,
	\begin{align}\label{eq:obj_pckmeans}
		L(G) = & \underbrace{ \frac{1}{2}\sum_{k=1}^{K} \sum_{i \in B_k} \left\|z_i - \mu_{k}\right\|^{2}}_{\text{within-cluster sum of squares}} + \underbrace{\sum_{\left(i,j\right) \in \mathcal{P}} \omega^m_{ij} \mathbf{1}\left( g_{i} \neq g_j\right) + \sum_{\left(i,j\right) \in \mathcal{N}} \omega^c_{i j} \mathbf{1}\left( g_{i}=g_j\right) }_{\text{cost of violation}},
	\end{align}
	where $\mu_{k}$ is the centroid of group $k$, i.e., $\mu_{k}=\frac{1}{\left|B_k\right|} \sum_{i \in B_k} y_i$, $B_k$ is the set of units that are assigned to group $k$, and $\left|B_k\right|$ is the size of group $k$. Equation (\ref{eq:obj_pckmeans}) can be rewritten as
	\begin{align}\label{eq:obj_pckmeans2}
		L(G) = \;& \frac{1}{2}\sum_{i=1}^{N} \left\|y_i - \mu_{g_i}\right\|^{2} - \sum_{i, j} c  W_{i j} \delta_{ij} + Const,
	\end{align}
	where $Const = c \left( \sum_{\left(i,j\right) \in \mathcal{P}} W_{i j} -  \sum_{\left(i,j\right) \in \mathcal{N}} W_{i j}\right)$ is a constant, $c$ is the scaling constant introduced in (\ref{eq:prior_G_soft}), and
	\begin{align*} \label{eq:weight_kmeans_soft}
		W_{ij}
		\; = \; &
		\begin{cases}
			\frac{\omega^m_{ij}}{2c}& \text{ if } \left(i,j\right) \in \mathcal{P}\\
			-\frac{\omega^c_{ij}}{2c}& \text{ if } \left(i,j\right) \in \mathcal{N}\\
			0 &\text{ otherwise}.
		\end{cases} \numberthis
	\end{align*}
	The clustering process includes minimizing the objective function over both group partition $G$ (assignment step) and the model parameters $\mu=\left\{\mu_{1}, \mu_{2}, \ldots, \mu_{K}\right\}$ (update step). Next, we will show that the \textit{PC-KMeans} algorithm is embodied in our proposed Gibbs sampler with soft constraints.

	Under assumption (\ref{soft_as0}), we can rewrite the baseline model with a set of variables $z_{it}$ that don't have grouped heterogeneous effects on $y_{it}$, 
	\begin{align*}
		y_{it} & = \alpha'_{g_i} x_{it} + \beta'_{i} z_{it} + \varepsilon_{i t} =  \alpha_{g_i} + \beta'_{i} z_{it} + \varepsilon_{i t},
	\end{align*}
	where the second equality holds due to the assumption of $x_{it} = 1$. $\beta_{i}$ can be equal across units, i.e., $\beta_{i} = \beta$.
	
	Under assumption (\ref{soft_as1}), we fix the number of groups upfront and thus we don't rely on slice sampling in which $K$ is unknown and determined dynamically. Hereinafter, we focus on posterior distribution without the auxiliary variable $u_i$. Notice that the indicator function $\mathbf{1} (u_i < \pi_{g_i})$ in the posterior density reduces to $\pi_{g_i}$.
	
	\vspace{0.5cm}
	
	{\bf Part 1: Assignment Step}
	
	Assume we have soft pairwise constraints and weights are specified in (\ref{eq:weight_kmeans_soft}). Under the assumptions (\ref{soft_as2}) and (\ref{soft_as3}), the posterior density of the group membership indicators $G$ is,
	\begin{align*} \label{eq:post_G_soft_thm}
		& p(G | \alpha,\beta,  \sigma^2, Y,X, Z, W) \\
		\; = \; & \frac{1}{Z_S} \prod_{i=1}^N \left[p(y_i | \beta_i, \alpha_{g_i}, \sigma^2_{g_i}, x_i, z_i) \pi_{g_i} \right] p(W | G)  \\
		= \; & \frac{1}{Z_S} \prod_{i =1}^N  p(y_i | \beta_i, \alpha_{g_i}, \sigma^2_{g_i}, x_i, z_i)\pi_{g_i}\prod_{i, j = 1}^N \exp \left( \frac{cW_{i j}}{\sigma^2} \delta_{ij} \right) \\
		= \;  & \frac{1}{Z_S} \prod_{i =1}^N (2\pi \sigma^2)^{-\frac{T}{2}}  \pi_{g_i}  \exp \left[ - \frac{1}{2\sigma^2} \left\| \tilde{y}_i - \alpha_{g_i} \right\|^{2} \right]  \prod_{i, j = 1}^N \exp \left( \frac{cW_{i j}}{\sigma^2} \delta_{ij}\right), \numberthis
	\end{align*}
	where $\tilde{y}_i = y_i - \beta_{i}' z_i$, $z_i = [z_{i1} \; z_{i2} \ldots \; z_{iT}]'$ and $Z_S$ is the normalization constant.
	
	Next, we define the optimal group partition $G^*$ that minimizes the objective function of \textit{PC-KMeans} defined in (\ref{eq:obj_pckmeans2}) with $x_i = \tilde{y}_i$ and $\mu_{k} = \alpha_{k}$, that is,
	\begin{align*} \label{eq:sol_kmean_soft}
		G^* \equiv & \arg \min_{G} L(G) \\
		= & \arg \min_{G}  \frac{1}{2}\sum_{i=1}^{N} \left\| \tilde{y}_i - \alpha_{g_i} \right\|^{2} - \sum_{i, j} c W_{i j} \delta_{ij}. \numberthis
	\end{align*}

	Now we consider the asymptotic behavior of the posterior probability in (\ref{eq:post_G_soft_thm}). We will show that as $\sigma^2$ goes to 0, the posterior probability of $G$ approaches 0 for all group partitions except for $G^*$:
	\begin{align*}
		\lim_{\sigma^2 \to 0} p(G | \rho, \beta, \alpha, \Sigma, Y,X, W) \to
		\begin{cases}
			1 & \text{ if } G = G^*; \\
			0 & otherwise. \\
		\end{cases}
	\end{align*}
	
	We start with the log posterior density of $G$ in (\ref{eq:post_G_soft_thm}),
	\begin{align*} \label{eq:log_post_G}
		l (G)
		\equiv  \; & \log p (G | \rho, \beta, \alpha, \Sigma, Y,X, W) \\
		= \;  & -\frac{1}{2\sigma^2} \sum_{i=1}^{N} \left\| \tilde{y}_i - \alpha_{g_i} \right\|^{2} +  \sum_{i, j = 1}^N \left( \frac{cW_{i j}}{\sigma^2} \delta_{ij}\right) \\
		& -\frac{NT}{2} \log (2\pi \sigma^2) + \sum_{i=1}^N \log( \pi_{g_i}) - \log Z_S. \numberthis
	\end{align*}
	
	The difference between two log posterior probabilities evaluated at $G^*$ and any other $G$ is
	\begin{align*}
		& l (G^*) - l (G) \\
		= \; & \frac{1}{\sigma^2}\left[ \left( \frac{1}{2} \sum_{i=1}^{N} \left\| \tilde{y}_i - \alpha_{g_i} \right\|^{2} - \sum_{i, j = 1}^N c W_{ij} \delta_{ij} \right) - \left( \frac{1}{2} \sum_{i=1}^{N} \left\| \tilde{y}_i - \alpha_{g^*_i} \right\|^{2} - \sum_{i, j = 1}^N c W_{ij} \delta^*_{ij} \right) \right]\\
		& + \sum_{i=1}^{N} \left[ \log( \pi_{g^*_i}) - \log( \pi_{g_i}) \right]. \numberthis
	\end{align*}
	The first term is strictly positive according the definition of $G^*$ in (\ref{eq:sol_kmean_soft}). For simplicity, we denote the expression within the first square brace as $V$ and $V>0$. The second term is finite since
	\begin{align*}
		\left| \sum_{i=1}^{N} \left[ \log( \pi_{g^*_i}) - \log( \pi_{g_i}) \right] \right| \le  N \left| \max(\log (\pi_j)) - \min( \log (\pi_j))\right| < +\infty
	\end{align*}
	
	Thus, for any $G \ne G^*$, in the limit as $\sigma^2 \to 0$, we have
	\begin{align*}
		& \lim_{\sigma^2 \to 0} l (G^*) - l (G) = \lim_{\sigma^2 \to 0} \frac{V}{\sigma^2} + \sum_{i=1}^{N} \left[ \log( \pi_{g^*_i}) - \log( \pi_{g_i}) \right] = +\infty. \numberthis
	\end{align*}
	This indicates that
	\begin{align*}
		\lim_{\sigma^2 \to 0} \frac{ p(G | \alpha, \sigma^2, Y, X, Z, W)}{ p(G^* | \alpha, \sigma^2, Y, X, Z, W)} =  \lim_{\sigma^2 \to 0} \exp \left[l (G) - l (G^*) \right] = \exp (-\infty) = 0.
	\end{align*}
	We take the sum over all possible group partitions and get,
	\begin{align*}
		& \lim_{\sigma^2 \to 0} \frac{ \sum_{G'} p(G' |  \alpha, \sigma^2, Y, X, Z, W)}{ p(G^* |  \alpha, \sigma^2, Y, X, Z, W)} \\
		= \; & \lim_{\sigma^2 \to 0} \frac{ \sum_{G' \ne G} p(G' | \alpha, \sigma^2, Y, X, Z, W) + p(G^* |  \alpha, \sigma^2, Y, X, Z, W)}{ p(G^* |  \alpha, \sigma^2, Y, X, Z, W)} \\
		= \; & 1.
	\end{align*}
	Since $\sum_{G'} p(G' | \alpha, \sigma^2, Y, X, Z, W) = 1$, we have
	\begin{align}
		\lim_{\sigma^2 \to 0}  p(G^* | \alpha, \sigma^2, Y, X, Z, W) = 1.
	\end{align}
	
	Therefore, when $\sigma^2 \to 0$, every posterior draw of $G$ from the proposed Gibbs sampler is the solution to the assignment step of the \textit{PC-KMeans} algorithm, conditional on  the posterior draws of other parameters.
	
	\vspace{0.5cm}
	
	{\bf Part 2: Update Step}
	
	During Gibbs sampling, once we have performed one complete set of Gibbs moves on the group assignments and non-group-specific parameters including $\beta_i$ and $\sigma^2$, we need to sample the $\alpha_{k}$ conditioned on all assignments and observations.
	
	Let $|B_k|$ be the number of units assigned to group $k$, then the posterior density for $\alpha_k$ read as
	\begin{align*}
		p(\alpha_k | \beta, \sigma^2, Y, X, Z) \; \propto \;  \exp \left[ - \frac{1}{2}\left( \alpha_k - \bar{\mu}_{\alpha_k} \right)' \bar{\Sigma}_{\alpha_k}^{-1} \left( \alpha_k - \bar{\mu}_{\alpha_k} \right) \right], \numberthis
	\end{align*}
	where
	\begin{align*}
		\bar{\Sigma}_{\alpha_k} &= \left( \Sigma_\alpha^{-1} + |B_k| \sigma^{-2} I_T \right)^{-1}, \\
		\bar{\mu}_{\alpha_k} &= \bar{\Sigma}_{\alpha_k} \left( \Sigma_\alpha^{-1} \mu_\alpha + \sigma^{-2} \sum_{i \in B_k}  \tilde{y}_i \right), \\
		\tilde{y}_i &= y_i - \rho y_{-1,i} - x_i\beta_i.
	\end{align*}
	We can see that the mass of the posterior distribution becomes concentrated around the posterior group mean $\bar{\mu}_{\alpha_k}$ as $\sigma^2 \to 0$. Meanwhile, the posterior group mean $\bar{\mu}_{\alpha_k}$ equals the group ``sample'' mean in the limit:
	\begin{align*}
		\lim_{\sigma^2 \to 0} \bar{\mu}_{\alpha_k} & = \lim_{\sigma^2 \to 0}  \left( \Sigma_\alpha^{-1} + |B_k| \sigma^{-2} I_T \right)^{-1} \left( \Sigma_\alpha^{-1} \mu_\alpha + \sigma^{-2} \sum_{i \in B_k}  \tilde{y}_i \right) \\
		& = \lim_{\sigma^2 \to 0}  \left( \sigma^{2} \Sigma_\alpha^{-1} + |B_k| I_T \right)^{-1} \left( \sigma^{2}\Sigma_\alpha^{-1} \mu_\alpha + \sum_{i \in B_k}  \tilde{y}_i \right) \\
		& = |B_k|^{-1}  \sum_{i \in B_k}  \tilde{y}_i.
	\end{align*}
	In other words, after we determine the assignments of units to groups, we update the means as the ``sample" mean of the units in each group. This is equivalent to the standard \textit{KMeans} cluster update step in general. Of course, we need additional steps to draw $\beta_i$ and $\sigma^2$ before updating group means.
\end{proof}


\section{Monte Carlo Simulation} \label{sec:mcmc}

In this section, we conducted Monte Carlo simulations to examine the performance of various constrained BGFE estimators under different data generating processes (DGPs) and prior belief on $G$. Two sets of DGPs are considered. For the simple DGPs, we introduce various group pattern in the fixed-effects only. The general DGPs, on the other hand, include more covariates with group-specific slope coefficients. Such designs enable us to investigate not only how our proposed estimators perform under various DGPs with specific features, but also the accuracy of estimating the number of groups.

We consider a short-panel environment in which the sample size is $N = 200$ and the time span is $T = 11$. As we focus on one-step ahead forecasts, the last observation of each unit serves as the hold-out sample for evaluation. A similar framework can be applied to $H$-step ahead forecasts by generating additional $H$ observations. The true number of groups is set to $K_0 = 4$. Given $N$ and $K^0$, we divide the entire sample into $K^0$ balanced blocks with $N/K_0$ units in each block.\footnote{If $N/K_0$ is not an integer, we assign $\lfloor N/K_0 \rfloor$ units for group 1,2,..,$K_0-1$ and the last group contains the remainder.} For each DGP, 1,000 datasets are generated, and we run the block Gibbs samplers for each data set with $M = 5,000$ iterations after a burn-in of 5,000 draws.




\subsection{Data Generating Processes} \label{sec:dgp}

\subsubsection{Simple DGPs}
We begin with a simple dynamic panel data model with group pattern in the fixed-effects and no covariates or heteroskedasticity.

\noindent{\bf DGP 1 \& 2:}
\begin{align}
	y_{it} &= \alpha_{g_{i}} + \rho y_{it-1} + \varepsilon_{i t},
\end{align}
where $\rho = 0.7$ and $ \varepsilon_{i t} \sim N \left(0, 1 \right)$. The distributions of initial values are selected to ensure the simulation paths are stationary. Idiosyncratic error $\varepsilon_{i t}$ are independent across $i$ and $t$, and mutually independent. $\varepsilon_{i t}$ is also independent of all regressors. 

We assume $\alpha_{k}$ has zero mean and takes the form $\alpha_{k}= m(k-2.5)$, where $m$ controls the cross-sectional variance of $\alpha_{i}$. Two sets of $\alpha_{k}$ are specified: $m = 1.79$ such that $\operatorname{var}\left(\alpha_{k}\right) = 1/4$ in DGP 1 and $m = 0.51$ such that $\operatorname{var}\left(\alpha_{k}\right) = 1/50$ in DGP 2, see details in Appendix \ref{sec:simple_dgp_detail}. The difference in $\alpha_k$ between these two DGPs distinguishes their properties. As depicted in Figure \ref{fig:MC_DGP}, the group pattern is readily apparent in DGP 1. Different groups of units are perfectly divided, and the simulated paths are pretty flat. DGP 2 has a less visible group structure than DGP 1 because the difference between group means of $\alpha_{k}$ is smaller. The simulated pathways are considerably noisier and fluctuate around the unconditional mean.

\begin{figure}[htp]
	\caption{Simulated Paths for Units in Different DGPs}
	\label{fig:MC_DGP}
	\centering
	\begin{subfigure}[b]{0.35\textwidth}
		\centering
		\caption{DGP 1: Sharp Group}
		\includegraphics[width=\textwidth]{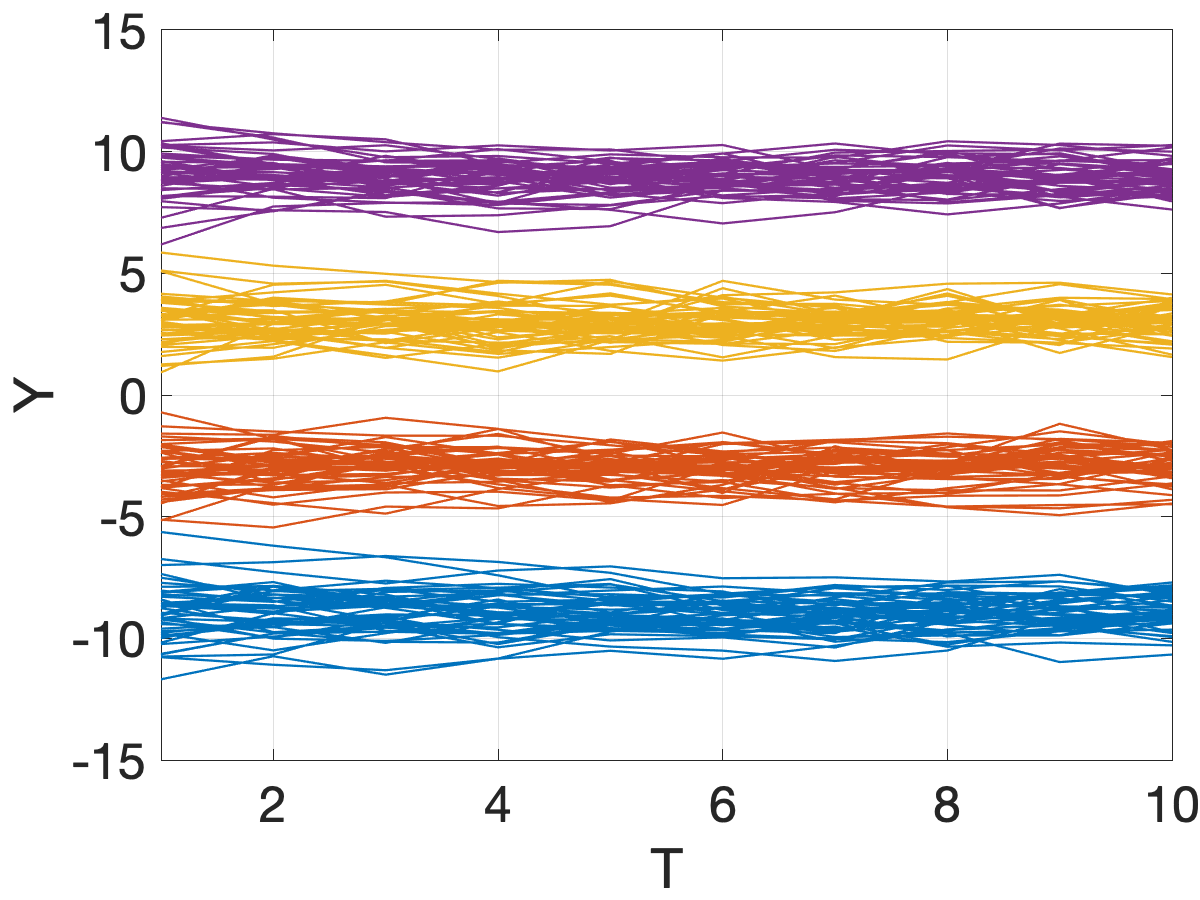}
	\end{subfigure}
	\begin{subfigure}[b]{0.35\textwidth}
		\centering
		\caption{DGP 2: Noisy Group}
		\includegraphics[width=\textwidth]{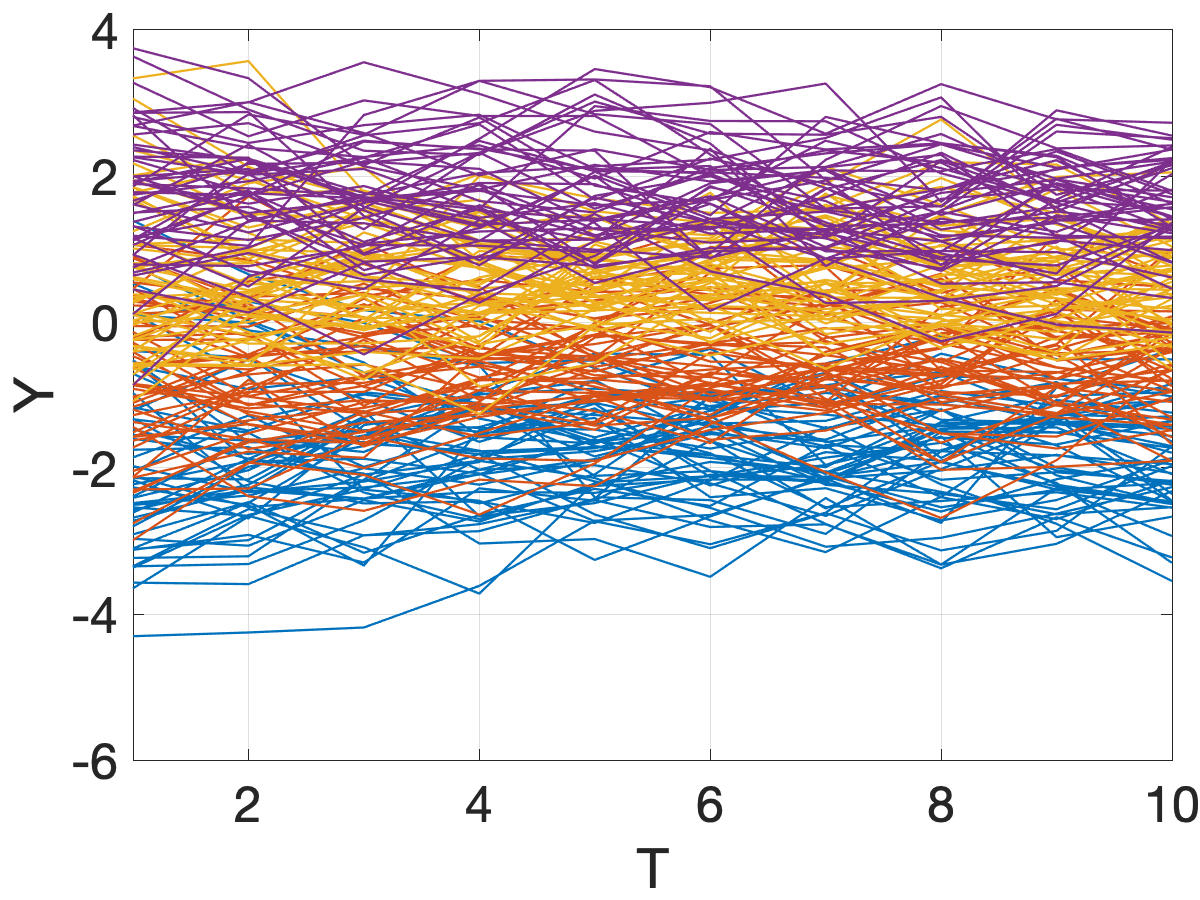}
	\end{subfigure}
\end{figure}

\subsubsection{Details of the Simple DGPs}  \label{sec:simple_dgp_detail}

We start with the mean of $\alpha_{k}$. Assume the values of $\alpha_{k}$ for group $k$ take the form of $\alpha_{k} = m(k - c)$, where $c$ is a shifting constant and $m$ is a scaling constant. With loss of generality, we fix the mean of $\alpha_{k}$ to 0, 
\begin{align}
	\sum_{k=1}^{K_0} \alpha_{k} = m\sum_{k=1}^{K_0} \left( k - c \right) = 0
\end{align}
It immediately follows that $c = \frac{K_0 + 1}{2}$ and 
\begin{align} \label{eq:DGP_a_k_diff_G_app}
	\alpha_{k}  = m \left( k - \frac{K_0 + 1}{2} \right), \quad \text{ for } k = 1,2,.., K_0.
\end{align}
%

Next,  $m$ is the only unknown coefficient in the DGP. To find a reasonable value for $m$, we connect it to the variance of $\alpha_i$. As $\alpha_{i}$ are assume to be identical within a group, the sample variance of $\alpha_{i}$ is given by
\begin{align*}
	V(\alpha_{i})  = \frac{1}{N} \sum_{i=1}^{N} \alpha_i^2 = \frac{1}{N} \frac{N}{K_0} \sum_{k=1}^{K_0} \alpha_k^2 = \frac{1}{K_0} \sum_{k=1}^{K_0} \alpha_k^2
\end{align*}

Plugging in the expression of $\alpha_{k}$ in (\ref{eq:DGP_a_k_diff_G_app}), we have
\begin{align}  \label{eq:DGP_V_diff_G_app}
	V(\alpha_{i}; m, K_0)  =  \frac{m^2}{K_0} \sum_{k=1}^{K_0} \left( k - \frac{K_0 + 1}{2} \right)^2.
\end{align}

To make the DGPs more comparable as more groups are considered, we assume $V(\alpha_{i}; m, K_0)$ is monotonically increasing in $K^2_0$, e.g., $V(\alpha_{i}; m, K_0) = V_0 K^2_0$ for some constant $V_0$. As a result, we can deduct the value of $m$ from (\ref{eq:DGP_V_diff_G_app}),
\begin{align} \label{eq:DGP_m_diff_G_app}
	m(K_0, V_0) =  \left[ \frac{V_0 K_0}{ \sum_{k=1}^{K_0} \left( k - \frac{K_0 + 1}{2} \right)^2} \right]^{\frac{1}{2}}.
\end{align}

It is straightforward to find $V_0$ controls the dispersion of the underlying DGP. A larger $V_0$ indicate $\alpha_{k}$ are more separated and hence the group pattern become sharper, and vice versa.

\subsubsection{General DGPs}

The general DGP is based on the dynamic panel data model specified in (\ref{simple_model}) with an exogenous predictor $z_{it}$ that has common effect for all units. This DGP incorporates group heterogeneity in the fixed-effects, the lagged term $ x^{(1)}_{it} = y_{it-1}$ and an exogenous predictor $x^{(2)}_{it}$, as well as error variance $\sigma^2_{g_i}$.

\noindent{\bf DGP 3:}
\begin{align}
	y_{it} & = \alpha_{g_i}' x_{it} + \gamma z_{it} + \sigma_{g_i} \varepsilon_{it},
\end{align}
where $x_{it} = [1, x^{(1)}_{it}, x^{(2)}_{it}]'$, $\gamma = 1.5$, $y_{i0} \sim N(0,1)$ and $\varepsilon_{it} \sim N (0,1)$. For each $i$, the initial value is specified to guarantee that the time series $(y_{i0}, y_{i1},..., y_{iT})$ is strictly stationary. We assume there are $K^0 = 4$ balanced groups, with the true grouped coefficients summarized in Table (\ref{fig:MC_general_DGP}). The AR(1) coefficients represent different degree of persistence persistence. The exogenous variable $ x^{(2)}_{it} $ is drawn from $N(0,1)$ and $z_{it}$ is drawn from $Gamma(1,1)$, capped by 10.

\begin{table}[htp]
	\begin{center}
		\caption{True Grouped Coefficients in the General DGP}
		\label{fig:MC_general_DGP}
		\begin{tabular}{x{2cm} x{2cm} x{2cm} x{2cm} x{2cm}}
			\toprule
			& $\alpha_{0, k}$ & $\alpha_{1,k}$ & $\alpha_{2,k}$ & $\sigma_k^2$ \\ 
			\cmidrule(lr){2-5} 
			& (FE) & (lagged) & (exo var.) & (variance) \\
			\midrule 
			Group 1 & -0.15 & 0.4 & 0.16 & 0.500  \\ 
			Group 2 & -0.05 & 0.8 & 0.14 & 0.375 \\ 
			Group 3 & \textit{ }0.05 & 0.5 & 0.12 & 0.250 \\ 
			Group 4 & \textit{ }0.15 & 0.7 & 0.10 & 0.125 \\
			\bottomrule 
		\end{tabular}
	\end{center}
\end{table}

\subsection{Construction of Soft Pairwise Constraints} \label{sec:pair_cstr_doc}

\noindent {\bf Number of constraints:} We set the number of constraints $N_{PL}$ and $N_{NL}$ as a function of $N$ and $K_0$ to facilitate performance comparisons across different settings and to ensure that the information of constraints does not vanish as $N$ increases. Specifically, $N_{PL}$ and $N_{NL}$ are a predetermined proportion of the total number of correct constraints for each type which are given by,
\begin{align}
	N^*_{PL}(N,K) & = K C^2_{N/K} = K \frac{N/K(N/K-1)}{2} = \frac{N(N-K)}{2K}, \\
	N^*_{NL}(N,K ) & = \left( \frac{N}{K}\right)^2 C^2_{K} =  \left( \frac{N}{K}\right)^2 \frac{K(K-1)}{2} = \frac{N^2(K-1)}{2K}.
\end{align}

In the setting with $(N, K_0) = (200,4)$, we have $N^*_{PL}(200,4) = 4,900$ and $N^*_{NL}(200,4) = 15,000$. We choose randomly select 5\% of these constraints, leading to $N_{PL} = 245$  and $N_{NL} = 750$.

\noindent {\bf Type of pairwise constraints:} The pairwise constraints are generated randomly. Given the number of PL constraints $N_{PL}$, each PL constraint is generated by randomly selecting a group and uniformly selecting two units within that group to be positive-linked. Similarly, for each of $N_{NL}$ NL constraints, two unique groups are chosen at random and one unit is randomly selected from each. Regarding the remaining unselected units, we assume they are not restricted and have no prior belief on them.

\noindent {\bf Accuracy of pairwise constraints:} Each constraint is annotated with a level of accuracy $\psi$ generating from a transformed Beta distribution defined on $[0.5, 1]$. We begin by drawing $\nu$ from a Beta distribution: if the constraint is correct, $\nu \sim \operatorname{Beta}(3, 2)$ for some $\alpha>1$; otherwise, $\nu \sim \operatorname{Beta}(2, 3)$. Then the level of confidence is $\psi = \frac{\nu}{2} + 0.5$ so that its domain is $[0.5,1]$. We derive $\psi$ in this manner to reflect the assumption that an expert should have less certainty in erroneous constraints than in correct ones.

\noindent {\bf Perturbation in pairwise constraints:} To examine the performance with soft constraints under inaccurate prior belief, we artificially add errors to the randomly generated constraints. A fraction $e$ of the constraints are mislabeled -- a positive-link would be mislabeled as a negative-link and vice versa. We turn $eN_{PL}$ true PL into NL and $eN_{NL}$ true NL into PL with $e = 20\%$.

All DGPs are equipped with the same set of pairwise constraints, e.g., we only draw pairwise constraints and construct weights once.

\subsection{Alternative Estimators}
We explore various types of estimators that differ in the prior belief on $G$.
\begin{enumerate}[(i)]
	\item \textit{BGFE}: The baseline Bayesian grouped fixed-effects (BGFE) estimator are correctly-specified, i.e. assuming that the true model exhibits time-invariant grouped heterogeneity and that variance of error term is constant (varying) across units in the simple (general) DGPs. No prior belief on $G$ is available for this estimator.
	
	\item \textit{BGFE-cstr}: The baseline BGFE estimator that takes pairwise constraints into consideration.
	
	\item \textit{BGFE-oracle}: This estimator is a variant of the BGFE estimator equipped with \textit{known} true G.
\end{enumerate}

We also evaluate the other Bayesian estimators with different prior assumptions on $\alpha_i$ that don't model group structure.
\begin{enumerate}[(i)]
	\setcounter{enumi}{3}
	\item \textit{Pooled}: Bayesian pooled estimator views $\alpha_i$ as a common parameter and, consequently, all units have the same prior level of $\alpha_i$.
	
	\item \textit{Flat}: flat-prior estimator assumes $p(\alpha_i) \; \propto \; 1$. There is no pooling across units in this case and $\alpha_i$'s are individually estimated using their own history. This also amounts to sampling from a posterior whose mode is the MLE estimate. 
	
\end{enumerate}

\subsection{Posterior Predictive Densities and Performance Evaluation}

\subsubsection{Posterior Predictive Densities}

Given $S$ posterior draws, the posterior predictive distribution estimated from the MCMC draws is
\begin{align}
	\hat{p}(y_{i T+1} | Y,X) = \frac{1}{S} \sum_{j=1}^{S} \left[ \sum_{k=1}^{K^{(j)}(G)} \mathbf{1}(g_i = k) p\left(y_{i T+1}| Y, X, \Theta^{(j)} \right) \right].
\end{align}

We can therefore draw samples from $\hat{p}(y_{i T+1} | Y,X)$ by simulating (\ref{simple_model}) forward conditional on the posterior draws of $\Theta $ and observations.

\subsubsection{Point Forecasts}

The optimal posterior forecast under quadratic loss function is obtain by minimizing the posterior risk, with is the posterior mean. Conditional on posterior draws of parameters, the mean forecast can be approximated by the Monte Carlo averaging,
\begin{align}
	\hat{y}_{i, T+1 | T}  \approx \frac{1}{S}\sum_{j=1}^{S} \hat{y}_{i T+1 | T}^{(j)} = \frac{1}{S}\sum_{j=1}^{S} \hat{\alpha}^{(j)'}_{g_i} x_{i T+1},
\end{align}
and the RMSFE across units is given by
\begin{align}
	RMSFE_{T+1} = \sqrt{ \frac{1}{N} \sum_{i=1}^{N} \left(y_{i, T+1} -\hat{y}_{i, T+1}\right)^2 }.
\end{align}

\subsubsection{Set Forecasts}
We construct set forecasts $CS_{i T+1}$ from the posterior predictive distribution of each unit. In particular, we adopt a Bayesian approach and report the highest posterior density interval (HPDI), which is the narrowest connected interval with coverage probability of $1-\alpha$. In other words, it requires that the probability of $y_{i T+1} \in CS_{iT+1}$ conditional on having observed the history $Y$ be at least $1-\alpha$, i.e.,
\begin{align}
	P ( y_{i T+1} \in CS_{i T+1} ) \geq 1-\alpha, \quad \text {for all } i,
\end{align}
and this interval is the shortest among all possible single connected candidate sets. Let $\delta^{l}$ be the lower bound and $\delta_{u}$ be the upper bound, then $CS_{i T+1} = \left[\delta_{i}^{l}, \delta_{i}^{u}\right]$.

The assessment of set forecasts in simulation studies and empirical applications is based on two metrics: (1) the cross-sectional coverage frequency,
\begin{align}
	Cov_{T+1} = \frac{1}{N} \sum_{i=1}^{N} \mathbf{1} \left( y_{i T+1} \in CS_{i T+1} \right),
\end{align}
and (2) the average length of the sets $C_{i T+1}$,
\begin{align}
	AvgL_{T+1} = \frac{1}{N} \sum_{i=1}^{N} \left(\delta_{i}^{u} - \delta_{i}^{l} \right).
\end{align}

\subsubsection{Density Forecasts}
To compare the performance of density forecasts for various estimators, we examine the continuous ranked probability score \citep{matheson1976, hersbach2000} across units. The continuous ranked probability score (CRPS)  is frequently used to assess the respective accuracy of two probabilistic forecasting models. It is a quadratic measure of the difference between the predictive cumulative distribution function, $F^{T+1|T}_i(y)$, and the empirical CDF of the observation with the formula as follows,
\begin{align*}
	CRPS_{T+1} &= \frac{1}{N} \sum_{i=1}^{N} CRPS(F^{T+1|T}_i, y_{i T+1}) \\
	&= \frac{1}{N} \sum_{i=1}^{N} \int_{0}^{\infty} \left[ F^{T+1|T}_{i}(y)-\mathbf{1}\left( y_{i T+1} \leq y \right)\right]^{2} dy, \numberthis
\end{align*}
where $y_{i T+1}$ is the realization at $T+1$.

In practice, the true predictive cumulative distribution function $F^{T+1 | T}_i(y)$ or the PIT of $y_{i T+1}$ is not available. We approximate it via the empirical distribution function for each unit based on the posterior draws from the predictive density,
\begin{align}
	\hat{F}^{T+1 | T}_i (y) = \frac{1}{S} \sum_{j=1}^{S} \mathbf{1} \left( y^{(j)}_{i T+1 | T} \leq y \right),
\end{align}
Based on sorted posterior draws $\widetilde{y}_{iT+1}^{(j)}$, we can calculate CRPS using the below representation by \citet{laio2007},
\begin{align}  \label{eq:CRPS}
	CRPS\left(\hat{F}^{T+1 | T}_i, y_{iT+1}\right)=\frac{2}{S^{2}} \sum_{j = 1}^{S}\left( \widetilde{y}_{i T+1| T}^{(j)} - y_{i T+1}\right)\left( 1\left\{y_{i T+1} < \widetilde{y}_{i T+1 | T}^{(j)} \right\} s - i + \frac{1}{2}\right).
\end{align}


Moreover, we report the average log predictive scores (LPS) to assess the performance of the density forecast from the view of the probability distribution function. As suggested in \citet{geweke2010}, the LPS for a panel reads as,
\begin{align}
	LPS_{T+1}  =  &  - \frac{1}{N} \sum_{i=1}^{N} \ln \int p\left(y_{i T+1} | Y, X, \Theta\right) p(\Theta | Y, X) d \Theta ,
\end{align}
where the expectation can be approximated using posterior draws,
\begin{align}
	\int p\left(y_{i T+1} | Y, X, \Theta\right) p(\Theta | Y, X) d \Theta &\approx \frac{1}{S} \sum_{j = 1} ^{S} p \left(y_{i T+1} | Y, X, \Theta^{(j)} \right).
\end{align}

\subsection{Simulation Results}


\subsubsection{Simple Dynamics Panel Data}

To evaluate the advantage of pooling units into groups, we report the RMSE, bias, standard deviation, average length of 95\% credible set, and frequentist coverage of the posterior estimate of $\rho$ across Monte Carlo repetitions. For the fixed effects $\alpha$, we only present the average bias as it may not be of importance for most empirical study.

The comparison across alternative estimators is shown in Table \ref{tab:MC_est_soft_NL80_CL40}. In DGP 1,the BGFE-cstr and BGFE estimators are equally accurate as the oracle estimator. This is not surprising because the units are well-separated by design, and the data provide sufficient information for the BGFE estimator to determine the group pattern. In this situation, prior knowledge of $G$ or the true group indices has quite marginal influence. The pooled estimator, on the other hand, erroneously pools all groups together, resulting in inaccurate estimates of $\alpha_i$ and $\rho$. Despite the fact that the flat estimator treats units separately, it is still inferior to the BGFE-type estimators. This is because it cannot utilize cross-sectional information to estimates parameters in this short panel and hence bears much larger bias.

In DGP 2, where the group pattern is less apparent, the BGFE-cstr estimator is arguably the most accurate. In contrast to the standard BGFE estimator, it uses cross-sectional data and pairwise constraints to determine the group pattern. These properties substantially reduce the biases of $\hat{\beta}$ and $\hat{\alpha}_i$, enabling the BGFE-cstr estimator to outperform the unconstrained estimator by a significant margin and to perform comparable to the oracle estimator. Remember that we manually add 20\% incorrect constraints into the prior knowledge. Despite the presence of these misspecified constraints, the BGFE-cstr estimator is still able to extract relevant information from constraints in order to enhance the overall performance. The BGFE estimator, however, is unable to correctly reconstruct the group structure due to the noisy data, which results in the algorithm improperly grouping the units and hence generating inaccurate estimates.

\begin{table}[h]
	\begin{center}
		\caption{Monte Carlo: Estimates, Soft Constraint}
		\label{tab:MC_est_soft_NL80_CL40}
		\scalebox{0.8}{
			\begin{tabular}{ll|rrrrr|r|l}
				\toprule
				& & \multicolumn{5}{c}{$\hat{\rho}$} & \multicolumn{1}{c}{$\hat{\alpha}_i$} & Group \\
				\cmidrule(lr){3-7} \cmidrule(lr){8-8} \cmidrule(lr){9-9} \noalign{\smallskip}
				& & \multicolumn{1}{c}{RMSE } & \multicolumn{1}{c}{Bias } & \multicolumn{1}{c}{Std } & \multicolumn{1}{c}{AvgL } & \multicolumn{1}{c}{Cov } & \multicolumn{1}{c}{Bias } & \multicolumn{1}{c}{Avg K} \\
				
				\midrule
				\multirow{5}{*}{\shortstack{DGP 1}}
				& BGFE-oracle & 0.0104 & 0.0037 & 0.0072 & 0.0276 & 0.92 & 0.0371 & 4 \\ 
				& BGFE-cstr & 0.0102 & 0.0030 & 0.0072 & 0.0282 & 0.94 & 0.0369 & 4.92 \\ 
				& BGFE & 0.0103 & 0.0037 & 0.0071 & 0.0274 & 0.92 & 0.0377 & 4.4 \\ 
				& Pooled & 0.3543 & 0.3543 & 0.0032 & 0.0125 & 0 & 1.7889 & - \\ 
				& Flat & 0.1713 & 0.1711 & 0.0073 & 0.0283 & 0 & 0.8668 & - \\ 
				
				\midrule
				\multirow{5}{*}{\shortstack{DGP 2}}
				& BGFE-oracle & 0.0186 & 0.0030 & 0.0137 & 0.0527 & 0.95 & 0.0235 & 4 \\ 
				& BGFE-cstr & 0.0202 & 0.0058 & 0.0143 & 0.0557 & 0.93 & 0.0373 & 5.06 \\ 
				& BGFE & 0.0546 & 0.0443 & 0.0212 & 0.0809 & 0.66 & 0.1357 & 4.78 \\ 
				& Pooled & 0.2920 & 0.2919 & 0.0077 & 0.0298 & 0 & 0.5060 & - \\ 
				& Flat & 0.1170 & 0.0834 & 0.0131 & 0.0509 & 0.14 & 0.2344 & - \\ 
				
				\bottomrule
			\end{tabular}
		}
	\end{center}
\end{table}

Table \ref{tab:MC_fcst_soft_NL80_CL40} provides a summary of the prediction performance of each estimator. In general, the conclusions of the one-step-ahead forecast agree with those of the estimation. In DGP 1, the performance of the three BGFE estimators are quite similar, followed by the flat and pooled estimators. In DGP 2, the BGFE-cstr estimator, which utilizes prior belief on $G$, beats the other feasible estimators in point, set, and density forecast and is comparable to the oracle estimator.

\begin{table}[h]
	\begin{center}
		\caption{Monte Carlo: Forecast, Soft Constraint}
		\label{tab:MC_fcst_soft_NL80_CL40}
		\scalebox{0.8}{
			\begin{tabular}{ll|rrr|rr|rr}
				\toprule
				& & \multicolumn{3}{c}{Point Forecast}  & \multicolumn{2}{c}{Set Forecast} & \multicolumn{2}{c}{Density Forecast} \\
				\cmidrule(lr){3-5} \cmidrule(lr){6-7} \cmidrule(lr){8-9} \noalign{\smallskip}
				& & \multicolumn{1}{c}{RMSFE } & \multicolumn{1}{c}{Error } & \multicolumn{1}{c}{Std } & \multicolumn{1}{c}{AvgL } & \multicolumn{1}{c}{Cov } & \multicolumn{1}{c}{LPS } & \multicolumn{1}{c}{CRPS } \\
				
				\midrule
				\multirow{5}{*}{\shortstack{DGP 1}}
				& BGFE-oracle & 0.4989 & 0.0001 & 0.4989 & 1.9627 & 0.95 & 0.7254 & 0.2818 \\ 
				& BGFE-cstr & 0.4991 & 0.0004 & 0.4990 & 1.9666 & 0.95 & 0.7256 & 0.2819 \\ 
				& BGFE & 0.4990 & 0.0001 & 0.4990 & 1.9616 & 0.95 & 0.7255 & 0.2818 \\ 
				& Pooled & 0.6401 & 0.0006 & 0.6404 & 3.0114 & 0.98 & 1.0064 & 0.3657 \\ 
				& Flat & 0.5620 & 0.0003 & 0.5622 & 2.4265 & 0.97 & 0.8544 & 0.3184 \\ 
				
				\midrule
				\multirow{5}{*}{\shortstack{DGP 2}}
				& BGFE-oracle & 0.4990 & 0.0001 & 0.4989 & 1.9629 & 0.95 & 0.7254 & 0.2819 \\ 
				& BGFE-cstr & 0.5021 & 0.0001 & 0.5021 & 1.9790 & 0.95 & 0.7314 & 0.2836 \\
				& BGFE & 0.5186 & 0.0002 & 0.5187 & 2.0546 & 0.95 & 0.7633 & 0.2930 \\ 
				& Pooled & 0.5396 & 0.0005 & 0.5395 & 2.2444 & 0.96 & 0.8079 & 0.3052 \\ 
				& Flat & 0.5286 & 0.0002 & 0.5287 & 2.1165 & 0.95 & 0.7841 & 0.2987 \\ 
				
				\bottomrule
			\end{tabular}
		}
	\end{center}
\end{table}

\subsubsection{General Panel Data} 

As the number of parameters increases for DGP 3, we present the RMSE and absolute bias of $\alpha_{g_i} = [\alpha_{1, g_i} \; \alpha_{2, g_i} \; \alpha_{3, g_i}]'$ and $\gamma$, as well as metrics for point and density prediction. In addition, all BGFE estimators now account for heteroskedasticity because the cross-sectional variance in DGP 3 is informative to group structure. As a result, we have the \textit{BGFE-he-oracle}, \textit{BGFE-he-cstr} and \textit{BGFE-he} estimators in this exercise, where "\textit{he}" denotes heteroskedasticity. For comparison, we also offer the \textit{BGFE-ho-cstr} estimator, which assumes homoskedasticity, and \textit{Flat-he} estimator, which is the heteroskedastic flat estimator.

\begin{table}[h]
	\begin{center}
		\caption{Results for Estimation, Point Forecast and Estimated $K$}
		\label{tab:MC_general_DGP}
		\resizebox{0.99\textwidth}{!}{%
			\begin{tabular}{l c c c c c c c c c c l l}
				\toprule
				& \multicolumn{8}{c}{Estimates}  & \multicolumn{2}{c}{Forecast} & \multicolumn{2}{c}{Group} \\
				\cmidrule(lr){2-9} \cmidrule(lr){10-11} \cmidrule(lr){12-13} \noalign{\smallskip}
				& $R(\hat{\alpha}_0)$ & $B(\hat{\alpha}_0)$ & $R(\hat{\alpha}_1)$ & $B(\hat{\alpha}_1)$ & $R(\hat{\alpha}_2)$ & $B(\hat{\alpha}_2)$ & $R(\hat{\gamma})$ & $B(\hat{\gamma})$ & RMSFE & LPS & AvgK & PctK\\ 
				\midrule 
				Flat-he & 0.258 & 0.199 & 0.131 & 0.098 & 0.200 & 0.149 & 0.026 & 0.014 & 0.667 & 1.108 & - & - \\ 
				\midrule
				BGFE-he-oracle & 0.126 & 0.137 & 0.092 & 0.101 & 0.119 & 0.132 & 0.569 & 0.602 & 0.840 & 0.706 & 4 & 1 \\ 
				BGFE-he-cstr & 0.171 & 0.164 & 0.290 & 0.179 & 0.125 & 0.135 & 0.566 & 0.608 & 0.847 & 0.716 & 4.087 & 0.914 \\ 
				BGFE-ho-cstr & 0.218 & 0.198 & 0.328 & 0.220 & 0.140 & 0.144 & 0.625 & 0.671 & 0.850 & 0.769 & 4.342 & 0.682 \\ 
				BGFE-he & 0.303 & 0.317 & 0.560 & 0.482 & 0.137 & 0.147 & 0.601 & 0.640 & 0.871 & 0.756 & 3.575 & 0.534 \\ 
				Pooled & 0.444 & 0.503 & 1.262 & 1.527 & 0.131 & 0.148 & 0.734 & 0.805 & 0.993 & 0.910 & - & - \\ 
				\bottomrule 
			\end{tabular}
			
		}
	\end{center}
	{\footnotesize \textit{Notes: The first line gives the levels of the each metrics based on the \textit{Flat-he} estimator, which is the benchmark model, and the following lines in the columns head "Estimates" and "Forecast" present ratios of the respective method relative those based on the flat-he estimator. In the columns head "Group", we show the average of number of groups (AvgK) and the percentage of iterations that the posterior sampler selects $K_0$ (PctK) averaged over 1,000 runs of algorithm. $R(\cdot)$ is RMSE of the posterior mean estimator. $B(\cdot)$ is the absolute bias of the posterior mean estimator. }\par}
\end{table}

Table \ref{tab:MC_general_DGP} presents the relative performance of estimation and forecasting for DGP 3. The benchmark model is \textit{Flat-he}. Several findings arise. First, the gain from incorporating pairwise constraints is evident. It reduces the RMSE and bias for all parameters and improve both point and density forecast, when comparing BGFE-he-cstr to BGFE-he. The percentage of the Gibbs sampler that visits the true number of groups $K_0$ grows considerably from 53.4\% to 91.4\%. Even when pairwise constraints are taken into account, this percentage is just 68.2\% if heteroskedasticity is ignored. Second, when we include prior belief on $G$, the improvement in $\alpha_{1, g_i}$, the AR coefficient, is the greatest among all three grouped coefficients with a bias reduction of more than 60\%. This also suggests that the AR coefficient may be more sensitive to the estimated group structure. Thirdly, BGFE-he-cstr and BGFE-ho-cstr have comparable RMSFE values, but BGFE-he-cstr has a significantly lower LPS, showing that modeling heteroskedasticity in the current setting is favorable for the density forecast. The empirical results below also confirm this finding. Lastly, all BGFE-type estimators generate similar estimates for the exogenous variables that don't have group effects on $y_{it}$ as the improvement in $\alpha_{2, g_i}$ and $\gamma$ are marginal when prior belief on group is included or when true group is imposed.

\subsubsection{Computational Time}

We compare the running times of each estimator using the simulation with the general DGPs. For each estimator, we run the block Gibbs samplers with 5,000 iterations after a burn-in of 5,000 draws. We calculate the running time by averaging over 1,000 repetitions. All programs, including this exercise, the simulations above, and the empirical analysis in Section \ref{sec:emp_result} were executed in Matlab 2020a running on a server with Intel Skylake CPUs and 32 GB of RAM.

Table $\ref{tab:MC_speed}$ shows the running time in seconds. It suggests that modeling group structure does not considerably increase computational load. The running time of the BGFE-he estimator is comparable to that of the Flat estimator when the number of observations is small. As $T$ increases, the BGFE-he estimators take slightly more time to finish than the flat estimator. The BGFE-he-cstr estimator requires approximately 20\% more time to complete the calculations than the BGFE-he estimator. This is because BGFE-he-cstr adapts pairwise constraints and has additional term to evaluate in the posterior distribution of the group indices. Notice that we fix the scaling constant $c$ in this experiment. If we find the optimal $c$ using grid search as proposed in Section \ref{app:det_c}, the running time of BGFE-he-cstr will be $m$ times as great as what is shown below, where $m$ is the size of the grid.

\begin{table}[h]
	\begin{center}
		\caption{Running Time in Seconds}
		\label{tab:MC_speed}
		\begin{tabular}{c c | x{2.5cm}  x{2.5cm}  x{2.5cm}  x{2.5cm}}
			\toprule
			N & T & BGFE-he-cstr & BGFE-he & Flat & Pooled \\ 
			\midrule
			\multirow{3}{*}{\shortstack{100}}
			&  5 & 74.9 & 62.1 & 1.2 & 61.7 \\ 
			& 10 & 82.5 & 69.8 & 1.3 & 64.7 \\ 
			& 20 & 86.1 & 75.7 & 1.3 & 62.4 \\ 
			\midrule
			\multirow{3}{*}{\shortstack{200}}
			&  5 & 145.0 & 111.9 & 1.2 & 114.3 \\ 
			& 10 & 146.2 & 118.8 & 1.3 & 113.0 \\ 
			& 20 & 159.4 & 133.2 & 1.5 & 117.0 \\ 
			\bottomrule 
		\end{tabular}
	\end{center}
\end{table}

It is possible to speed the algorithm up especially when we select the optimal value for $c$. \citet{wang2011} propose the sequential updating and greedy search (SUGS) algorithm, which essentially simplifies the algorithm \ref{algo:RC_GH} with approximation. The SUGS algorithm allows fast yet accurate approximate Bayes inferences under DP priors with just a single cycle of simple calculations for each unit, while also producing marginal likelihood estimates to be used in selecting the constant $c$. This will be left for future study.

\section{Data Description} \label{appendix:data}

\subsection{Inflation of the U.S. CPI Sub-Indices}  \label{appendix:data_cpi}
The seasonally adjusted series of CPI for All Urban Consumers (CPI-U) for subcategories at all display level are obtained from the BLS Online Databases.\footnote{\url{https://data.bls.gov/PDQWeb/cu}} The raw data contains 318 series, which are recorded on a monthly basis and spanned the period from January 1947 to August 2022.  Notice that the raw dataset doesn't include an indicator for the expenditure categories. We manually merge the raw dataset with the table of content of CPI entry level items\footnote{\url{https://www.bls.gov/cpi/additional-resources/entry-level-item-descriptions.xlsx}} by entry level item (ELI) code, the series description, and universal classification codes (UCC), if necessary.\footnote{Some series are labeled by UCC rather than ELI. The concordance provided by the BLS can be found here: \url{https://www.bls.gov/cpi/additional-resources/ce-cpi-concordance.htm}.} 

Series can enter and exit the sample. The BLS discontinued and launched series on a regular basis owing to changes in source data and methodology, for example, see the \href{https://www.bls.gov/cpi/additional-resources/recent-upcoming-methodology-changes.htm}{Post} for the updates on series since 2017. The measure of certain subcategories was impacted by the Pandemic and hence missing. Since the Pandemic, the related activities and venues (sports events, bars, schools) were canceled and close temporarily, such as admission to sporting events (SS62032), distilled spirits away from home (SS20053), food at elementary and secondary schools (SSFV031A), etc. We chose to not impute the missing values since there was no clear benchmark to compare with, especially given the depressed economic conditions.

The CPI-U consists of eight major expenditure categories (1) Apparel; (2) Education and Communication; (3) Food and beverages; (4) Housing; (5) Medical Care; (6) Recreation; (7) Transportation; (8) Other goods and services. Each major category contains multiple subcategories, resulting in a hierarchy of categories with increasing specificity. BLS provides a detailed table\footnote{ \url{https://download.bls.gov/pub/time.series/cu/cu.item}.} that records the series code, series name, and display level. We resort to the display level to build the tree structure of the CPI sub-indices and eliminate those parent nodes, as illustrated in Figure \ref{fig_app:CPI_tree}.

\begin{figure}[H]
	\centering
	\caption{Hierarchical Structure of CPI: Eliminating Parent Nodes} \label{fig_app:CPI_tree}
	\begin{center}
		\resizebox*{.99\linewidth}{!}{%
			\begin{forest}
				forked edges,
				for tree={
					grow = east,
					align = center,
					font=\sffamily,
				}
				[\cancel{All Items}
				[\cancel{Housing}
				[\cancel{Shelter}
				[Rent of primary residence]
				[Owners' equivalent rent of residences]
				]
				[\cancel{Fuels and utilities}
				[Water sewer and trash]
				[\cancel{Household energy}
				[\cancel{Fuel oil and other fuels}
				[Fuel oil]
				[Propane kerosene and firewood]]
				[\cancel{Energy services}
				[Electricity]
				[Utility gas service]]]
				]
				]
				[\cancel{Transportation}
				[...]
				]
				[\cancel{Food and beverages}
				[...]
				]
				]
			\end{forest}
		}
	\end{center}
\end{figure}
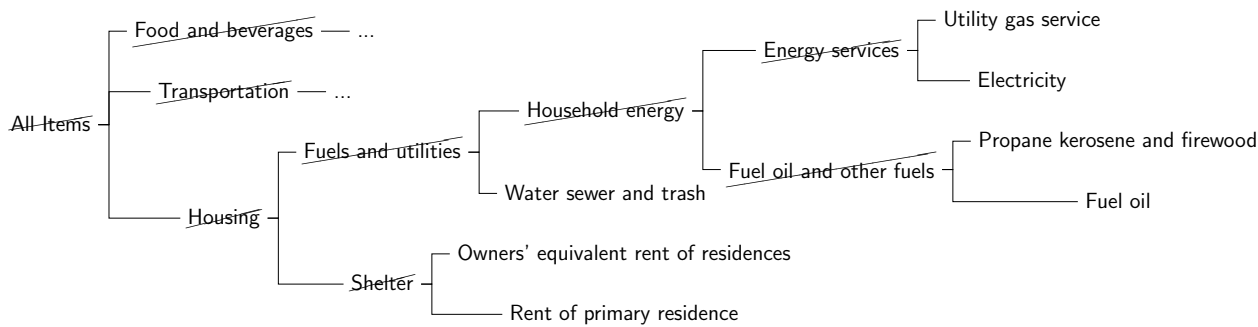

Because all CPI data is available on a monthly basis, we use the unemployment gap as labor market slack measures in the Phillips curve model. We use the seasonally adjusted unemployment rate\footnote{\url{https://fred.stlouisfed.org/series/UNRATE}} from FRED and construct the ‘‘gap’’ measures using the Hamilton filter \citep{hamilton2018}. The Hamilton filter has two parameters: number of lags $p$ and number of lookahead periods $h$. We follow Hamilton's suggestion and set $h=24$ and $p=12$, or an AR(12) process, additionally lagged by 24 lookahead periods for the monthly time series.




\subsection{Income and Democracy}

All data in this section are taken from the replication files of BM.\footnote{ \url{https://www.dropbox.com/s/ssjabvc2hxa5791/Bonhomme_Manresa_codes.zip?dl=0}} The data set contains a balanced panel of 89 countries and 7 periods at a five-year interval over 1970-2000. The main measure of democracy is the Freedom House Political Rights Index. A country receives the highest score if political rights come closest to the ideals suggested by a checklist of questions, beginning with whether there are free and fair elections, whether those who are elected rule, whether there are competitive parties or other political groupings, whether the opposition plays an important role and has actual power, and whether minority groups have reasonable self-government or can participate in the government through informal consensus. See more details in \citet{acemoglu2008}, Section 1.

Table \ref{tab:app_demo_data_summary} contains descriptive statistics for the main variables. The sample period is 1970–2000, and each observation corresponds to five-year intervals. The table shows these statistics for all countries and also for high- and low-income countries, split according to the median of the countries' averaged income. The comparison of high- and low-income countries in the medium and lower panels reveals the pattern that richer countries tend to be more democratic.

\begin{table}[htp]
	\begin{center}
		\caption{Summary Statistics for the Democracy Data Set}
		\label{tab:app_demo_data_summary}
		\begin{tabular}{l  r r r r r }
			\toprule
			& Mean & Median & S.E. & Min & Max  \\
			\midrule
			Full Sample & & & & & \\
			\cmidrule(lr){1-1} 
			Democracy index & 0.5535  &  0.5000  &   0.3727  & 0  &  1.0000 \\
			GDP per capita (in log) & 8.2534   &  8.2444  &  1.0763   &  5.7739  & 10.4450 \\
			
			\midrule
			High-Income & & & & & \\
			\cmidrule(lr){1-1} 
			Democracy index & 0.7852  &  1.0000   &   0.2934  &  0   &  1.0000 \\
			GDP per capita (in log) & 9.1490  &  9.1975  &  0.6079  &  7.4970  &  10.4450 \\
			
			\midrule
			Low-Income & & & & & \\
			\cmidrule(lr){1-1} 
			Democracy index & 0.3247  &   0.1667   &  0.2973   &  0   &  1.0000 \\
			GDP per capita (in log) & 7.3576   &  7.3208   &  0.6051  &  5.7739  &  8.81969 \\
			
			\bottomrule
		\end{tabular}
	\end{center}
\end{table}


\section{Additional Empirical Results} \label{appendix:extra_empical}

\subsection{Inflation of the U.S. CPI Sub-Indices} \label{appendix:res_cpi}

Figure \ref{fig_app:app_cpi_ngroup} shows the posterior mean of the number of active groups (red) along with total number of available CPI sub-indices (dark blue). Note that we choose not to show the credible set or the distribution of $K$ because the distributions are concentrated around a single number in most samples.
 
\begin{figure}[h]
	\caption{Number of Active Groups, BGFE-he-cstr}
	\label{fig_app:app_cpi_ngroup}
	\begin{center}
		\includegraphics[scale=0.4]{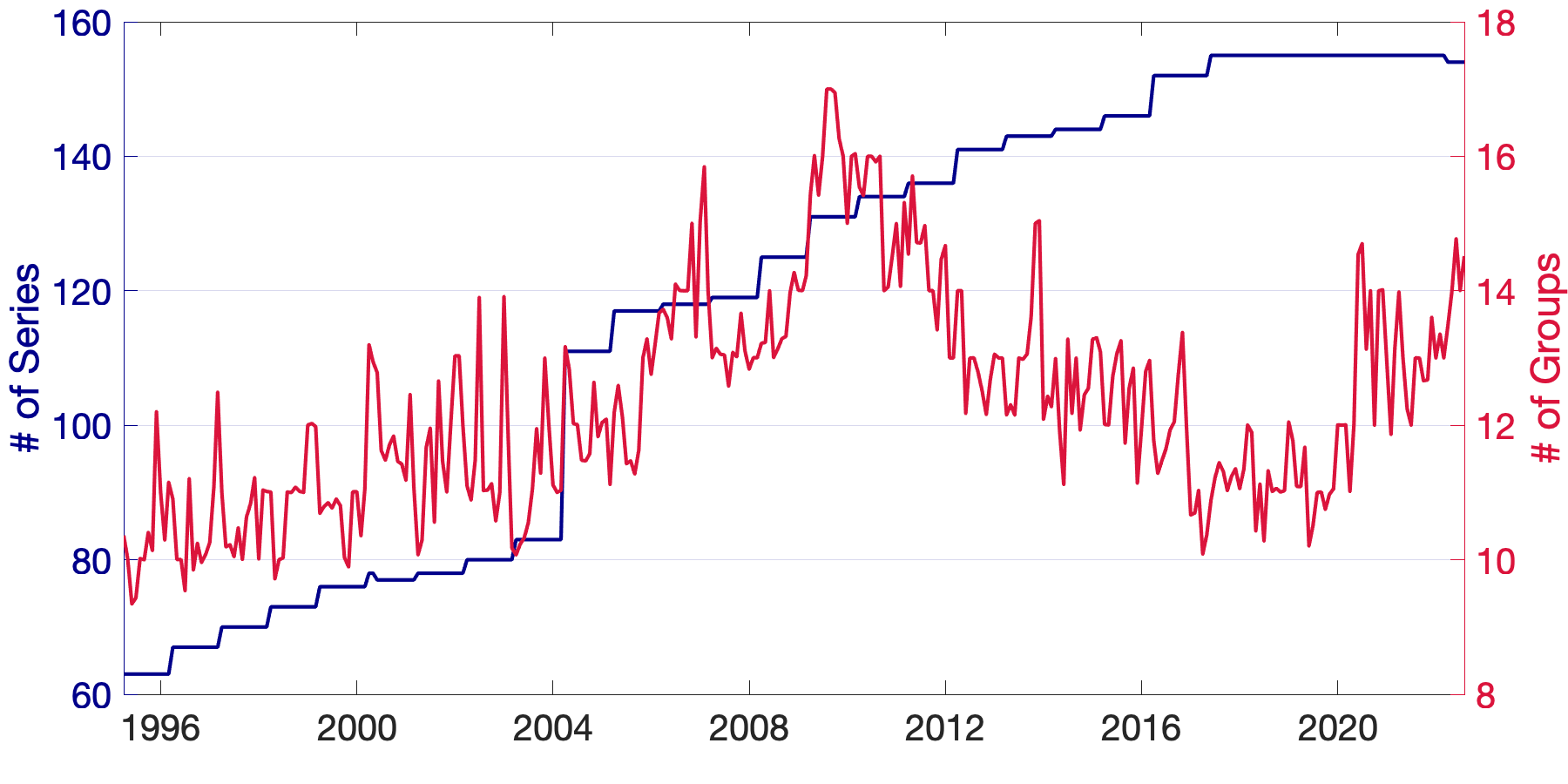}
	\end{center}
\end{figure}

Figure \ref{fig_app:app_cpi_c} shows the selected scaling constant $c$ over time.
\begin{figure}[h]
	\caption{Scaling Constant, BGFE-he-cstr}
		\label{fig_app:app_cpi_c}
	\begin{center}
		\includegraphics[scale=0.4]{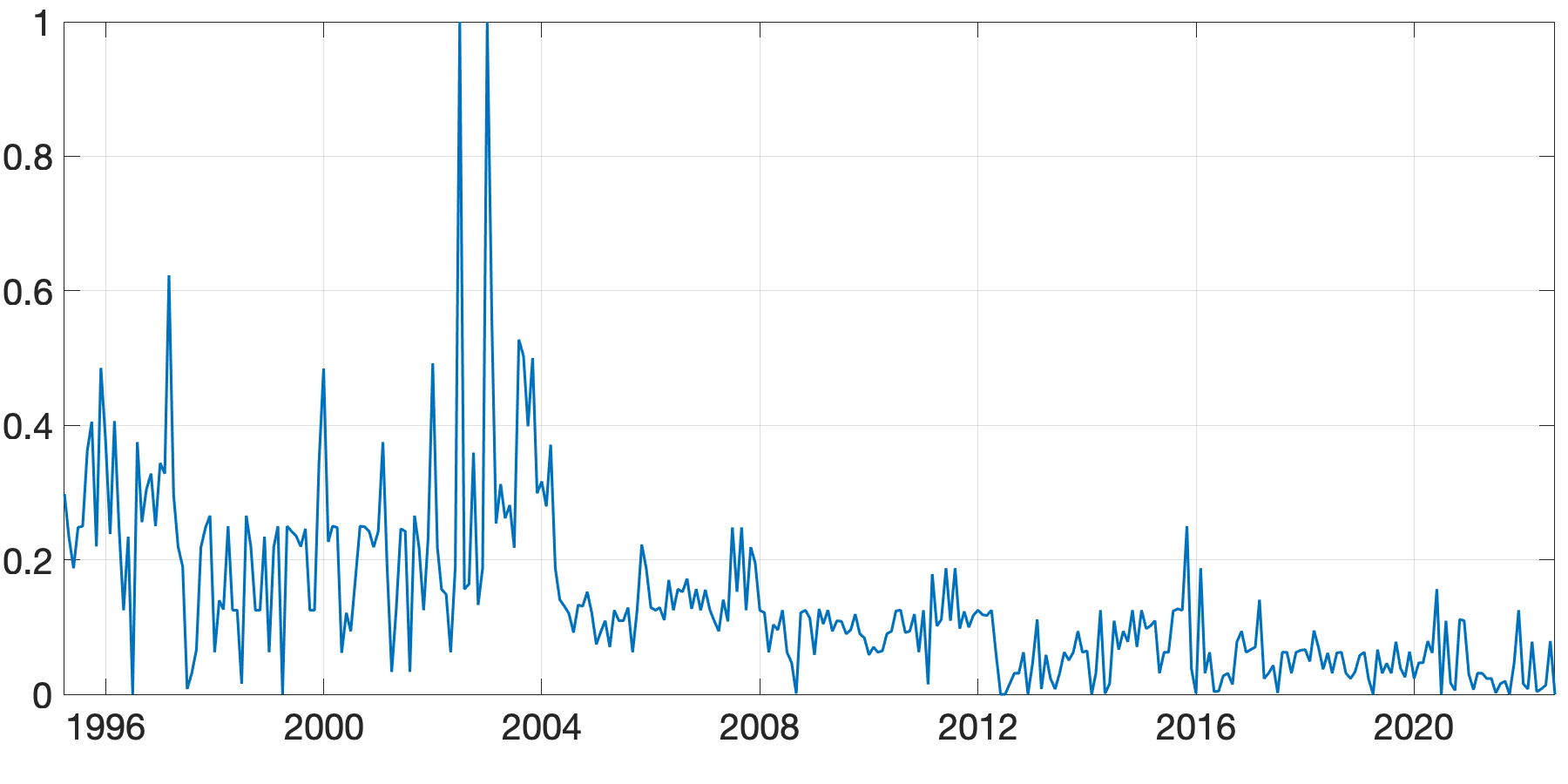}
	\end{center}
\end{figure}

Table \ref{tab_app:cpi_sub_rmse1} and \ref{tab_app:cpi_sub_rmse2} present the ratio of RMSE between each estimator and AR-he for each expenditure category in five periods: 1995-1999, 2000-2004, 2005-2009, 2010-2014, 2015-2019, 2020-2022. Similar results for the difference of LPS are shown in Table \ref{tab_app:cpi_sub_lps1} and \ref{tab_app:cpi_sub_lps2}.

\begin{table}[htp]
	\begin{center}
		\caption{Relative RMSE by Expenditure Category and Period}
		\label{tab_app:cpi_sub_rmse1}
		\begin{tabular}{lSSSSSSS}
			\toprule
			& \text{Full} & \text{95 - 99} & \text{00 - 04} & \text{05 - 09} & \text{10 - 14} & \text{15 - 19} & \text{20 - 22} \\ 
			\midrule 
			& \multicolumn{7}{l}{Average of All Series} \\
			\cmidrule(lr){2-8} 
			BGFE-he-cstr & 0.97 & 0.97 & 0.97 & 0.99 & 0.96 & 0.97 & 0.94 \\ 
			BGFE-he & 0.98 & 0.98 & 0.98 & 1.00 & 0.97 & 0.98 & 0.95 \\ 
			BGFE-ho & 1.00 & 1.04 & 0.98 & 1.00 & 0.98 & 0.99 & 1.00 \\ 
			AR-he-PC & 1.01 & 1.02 & 1.03 & 0.99 & 1.00 & 1.01 & 1.04 \\ 
			Pooled & 1.00 & 1.04 & 0.98 & 1.03 & 1.01 & 0.99 & 0.94 \\ 
			\midrule 
			& \multicolumn{7}{l}{Category 1: Apparel} \\
			\cmidrule(lr){2-8} 
			BGFE-he-cstr & 0.97 & 0.98 & 0.99 & 0.97 & 0.96 & 0.98 & 0.91 \\ 
			BGFE-he & 0.97 & 0.98 & 0.99 & 0.98 & 0.97 & 0.98 & 0.92 \\ 
			BGFE-ho & 0.98 & 0.99 & 1.00 & 0.98 & 0.98 & 0.99 & 0.93 \\ 
			AR-he-PC & 1.02 & 1.02 & 1.00 & 1.03 & 1.01 & 1.01 & 1.08 \\ 
			Pooled  & 0.99 & 1.00 & 0.97 & 0.99 & 1.03 & 1.03 & 0.90 \\ 
			\midrule 
			& \multicolumn{7}{l}{Category 2: Education and Communication} \\
			\cmidrule(lr){2-8} 
			BGFE-he-cstr & 0.96 & 0.93 & 0.99 & 0.95 & 0.94 & 0.97 & 0.93 \\ 
			BGFE-he & 0.95 & 0.95 & 1.00 & 0.94 & 0.95 & 0.97 & 0.92 \\ 
			BGFE-ho & 1.08 & 2.18 & 1.02 & 1.00 & 1.00 & 0.97 & 0.91 \\ 
			AR-he-PC & 1.01 & 1.03 & 0.99 & 1.00 & 0.95 & 1.01 & 1.08 \\ 
			Pooled  & 1.16 & 2.24 & 1.21 & 1.05 & 1.12 & 1.03 & 0.97 \\ 
			\midrule 
			& \multicolumn{7}{l}{Category 3: Food and Beverages} \\
			\cmidrule(lr){2-8} 
			BGFE-he-cstr & 0.98 & 0.99 & 0.99 & 1.00 & 0.97 & 0.98 & 0.95 \\ 
			BGFE-he & 0.98 & 0.99 & 1.00 & 1.00 & 0.97 & 0.98 & 0.95 \\ 
			BGFE-ho & 0.99 & 1.03 & 0.99 & 0.97 & 0.97 & 0.98 & 0.95 \\ 
			AR-he-PC & 1.00 & 1.01 & 1.02 & 0.99 & 1.00 & 1.01 & 0.98 \\ 
			Pooled  & 1.01 & 1.03 & 1.01 & 1.05 & 1.01 & 0.99 & 0.93 \\ 
			\midrule 
			& \multicolumn{7}{l}{Category 4: Housing} \\
			\cmidrule(lr){2-8} 
			BGFE-he-cstr & 0.97 & 0.96 & 0.95 & 1.03 & 0.94 & 0.96 & 0.96 \\ 
			BGFE-he & 0.96 & 0.97 & 0.93 & 1.03 & 0.94 & 0.96 & 0.96 \\ 
			BGFE-ho & 0.98 & 1.14 & 0.93 & 1.02 & 0.96 & 0.97 & 0.95 \\ 
			AR-he-PC & 1.03 & 1.03 & 1.04 & 1.04 & 0.99 & 1.01 & 1.07 \\ 
			Pooled & 0.99 & 1.14 & 0.92 & 1.06 & 0.97 & 0.97 & 0.94 \\ 
			\bottomrule 
		\end{tabular}
	\end{center}
	{\footnotesize \textit{Notes: Benchmark model = AR-he. }\par}
\end{table}

\begin{table}[htp]
	\begin{center}
		\caption{Relative RMSE by Expenditure Category and Period, \textit{cont.}}
		\label{tab_app:cpi_sub_rmse2}
		\begin{tabular}{lSSSSSSS}
			\toprule
			& \text{Full} & \text{95 - 99} & \text{00 - 04} & \text{05 - 09} & \text{10 - 14} & \text{15 - 19} & \text{20 - 22} \\ 
			\midrule 
			& \multicolumn{7}{l}{Category 5: Medical Care} \\
			\cmidrule(lr){2-8} 
			BGFE-he-cstr & 0.95 & 0.93 & 0.98 & 0.95 & 1.01 & 0.96 & 0.87 \\ 
			BGFE-he & 0.95 & 0.93 & 0.98 & 0.95 & 1.01 & 0.96 & 0.85 \\ 
			BGFE-ho & 1.06 & 1.32 & 1.15 & 1.07 & 1.09 & 0.99 & 0.86 \\ 
			AR-he-PC & 1.03 & 1.04 & 0.99 & 1.02 & 1.03 & 1.02 & 1.07 \\ 
			Pooled  & 1.15 & 1.30 & 1.38 & 1.23 & 1.15 & 1.04 & 0.92 \\ 
			\midrule 
			& \multicolumn{7}{l}{Category 6: Recreation} \\
			\cmidrule(lr){2-8} 
			BGFE-he-cstr & 0.97 & 0.97 & 0.99 & 0.99 & 0.98 & 0.96 & 0.94 \\ 
			BGFE-he & 0.97 & 0.97 & 1.00 & 0.99 & 0.99 & 0.96 & 0.92 \\ 
			BGFE-ho & 1.02 & 1.18 & 1.09 & 1.05 & 1.00 & 0.97 & 0.94 \\ 
			AR-he-PC & 1.03 & 1.03 & 1.03 & 1.02 & 1.01 & 1.01 & 1.08 \\ 
			Pooled  & 1.12 & 1.17 & 1.24 & 1.31 & 1.17 & 0.99 & 0.96 \\ 
			\midrule 
			& \multicolumn{7}{l}{Category 7: Transportation} \\
			\cmidrule(lr){2-8} 
			BGFE-he-cstr & 0.99 & 0.97 & 1.01 & 0.99 & 0.99 & 1.00 & 0.96 \\ 
			BGFE-he & 0.99 & 0.97 & 1.01 & 0.98 & 0.99 & 1.00 & 0.96 \\ 
			BGFE-ho & 1.03 & 1.07 & 1.02 & 1.00 & 1.01 & 1.02 & 1.07 \\ 
			AR-he-PC & 1.02 & 1.04 & 1.05 & 0.97 & 1.01 & 0.99 & 1.07 \\ 
			Pooled  & 0.99 & 1.07 & 1.02 & 0.99 & 0.97 & 0.96 & 0.97 \\ 
			\midrule 
			& \multicolumn{7}{l}{Category 8: Other Goods and Services} \\
			\cmidrule(lr){2-8} 
			BGFE-he-cstr & 0.97 & 0.98 & 0.99 & 1.02 & 0.89 & 0.96 & 0.93 \\ 
			BGFE-he & 0.95 & 0.94 & 0.95 & 1.05 & 0.89 & 0.97 & 0.91 \\ 
			BGFE-ho & 0.96 & 0.99 & 0.97 & 0.99 & 0.89 & 0.97 & 0.91 \\ 
			AR-he-PC & 1.02 & 1.01 & 1.02 & 1.02 & 0.97 & 1.00 & 1.11 \\ 
			Pooled  & 0.98 & 0.99 & 1.04 & 0.99 & 0.89 & 0.99 & 0.94 \\ 
			\bottomrule 
		\end{tabular}
	\end{center}
	{\footnotesize \textit{Notes: Benchmark model = AR-he. }\par}
\end{table}

\begin{table}[htp]
	\begin{center}
		\caption{Relative LPS, by Expenditure Category and Period}
		\label{tab_app:cpi_sub_lps1}
		\begin{tabular}{lSSSSSSS}
			\toprule
			& \text{Full} & \text{95 - 99} & \text{00 - 04} & \text{05 - 09} & \text{10 - 14} & \text{15 - 19} & \text{20 - 22} \\ 
			\midrule 
			& \multicolumn{7}{l}{Average of All Series} \\
			\cmidrule(lr){2-8} 
			BGFE-he-cstr & -0.08 & -0.08 & -0.07 & -0.08 & -0.09 & -0.07 & -0.08 \\ 
			BGFE-he & -0.06 & -0.06 & -0.05 & -0.06 & -0.07 & -0.05 & -0.08 \\ 
			BGFE-ho & 0.64 & 0.77 & 0.72 & 0.77 & 0.48 & 0.53 & 0.45 \\ 
			AR-he-PC & 0.01 & 0.02 & 0.02 & 0.01 & 0.01 & 0.01 & 0.01 \\ 
			Pooled OLS & 0.66 & 0.84 & 0.73 & 0.79 & 0.53 & 0.52 & 0.44 \\ 
			\midrule 
			& \multicolumn{7}{l}{Category 1: Apparel} \\
			\cmidrule(lr){2-8} 
			BGFE-he-cstr & -0.05 & -0.03 & -0.02 & -0.05 & -0.05 & -0.02 & -0.14 \\ 
			BGFE-he & -0.04 & -0.02 & -0.02 & -0.04 & -0.05 & -0.01 & -0.14 \\ 
			BGFE-ho & 0.10 & 0.25 & 0.07 & 0.07 & 0.11 & 0.06 & 0.03 \\ 
			AR-he-PC & 0.03 & 0.03 & 0.01 & 0.03 & 0.02 & 0.01 & 0.09 \\ 
			Pooled & 0.12 & 0.25 & 0.07 & 0.10 & 0.18 & 0.11 & -0.05 \\ 
			\midrule 
			& \multicolumn{7}{l}{Category 2: Education and Communication} \\
			\cmidrule(lr){2-8} 
			BGFE-he-cstr & -0.12 & -0.06 & -0.07 & -0.19 & -0.17 & -0.10 & -0.11 \\ 
			BGFE-he & -0.13 & -0.08 & -0.11 & -0.20 & -0.16 & -0.09 & -0.12 \\ 
			BGFE-ho & 0.92 & 1.44 & 1.04 & 0.87 & 0.91 & 0.56 & 0.56 \\ 
			AR-he-PC & 0.01 & 0.00 & 0.01 & 0.02 & 0.00 & 0.01 & 0.07 \\ 
			Pooled & 0.98 & 1.45 & 1.09 & 0.95 & 0.99 & 0.62 & 0.62 \\ 
			\midrule 
			& \multicolumn{7}{l}{Category 3: Food and Beverages} \\
			\cmidrule(lr){2-8} 
			BGFE-he-cstr & -0.04 & -0.03 & -0.02 & -0.06 & -0.05 & -0.05 & -0.07 \\ 
			BGFE-he & -0.04 & -0.04 & -0.02 & -0.05 & -0.05 & -0.05 & -0.08 \\ 
			BGFE-ho & 0.37 & 0.66 & 0.36 & 0.34 & 0.27 & 0.39 & 0.03 \\ 
			AR-he-PC & 0.00 & 0.01 & 0.03 & -0.01 & 0.01 & 0.00 & -0.08 \\ 
			Pooled & 0.44 & 0.84 & 0.41 & 0.44 & 0.33 & 0.41 & 0.03 \\ 
			\midrule 
			& \multicolumn{7}{l}{Category 4: Housing} \\
			\cmidrule(lr){2-8} 
			BGFE-he-cstr & -0.12 & -0.13 & -0.14 & -0.10 & -0.11 & -0.12 & -0.07 \\ 
			BGFE-he & -0.11 & -0.12 & -0.13 & -0.10 & -0.11 & -0.12 & -0.07 \\ 
			BGFE-ho & 0.81 & 0.96 & 1.13 & 0.83 & 0.60 & 0.60 & 0.70 \\ 
			AR-he-PC & 0.02 & 0.00 & 0.03 & 0.04 & 0.01 & 0.01 & 0.06 \\ 
			Pooled & 0.82 & 0.97 & 1.13 & 0.88 & 0.63 & 0.60 & 0.68 \\ 
			\bottomrule 
		\end{tabular}
	\end{center}
	{\footnotesize \textit{Notes: Benchmark model = AR-he. }\par}
\end{table}

\begin{table}[htp]
	\begin{center}
		\caption{Relative LPS by Expenditure Category and Period, \textit{cont.}}
		\label{tab_app:cpi_sub_lps2}
		\begin{tabular}{lSSSSSSS}
			\toprule
			& \text{Full} & \text{95 - 99} & \text{00 - 04} & \text{05 - 09} & \text{10 - 14} & \text{15 - 19} & \text{20 - 22} \\ 
			\midrule 
			& \multicolumn{7}{l}{Category 5: Medical Care} \\
			\cmidrule(lr){2-8} 
			BGFE-he-cstr & -0.15 & -0.29 & -0.14 & -0.14 & -0.16 & -0.05 & -0.09 \\ 
			BGFE-he & -0.15 & -0.29 & -0.13 & -0.14 & -0.16 & -0.05 & -0.09 \\ 
			BGFE-ho & 1.16 & 1.64 & 1.34 & 1.11 & 1.14 & 0.78 & 0.78 \\ 
			AR-he-PC & 0.02 & 0.02 & 0.01 & 0.03 & 0.02 & 0.01 & 0.04 \\ 
			Pooled & 1.21 & 1.64 & 1.39 & 1.20 & 1.20 & 0.84 & 0.84 \\ 
			\midrule 
			& \multicolumn{7}{l}{Category 6: Recreation} \\
			\cmidrule(lr){2-8} 
			BGFE-he-cstr & -0.04 & -0.04 & -0.05 & -0.02 & -0.02 & -0.05 & -0.05 \\ 
			BGFE-he & -0.03 & -0.03 & -0.05 & -0.01 & -0.01 & -0.04 & -0.07 \\ 
			BGFE-ho & 0.55 & 0.91 & 0.79 & 0.58 & 0.49 & 0.25 & 0.12 \\ 
			AR-he-PC & 0.02 & 0.03 & 0.03 & 0.02 & 0.01 & 0.01 & 0.06 \\ 
			Pooled & 0.61 & 0.91 & 0.84 & 0.70 & 0.58 & 0.29 & 0.13 \\ 
			\midrule 
			& \multicolumn{7}{l}{Category 7: Transportation} \\
			\cmidrule(lr){2-8} 
			BGFE-he-cstr & -0.04 & -0.10 & -0.03 & -0.05 & -0.07 & 0.00 & 0.00 \\ 
			BGFE-he & -0.03 & -0.10 & 0.00 & -0.04 & -0.06 & 0.00 & 0.00 \\ 
			BGFE-ho & 1.59 & 0.98 & 1.56 & 2.83 & 0.79 & 1.51 & 2.11 \\ 
			AR-he-PC & 0.01 & 0.03 & 0.03 & -0.02 & -0.01 & 0.00 & 0.07 \\ 
			Pooled & 1.42 & 0.98 & 1.41 & 2.38 & 0.75 & 1.24 & 1.98 \\ 
			\midrule 
			& \multicolumn{7}{l}{Category 8: Other Goods and Services} \\
			\cmidrule(lr){2-8} 
			BGFE-he-cstr & -0.03 & 0.22 & -0.09 & -0.07 & -0.14 & -0.06 & -0.05 \\ 
			BGFE-he & -0.03 & 0.25 & -0.09 & -0.06 & -0.13 & -0.06 & -0.06 \\ 
			BGFE-ho & 0.64 & 0.59 & 0.66 & 0.69 & 0.83 & 0.53 & 0.49 \\ 
			AR-he-PC & 0.01 & 0.02 & 0.00 & 0.03 & 0.00 & 0.01 & 0.04 \\ 
			Pooled & 0.69 & 0.61 & 0.71 & 0.77 & 0.89 & 0.57 & 0.53 \\ 
			\bottomrule 
		\end{tabular}
	\end{center}
	{\footnotesize \textit{Notes: Benchmark model = AR-he. }\par}
\end{table}

\subsubsection{Network Visualization of Posterior Similarity Matrix}
In our empirical work, we estimate the posterior similarity matrix (PSM) among 154 series for the last sample (August 2022). Presenting and examining $154 \times 154 = 23,716$ estimates pairwise posterior probabilities in PSM would be thoroughly uninformative. Hence we characterize the estimated PSM graphically as network graphs, which contain node names, node color, and link size (one per link since the network is undirected).
\begin{itemize} \itemsep0em 
%
	
	\item \textit{Node name} shows the item names of CPI sub-indices.
	
	\item \textit{Node color} indicates the group structure used in the prior, e.g, expenditure category.
	
	\item \textit{Link size} represents the pairwise probabilities in the PSM.
\end{itemize}

We use the \textit{qgraph} package in R for network visualization. Node locations are determined by a modified version of a force-embedded algorithm proposed by \citet{fruchterman1991}. Figure \ref{fig:app_cpi_nework_plot} show the full-sample CPI sub-index network graphs.

\begin{landscape}
	\begin{figure}[h]
		\centering
		\caption{Individual CPI Sub-Index Network Graph based on Posterior Similarity Matrix, August 2022}
		\label{fig:app_cpi_nework_plot}
		\includepdf[landscape=true, scale=0.83]{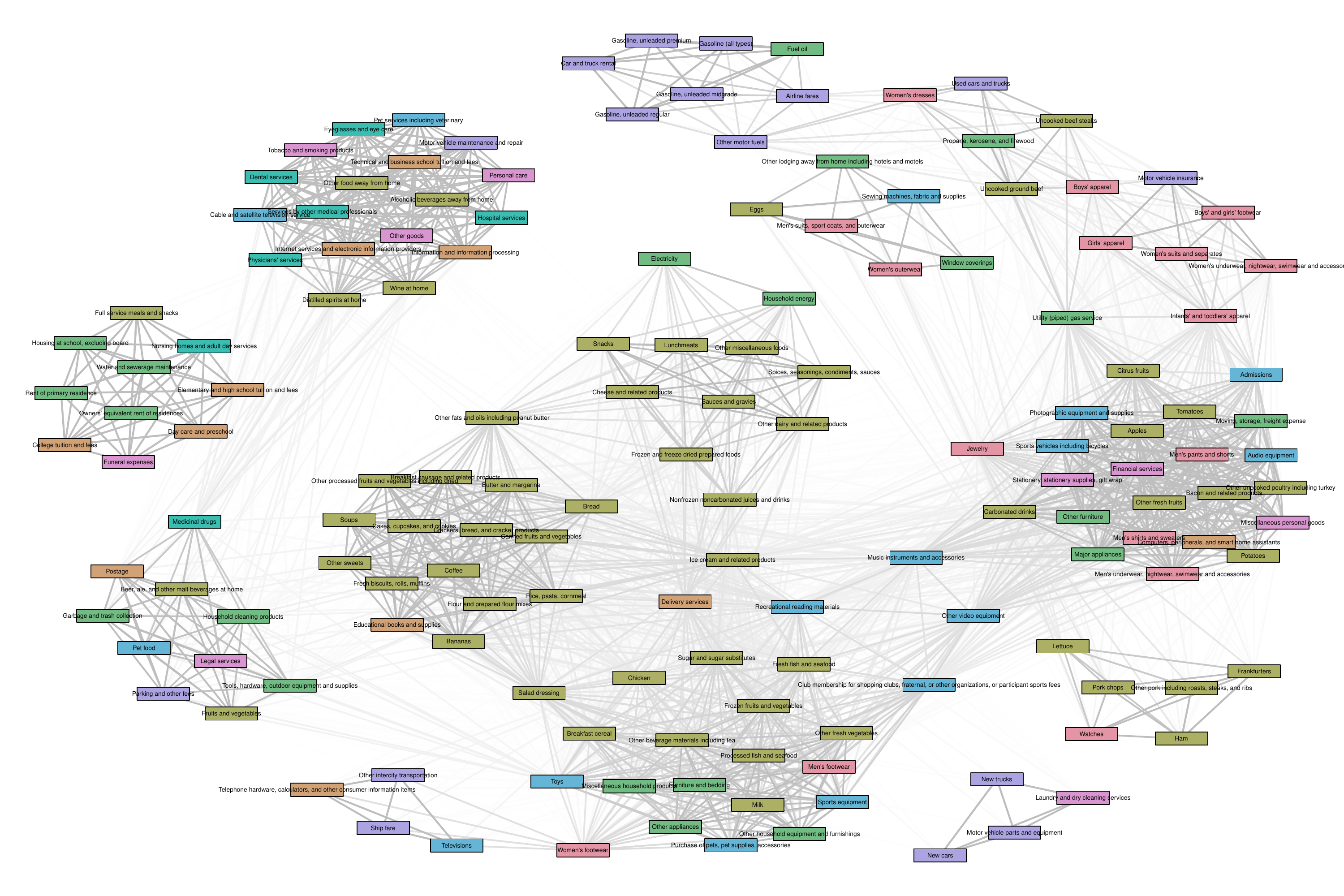}
	\end{figure}
\end{landscape}

\subsubsection{Heteroskedasticity vs. Homoskedasticity} \label{subsec:hetero_vs_homo}

We conclude by examining how grouped heteroskedasticity impacts forecast accuracy and why this is important. For illustrative purposes, we focus on the two BGFE estimators, BGFE-he and BGFE-ho, that do not involve pairwise constraints.


A distinguishing characteristic between BGFE-he and BGFE-ho is the estimated number of groups. Figure \ref{fig:app_cpi_hetsk_vs_homsk_K} depicts the number of groups over samples. BGFE-he estimator forms 9 groups for the beginning of the sample, and increase it during the Great Recession and the Pandemic. However, the estimated number of groups for BGFE-ho is rather low in the 1990s, and progressively increases to around seven by the end of the sample. It is noticeable that when heteroskedasticity is allowed, there are more groups than when it is not. This is intuitive. Two groups can be expected to have comparable estimates of grouped fixed-effects and slope coefficients, but vastly different error variances. As a result, allowing for heteroskedasticity would result in a more refined group structure and increase the overall number of groups. 

\begin{figure}[h]
	\centering
	\caption{Number of Groups}
	\label{fig:app_cpi_hetsk_vs_homsk_K}
	\includegraphics[scale= 0.4]{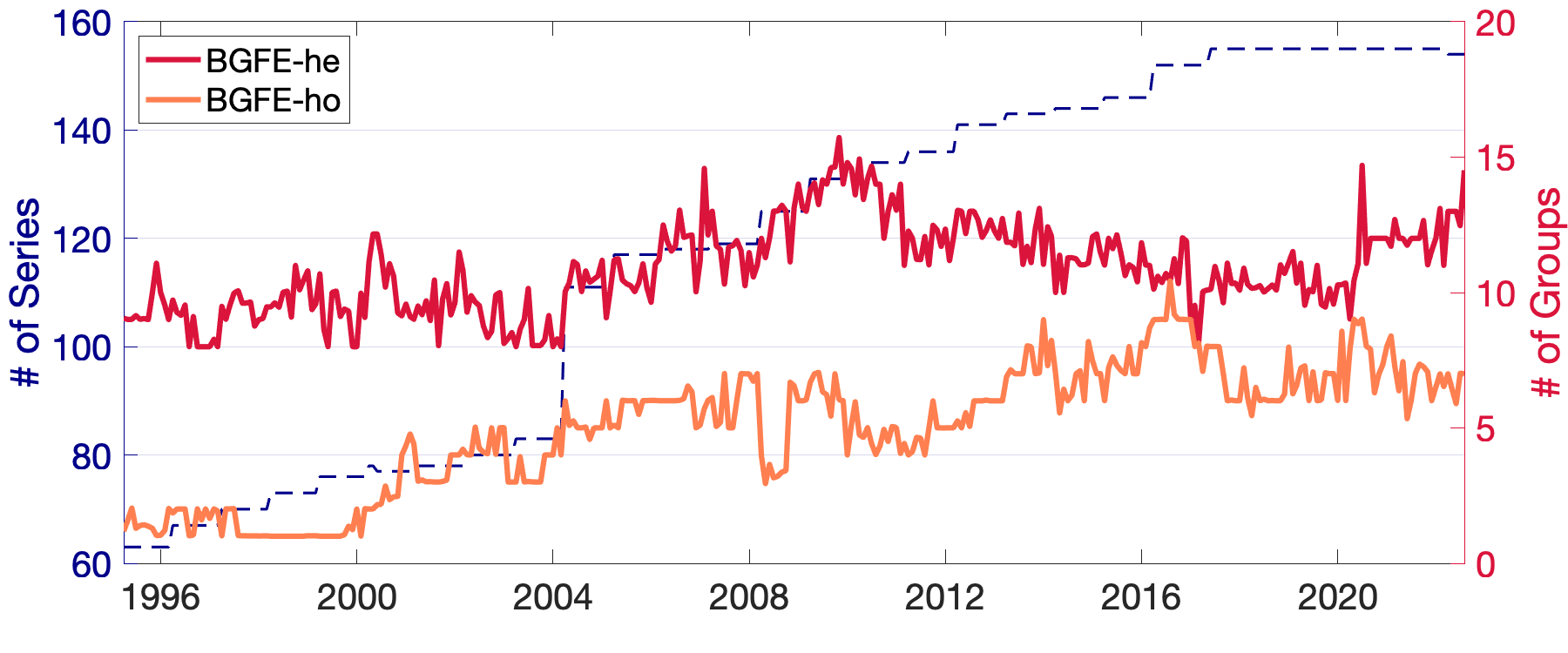}
\end{figure}

As seen in Figure \ref{fig:app_cpi_RMSE} and \ref{fig:app_cpi_LPS}, the grouped heteroskedasticity doesn't improve the point forecast but the density forecast. Figure \ref{fig:app_cpi_hetsk_vs_homsk_fcst} depicts a clear perspective of it and demonstrates the performance of point and density forecasts through time. In panel (a), we observe that the ratio of RMSE is generally around one over the whole sample, meaning that heteroskedasticity cannot improve the point forecast in general.  In panel (b), the difference in LPS is consistently negative. This demonstrates that the improved density prediction performance is not a fluke and that enabling heteroskedasticity improves the density forecast regardless of sample. This is actually in line with the simulation results presented in Table \ref{tab:MC_general_DGP}.

\begin{figure}[h]
	\caption{Results of BGFE-he vs. BGFE-ho}
	\label{fig:app_cpi_hetsk_vs_homsk_fcst}
	\centering
	\begin{subfigure}[b]{\textwidth}
		\centering
		\caption{$\text{RMSE}_{he}$ / $\text{RMSE}_{ho}$}
		\includegraphics[scale= 0.4]{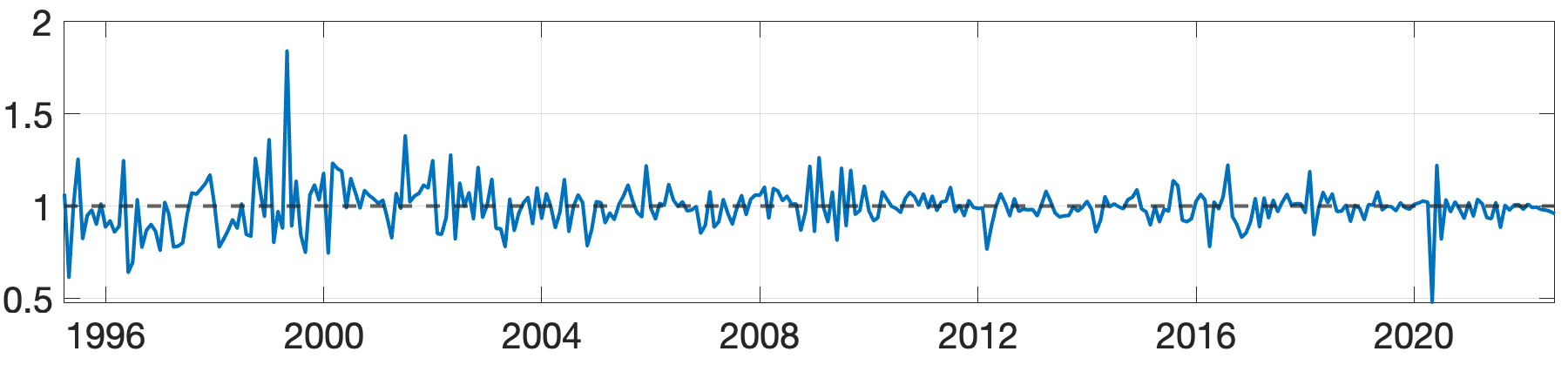}
	\end{subfigure}
	~
	\begin{subfigure}[b]{\textwidth}
		\centering
		\caption{$\text{LPS}_{he}$ - $\text{LPS}_{ho}$}
		\includegraphics[scale= 0.4]{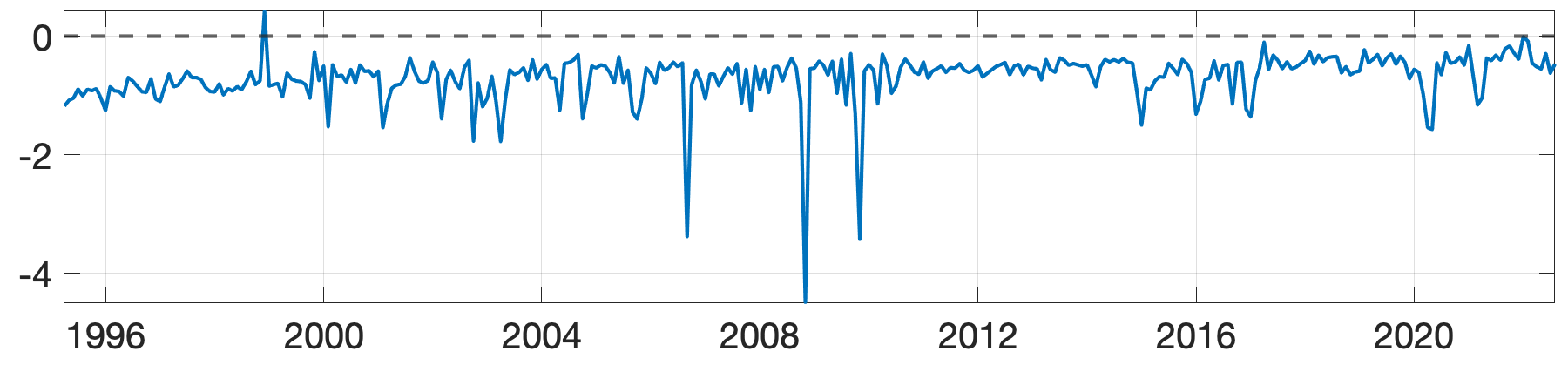}
	\end{subfigure}
\end{figure}

Density forecasts vary substantially across categories. We pick three typical sub-categories and plot their posterior predictive densities of August 2022 in Figure \ref{fig:app_cpi_hetsk_vs_homsk_example}. The vertical dashed black lines represent the actual values. Several insights emerge while comparing these three subcategories. First, BGFE-he and BGFE-ho provide comparable posterior means for all three subcategories - the posterior predictive densities concentrate around the similar price levels. This explains why BGFE-he and BGFE-ho have comparable results in point forecasting. Second, BGFE-he reveals different predictive variance. As the rolling sample size is set to 4 years, all observations throughout the Pandemic are included, and it is anticipated that the price levels of elementary and high school (basic education) tuition and fees, major appliances, and airline prices would respond differently to the shock. Intuitively, education tuition and fees should not fluctuate as much as other prices, while airline fares have been strongly influenced by the fluctuating oil prices since the beginning of the Pandemic.  Consequently, accounting for heteroskedasticity successfully captures this characteristic, such that college tuition and fees have a smaller predictive variance than that of BGFE-ho, but airline fares have a greater predictive variance.

Combing these two observations together reveals why BGFE-he has a better density forecast: the capacity to optimally cluster units according to the error variance and accommodate heteroskedasticity. For elementary and high tuition and fees, providing that both BGFE-he and  BGFE-ho yields accurate posterior mean, BGFE-he yields much lower predictive variance, decreasing the LPS dramatically. Both BGFE-he and BGFE-ho underestimate the inflation rate for airline fares, but BGFE-he subtly creates a greater predicted variance to account for the wild probable shift in this sub-category and hence reduces the LPS significantly. Major appliances is an example to show that BGFE-he and BGFE-ho generate comparable density forecasts for some sub-categories.

\begin{figure}[htp]
	\caption{Predictive Posteriors for Selected Series: BGFE-he vs. BGFE-ho}
	\label{fig:app_cpi_hetsk_vs_homsk_example}
	\centering
	\begin{subfigure}[b]{0.32\textwidth}
		\centering
		\caption{Basic Edu. Tuition and Fees}
		\includegraphics[scale= 0.35]{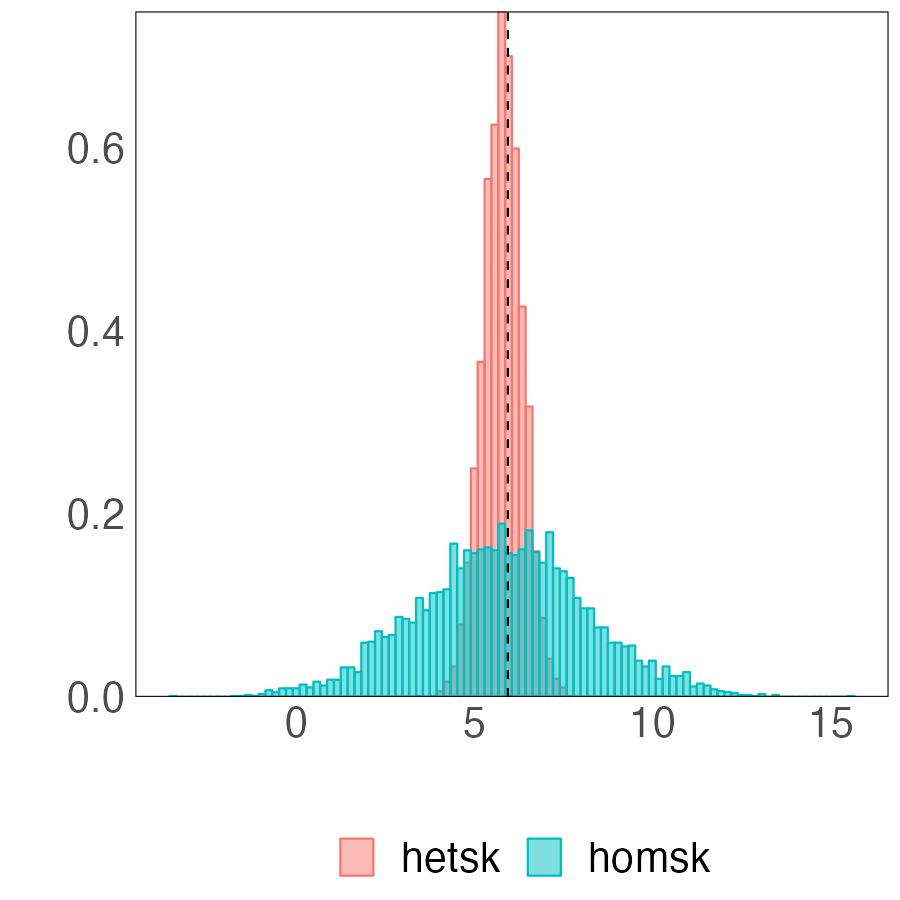}
	\end{subfigure}
	~
	\begin{subfigure}[b]{0.32\textwidth}
		\centering
		\caption{Major Appliances}
		\includegraphics[scale= 0.35]{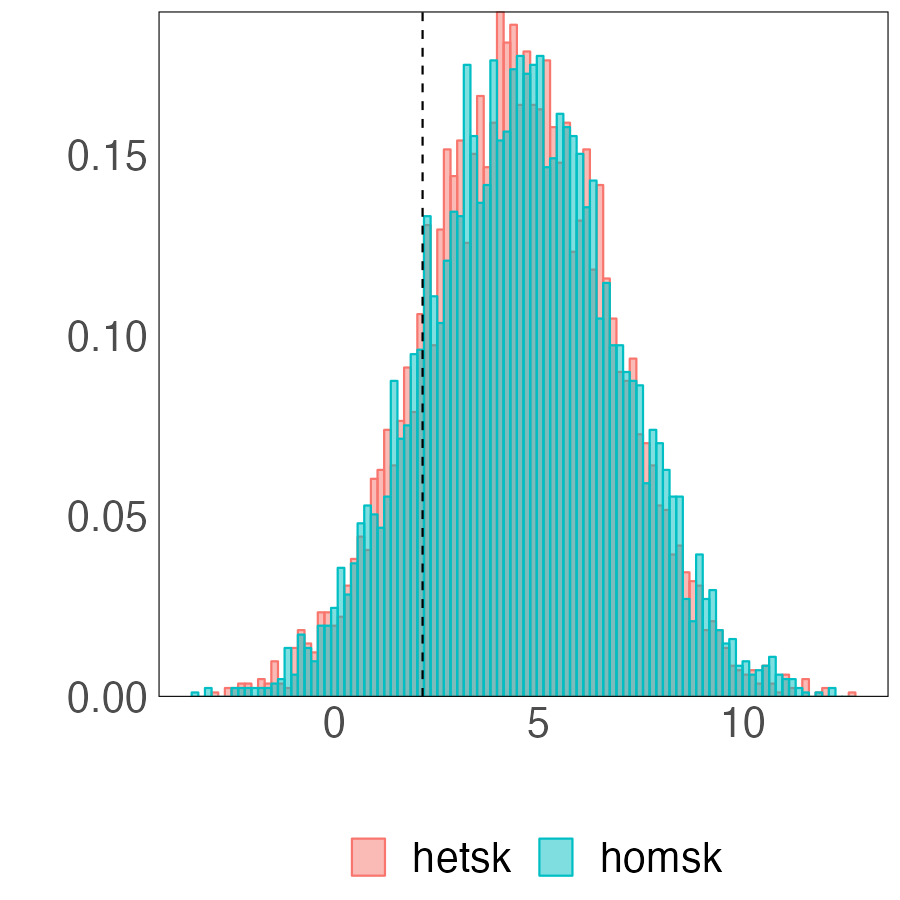}
	\end{subfigure}
	~
	\begin{subfigure}[b]{0.32\textwidth}
		\centering
		\caption{Airline Fares}
		\includegraphics[scale= 0.35]{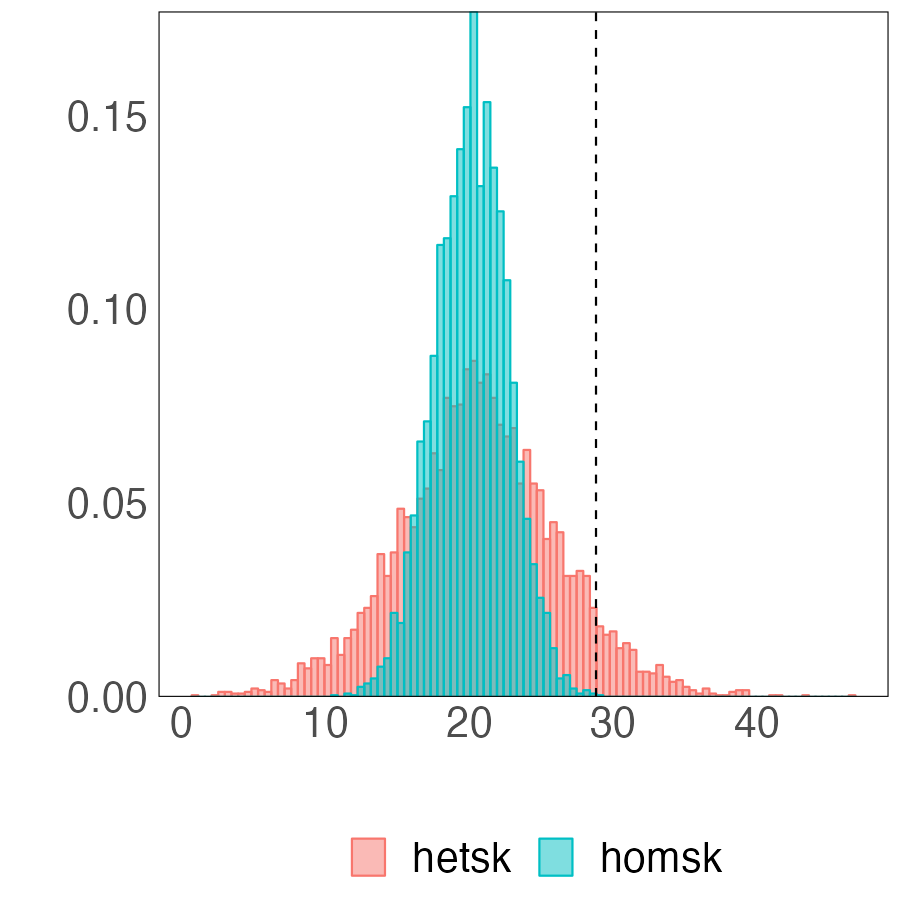}
	\end{subfigure}
\end{figure}

\newpage
\subsection{Income and Democracy} \label{appendix:res_demo}

%
%
%

\subsubsection{Results of Specification 1}

We start our analysis with the specification 1 in (\ref{eq:emp_app_2_sp1}). Table \ref{tab:emp_app_dem_sp1_ngroup} demonstrates the posterior probability of the number of groups utilizing various estimators. Notably, the BGFE-ho in this specification is identical to the primary model in BM, allowing us to evaluate the optimal number of groups. BGFE-ho creates 8 groups in all posterior draws, which is consistent with BM's conclusion of using BIC: the upper bound of the true number of groups is 10. Despite the fact that BM is unable to validate the ideal number of groups for their study, our BGFE-ho estimator provides an accurate estimate of it. Intriguingly, accounting for heteroskedasticity drastically reduces the number of groups, with BGFE-he identifying three groups in 92.9\% of posterior draws. Adding pairwise constraints based on geographic information increase the number groups. Two-third of posterior draws from BGFE-he-cstr generate 5 group.

\begin{table}[ht]
	\begin{center}
		\caption{Probability for the number of groups}
		\label{tab:emp_app_dem_sp1_ngroup}
		\begin{tabular}{l | c c c}
			\toprule
			& BGFE-he-cstr & BGFE-he & BGFE-ho \\
			\midrule
			$Pr (K < 3)$ & 0.000 & 0.000 & 0.000 \\
			$Pr (K = 3)$ & 0.000 & {\bf 0.929} & 0.000 \\
			$Pr (K = 4)$ & 0.344 & 0.071 & 0.000 \\
			$Pr (K = 5)$ & {\bf 0.656} & 0.000 & 0.000 \\
			$Pr (K > 5)$ & 0.000 & 0.000 & {\bf 1.000} \\
			\bottomrule 
		\end{tabular}
		
	\end{center}
\end{table}

The marginal data density (MDD) of each estimators in Table \ref{tab:emp_app_dem_sp1_mdd} provides some insight on different models. Even while BGFE-ho produces eight groups and has a tendency to overfit, its MDD is the lowest of the three estimators. BGFE-he with fewer groups is superior to BGFE-ho with higher MDD. BGFE-he-cstr has the highest MDD because the pairwise constraints give direction on grouping and identify the ideal group structure, which BGFE-he cannot uncover without our prior knowledge.

\begin{table}[h]
	\begin{center}
		\caption{Marginal Data Density}
		\label{tab:emp_app_dem_sp1_mdd}
		\begin{tabular}{ x{2.5cm} x{2.5cm} x{2.5cm}}
			\toprule
			BGFE-he-cstr & BGFE-he & BGFE-ho \\
			\midrule
			425.690 & 381.218 & 368.918 \\
			\bottomrule 
		\end{tabular}
		
	\end{center}
\end{table}

We focus on the BGFE-he-cstr estimator and use the approach outlined in Section \ref{subsec:det_partition} to identify the unique group partitioning $\widehat{G}$. The left panel of Figure \ref{fig:app_dem_result_sp1} presents the world map colored by $\widehat{G}$, while the right panel present the group-specific averages of democracy index over time. The estimated group structure $\widehat{G}$ features four distinct groups, which is coincident to the choice of BM. As described in BM, we refer to groups 1-4 as the ``high-democracy", ``low-democracy", ``early transition",  and ``late transition" group, respectively. With the exception of the ``early transition" group that is slight at odd with the counterpart in BM, the group-specific averages of the democracy index for all other groups are relatively similar to those in BM.  Notice that BM manually sets the number of groups to four, but we discover that four is the optimal number. Consequently, by employing model specification 1 and accounting for heteroskedasticity, we find the support for BM's main results.

\begin{figure}[h]
	\caption{Point Estimation of Group Partitioning and Average Democracy}
	\label{fig:app_dem_result_sp1}
	\centering
	\begin{subfigure}[b]{0.45\textwidth}
		\centering
		\includegraphics[scale= 0.4]{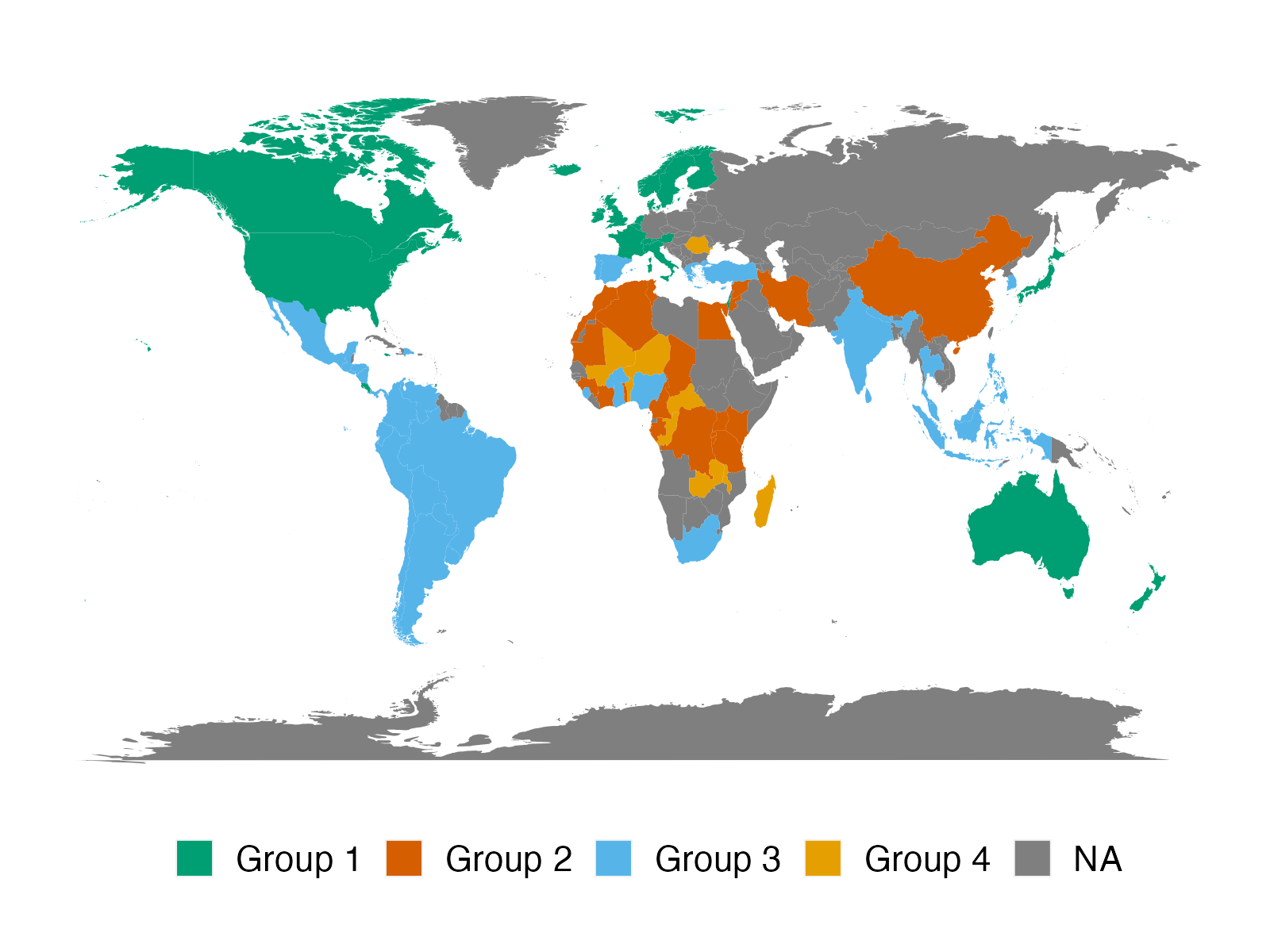}
	\end{subfigure}
	\begin{subfigure}[b]{0.45\textwidth}
		\centering
		\includegraphics[scale= 0.35]{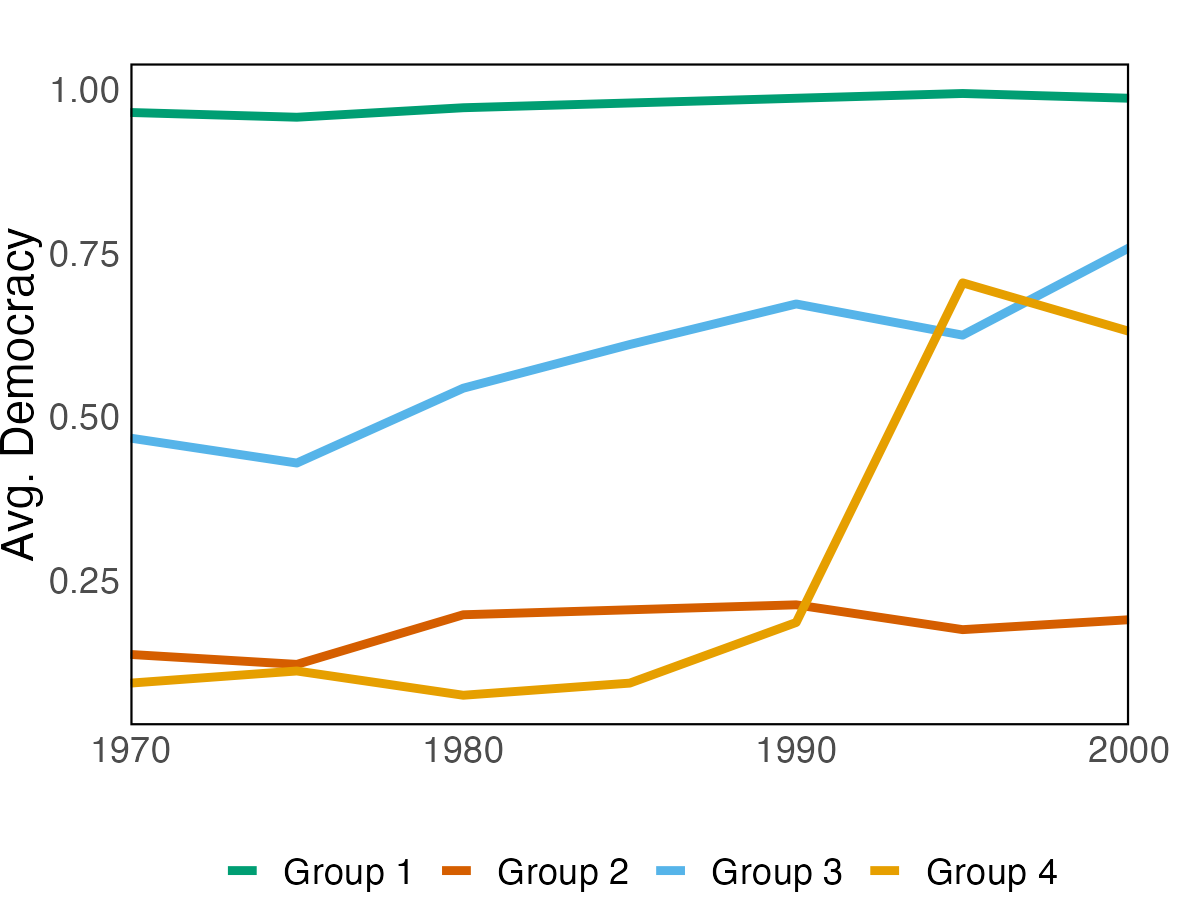}
	\end{subfigure}
\end{figure}

Table \ref{fig:app_dem_est_sp1} shows the posterior mean and 90\% credible set for each coefficient, with $G$ fixing at the point estimate $\widehat{G}$. Comparing to the pooled OLS, $\hat{\rho}$ and $\hat{\beta}$ once we incorporate the group-specific time patterns. The results are essentially consistent with the conclusion in BM: there is modest persistence and a positive effect of income on democracy, but the cumulative income effect $\beta / (1-\rho) = 0.08$ is quantitatively small.

\begin{table}[h]
	\begin{center}
		\caption{Coefficient estimates across groups}
		\label{fig:app_dem_est_sp1}
		\resizebox{0.6\textwidth}{!}{%
			
			\begin{tabular}{l  p{1.1cm}  p{3cm}  p{1.1cm}  p{3cm} }
				\toprule
				& \multicolumn{2}{c}{Lagged democracy ($\rho$)}  & \multicolumn{2}{c}{Lagged Income ($\beta$)} \\
				\cmidrule(lr){2-3} \cmidrule(lr){4-5}   \noalign{\smallskip}
				& Coef. & Cred. Set & Coef. & Cred. Set \\ 
				\midrule
				BGFE-he-cstr & 0.499 & [0.438, 0.558] & 0.040 & [0.027, 0.053]  \\ 
				\midrule
				Pooled OLS & 0.665 & [0.616, 0.718] & 0.082 & [0.065, 0.100]  \\ 
				\bottomrule 
			\end{tabular}
		}
	\end{center}
\end{table}

%
%

\subsubsection{Network Visualization of Posterior Similarity Matrix}

Figure \ref{fig:app_dem_sp1_nework_plot} and \ref{fig:app_dem_sp2_nework_plot} show the full-sample country network graphs, for specification 1 and 2 respectively. Node color reflects the geographic information used to construct the group structure in the prior.

\begin{landscape}
	\begin{figure}[h]
		\centering
		\caption{Individual Country Network Graph based on Posterior Similarity Matrix, Specification 1}
		\label{fig:app_dem_sp1_nework_plot}
		\includepdf[landscape=true, scale=0.83]{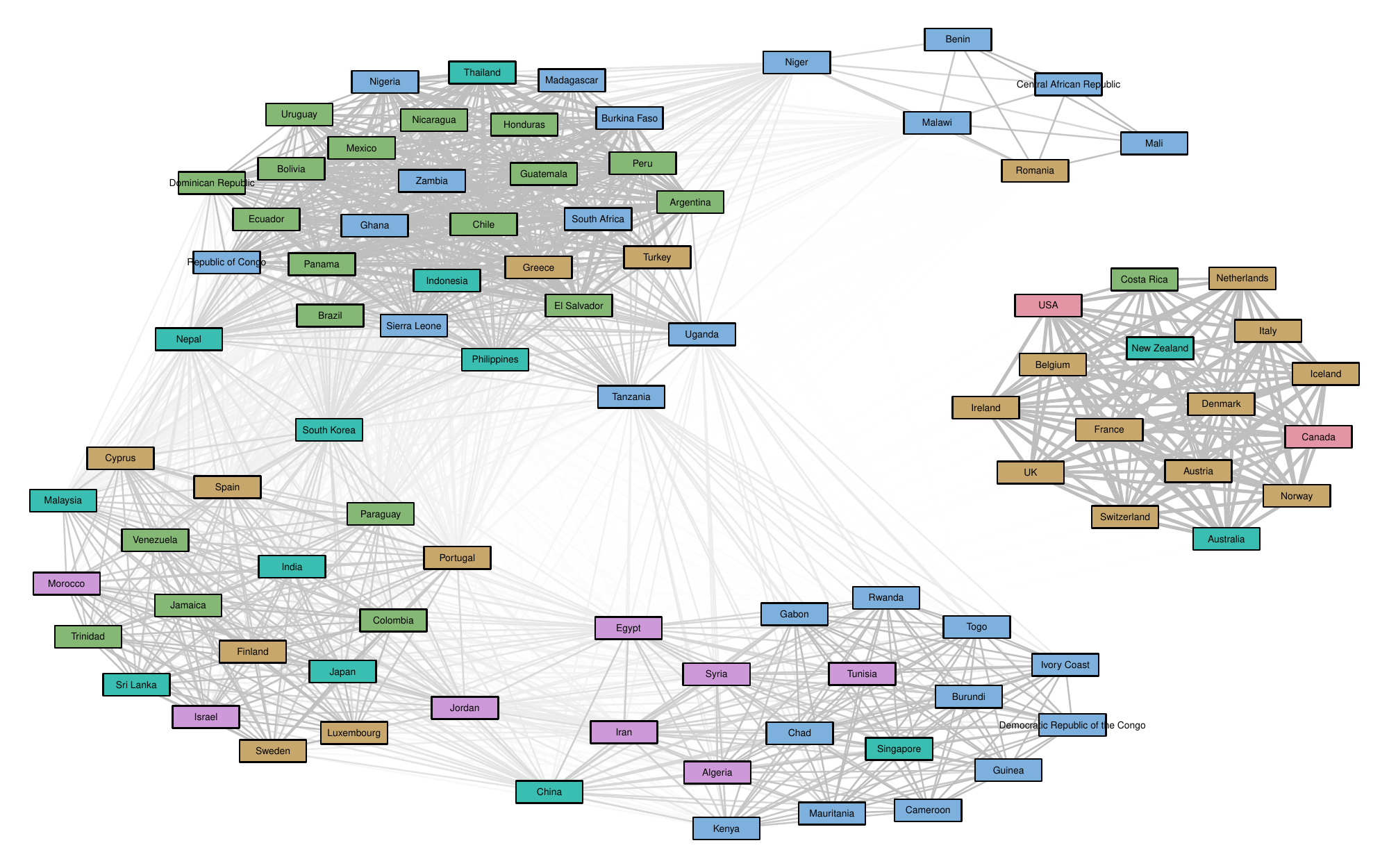}
	\end{figure}
\end{landscape}

\begin{landscape}
	\begin{figure}[h]
		\centering
		\caption{Individual Country Network Graph based on Posterior Similarity Matrix, Specification 2}
		\label{fig:app_dem_sp2_nework_plot}
		\includepdf[landscape=true, scale=0.83]{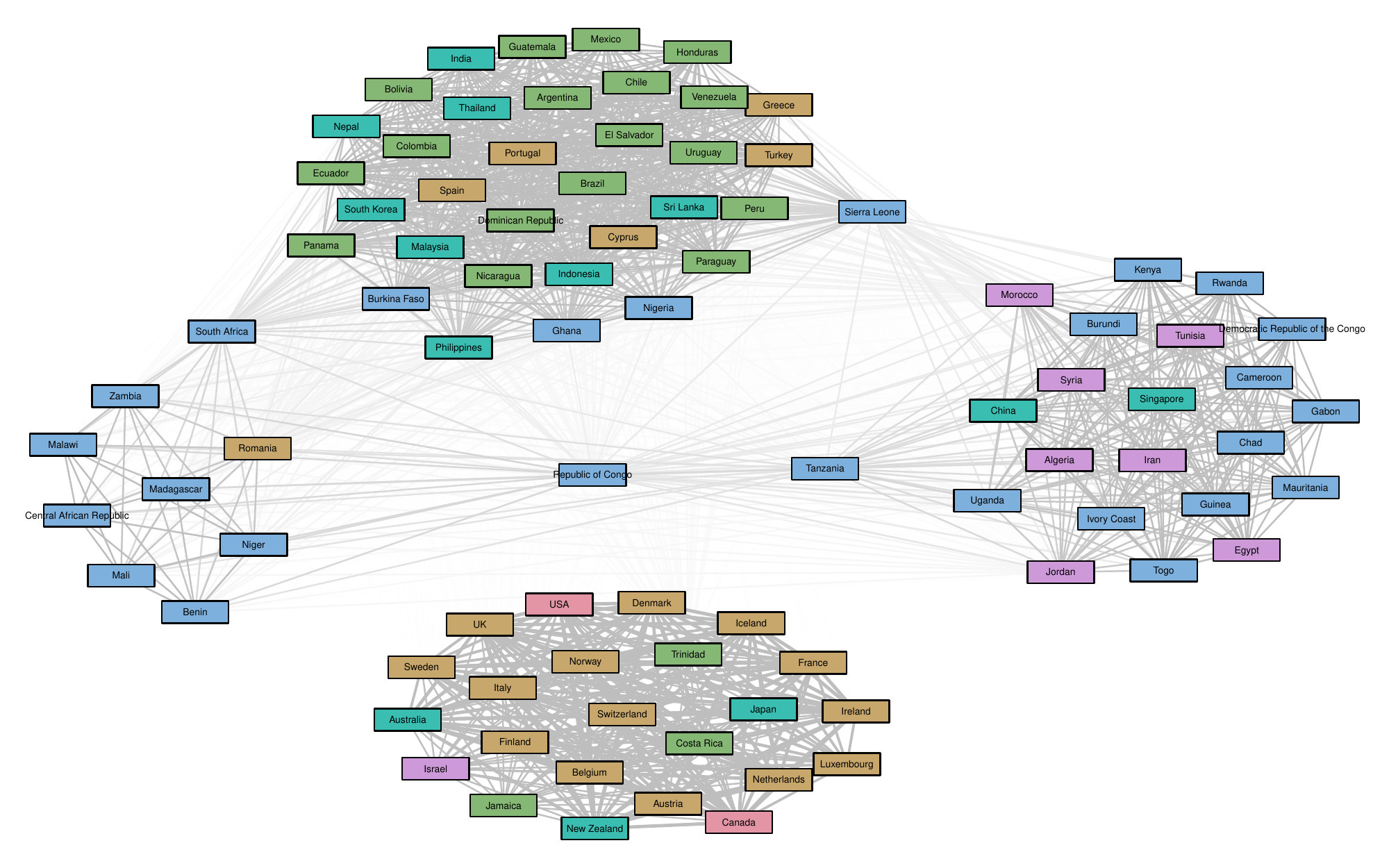}
	\end{figure}
\end{landscape}

\end{document}